\documentclass[reqno]{amsart}
\usepackage{amsmath,amsthm,amsfonts,amssymb,srcltx}
\usepackage{graphicx} 
\usepackage{epstopdf}
\usepackage{enumerate}
\usepackage{url}
\usepackage{bm}
\usepackage{epsfig}
\usepackage{color}
\usepackage[labelfont={rm},labelformat=simple]{subcaption}
\usepackage{hyperref}
\usepackage{pdflscape}
\usepackage[all]{xy}

\allowdisplaybreaks

\theoremstyle{plain}
 \newtheorem{thm}{Theorem}[section]
  \newtheorem{prop}[thm]{Proposition}
    \newtheorem{lemm}[thm]{Lemma}
  \newtheorem{conj}[thm]{Conjecture}

 \theoremstyle{definition}    
  \newtheorem{dfn}[thm]{Definition}
  \newtheorem{ass}[thm]{Assumption}
  \newtheorem{prob}[thm]{Problem}
  \newtheorem{exa}[thm]{Example}

\theoremstyle{remark}
  \newtheorem{rem}[thm]{Remark}

\numberwithin{equation}{section}\numberwithin{figure}{section}

\def\rchi{{\hbox{\raise1.5pt\hbox{$\chi$}}}}

\def\isom{\cong}

\newcommand{\bea}{\begin{eqnarray}}
\newcommand{\eea}{\end{eqnarray}}
\newcommand{\be}{\begin{equation}}
\newcommand{\ee}{\end{equation}}

\newcommand{\Res}{\mathop{\rm Res}}

\title[Topological recursion and uncoupled BPS structures I]
{Topological recursion and uncoupled BPS structures I: 
BPS~spectrum and free energies}

\author{Kohei Iwaki}%
\address{Graduate School of Mathematical Sciences, 
The University of Tokyo, 
3-8-1 Komaba, Meguro-ku, Tokyo, 153-8914, Japan}
\email{iwaki@ms.u-tokyo.ac.jp}

\author{Omar Kidwai}%
\address{Graduate School of Mathematical Sciences, 
The University of Tokyo, 
3-8-1 Komaba, Meguro-ku, Tokyo, 153-8914, Japan}
\email{kidwai@ms.u-tokyo.ac.jp}

\date{}

\begin{document}

\large
\setcounter{section}{0}
\maketitle
\begin{abstract}  For the hypergeometric spectral curve and its confluent degenerations (spectral curves of ``hypergeometric type"), we obtain a simple formula expressing the topological recursion free energies as a sum over BPS states (degenerate spectral networks) for a corresponding quadratic differential $\varphi$. In doing so, we generalize Gaiotto-Moore-Neitzke's construction of BPS structures to include the case where $\varphi$ has simple poles or supports a degenerate ring domain. For the nine spectral curves of hypergeometric type, we provide a complete description of the corresponding BPS structures over a generic locus in the relevant parameter space; in particular, we prove the existence of saddles trajectories at the expected parameter values. We determine the corresponding BPS cycles, central charges, and BPS invariants, and verify our formula in each case. We conjecture that a similar relation should hold more generally whenever the corresponding BPS structure is uncoupled, and provide experimental evidence in two simple higher rank examples.
\end{abstract}

\tableofcontents

\section{Introduction}

The purpose of this paper and its sequel is to 
describe a relationship between two enumerative theories of \emph{a priori} different origins --- the theory of \emph{BPS structures} and the formalism of \emph{topological recursion} --- in the special case where the former is \emph{uncoupled}.  Throughout, our testing ground will be concrete examples coming from spectral curves related to the Gauss hypergeometric equation and its confluent degenerations, which we refer to as \emph{spectral curves of hypergeometric type}, depicted in Figure \ref{fig:confdiag} below. In short, we will find that in all examples we consider:
\begin{enumerate}[(1)]
        \item The topological recursion free energies $F_g$ can be written as a sum of powers of central charges of BPS cycles $\gamma_{\rm BPS}$, weighted by their BPS indices $\Omega(\gamma_{\rm BPS})$, times a simple explicit expression involving Bernoulli numbers,
        \item The Borel-resummed topological recursion partition function coincides with Bridgeland's $\tau$-function associated to the solution of a BPS Riemann-Hilbert problem, up to a simple factor.
\end{enumerate}

The present paper deals with the former of these two results. To make the statement meaningful, let us begin by recalling the notions of BPS structures and topological recursion.

\subsection{BPS structures}
The notion of BPS structure was introduced by Bridgeland \cite{Bri19} by axiomatizing the structure of Donaldson-Thomas (DT) invariants of a Calabi-Yau 3 triangulated category with a stability condition.  
This is a special case of a stability structure in the sense of Kontsevich-Soibelman \cite{KS08}.  
A BPS structure $(\Gamma,Z,\Omega)$ consists of
\begin{itemize}
\item 
a lattice $\Gamma$ equipped with an anti-symmetric pairing $\langle \cdot , \cdot \rangle$, 
\item 
a group homomorphism $Z : \Gamma \to {\mathbb C}$ called the \emph{central charge}, and
\item 
a map $\Omega : \Gamma \to {\mathbb Q}$ or ${\mathbb Z}$, called the \emph{BPS invariants} which has a close relationship to DT invariants, 
\end{itemize}
satisfying certain conditions. We may also discuss a \emph{variation} of BPS structures by considering a family of these data parametrized by a complex manifold, which also axiomatizes the behavior of DT invariants/BPS invariants under a variation of the stability condition (i.e., the Kontsevich-Soibelman wall-crossing formula).

Given a BPS structure, Bridgeland also associated a Riemann-Hilbert type problem on ${\mathbb C}^\ast$, which we call the \emph{BPS Riemann-Hilbert problem}. 
The Riemann-Hilbert jump contour here is a collection of straight half-lines, called \emph{BPS rays}, of the form $\ell = {\mathbb R}_{> 0} e^{i \vartheta} \subset {\mathbb C}^\ast$ for some $\vartheta \in {\mathbb R}$.  
In sufficiently nice situations, the jump factor across a BPS ray $\ell$, which is called the \emph{BPS automorphism}, is given by a cluster-like birational transformation of a (twisted) algebraic torus, and encodes the BPS invariants $\Omega(\gamma)$ of those $\gamma \in \Gamma$ satisfying $Z(\gamma) \in \ell$. 
The BPS Riemann-Hilbert problem is closely related to the one studied by Gaiotto-Moore-Neitzke in \cite{GMN08, GMN09}, where the BPS spectra for a class of four-dimensional ${\mathcal N} = 2$ supersymmetric gauge theories (class ${\mathcal S}$ theories) were investigated. We note also that similar Riemann-Hilbert problems were studied in \cite{BTL, FFS}.

\begin{figure}[t]
$$
\xymatrix@!C=35pt@R7pt{
&&
&& \underset{\rm (Web)}
{\text{Weber}} \ar@{~>}[rrd]
&&
\\
&& \underset{\rm (Kum)}
{\text{Kummer}} \ar@{->}[rru] \ar@{->}[rrd] \ar@{~>}[rrddd]
&&
&& \underset{\rm (Ai)}
{\text{Airy}}
\\
&&
&& \underset{\rm (Whi)}
{\text{Whittaker}} \ar@{~>}[rru] \ar@{~>}[rrddd]
&&
&&
\\
\underset{\rm (HG)}
{\text{Gauss}} \ar@{->}[rruu] \ar@{~>}[rrdd]
&&
&&
&&
&&
\\
&&
&& \underset{\rm (Bes)}
{\text{Bessel}} 
\ar@{->}[rruuu]|(0.32)\hole \ar@{~>}[rrd]
&&
&&
\\
&& \underset{\rm (dHG)}
{\text{Degenerate Gauss}} 
\ar@{~>}[rru] \ar@{~>}[rrd] \ar@{~>}[rruuu]|(0.67)\hole
&&
&& \underset{\rm (dBes)}
{\text{Degenerate Bessel}}
\\
&&
&& \underset{\rm (Leg)}
{\text{Legendre}}\ar@{~>}[rru]
&&
\\
}
$$
\caption{Confluence diagram of the Gauss hypergeometric equation/spectral curves.
A straight line (resp., a wiggly line) in the figure
denotes the confluence of singular points
(resp., the coalescence of a brach point and a singular point).
The number of mass parameters is reduced by 1 at each confluence.
}
\label{fig:confdiag}
\end{figure}
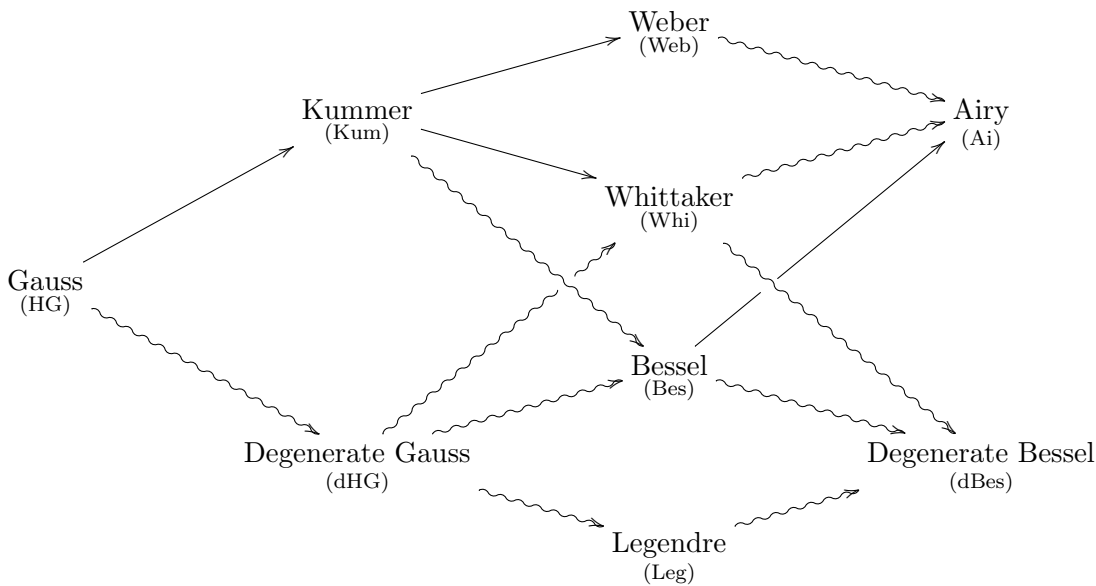

\begin{table}[h]
\begin{center}
\begin{tabular}{ccc}\hline
$\underset{\rm (label)}{\rm Equation}$ & $Q_{\bullet}(x)$  & Assumption
\\\hline\hline
\parbox[c][3.25em][c]{0em}{}
${\rm Gauss \; (HG)}$
&
\begin{minipage}{.35\textwidth}
\begin{center}
$\dfrac{{m_\infty}^2 x^2 
- ({m_\infty}^2 + {m_0}^2 - {m_1}^2)x 
+ {m_0}^2}{x^2 (x-1)^2}$
\end{center}
\end{minipage}
& ~~~\quad
\begin{minipage}{.35\textwidth}
\begin{center}
$m_0, m_1, m_\infty \neq 0$,\\
$m_0 \pm m_1 \pm m_\infty \ne 0$.
\end{center}
\end{minipage}
\\\hline
\parbox[c][2.75em][c]{0em}{}
${\rm Degenerate~Gauss \; (dHG)}$
&
\begin{minipage}{.35\textwidth}
\begin{center}
$\dfrac{{m_\infty}^2 x + {m_1}^2 - {m_\infty}^2}{x(x-1)^2}$
\end{center}
\end{minipage}
&
\begin{minipage}{.35\textwidth}
\begin{center}
$m_1, m_\infty \neq 0$,\\
$m_1 \pm m_\infty \neq 0$.
\end{center}
\end{minipage}
\\\hline
\parbox[c][2.75em][c]{0em}{}
${\rm Kummer \;(Kum)}$ 
&
\begin{minipage}{.3\textwidth}
\begin{center}
$\dfrac{x^2 + 4 m_\infty x + 4 {m_0}^2}{4x^2}$
\end{center}
\end{minipage}
&
\begin{minipage}{.3\textwidth}
\begin{center}
$m_0\neq 0$,\\
$m_0 \pm m_\infty \neq 0$.
\end{center}
\end{minipage}
\\\hline
\parbox[c][2.75em][c]{0em}{}
${\rm Legendre \; (Leg)}$ 
& $\dfrac{m_\infty^2}{x^2-1}$
&
$m_\infty \neq 0$.
\\\hline
\parbox[c][2.75em][c]{0em}{}
${\rm Bessel \; (Bes)}$ 
& $\dfrac{x + 4m_0^2}{4x^2}$
&
$m_0 \neq 0$.
\\\hline
\parbox[c][2.75em][c]{0em}{}
${\rm Whittaker \; (Whi)}$ 
& $\dfrac{x - 4m_\infty}{4 x}$
& $m_\infty \neq 0$.
\\\hline
\parbox[c][2.75em][c]{0em}{}
${\rm Weber \; (Web)}$ 
& $\dfrac{1}{4} x^2 - m_\infty$
& $m_\infty \neq 0$.
\\\hline
\parbox[c][2.75em][c]{0em}{}
${\rm Degenerate~Bessel \; (dBes)}$ 
& $\dfrac{1}{x}$
& --
\\\hline
\parbox[c][2.75em][c]{0em}{}
${\rm Airy \; (Ai)}$ 
& $x$
& --
\\\hline
\end{tabular}
\end{center}
\caption{Quadratic differentials corresponding to hypergeometric type spectral curves, where 
$\bullet \in \{ {\rm HG}, {\rm dHG}, {\rm Kum}, 
{\rm Leg}, {\rm Bes}, {\rm Whi}, {\rm Web}, {\rm dBes},  {\rm Ai} \}$ 
labels the equation in Figure \ref{fig:confdiag} to which $Q_\bullet$ is related.} 
\label{table:classical}
\end{table}

In the present paper, we will focus on the BPS structure itself, rather than the Riemann-Hilbert problem, which will be the subject of the second paper. Furthermore, we will consider a specific class of BPS structures arising in a natural way from meromorphic quadratic differentials $\varphi = Q(x) dx^2$ on the Riemann surface $X={\mathbb P}^1$ (c.f., \cite{Bri19, All19}).
The corresponding BPS structure is described through the geometry of the associated \emph{spectral cover}\footnote{In the literature, $\Sigma$ is usually called a \emph{spectral curve}, owing to its role in the theory of Hitchin systems. It is closely related but distinct from the notion of spectral curve in topological recursion. We stick with the terminology of ``spectral cover" to keep the distinction clear.}
\begin{equation}
\Sigma = \{\lambda ~|~ \lambda^2 - \varphi = 0 \} \subset T^\ast X, 
\end{equation}
which is the Riemann surface on which the square root $\sqrt{\varphi}$ is well-defined. 
The lattice $\Gamma$ is a sublattice of the homology group $H_1(\Sigma; {\mathbb Z})$ of the spectral cover, and the central charge $Z$ is the period map; i.e., 
$Z(\gamma) =  \oint_\gamma \lambda  = \oint_\gamma \sqrt{\varphi}$
for $\gamma \in \Gamma$.
The BPS invariants are given by a weighted counting of saddle trajectories or ring domains on $X$ associated to $\varphi$; this is based on the idea of Gaiotto-Moore-Neitzke \cite{GMN09} in which they find a combinatorial algorithm to describe the  BPS spectra of certain four-dimensional supersymmetric QFTs through a \emph{spectral network} (known as a \emph{Stokes graph} in WKB literature such as \cite{KT98, IN14}). 

 The algorithm developed by Gaiotto-Moore-Neitzke outputs BPS structures given a quadratic (or higher) differential. However, they do not consider the case when $\varphi$ has simple poles. In $\S \ref{subsection:constructing-BPS-str}$, we generalize their definition of $\Omega$ to include BPS cycles coming from more general degenerations in order to include these cases, as well as degenerate ring domains\footnote{Due to the behaviour of the foliation in the presence of two simple poles (e.g. the Legendre case), it is necessary to view these as degenerations of the entire foliation rather than just the spectral network. We explain this in further detail in \S \ref{section:Legendre}}. In \S \ref{sec:maincomputation}, we will give a complete description of the BPS structures arising in this fashion for all nine examples of hypergeometric type, over generic loci in their respective parameter spaces. In particular, we give an explicit description of parameter values for which the spectral network of $\varphi({\bm m})$ degenerates, and prove the existence of the saddle trajectory at these values. We also determine the corresponding BPS cycles, central charges, and BPS invariants, and summarize our computation in Tables \ref{table:bpsangles-Web}-\ref{table:bpsangles-Leg}.

These examples all correspond to spectral covers of degree 2. Although the theory of Gaiotto-Moore-Neitzke \cite{GMN12} is on much less solid footing mathematically when the spectral cover has degree 3 or higher (corresponding to $k$-differentials for $k>2$), we will also consider some higher degree examples experimentally. We find evidence that in both of these cases, our main result continues to hold.

\subsection{Topological recursion}
Topological recursion (TR) was introduced by Eynard-Orantin \cite{EO} and Chekhov-Eynard-Orantin \cite{CEO} as a generalization of the loop equations in the theory of matrix models, describing random matrices and physical applications thereof. 
To a given algebraic curve, TR constructs a doubly-indexed sequence $\{ W_{g,n} \}_{g \ge 0, n \ge 1}$ of meromorphic multidifferentials on the curve, called the \emph{Eynard-Orantin correlators}, satisfying a number of nice properties. 
In the TR formalism, the curve taken as the input for TR is called the \emph{spectral curve} since it generalizes the spectral curve of matrix models. 
The Eynard-Orantin correlator $W_{g,n}$ is closely related to the coefficients of the topological expansion 
of $n$-point correlation functions of the matrix model. 
The analogue of the genus $g$ contribution $F_g$ to the topological expansion of the \emph{free energy} is also introduced by integrating the Eynard-Orantin correlators. We call the generating formal series 
$Z_{\rm TR}(\hbar) := \exp(\sum_{g \ge 0} \hbar^{2g-2} F_g)$ the \emph{TR partition function} for the given spectral curve, which is the most fundamental output of TR. Note that if there is a family of spectral curves parametrized by a complex manifold $M$, then the TR invariants (i.e., $W_{g,n}$ and $F_g$) depend also on ${\bm m} \in M$; in that situation, we will write the genus $g$ free energy as $F_g({\bm m})$ to indicate the parameter dependence. 

Since its introduction, it has been noted that TR often produces various enumerative invariants such as Gromov-Witten invariants, Hurwitz numbers, Mirzakhani's Weil-Petersson volume, etc \cite{EO2, Eynard-11, DMNPS13, DN16}; see the review article \cite{EO-08} for more information. 
In the context of the mirror symmetry, the topological recursion is regarded as the higher genus theory of the $B$-model. 
Due to its wide range of applications, TR continues to attract many researchers including both mathematicians and physicists.

As represented by the celebrated Kontsevich-Witten theorem \cite{Kon, Wit}, the generating series of such enumerative invariants are expected to be related to the $\tau$-functions of various integrable systems. In fact, at the level of formal power series of $\hbar$, KdV and Painlev\'e $\tau$-functions have been constructed via the TR partition function for certain spectral curves; see \cite{EO, IS, IM, IMS, MO18, I19, EG19, MO19} for example. The actual $\tau$-function should then be given as the Borel sum of those formal series. 
Although the Borel summability of the TR partition function is not proved in full generality (see \cite{Eynard-19} for the growth estimate of the coefficients), those which appear in this paper are Borel summable except for finitely many singular directions, and we will discuss their properties in the sequel to this paper.

On the other hand, the formalism of \emph{quantum curves}, which was developed relatively recently in \cite{GS12, DM13, DMNPS13, DN16, BE-16, KSTR} etc., relates topological recursion to WKB analysis. 
It claims that, under certain conditions, a generating series of (integrals of) Eynard-Orantin correlators gives a WKB solution of a Schr\"odinger-type differential (or difference) equation whose classical limit recovers the original spectral curve; therefore, the resulting Schr\"odinger-type equation is called a quantum (spectral) curve. To avoid overloading the reader, we omit a detailed discussion of quantum curves in the present paper, but we emphasize that the relation between BPS structures and TR which we will observe may be viewed as a consequence of the quantum curve machinery (both theories are related through the {\em exact WKB analysis} \cite{Voros83, KT98}). We will return to this in the second paper.

\subsection{Main results for the hypergeometric spectral curve}

Let us state the main result of the present paper for the family of \emph{(Gauss) hypergeometric spectral curves} 
$\Sigma_{\rm HG}$ defined by the quadratic differentials $\varphi_{\rm HG} = Q_{\rm HG}(x) dx^2$ on $X={\mathbb P}^1$ with
\begin{equation}
Q_{\rm HG}(x) = \frac{m_{\infty}^2 x^2 - (m_\infty^2 - m_1^2 + m_0^2) x + m_0^2}{x^2(x-1)^2}.
\end{equation}
Here ${\bm m} = (m_0, m_1, m_\infty)$ is a tuple of complex numbers parametrizing $\Sigma_{\rm HG}$. 
We assume that ${\bm m}$ lies in $M_{\rm HG} \subset {\mathbb C}^3$ defined by the condition

\begin{equation} \label{eq:generic-mass-intro}
m_0 m_1 m_\infty (m_0+m_1+m_\infty)(m_0+m_1-m_\infty)
(m_0-m_1+m_\infty)(m_0-m_1-m_\infty) \ne 0.
\end{equation}

The curve $\Sigma_{\rm HG}$ is equivalent to the classical limit of the hypergeometric equation with a small parameter $\hbar$ which is introduced appropriately (see \cite{Aoki-Tanda, IKoT-II}). 

Let us summarize the properties of the BPS structure/TR invariants arising from $\Sigma_{\rm HG}$.

\subsubsection*{\bf BPS side}
We define the lattice
\begin{equation}
\Gamma = \{ \gamma \in H_1(\Sigma_{\rm HG}; \mathbb Z) ~|~ \sigma_\ast \gamma = - \gamma \},
\end{equation} where 
$\sigma_\ast$ is the action induced by the covering involution of $\Sigma_{\rm HG}$. 
$\Gamma$ is a rank three lattice whose generators are written as a ${\mathbb Z}$-linear combinations of $\{\gamma_{0_+}, \gamma_{0_-}, \gamma_{1_+}, \gamma_{1_-}, \gamma_{\infty_+}, \gamma_{\infty_-}\}$, where $\gamma_{s_+}$ and $\gamma_{s_-}$ denote the homology classes represented by two preimages of a positively oriented small circle around $s \in \{0, 1, \infty\}$ by the natural projection $\Sigma_{\rm HG} \to {\mathbb P}^1$. 
The corresponding central charge is given explicitly by 
$Z(\gamma_{s_\pm}) =  
\pm 2 \pi i  m_s$. 

Several properties of the spectral networks (Stokes graphs) defined by $\Sigma_{\rm HG}$ were 
studied by Aoki-Tanda \cite{Aoki-Tanda}. 
If we look at spectral networks with phase $\vartheta \in [0, 2\pi)$, 
we can show that, for any generic ${\bm m} \in M$, 
there are always 8 (resp., 6) phases for which the spectral networks of $e^{-2i \vartheta}$ contain 
a saddle trajectory between distinct branch points (resp., a saddle trajectory of loop-type). 
We assign the BPS invariants $\Omega(\gamma) = 1$ (resp., $\Omega(\gamma) = -1$)  
for the class $\gamma \in \Gamma$ which encircles the saddle trajectory 
(resp., the saddle trajectory of loop-type).
More concretely, the eight cycles  
\begin{equation} \label{eq:saddle-class-intro}
\pm(\gamma_{0_+}+\gamma_{1_+}+\gamma_{\infty_+}), ~
\pm(\gamma_{0_+}+\gamma_{1_+}+\gamma_{\infty_-}), ~
\pm(\gamma_{0_+}+\gamma_{1_-}+\gamma_{\infty_+}), ~
\pm(\gamma_{0_+}+\gamma_{1_-}+\gamma_{\infty_-}) ~ \in \Gamma
\end{equation}
give $\Omega(\gamma) = 1$, and the six cycles
\begin{equation}  \label{eq:loop-class-intro}
\pm (\gamma_{0_+} - \gamma_{0_-}), ~
\pm (\gamma_{1_+} - \gamma_{1_-}), ~
\pm (\gamma_{\infty_+} - \gamma_{\infty_-}) ~ \in \Gamma
\end{equation}
give $\Omega(\gamma) = -1$, and 
$\Omega(\gamma) = 0$ for other classes $\gamma \in \Gamma$.
Like all of our examples in this paper, this BPS structure is \emph{uncoupled}; i.e., any pair $\gamma, \gamma' \in \Gamma$
with $\Omega(\gamma), \Omega(\gamma') \ne 0$ satisfy $\langle \gamma, \gamma' \rangle = 0$, which follows in this case from the triviality of the intersection pairing.

The associated BPS structure 
has fourteen BPS rays in ${\mathbb C}^\ast$, and among them, there are seven BPS rays on 
any generic half plane 
(four of them come from the cycles in \eqref{eq:saddle-class-intro} which give $\Omega(\gamma) = 1$, 
while the other three come from \eqref{eq:loop-class-intro} which give $\Omega(\gamma) = -1$). 

\subsubsection*{\bf TR side}
Let $W_{g,n}^{\rm HG}$ and $F_g^{\rm HG}$ be the TR invariants of of the spectral curve $\Sigma_{\rm HG}$.
Properties of these TR invariants (and those for the spectral curves appearing 
in Table \ref{table:classical}) were studied in \cite{IKoT-I, IKoT-II}. 
This example also appears in \cite{CPT19}.
One of the main results of these papers gives an explicit expression for all free energies of hypergeometric type spectral curves. In particular, $F_g^{\rm HG}$ was obtained: 
\begin{align} \label{eq:Fg-HG-intro}
F^{\rm HG}_g({\bm m}) 
= \dfrac{B_{2g}}{2g(2g-2)} 
&\biggl( \dfrac{1}{(m_{0} + m_{1} + m_{\infty})^{2g-2}} 
+ \dfrac{1}{(m_{0} + m_{1} - m_{\infty})^{2g-2}} \notag \\
& 
\hspace{3mm}  + \dfrac{1}{(m_{0} - m_{1} + m_{\infty})^{2g-2}} 
+ \dfrac{1}{(m_{0} - m_{1} - m_{\infty})^{2g-2}} \notag \\
&\hspace{1.7cm} 
- \dfrac{1}{(2m_0)^{2g-2}} 
- \dfrac{1}{(2m_1)^{2g-2}} 
- \dfrac{1}{(2m_\infty)^{2g-2}} 
\biggr).
\end{align}
Here, $B_{2g}$ denotes the $2g$-th Bernoulli number, whose definition is recalled below. This expression is valid for $g \ge 2$ and ${\bm m}$ satisfying \eqref{eq:generic-mass-intro}, and we refer the reader to \cite[Theorem 3.1]{IKoT-II} for further details. The explicit formulas of the free energies for the other spectral curves 
are summarized in Table \ref{table:free-energy-0}.


Once we know the BPS structure for $\Sigma_{\rm HG}$, comparing both sides we obtain the result:

\begin{thm}[{Theorem \ref{thm:maintheorem-2}}]  \label{thm:main-intro}
For $g\geq 2$, the $g$-th free energy $F_g^{\rm HG}$ for $\Sigma_{\rm HG}$ has the following expression
\begin{equation}  \label{eq:main-result2-intro}
F^{\rm HG}_{g}({\bm m})=\dfrac{B_{2g}}{2g(2g-2)} 
\sum_{\substack{ \gamma \in \Gamma \\ Z(\gamma)\in  \mathbb{H}}} 
\Omega(\gamma) \left( 
{\dfrac{2 \pi i}{Z(\gamma)}} \right)^{2g-2},
\end{equation}
where $\mathbb{H}$ is any half plane whose boundary rays are not BPS.
\end{thm}

In other words we find that for BPS structures corresponding to spectral curves of hypergeometric type
(in Table \ref{table:classical}), we may express the ($g\geq2$) free energy as a sum over BPS states, times some well-known sequence (in this case, $\frac{B_{2g}}{2g(2g-2)}$). We will also give similar statements in {Theorem \ref{thm:maintheorem-2}} for the other spectral curves 
appearing in Table \ref{table:classical}.

We will make a conjecture in {Section \ref{sec:conjectures}} which claims that 
our result {(Theorem \ref{thm:maintheorem-2})} for the spectral curves of hypergeometric type should be generalized to spectral curves (possibly with higher degrees)  whose corresponding BPS structure is \emph{uncoupled}. 
See also \cite[\S 10.2]{CLT20} for a related discussion in the context of topological strings.

Although TR is applicable to higher degree spectral curves (\cite{BHLMR-12, BE-12, BE-16}), the existence of the associated BPS structure remains conjectural: we expect that Gaiotto-Moore-Neitzke's approach in \cite{GMN12} allows us to define a natural BPS structure for spectral curves which may come from higher differentials (e.g., $SL(N, {\mathbb C})$ Hitchin spectral curves with $N \ge 3$).  
In particular, our conjecture suggests that, at least for uncoupled BPS structures, it should be possible to deduce the BPS invariants without ever drawing a spectral network, avoiding the intensive computations required by the algorithm of \cite{GMN12}. 
While other such methods have been known previously \cite{bpsquivers1, bpsquivers2}, our approach is distinct, and we hope it may provide new insight. 
In this direction, we will look at two higher degree examples (whose free energy was recently computed in \cite{YM-Takei20} by Y.\,M.\,Takei) 
for which the conjecture is testable and appears to hold.

\subsection{Organization} 
The paper is structured as follows. In \S \ref{sec:TR} we review the topological recursion formalism, and review its application to the spectral curves of hypergeometric type. We will recall the relevant results on the free energies from \cite{IKoT-I, IKoT-II}. In \S \ref{sec:BPS} we recall the notion of BPS structures and spectral networks, and describe the construction of BPS structures from meromorphic quadratic differentials. Our main technical results computing the BPS structure are contained 
in \S \ref{sec:maincomputation}. 
Finally, in \S \ref{sec:final}, we will assemble the pieces into the formula for $F_g$. We will also propose a conjectural statement generalizing our main result for a wider class of spectral curves, and consider two examples of higher degree spectral curves experimentally as evidence.

\subsection{Acknowledgements} 
We thank 
Dylan Allegretti, 
Tom Bridgeland,  
Aaron Fenyes and
Yumiko Takei
for helpful conversations. 
O.K.'s work was supported by 
a JSPS Postdoctoral Fellowship for Research in Japan (Standard). 
This work was also supported by JSPS KAKENHI Grant Numbers 
16K17613, 16H06337, 17H06127, 19F19738 and 20K14323.

Many figures in this paper were produced using computer programs written by Andy Neitzke \cite{swnplotter} and Takashi Aoki \cite{aokiprogram}. We thank them for sharing these with us.

\section{Topological recursion for spectral curves of hypergeometric type}
\label{sec:TR}

Here we briefly recall several facts from Eynard-Orantin's theory (\cite{EO}), 
and recall the explicit form of the free energies for our examples computed recently in \cite{IKoT-I, IKoT-II}.

\subsection{Definition of correlators and partition function}

Let us recall the notion of spectral curves which are input for the topological recursion. 
\begin{dfn} \label{def:spectral-curve-TR}
A \emph{spectral curve} is a tuple $({\mathcal C}, x, y, B)$ of the following data:
\begin{itemize}
\item
a compact Riemann surface ${\mathcal C}$, 
\item 
a pair of non-constant meromorphic functions $x,y$ on ${\mathcal C}$ such that 
$dx$ and $dy$ never vanish simultaneously, and 
\item
a symmetric meromorphic bidifferential\footnote{
A bidifferential is a multidifferential with $n = 2$.} 
$B$ on ${\mathcal C}$ 
having a second order pole with no residue along the diagonal, and holomorphic elsewhere.
\end{itemize}
\end{dfn}

We usually denote by $z$ a local coordinate of ${\mathcal C}$, 
and by $z_i$ a copy of the coordinate when we consider {\em multidifferentials}.  
Here, a meromorphic multidifferential is a meromorphic section of the line bundle 
$\pi_1^\ast(T^\ast {\mathcal C}) \otimes \cdots \otimes \pi_n^\ast(T^\ast {\mathcal C})$
on ${\mathcal C}^n$, where $\pi_j : {\mathcal C}^n \to {\mathcal C}$ is the $j$-th projection map
\cite{DN16}. We often drop the symbol $\otimes$ when we express multidifferentials.

We denote by ${\mathcal R}$ the set of \emph{ramification points} of $x$, 
that is, ${\mathcal R}$ consists of zeros of $dx$;
here we consider $x$ as a branched covering map
${\mathcal C} \to \mathbb{P}^1$, and we call the images of ramification points
by $x$ the {\em branch points}.
We also assume that all ramification points of $x$ 
are simple so that the local conjugate map $z \mapsto \bar{z}$ 
near each ramification point is well-defined. 
Then, the \emph{topological recursion} (TR) is formulated as follows. 

\begin{dfn}[{\cite[Definition 4.2]{EO}}]
The \emph{Eynard-Orantin correlator} 
$W_{g, n}(z_1, \cdots, z_n)$ for $g \geq 0$ and $n \geq 1$
is defined as a meromorphic multidifferential on ${\mathcal C}$ 
by the recursive relation
\begin{align}
\label{eq:TR}
W_{g, n+1}(z_0, z_1, \cdots, z_n)
&:= \sum_{r \in {\mathcal R}}
\Res_{z = r} K_r(z_0, z)
\Bigg[
W_{g-1, n+2} (z, \overline{z}, z_1, \cdots, z_n)
\\
&\qquad\qquad
+
\sum_{\substack{g_1 + g_2 = g \\ I_1 \sqcup I_2 = \{1, 2, \cdots, n\}}}'
W_{g_1, |I_1| + 1} (z, z_{I_1})
W_{g_2, |I_2| + 1} (\overline{z}, z_{I_2})
\Bigg]
\notag
\end{align}
for $2g + n \geq 2$ with initial conditions given by
\begin{align}
W_{0, 1}(z_0) &:= y(z_0) dx(z_0),
\quad
W_{0, 2}(z_0, z_1) := B(z_0, z_1).
\end{align}
Here we set $W_{g,n} \equiv 0$ for a negative $g$, 
$K_r$ denotes the so-called ``recursion kernel" given by
\begin{equation}
\label{eq:RecursionKernel}
K_r(z_0, z)
:= \frac{1}{2\big(y(z) - y(\overline{z})\big) dx(z)}
\int^{\zeta = z}_{\zeta = \overline{z}} B(z_0, \zeta)
\end{equation}
defined near a ramification point $r \in {\mathcal R}$,
$\sqcup$ denotes the disjoint union,
and the prime $'$ on the summation symbol in \eqref{eq:TR}
means that we exclude terms for
$(g_1, I_1) = (0, \emptyset)$
and
$(g_2, I_2) = (0, \emptyset)$
(so that $W_{0, 1}$ does not appear) in the sum.
We have also used the multi-index notation:
for $I = \{i_1, \cdots, i_m\} \subset \{1, 2, \cdots, n\}$
with $i_1 < i_2 < \cdots < i_m$, $z_I:= (z_{i_1}, \cdots, z_{i_m})$.
\end{dfn}

For $2g-2+n \ge 1$, it is known that the correlators 
$W_{g,n}(z_1,\dots, z_n)$ are symmetric meromorphic 
multidifferential on ${\mathcal C}^n$. 
As a differential in each variable $z_i$, it has poles only at ramification points with no residue; 
that is, it is a second kind differential on ${\mathcal C}$. 
We refer the reader to \cite{EO} for further properties of $W_{g,n}$. 
See also \cite{BHLMR-12, BE-12} for the {\em generalized topological recursion} 
which admits non-simple ramification points. 

One of the central objects in Eynard-Orantin's theory is 
the \emph{genus $g$ free energy} 
$F_g$ ($g\geq 0$) associated to the spectral curve:
\begin{dfn}[{\cite[Definition 4.3]{EO}}]
For $g \geq 2$, the \emph{genus $g$ free energy} $F_g$ is defined by
\begin{equation}
\label{def:Fg2}
F_g := \frac{1}{2- 2g} \sum_{r \in {\mathcal R}} \Res_{z = r}
\big[\Phi(z) W_{g, 1}(z) \big],
\end{equation}
where $\Phi(z)$ is any primitive of $y(z) dx(z)$.
\end{dfn}
The free energies $F_0$ and $F_1$ for $g=0$ and $1$ 
are also defined, but in a different manner which we omit
(see \cite[\S 4.2.2 and \S 4.2.3]{EO} for the definition). Their expressions in the case of hypergeometric type spectral curves are given in \S\ref{sec:free-energies}.
We note that the right-hand side of \eqref{def:Fg2} indeed does not
depend on the choice of the primitive $\Phi(z)$. 

\begin{dfn}
The generating series
\begin{equation} \label{eq:total-free-energy}
F_{\rm TR}(\hbar) := \sum_{g = 0}^{\infty} \hbar^{2g-2} F_g
\end{equation}
of $F_g$ is called the \emph{free energy} of the spectral curve. 
Its exponential 
\begin{equation}
Z_\mathrm{TR}(\hbar) := e^{F_{\rm TR}(\hbar)}
\end{equation} 
is called the \emph{partition function} of the spectral curve. 
\end{dfn}

The free energy and the partition functions are regarded as formal series in $\hbar$. 
It is known that, usually, these formal series are divergent 
(see \cite{Eynard-19} for the growth estimate for the coefficients). Nonetheless, the series arising in our examples turn out to be {\em Borel summable} 
in any direction except for finitely many singular directions.

\subsection{Structure of hypergeometric type curves}
\label{section:hypergeometric-curvs-structure}

Our main examples are the {\em spectral curves of hypergeometric type} which are 
tabulated in Table \ref{table:classical}. 
Each of them is the classical limit of a certain 2nd order Schr\"odinger-type ODE which is equivalent to one of the members in the confluence diagram in Figure \ref{fig:confdiag}. 
In this section, we recall relevant facts about these spectral curves which were already 
found in \cite[Section 2.3]{IKoT-II}, though we use slightly different notations from \cite{IKoT-I, IKoT-II}.
We also sometimes borrow terminology from {\em exact WKB analysis}; 
see \cite{KT98} for background.

In the rest of Section \ref{sec:TR}, we denote by $Q(x)$ one of the rational functions 
$Q_{\bullet}(x)$ for 
$\bullet \in \{{\rm HG}, {\rm dHG}, {\rm Kum}, {\rm Leg}, {\rm Bes}, {\rm Whi}, {\rm Web}, 
{\rm dBes}, {\rm Ai} \}$, and by $\varphi = Q(x) dx^2$ the associated 
meromorphic quadratic differential on $X={\mathbb P}^1$.
We also introduce the following important loci:
\begin{itemize}
\item 
$P \subset X$ is the set of poles of $\varphi$.
It has a decomposition 
${P} = {P}_{\rm ev} \cup {P}_{\rm od}$
into the subsets of even order poles and odd order poles of $\varphi$. 
Note that all our examples in Table \ref{table:classical} satisfy $\infty \in P$.
\item 
$T \subset X$ is the set of zeros and simple poles of $\varphi$.
Borrowing terminology from the WKB literature, 
we call elements in $T$ {\em turning points}\,\footnote{
In the WKB analysis, $Q(x)$ plays the role of the potential function of 
a Schr\"odinger-type ODE. Traditionally, zeros of $\varphi$ are called 
turning points in WKB literature. 
It was pointed by \cite{Ko2} that simple poles of $\varphi$
also play similar roles to the usual turning points (called 
{\em turning points of simple pole type}). 
Taking this into account, 
in this paper, by a turning point we mean a zero or a simple pole of $\varphi$. 
}.

\item
${\rm Crit} = P \cup T$ is called the set of {\em critical points} 
of $\varphi$. Its subset ${\rm Crit}_{\infty} = P \setminus T$ consists 
of {\em infinite critical points} of $\varphi$; 
that is, the poles of $\varphi$ of order $\ge 2$. 
The points in $T = {\rm Crit} \setminus {\rm Crit}_\infty$ are also 
called {\em finite critical points} in this context. 

\end{itemize}
We will use the notation $P_{\bullet}$ etc. when we consider a specific 
example in Table \ref{table:classical}.

The quadratic differential $\varphi$ is parametrized by a tuple of parameters 
${\bm m} = ( m_s )_{s \in {P_{\rm ev}}}$ which we call {\em mass parameters} 
(or {\em temperatures} in TR language).  
In what follows, we assume the following condition on the mass parameters:

\begin{ass} \label{ass:genericity}
The mass parameters lie in the following set $M = M_\bullet$: 
\begin{align}
M_{\rm HG} & := \{(m_0, m_1, m_\infty) \in {\mathbb C}^3 ~|~ 
m_0 m_1 m_\infty \Delta_{\rm HG}(\bm m) \ne 0 \}. 
\\
M_{\rm dHG} & := \{(m_1, m_\infty) \in {\mathbb C}^2 ~|~ 
m_1 m_\infty (m_1 + m_\infty)(m_1 - m_\infty) \ne 0 \}. 
\\
M_{\rm Kum} & := \{(m_0, m_\infty) \in {\mathbb C}^2 ~|~ 
m_0 (m_0 + m_\infty)(m_0 - m_\infty) \ne 0 \}.
\\
M_{\rm Leg} & := \{m_\infty \in {\mathbb C} ~|~ m_\infty \ne 0 \}.
\\
M_{\rm Bes} & := \{m_0 \in {\mathbb C} ~|~ m_0 \ne 0 \}. 
\\
M_{\rm Whi} & := \{m_\infty \in {\mathbb C} ~|~ m_\infty \ne 0 \}.
\\
M_{\rm Web} & := \{m_\infty \in {\mathbb C} ~|~ m_\infty \ne 0 \}.
\end{align}
where we set 
\begin{equation} \label{eq:HG-Delta}
\Delta_{\rm HG}(\bm m) = (m_0 + m_1 + m_{\infty})(m_0 + m_1 -m_{\infty})
(m_0 - m_1 + m_{\infty})(m_0 - m_1 - m_{\infty}).
\end{equation}
For the degenerate Bessel curve (${\rm dBes}$) and Airy curve (${\rm Ai}$), 
we set $M_{\rm dBes} = M_{\rm Ai} = \emptyset$ since they have no mass parameters.
\end{ass}

We can verify that no collision of a pair of turning points or a pair of a turning point and a pole of $\varphi$ occur on $M$. Namely, a turning point is either a simple zero or a simple pole of $\varphi$ under the Assumption \ref{ass:genericity}.

Let us consider the corresponding \emph{spectral cover}:
\begin{equation} \label{eq:spectral-curve-x-y}
\Sigma := \{ (x, y) \in {\mathbb C}^2 ~|~ y^2 - Q(x) = 0 \}, 
\end{equation}
associated with the quadratic differential $\varphi$. This is a branched double cover of $X \setminus P$ 
with the projection map $\pi : (x, y) \mapsto x$. 
We denote by $\sigma : (x,y) \mapsto (x,- y)$ the covering involution of $\Sigma$.  
We will mainly use $x$ as a local coordinate\footnote{
Technically speaking, we must fix branch cuts and identify the coordinate $x$ of 
the base ${\mathbb P}^1$ with a coordinate of $\Sigma \setminus \pi^{-1}(T)$ on the first sheet 
(then, $\sigma(x)$ gives a coordinate in the second sheet).} of $\Sigma \setminus \pi^{-1}(T)$.
We note that, the Riemann surface $\Sigma$ for all the examples in Table \ref{table:classical} 
are of {\em genus $0$}, and hence topologically it is a sphere with several punctures, corresponding to poles of $\varphi$. 
As we will see below, $\Sigma$ has a parametrization by a pair of rational functions, making $\Sigma$ into a spectral curve in the sense above, which allows us to apply the TR formalism.

\begin{rem}
It is natural to regard $\Sigma$ as a subvariety of the holomorphic cotangent bundle 
$T^\ast X$ of $X$ defined by the equation $\lambda^2 - \varphi = 0$, 
where $\lambda$ is the fiber coordinate of $T^\ast X$. 
The presentation \eqref{eq:spectral-curve-x-y} can be understood as a local expression 
in a coordinate chart $(x,y) \in {\mathbb C}^2$ of $T^\ast X$ on which 
the tautological 1-form is expressed as $\lambda = y dx$. 
$T^\ast X$ has another coordinate chart $(\tilde{x}, \tilde{y}) \in {\mathbb C}^2$
on which $\Sigma$ is expressed as 
$\{(\tilde{x}, \tilde{y}) \in {\mathbb C}^2 ~|~ \tilde{y}^2 - \tilde{Q}(\tilde{x}) = 0 \}$ 
according to the gluing rule $(\tilde{x}, \tilde{y}) = (x^{-1}, -x^{2} y)$ in $T^\ast X$.  
Here $\tilde{Q}(\tilde{x}) = Q(\tilde{x}^{-1}) \tilde{x}^{-4}$ appears in the expression 
of $\varphi=\tilde{Q}(\tilde{x}) d\tilde{x}^2$ with respect to the coordinate 
$\tilde{x} = x^{-1}$ on $X = {\mathbb P}^1$. This perspective will be important on the BPS side of our correspondence.
\end{rem}

We denote by $\overline{\Sigma}$ the compactification of $\Sigma$, obtained by adding preimages of poles of $\varphi$, on which the square root $\sqrt{\varphi} = \sqrt{Q(x)} \, dx$ gives a meromorphic differential. 
More precisely, to obtain $\overline{\Sigma}$, we add to $\Sigma$ a single point for each $s \in P_{\rm od}$, and add a pair $s_{\pm}$ of two points for each $s \in P_{\rm ev}$:  
\begin{equation}
\overline{\Sigma} = \Sigma \cup P_{\rm od} \cup \{s_+, s_- ~|~ s \in P_{\rm ev} \}.
\end{equation}
Here we identify $s \in P_{\rm od}$ with its unique lift on $\overline{\Sigma}$. 
We keep using the same notations $\pi : \overline{\Sigma} \to X$ and ${\sigma}: \overline{\Sigma} \to \overline{\Sigma}$ for the projection map and the covering involution. 
The sign in the notation for the two points $s_{\pm}$, both of which are mapped to $s \in P_{\rm ev}$ by $\pi$, is chosen so that the following residue formula holds:
\begin{equation} \label{eq:sign-convention-preimages}
\Res_{x = s_{\pm}} \sqrt{Q(x)} \, dx = \pm m_s. 
\end{equation}
{For later use, we also introduce a partial compactification $\widetilde{\Sigma}$ of $\Sigma$ obtained by filling in punctures corresponding to simple poles of $\varphi$:  
\begin{equation}
\widetilde{\Sigma} := \Sigma \cup (P_{\rm od} \cap T) 
= \overline{\Sigma} \setminus D_{\infty}.
\end{equation}
Here we set 
\begin{equation}
D_\infty := \pi^{-1} ({\rm Crit}_{\infty}) ~\subset \overline{\Sigma}.
\end{equation}
If $\varphi$ has no simple poles, then $\widetilde{\Sigma} = \Sigma$. 
The homology classes on $\widetilde{\Sigma}$ are important in constructing a BPS structure associated with $\varphi$ in \S \ref{sec:maincomputation}. They will also be used when we define the Voros coefficients in the second part.
}

To apply TR to these curves, we should regard them as spectral curves in the sense of 
Definition \ref{def:spectral-curve-TR}. This was done in \cite[Section 2.3]{IKoT-II}, where an explicit rational parametrization of $\Sigma$ was given. 
That is, for each $\bullet$ in Table \ref{table:classical}, 
there exists a pair of rational functions 
$(x(z),y(z)) = (x_{\bullet}(z), y_{\bullet}(z))$ such that 
we have an isomorphism 
\begin{equation} \label{eq:meromorphic-parametrization}
\begin{array}{ccc}
{\mathcal C} \setminus {\mathcal P}
& \stackrel{\sim}{\longrightarrow} & \Sigma \\
\rotatebox{90}{$\in$} & & \rotatebox{90}{$\in$} \\
z & \longmapsto & (x(z), y(z))
\end{array}
\end{equation} 
of punctured Riemann surfaces, 
with the choice ${\mathcal C} = {\mathbb P}^1$ and the set ${\mathcal P}$ of poles of $x(z)$ and $y(z)$. We will give examples of these rational functions in Example \ref{exa:Gauss} and Example \ref{exa:dGauss} for Gauss and degenerate Gauss curves, respectively (see \cite[Section 2.3]{IKoT-II} for other examples). To avoid confusions, we also keep using the symbol 
${\mathcal C}$ for ${\mathbb P}^1$ to distinguish it from the target $X = {\mathbb P}^1$ of the map $x(z)$.
 
The set ${\mathcal P}$ is of the form
\begin{equation}
{\mathcal P} = \{ p_s ~|~ s \in P_{\rm od} \} \cup 
\{ p_{s_{+}}, p_{s_{-}} ~|~ s \in P_{\rm ev} \} \subset {\mathcal C},
\end{equation}
where the notation means that, for each $s \in P$, either $x(p_s) = s$ or $x(p_{s_\pm}) = s$ is satisfied depending on the parity of the pole order of $\varphi$ at $s$.  
{  We chose the sign so that 
\begin{equation} \label{eq:residue-and-mass-parameters}
\Res_{~~z = p_{s_\pm}} y(z) dx(z) = 
\pm m_{s} = \Res_{~x = s_{\pm}} \sqrt{Q(x)} \, dx
\end{equation}
holds for any $s \in P_{\rm ev}$.}
Hence, the set ${\mathcal P}$ is in bijection with the set $\overline{\Sigma} \setminus \Sigma$, and hence, \eqref{eq:meromorphic-parametrization} can be extended to an isomorphism ${\mathcal C} \isom \overline{\Sigma}$ of compact Riemann surfaces.
Together with the canonical choice 
\begin{equation} \label{eq:Bergman-P1}
B(z_1, z_2) = \frac{dz_1 dz_2}{(z_1 - z_2)^2}
\end{equation} 
of the bidifferential $B$ when ${\mathcal C} = {\mathbb P}^1$,
we obtain a spectral curve $({\mathcal C}, x, y, B)$ 
in the sense of Definition \ref{def:spectral-curve-TR}.
We will identify $\Sigma$ defined in \eqref{eq:spectral-curve-x-y}
(or its compactification $\overline{\Sigma}$) 
with the spectral curve $({\mathcal C}, x, y, B)$ through the isomorphism. 
Under the isomorphism, the set ${\mathcal R} \subset {\mathcal C}$ of ramification points 
is mapped to $\pi^{-1}(T) \subset \overline{\Sigma}$ bijectively.

\begin{exa}[Gauss curve {\cite[Section 2.3.1]{IKoT-II}}] \label{exa:Gauss}
The Gauss curve $\Sigma_{\rm HG}$ is defined from the meromorphic quadratic differential 
$\varphi_{\rm HG} = Q_{\rm HG}(x) dx^2$ with the rational function 
\begin{equation} \label{eq:Sigma-HG}
Q_{\rm HG}(x) := \frac{m_{\infty}^2 x^2 - (m_\infty^2 - m_1^2 + m_0^2) x + m_0^2}{x^2(x-1)^2}.
\end{equation}
Under the assumption ${\bm m} = (m_0, m_1, m_\infty) \in M_{\rm HG}$, 
$P_{\rm HG} = P_{{\rm HG, ev}} = \{0, 1, \infty\}$
and $T_{\rm HG} = \{b_1, b_2 \}$ consists of two simple zeros of $Q_{\rm HG}(x)$. 
Topologically, $\Sigma_{\rm HG}$ is a sphere with six punctures, and the  
compactification is given by 
\begin{equation}
\overline{\Sigma}_{\rm HG} := \Sigma_{\rm HG} \cup 
\{0_+, 0_-, 1_+, 1_-, \infty_+, \infty_- \}.  
\end{equation}
{  We note that $\widetilde{\Sigma}_{\rm HG} = \Sigma_{\rm HG}$ holds, 
and we have $D_{{\rm HG}, \infty} = \{0_+, 0_-, 1_+, 1_-, \infty_+, \infty_- \}$.}
A rational parametrization \eqref{eq:meromorphic-parametrization}
of $\Sigma_{\rm HG}$ is given by the pair of explicit rational functions 
\begin{equation}
\label{eq:Gauss_parameterization}
\begin{cases}
\displaystyle
x(z) = x_{\rm HG}(z)
:= \frac{\sqrt{\Delta_{\rm HG}({\bm m})}}{4 {m_{\infty}}^2} 
\left( z + z^{-1} \right)
+ \frac{{m_{\infty}}^2 + {m_0}^2 
- {m_1}^2}{2 {m_{\infty}}^2}, \\[10pt]
\displaystyle
y(z) = y_{\rm HG}(z)
:= \frac{4 {m_{\infty}}^3 z^2 \left(z - z^{-1} \right)}
{\sqrt{\Delta_{\rm HG}({\bm m})} \, (z - p_{0_+})(z - p_{0_-})(z - p_{1_+})(z - p_{1_-})},
\end{cases}
\end{equation}
where $\Delta_{\rm HG}(\bm m)$ is given in \eqref{eq:HG-Delta}, 
and the set 
${\mathcal P}_{\rm HG} = \{p_{0_+}, p_{0_-}, p_{1_+}, p_{1_-}, p_{\infty_+}, p_{\infty_-}  \}$
of poles consists of 
\begin{equation} 
p_{0_\pm} :=
- \frac{(m_{0} \pm m_\infty)^2 
- {m_1}^2}{\sqrt{\Delta_{\rm HG}({\bm m})}},
\quad
p_{1_\pm} := \frac{(m_{1} \pm m_\infty)^2 
- {m_0}^2}{\sqrt{\Delta_{\rm HG}({\bm m})}}, 
\quad
p_{\infty_+}=\infty, ~~ p_{\infty_-}=0.
\end{equation}
The points $p_{s_\pm} \in {\mathcal C}$ are the two preimages of $s \in \{0, 1, \infty \}$ 
by $x_{\rm HG}(z)$, and the labels are chosen so that the residue formula 
\eqref{eq:residue-and-mass-parameters} holds. 
The point $p_{s_\pm}$ is mapped to $s_\pm$ through the isomorphism 
${\mathcal C} \simeq \overline{\Sigma}_{\rm HG}$.
The set of ramification points is given by ${\mathcal R}_{\rm HG} = \{\pm 1 \} \subset {\mathcal C}$, and these two points are mapped to the turning points $b_{1}, b_{2} \in {\mathbb P}^1$ by $x_{\rm HG}(z)$.  The conjugate map is given by $\overline{z} = 1/z$, which corresponds to the covering involution of $\overline{\Sigma}_{\rm HG}$.
\end{exa}

\begin{exa}[Degenerate Gauss curve {\cite[Section 2.3.2]{IKoT-II}}] \label{exa:dGauss}
The degenerate Gauss curve $\Sigma_{\rm dHG}$ is defined from  
\begin{equation} 
Q_{\rm dHG}(x) := \frac{m_{\infty}^2 x - (m_\infty^2 - m_1^2)}{x(x-1)^2}.
\end{equation}
We have $P_{\rm dHG} = \{0, 1, \infty\}$, $P_{{\rm dHG}, {\rm od}} = \{0 \}$ 
and $P_{{\rm dHG}, {\rm ev}} = \{1, \infty \}$
under the assumption ${\bm m} = (m_1, m_\infty) \in M_{\rm dHG}$. 
The turning point set $T_{\rm dHG} = \{(m_\infty^2 - m_1^2)/ m_\infty^2, 0 \}$ consists of two points; 
$(m_\infty^2 - m_1^2)/ m_\infty^2$ is a simple zero, while $0$ is a simple pole of 
$\varphi_{\rm dHG} = Q_{\rm dHG}(x) dx^2$. 
Topologically, $\Sigma_{\rm dHG}$ is a sphere with 5 punctures, and we have
\begin{equation}
{  
\widetilde{\Sigma}_{\rm dHG} := \Sigma_{\rm dHG} \cup 
\{0 \}
~~\subset~~
}
\overline{\Sigma}_{\rm dHG} := \Sigma_{\rm dHG} \cup 
\{0, 1_+, 1_-, \infty_+, \infty_- \}.  
\end{equation}
The set $D_{{\rm dHG}, \infty}$
is given by $\{1_+, 1_-, \infty_+, \infty_- \}$.
The rational parametrization \eqref{eq:meromorphic-parametrization}
of $\Sigma_{\rm dHG}$ is given by the pair 
\begin{equation}
\label{eq:d-Gauss_parameterization}
\begin{cases}
\displaystyle
x(z) = x_{\rm dHG}(z)
:= \frac{m_1^2 - m_\infty^2}{z^2 - m_\infty^2}, \\[10pt]
\displaystyle
y(z) = y_{\rm dHG}(z)
:= - \frac{z(z^2 - m_\infty^2)}{z^2 - m_1^2},
\end{cases}
\end{equation}
where  ${\mathcal P}_{\rm dHG} = \{p_{0}, p_{1_+}, p_{1_-}, p_{\infty_+}, p_{\infty_-}  \}$ with
\begin{equation} 
p_{0} := \infty, 
\quad
p_{1_\pm} := \pm m_1, 
\quad
p_{\infty_\pm} := \pm m_\infty.
\end{equation}
Again, the set ${\mathcal P}_{\rm dHG}$ is bijectively mapped to the set 
$\overline{\Sigma}_{\rm dHG} \setminus \Sigma_{\rm dHG} = \{0, 1_+, 1_-, \infty_+, \infty_- \}$ 
through the isomorphism ${\mathcal C} \simeq \overline{\Sigma}_{\rm dHG}$. 
The labels are chosen so that the residue formula \eqref{eq:residue-and-mass-parameters} 
holds for each $s \in P_{\rm dHG, ev}$. 
The set of ramification points is given by ${\mathcal R}_{\rm dHG} = \{0, \infty \} \subset {\mathcal C}$, 
and these two points are mapped to the turning points by $x_{\rm dHG}(z)$. 
The conjugate map is given by $\overline{z} = -z$, 
which corresponds to the covering involution of $\overline{\Sigma}_{\rm dHG}$.
\end{exa}

\subsection{Free energy of hypergeometric type spectral curves}
\label{sec:free-energies}
In \cite{IKoT-I, IKoT-II}, an explicit expression of free energies of hypergeometric type spectral curves were obtained. For example, the $g$-th free energy $F_g^{\rm HG}$ of the Gauss hypergeometric curve is given as follows.

\begin{thm}[{\cite[Theorem 3.1 (iii)]{IKoT-II}}] \label{thm:Borel-sum-free-energy}
The $g$-th free energy of the Gauss hypergeometric curve $\Sigma_{\rm HG}$ are given explicitly as follows: 
\begin{align}
F^{\rm HG}_0({\bm m}) & = 
\sum_{\epsilon, \epsilon' \in \{\pm \}} \frac{(m_0 + \epsilon \, m_1 + \epsilon' m_\infty)^2}{2} \log(m_0 + \epsilon \, m_1 + \epsilon' m_\infty) - \sum_{s \in \{0,1,\infty \}} \frac{(2m_s)^2}{2} \log(2 m_s),
\\ 
F^{\rm HG}_1({\bm m}) & = - \frac{1}{12} \log\left( 
\frac{ \Delta_{\rm HG}(\bm m)
}
{ 
m_0 m_1 m_\infty
} 
\right), 
\\ 
F^{\rm HG}_g({\bm m}) & 
= \frac{B_{2g}}{2g(2g-2)} 
\biggl( 
\sum_{\epsilon, \epsilon' \in \{\pm \}} 
\frac{1}{(m_{0} + \epsilon \, m_{1} + \epsilon' m_{\infty})^{2g-2}} 
- \sum_{s \in \{0,1,\infty \}} \frac{1}{(2m_s)^{2g-2}} \biggr) \qquad (g \ge 2).
\end{align}
Here $B_{k}$ is the $k$-th Bernoulli number defined through its generating series
\begin{equation}
\label{def:Bernoulli-number}
\frac{w}{e^w - 1} = \sum_{k = 0}^{\infty} B_k \frac{w^k}{k!}.
\end{equation}
\end{thm}

\begin{table}[h]
\begin{center}
\begin{tabular}{cc}\hline
Spectral curve & $F_0$, $F_1$ and $F_g$ ($g \ge 2$)
\\\hline\hline
\parbox[c][3.5em][c]{0em}{}
\begin{minipage}{.15\textwidth}
Degenerate Gauss
\end{minipage}
&
{
$\displaystyle
F^{\rm dHG}_0 = \sum_{\epsilon \in \{\pm \}} 
(m_1+ \epsilon m_\infty)^2 \log(m_1 + \epsilon m_\infty) 
- \sum_{s \in \{1,\infty\}}  \dfrac{(2{m_s})^2}{2} \log{(2 m_s)}
$
}
\\[-.em]
\parbox[c][4.5em][c]{0em}{}
&
\hspace{-3.7cm}{$ 
F_1^{\rm dHG} = 
-  \dfrac{1}{12} \log \left( \dfrac{ (m_1+ m_\infty)^2(m_1 - m_\infty)^2}
{m_1 m_\infty} \right)
$
}
\\[-.em]
\parbox[c][3.5em][c]{0em}{}
&
{$ \displaystyle
F_g^{\rm dHG} = 
\dfrac{B_{2g}}{2g (2g-2)}
\left(
\sum_{\epsilon \in \{\pm \}}
\dfrac{2}{(m_1 + \epsilon m_\infty)^{2g-2}}
- 
\sum_{s \in \{1, \infty\}} 
\dfrac{1}{(2 m_s)^{2g-2}}
\right)
$
}
\\\hline
\parbox[c][3.5em][c]{0em}{}
Kummer 
&
{
$ \displaystyle 
F^{\rm Kum}_0 = 
\sum_{\epsilon \in \{ \pm \}}
\dfrac{(m_0+ \epsilon m_\infty)^2}{2} \log(m_0 + \epsilon m_\infty)
- \dfrac{(2 {m_0})^2}{2} \log{(2 m_0)}
$
}
\\[-.em]
\parbox[c][3.5em][c]{0em}{}
&
{
\hspace{-3.2cm}$F^{\rm Kum}_1 = 
- \dfrac{1}{12} \log \left(\dfrac{(m_0 + m_\infty)(m_0 - m_\infty)}{m_0}\right)$
}
\\[-.em]
\parbox[c][3.5em][c]{0em}{}
&
{\hspace{-0.3cm} $ \displaystyle 
F_g^{\rm Kum} = 
\dfrac{B_{2g}}{2g (2g-2)}
\left(
\sum_{\epsilon \in \{\pm \}}
\dfrac{1}{(m_0 + \epsilon m_\infty)^{2g-2}}
- 
\dfrac{1}{(2 m_0)^{2g-2}}
\right)
$
}
\\\hline
\parbox[c][3.5em][c]{0em}{}
Legendre 
& $\displaystyle  
F_0^{\rm Leg} = 2 m_\infty^2 \log(m_\infty) - 
\frac{(2m_\infty)^2}{2} \log(2 m_\infty)$
\qquad\quad~
$F_1^{\rm Leg} = - \dfrac{1}{12} \log \left( \dfrac{m_\infty^4}{2m_\infty} \right)$
\\[-.5em]
\parbox[c][3.5em][c]{0em}{}
&
{$ 
F_g^{\rm Leg} = 
\dfrac{B_{2g}}{2g(2g-2)} \left( \dfrac{4}{m_\infty^{2g-2}} 
 - \dfrac{1}{(2 m_\infty)^{2g-2}} \right)
$
}
\\\hline
\parbox[c][3.5em][c]{0em}{}
Bessel 
& $F_0^{\rm Bes} = - \dfrac{(2m_0)^2}{2} \log \left(2 m_0 \right)$ 
\qquad\quad~~
$F^{\rm Bes}_1 = \dfrac{1}{12} \log\left( 2m_0 \right)$
\\[-.5em]
\parbox[c][3.5em][c]{0em}{}
&
{$ 
F_g^{\rm Bes} = 
- \dfrac{B_{2g}}{2g (2g-2)} \dfrac{1}{(2m_0)^{2g-2}}
$
}
\\\hline
\parbox[c][3.5em][c]{0em}{}
Whittaker 
& $F_0^{\rm Whi} = 
m_\infty^2 \log (m_\infty)$  
\qquad\quad~~
$F^{\rm Whi}_1 = - \dfrac{1}{6} \log\left(m_\infty \right)$
\\[-.5em]
\parbox[c][3.5em][c]{0em}{}
&
{$ 
F_g^{\rm Whi} = \dfrac{B_{2g}}{2g (2g-2)} \dfrac{2}{m_\infty^{2g-2}}
$
}
\\\hline
\parbox[c][3.5em][c]{0em}{}
Weber 
& {$F_0^{\rm Web} = 
\dfrac{m_\infty^2}{2} \log m_\infty$  
\qquad\quad~~ 
$F^{\rm Web}_1 = - \dfrac{1}{12} \log\left( m_\infty \right)$
}
\\[-.5em]
\parbox[c][3.5em][c]{0em}{}
&
{$ 
F_g^{\rm Web} = \dfrac{B_{2g}}{2g (2g-2)} \dfrac{1}{m_\infty^{2g-2}}
$
}
\\\hline
\end{tabular}
\end{center}
\caption{$F_0$, $F_1$ and $F_g$ ($g \ge 2$) for the spectral curves $\Sigma$ 
in Table \ref{table:classical}. $B_{2g}$ denotes the $2g$-th Bernoulli number 
(see \eqref{def:Bernoulli-number} for the definition).}
\label{table:free-energy-0}
\end{table}

In Table \ref{table:free-energy-0}, we summarize the full expression of $F_g$'s for the other examples in Table \ref{table:classical}. 
We have excluded the Airy and degenerate Bessel cases since the free energy is trivial.
We note that the Weber curve appears as the spectral curve of the Gaussian matrix model, where the $g$-th free energy $F_g^{\rm Web}$ computes the Euler characteristic of the moduli space of genus $g$ Riemann surfaces \cite{HZ, Penner}.

\begin{rem}
In \cite{EO}, the genus $0$ free energy $F_0$ was defined through a regularization of a divergent integral of $W_{0,1}(z)$. 
The regularization process has an ambiguity because we must choose branches of logarithms in the quantity $\mu_\alpha$ appearing in the definition of $F_0$ (see \cite[eq.\,(4.14)]{EO}). 
For our spectral curves of hypergeometric type, we can verify that the ambiguity disappears if we mod out $F_0$ by polynomials of mass parameters at most degree two. 
Similarly, the genus $1$ free energy $F_1$ is defined up to additive constants since the Bergman $\tau$-function appearing in its definition is too. 
Our Table \ref{table:free-energy-0} shows expressions of $F_0$ and $F_1$ modulo these ambiguities, hence there are discrepancies between our presentations and the ones in \cite{IKoT-I, IKoT-II}. 
\end{rem}

From Table \ref{table:free-energy-0}, we can observe that the free energies $F_{\rm TR}({\bm m}, \hbar)$ have a superposition structure. That is, the free energy is described in the following schematic form 
\begin{equation}  \label{eq:free-energy-decomposition}
F_{\rm TR}({\bm m}, \hbar) \equiv 
\sum_{j_{\rm Web}} F^{\rm Web}(m_{j_{\rm Web}}, \hbar) 
+ \sum_{j_{\rm Whi}} F^{\rm Whi}(m_{j_{\rm Whi}}, \hbar)
+ \sum_{j_{\rm Bes}} F^{\rm Bes}(m_{j_{\rm Bes}}, \hbar).
\end{equation}
modulo the ambiguities in the first two terms $\hbar^{-2} F_0 + \hbar^{0} F_1$. 
For example, we have
\begin{equation}
F^{\rm HG}_{\rm TR}({\bm m}, \hbar) \equiv 
\sum_{\epsilon, \epsilon' \in \{\pm \}} 
F^{\rm Web}_{\rm TR}(m_0 + \epsilon \, m_1 + \epsilon' m_{\infty}, \hbar)
+ \sum_{s \in \{0,1,\infty \}} F^{\rm Bes}_{\rm TR}(m_s, \hbar). 
\end{equation}
This expression will be effectively used when we compute the Borel sum of the free energy of the partition function in the second paper.


\section{BPS structures and spectral networks}
\label{sec:BPS}

We now turn to the BPS side of the story. Let us recall several facts regarding BPS structures which are relevant for our paper, following \cite{Bri19}.

\subsection{Definition of BPS structure}

\begin{dfn}[{\cite[Definition 2.1]{Bri19}}]
 A \emph{BPS structure} is a tuple $(\Gamma,Z,\Omega)$ of the following data:
\begin{itemize}
\item a free abelian group of finite rank $\Gamma$ equipped with an antisymmetric pairing 
\begin{equation}
\langle \cdot , \cdot \rangle : \Gamma \times \Gamma \rightarrow \mathbb{Z},
\end{equation} 
\item a homomorphism of abelian groups $Z \colon \Gamma \rightarrow \mathbb{C}$, and
\item a map $\Omega : \Gamma \to {\mathbb Q}$,
\end{itemize}
satisfying the conditions
\begin{itemize}
\item \emph{Symmetry}: $\Omega(\gamma)=\Omega(-\gamma)$ for all $\gamma \in \Gamma$.
\item \emph{Support property}: for some (equivalently, any) choice of norm $\|\cdot\|$ on 
$\Gamma \otimes \mathbb{R}$, there is some $C>0$ such that
\begin{equation} \label{eq:support-property}
\Omega(\gamma)\neq  0 \implies |Z(\gamma)|>C\cdot \|\gamma\|.
\end{equation}
\end{itemize}
We call $\Gamma$ the \emph{charge lattice}. The homomorphism $Z$ is called the \emph{central charge}. 
The rational numbers $\Omega(\gamma)$ are called the \emph{BPS indices} or \emph{BPS invariants}.
\end{dfn}

Let us recall some useful terminology \cite{Bri19} for discussing BPS structures. 
\begin{itemize} 
\item
An \emph{active class} is an element $\gamma \in \Gamma$ which has 
nonzero BPS invariant, $\Omega(\gamma)\neq 0$.
\item 
A \emph{BPS ray} (or an \emph{active ray}) is a subset of ${\mathbb C}^\ast$  
which can be written as $\ell_\gamma = Z(\gamma) \, {\mathbb R}_{> 0}$ for some active class $\gamma$.


\end{itemize}

We note that the central charge $Z(\gamma)$ for an active class $\gamma$ is nonzero 
due to the support property. 

We can also consider certain classes of BPS structures with nice properties:
A BPS structure $(\Gamma, Z, \Omega)$ is said to be
\begin{itemize} 
\item 
\emph{finite} if there are only finitely many active classes,
\item
\emph{uncoupled} if $\langle \gamma_{1},\gamma_{2}\rangle=0$ holds whenever 
$\gamma_{1},\gamma_{2}$ are both active classes (otherwise, it is \emph{coupled})
\item
\emph{integral} if the BPS invariant $\Omega$ takes values in ${\mathbb Z}$.
\end{itemize}

In order to formulate the BPS Riemann-Hilbert problem in the next section, we may weaken these conditions. We call a BPS structure \emph{ray-finite} if there are finitely many active classes $\gamma$ with $Z(\gamma)\in \ell$ for a given BPS ray $\ell$, and \emph{generic} if $\langle \gamma_1,\gamma_2\rangle=0$ whenever $\gamma_1, \gamma_2$ are active and $Z(\gamma_1)$ and $Z(\gamma_2)$ lie on the same BPS ray.

The analysis of the BPS structures and the corresponding Riemann-Hilbert problem given in the next section is much more difficult in 
the coupled case, and all our calculations in this paper will be for 
finite, uncoupled and integral BPS structures.

\subsection{BPS Riemann-Hilbert problem}
\label{section:BPS-RHP}

Given a BPS structure, we may consider a certain Riemann-Hilbert type problem on ${\mathbb C}^\ast$, which we call the \emph{BPS Riemann-Hilbert problem}, following \cite{Bri19} (see also \cite{Bri20, Bar, BBS, Stoppa, Bri19-2, All19, Bri20-2} for further studies). 
Roughly speaking, this involves seeking a collection of functions $X_\gamma$, one for each $\gamma \in \Gamma$, that jumps whenever $\hbar$ crosses a BPS ray. The jump factor on a BPS ray $\ell$ is given by a certain {\em BPS automorphism} which encodes the BPS invariants $\Omega(\gamma)$ of those $\gamma \in \Gamma$ whose central charge $Z(\gamma)$ lies along $\ell$.
More precisely, we seek a function $X$ with values in the {\em twisted torus}
\begin{align}
\mathbb{T}_- := \left\{g : \Gamma \rightarrow \mathbb{C}^\ast  \; | \; 
g{(\gamma_1+\gamma_2)}= (-1)^{\langle\gamma_1,\gamma_2\rangle}g(\gamma_1)g(\gamma_2)\right\}
\end{align}
and denote $X_\gamma:=X(\gamma)$.

Here we give a condensed formulation of the BPS Riemann-Hilbert problem associated with a ray-finite, generic, integral BPS structure, following \cite{Bri19}. Defining the BPS automorphism needs careful analysis of convergence issues, and we need to work with some completion of the twisted torus in general. However, thanks to \cite[Proposition 4.2]{Bri19}, 
the BPS automorphisms for ray-finite, generic, and integral BPS structures are given by explicit birational automorphisms.

\begin{prob}[{\cite[Problem 3.1]{Bri19}}]
Let $(\Gamma, Z, \Omega)$ be a ray-finite, generic, and integral BPS structure, and fix a ``constant" $\xi\in\mathbb{T}_{-}$. Find a piecewise holomorphic map $X : {\mathbb C}^\ast \to {\mathbb T}_-$
such that for any $\gamma \in \Gamma$, we have:
\begin{itemize}
    \item[(i)] \emph{Jumping.} 
    As $\hbar\in \mathbb{C}^*$ crosses a BPS ray $\ell$ in the clockwise direction, $X_\gamma$ jumps according to the \emph{BPS automorphisms}:
       \begin{equation} \label{eq:RHP-jump-explicit-simplified}
    X^{+}_{\gamma} = 
    X^{-}_{\gamma} \,
    \prod_{\substack{\gamma' \in \Gamma \\ Z(\gamma') \in \ell}} 
    (1 - X^{-}_{\gamma'})^{\Omega(\gamma') \langle \gamma', \gamma \rangle},
    \end{equation}
     for $0 < |\hbar|\ll 1$, where $X^\pm_\gamma$ denotes the function before and after the jump.
    
    \item[(ii)] \emph{Asymptotics at 0.} 
    As $\hbar\rightarrow 0$, the leading asymptotics are controlled by $Z(\gamma)$:
    \begin{equation}
   e^{Z(\gamma)/\hbar} \, X_{\gamma}(\hbar)\rightarrow \xi(\gamma)
    \end{equation}
    
    \item[(iii)] \emph{Growth at $\infty$.} 
    The function $X_\gamma$ grows at most polynomially as $\hbar \rightarrow \infty$.
\end{itemize}
\end{prob}

This type of Riemann-Hilbert problem has been considered in Gaiotto-Moore-Neitzke's works \cite{GMN08, GMN09, GMN12}, where the {\em Fock-Goncharov coordinates} of the moduli space of framed local system solve the jump condition. In this context, the solution was used to construct semi-explicitly the hyperkahler metric on a Hitchin moduli space associated to the given theory. As is mentioned in \cite[\S 7]{Bri19}, the BPS Riemann-Hilbert problem is also closely related to the Stokes structure of the {\em Voros symbols} in the theory of exact WKB analysis of a Schr\"odinger-type ODE discussed in \cite{IN14, IN15}. 
See \cite{I16, All18, Kuw20} for a relation between Fock-Goncharov coordinates and Voros symbols, and see also \cite{All19} for the further development in this direction. In our context, the BPS structure for the BPS Riemann-Hilbert problem arises from the meromorphic quadratic differential which appears in the classical limit of the ODE. 

In the rest of this section, we will focus on how the BPS structures are constructed from a given meromorphic quadratic differential. We will discuss in greater detail the corresponding BPS Riemann-Hilbert problem (with a precise formulation) together with its relation to the Voros coefficient of the quantum curves in the sequel to this paper.


\subsection{WKB spectral networks (Stokes graphs)}

We recall the notion of a \emph{WKB spectral network} \cite{GMN12} associated to a quadratic differential\footnote{A more general construction can be made for arbitrary tuples of meromorphic $k$-differentials, but it is more complicated. We only consider quadratic differentials in this paper, so we omit the details.} on a compact Riemann surface $X$, also called a \emph{Stokes graph} in WKB literature (e.g., \cite{KT98}). 

\subsubsection{\bf Trajectories of meromorphic quadratic differentials}

Fix a meromorphic quadratic differential $\varphi = Q(x)dx^2$ on $X$ with at least one pole of order greater than one. 
We also assume that $\varphi$ has only simple zeros, and has at least one zero or one simple pole (i.e., we assume that $\varphi$ is a GMN differential in the sense of \cite[Definition 2.1]{BS13}). 
We will also use the same notations ($\Sigma, P, T, {\rm Crit}$ etc.) for the notions used in \S \ref{section:hypergeometric-curvs-structure}. 

\begin{dfn}
For a given $\vartheta \in {\mathbb R}$, the equation 
\begin{equation}
\mathrm{Im}\, e^{-i \vartheta} \int^{x}{ \sqrt{Q(x)} dx} = {\rm constant} 
\label{eq:trajectories}
\end{equation}
defines a foliation $\mathcal{F}_{\vartheta}$ on $X \setminus {\rm Crit}$.
A {\em trajectory of phase $\vartheta$} is any maximal leaf of the foliation $\mathcal{F}_{\vartheta}$. 
The trajectories of phase $\vartheta=0$ are called \emph{horizontal trajectories}, and trajectories of phase $\vartheta=\pi/2$ are called \emph{vertical trajectories}. 
\end{dfn}

{ 
At any point on $X \setminus {\rm Crit}$, there exists a distinguished local coordinate $w$ defined (up to the sign) by 
\begin{equation} \label{eq:distinguished-coordinate}
w(x) := \int^{x} \sqrt{Q(x)} \, dx,
\end{equation}
and trajectories are pullbacks of straight lines of phase $\vartheta$ in the $w$-plane by \eqref{eq:distinguished-coordinate}. 
In other words, trajectories of phase $\vartheta$ are the curves along which the integrand of \eqref{eq:distinguished-coordinate} has constant phase; $\sqrt{Q(x)}\, dx = e^{i \vartheta} |\sqrt{Q(x)}|\, |dx|$.
}
The general structure of trajectories of quadratic differentials is described in detail in the classic book \cite{St84} by K. Strebel (see also \cite{BS13}). 
Note that two non-overlapping trajectories of the same phase may never intersect; furthermore, any two trajectories of phases $\vartheta_1,\vartheta_2$ that intersect will always do so at an angle of $\vartheta_2-\vartheta_1$ for some choice of orientation.

It is known that every trajectory falls into exactly one of the following types
(\cite{St84}; see also \cite[\S 3.4]{BS13}):
\begin{enumerate}[i)]
\item 
{\em saddle trajectories} approache finite critical points at both ends;
\item 
{\em separating trajectories} approach critical points at each end, 
one finite and one infinite;
\item 
{\em generic trajectories} approach infinite critical points at both ends;
\item
{\em closed trajectories} are simple closed curves in $X \setminus {\rm Crit}$; 
\item 
{\em recurrent trajectories} are ``recurrent" in at least one direction.
\end{enumerate}

Since $X = {\mathbb P}^1$ and the number of poles of $\varphi$ are at most three in our examples, the Jenkins three pole theorem guarantees that recurrent trajectories do not appear (c.f., \cite{Jenkins}, \cite[Theorem 15.2]{St84}), so we omit their precise definition.

We will be primarily interested in trajectories with at least one end approaching a turning point (i.e., saddle and separating trajectories). We call these \emph{critical trajectories}.

It turns out the local structure of trajectories around any given point in $X$ can be classified fully. Here we summarize these local structures, following \cite{St84}. If $x$ is not a critical point of $\varphi$, we call it a \emph{regular point}. We consider a zero of order $k$ to be a critical point of order $k$, a pole of order $k$ 
a critical point of order $-k$, and may regard a regular point as a ``critical point" of order $0$.
Then, we have the following normal forms:

\begin{prop}[\cite{St84}]
Around any critical point (or regular point) of order $k$, there exists a neighbourhood $U$ and a holomorphic coordinate $z$ such that locally on $U$, 
\begin{equation}
\varphi 
=\begin{cases}
    z^k dz^2, \quad k \geq -1 \text{ or } k \text{ odd }\\
    \frac{r}{z^2}dz^2, \quad k=-2 \\ 
    (z^{k/2}+\frac{s}{z})^2 dz^2, \quad k <-2 \text{ and } k \text{ even}\quad 
    \end{cases}
\end{equation}
where $r \in \mathbb{C}^\ast$, $s\in \mathbb{C}$.
\end{prop}

It is easy to see the structure of the trajectories of these normal forms, some examples of which are depicted in Figures \ref{fig:turningpoints}, \ref{fig:secondorder} and \ref{fig:higherorder}. We depict the behaviour of both horizontal (green) and vertical (blue) foliations, and draw the critical trajectories in black. At poles of order greater than $k >2$, all horizontal trajectories asymptote to one of $k-2$ \emph{asymptotic directions} (which are themselves trajectories), which we depict in red.

\begin{figure}[h]
    \centering
    \begin{subfigure}[t]{.37\textwidth}
        \centering
        \includegraphics[width=\linewidth]{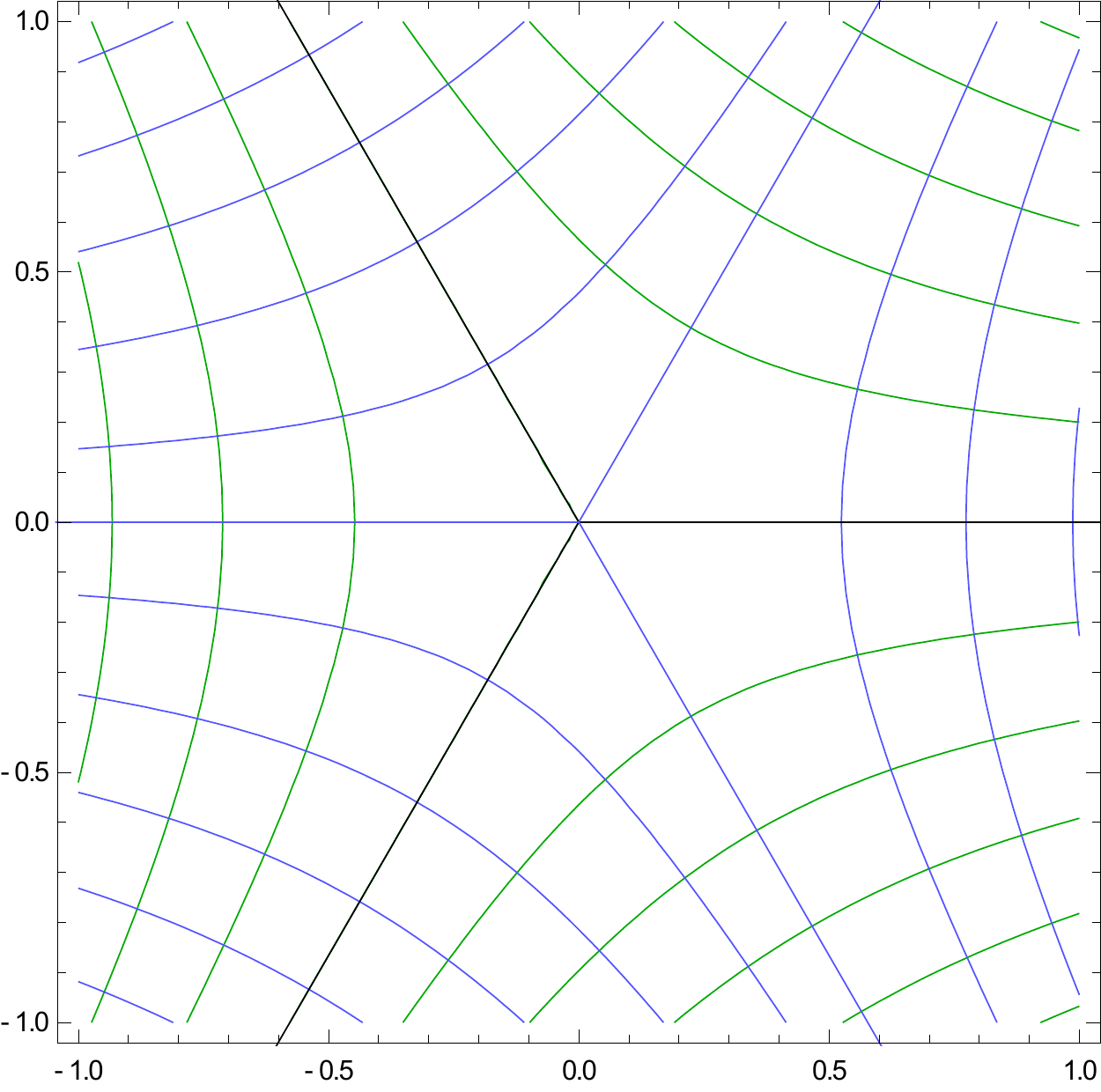}
           \caption{Local behaviour around a simple zero, where three critical trajectories emanate. $Q(x)=x$ is plotted here.}
      \label{fig:simplezero}
    \end{subfigure}
    \hspace{1.5cm}                  
    \begin{subfigure}[t]{.37\textwidth}
        \centering
        \includegraphics[width=\linewidth]{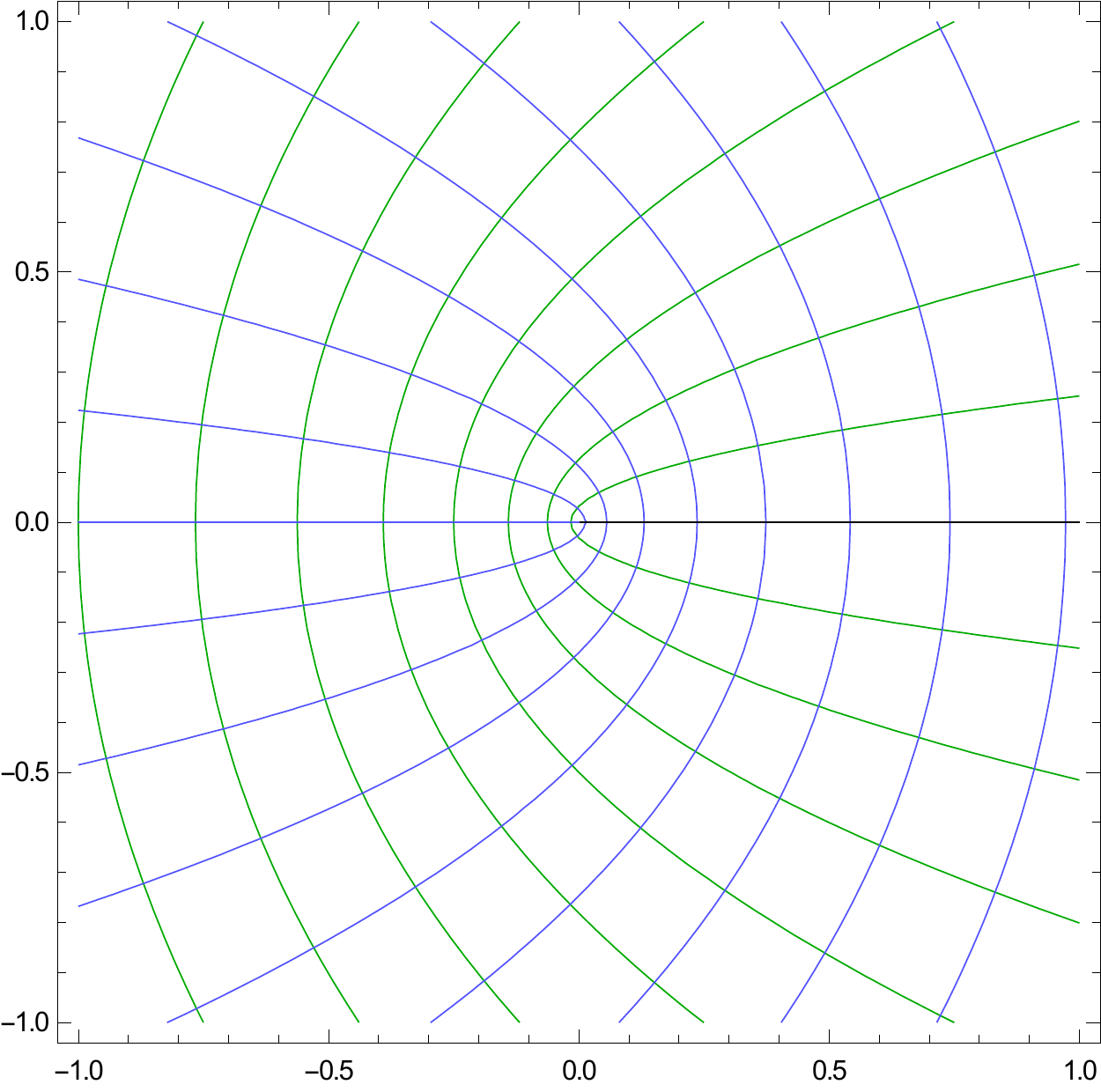}
        \caption{Trajectories around a simple pole, where a single critical trajectory emanates. $Q(x)=1/x$ is plotted here.}
      \label{fig:simplepole}
    \end{subfigure}
    \caption{Trajectory structure around simple zeroes and simple poles of $Q(x)$. Green and blue curves are horizontal and vertical trajectories respectively, the black curve are critical trajectories.}
    \label{fig:turningpoints}
\end{figure}

\begin{figure}[h]
    \centering
    \begin{subfigure}[t]{.20\textwidth}
        \centering
        \includegraphics[width=\linewidth]{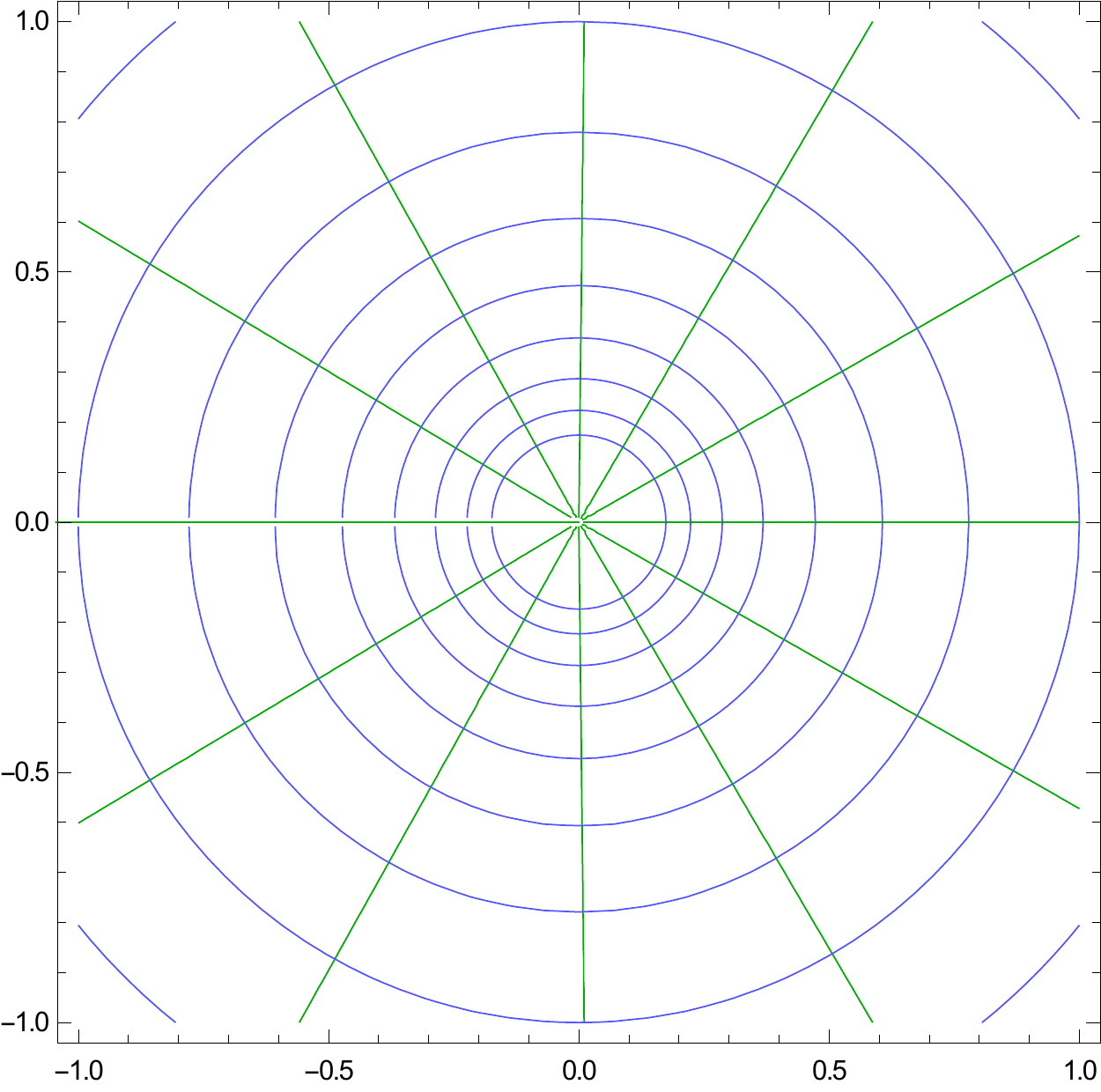}
           \caption{$r \in \mathbb{R}_{>0}$}
      \label{fig:secondorder1}
    \end{subfigure}
    \hfill
    \begin{subfigure}[t]{.20\textwidth}
        \centering
        \includegraphics[width=\linewidth]{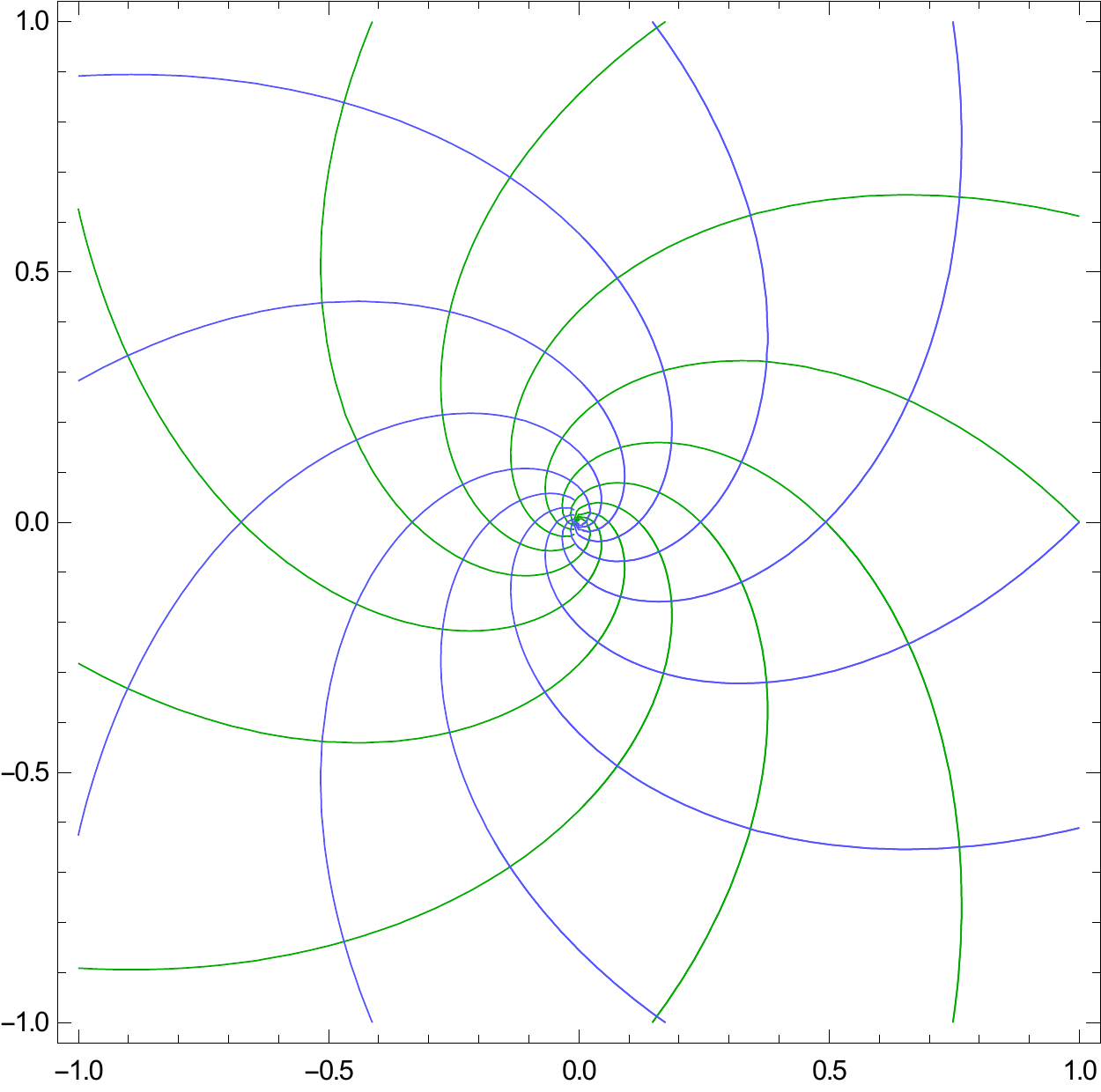}
           \caption{$r \in i\mathbb{R}_{>0}$}
      \label{fig:secondorder2}
    \end{subfigure}
    \hfill
    \begin{subfigure}[t]{.20\textwidth}
        \centering
        \includegraphics[width=\linewidth]{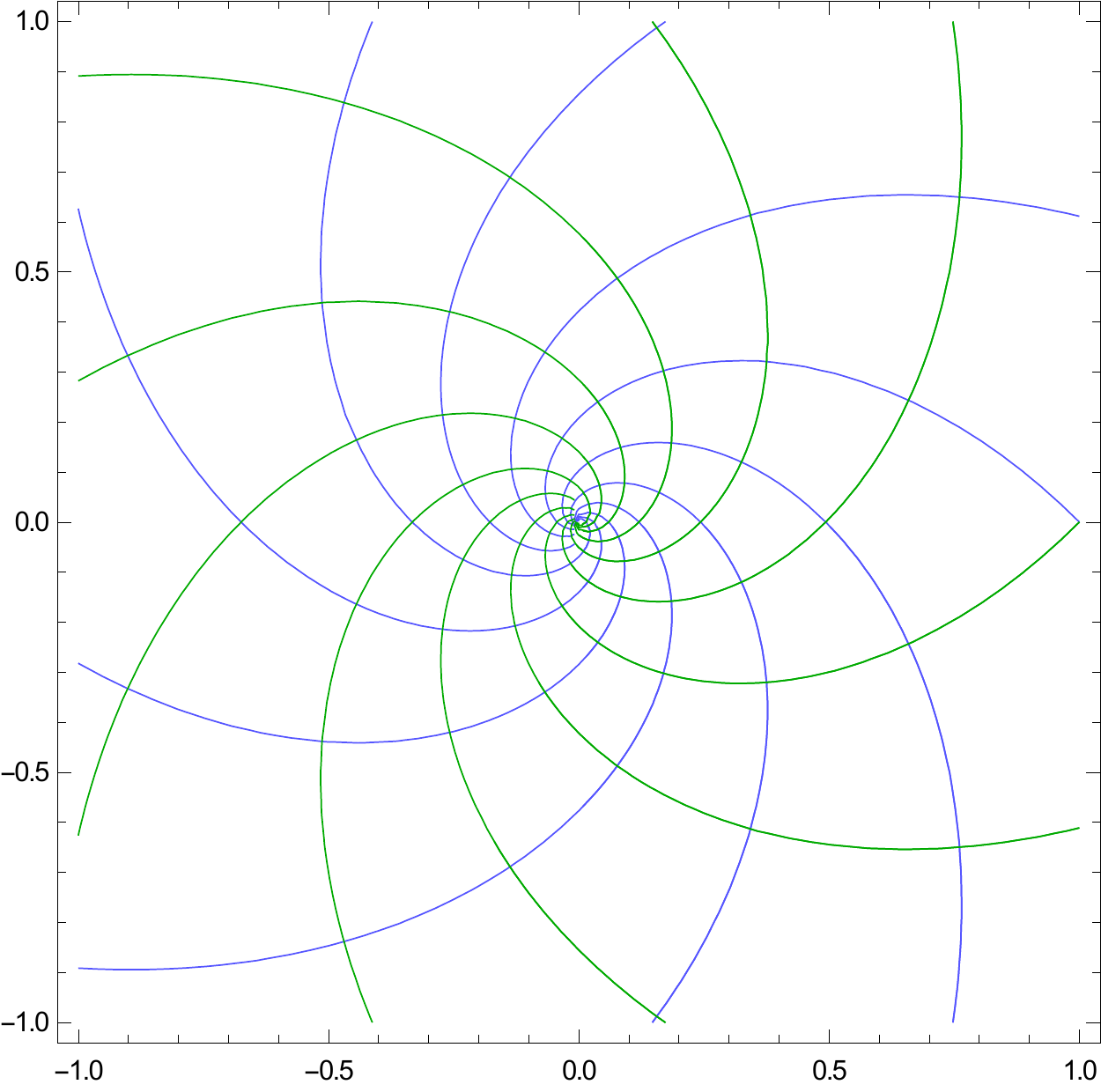}
           \caption{$r \in -i\mathbb{R}_{>0}$}
      \label{fig:secondorder3}
    \end{subfigure}
    \hfill
       \begin{subfigure}[t]{.20\textwidth}
        \centering
        \includegraphics[width=\linewidth]{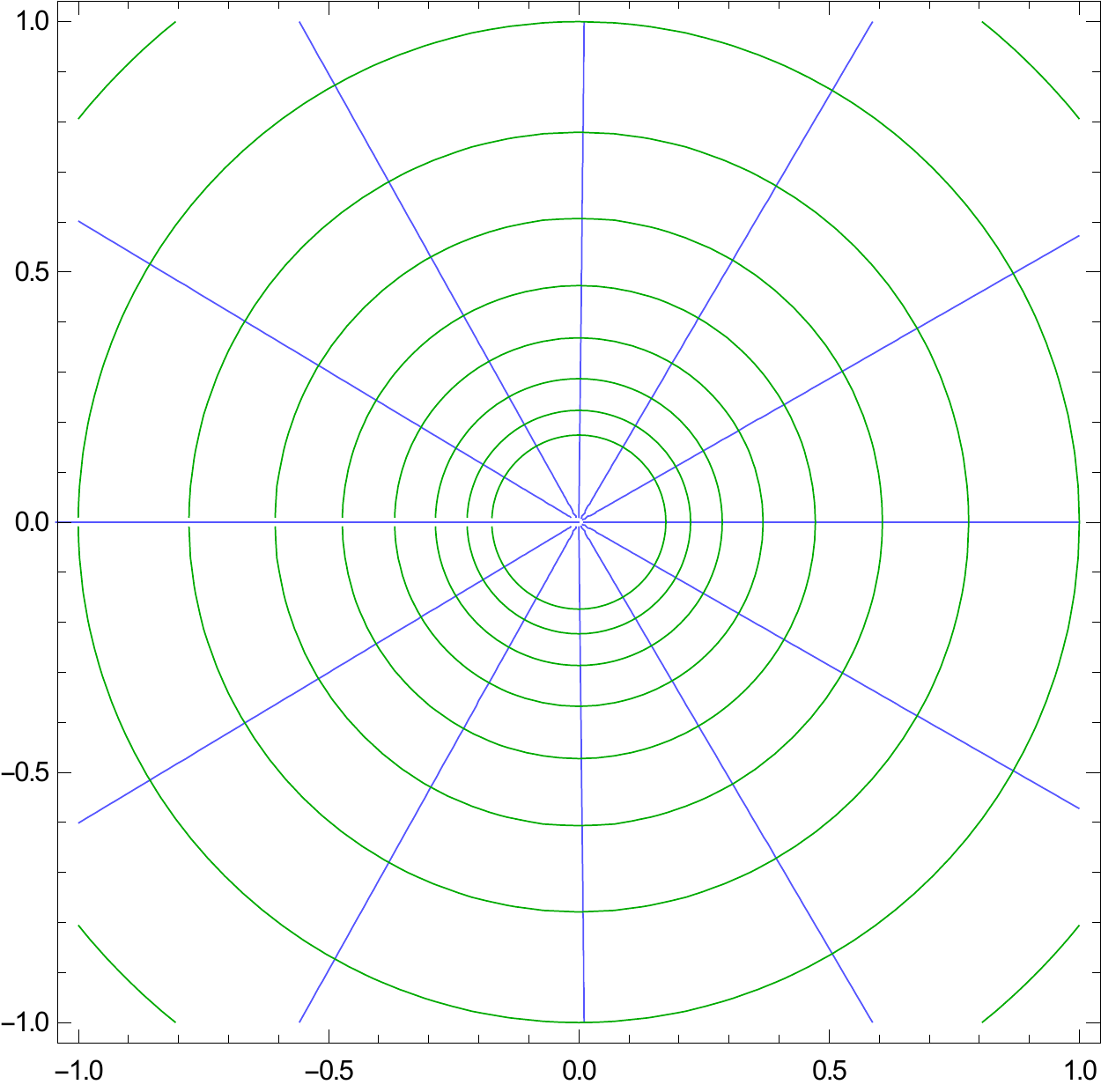}
           \caption{$r\in -\mathbb{R}_{>0}$}
      \label{fig:secondorder4}
    \end{subfigure}
            \caption{Trajectory structures around a second order pole.}
    \label{fig:secondorder}
\end{figure}

\begin{figure}[h]
    \centering
    \begin{subfigure}[t]{.43\textwidth}
        \centering
        \includegraphics[width=\linewidth]{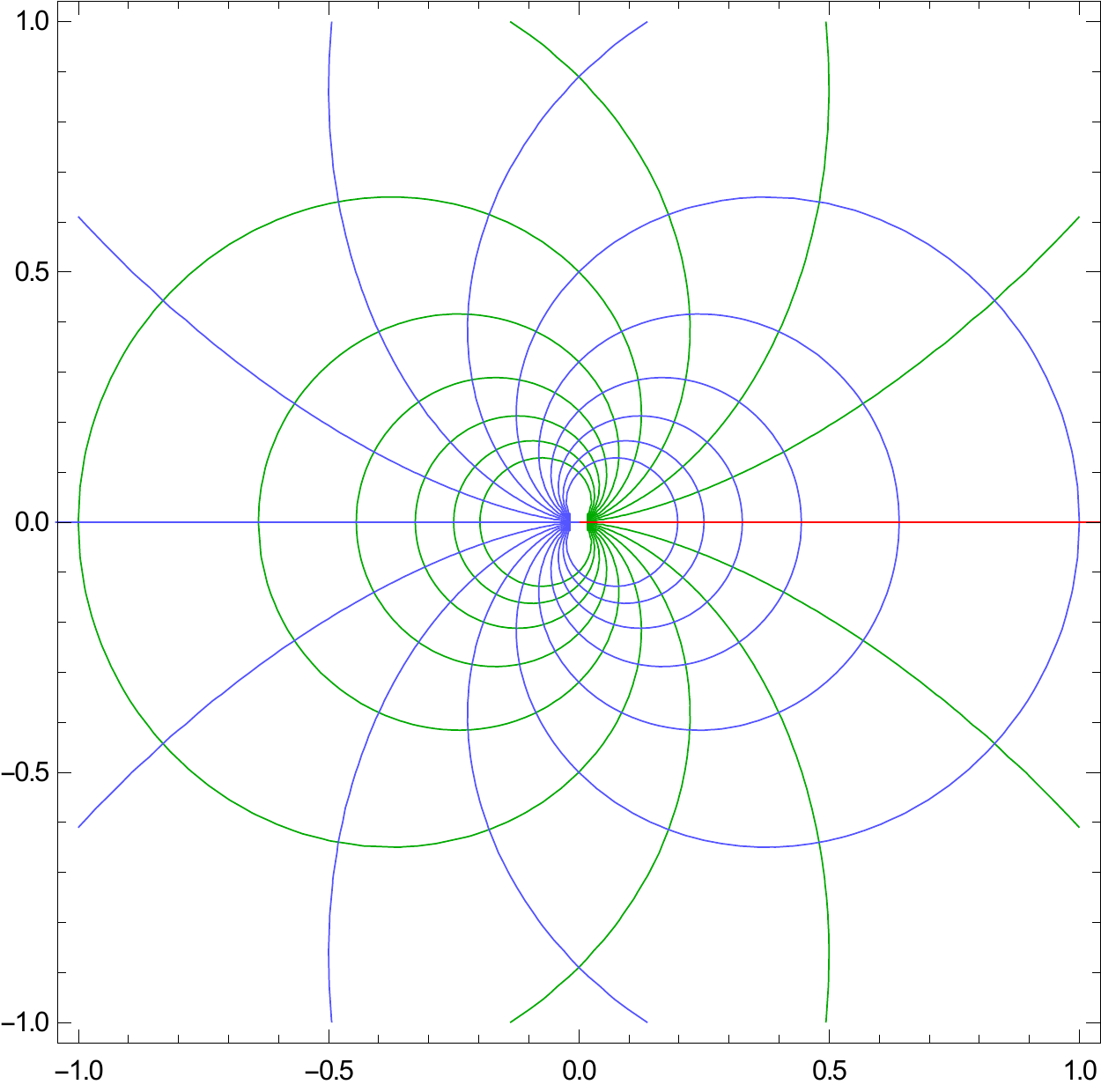}
           \caption{Trajectory structure around a pole of order $3$. The single asymptotic (horizontal) direction is shown in red.}
      \label{fig:thirdorder}
    \end{subfigure}
    \hspace{1.5cm}                  
    \begin{subfigure}[t]{.43\textwidth}
        \centering
        \includegraphics[width=\linewidth]{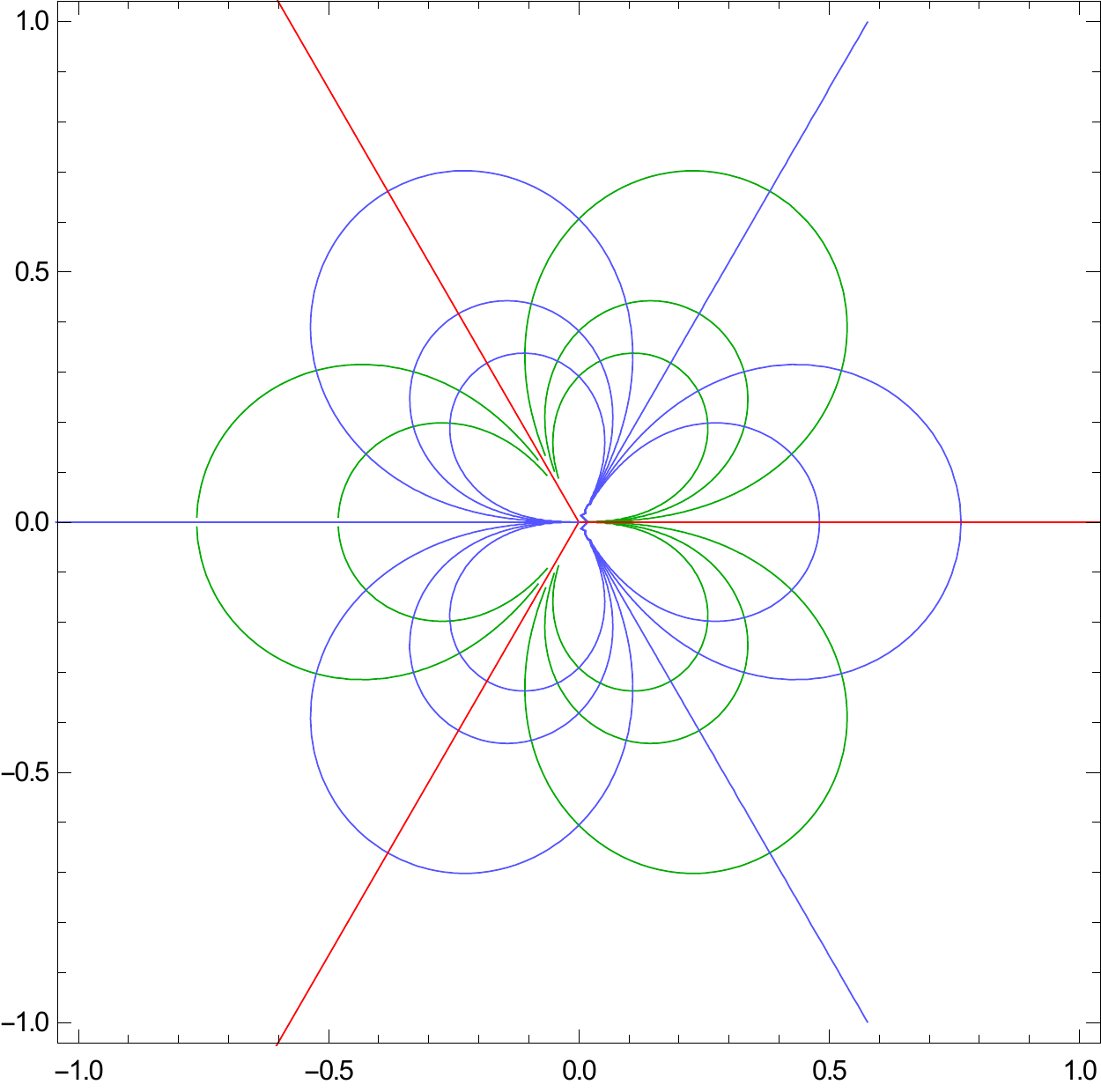}
        \caption{Trajectory structure around a pole of order $5$. The three asymptotic (horizontal) directions are shown in red.}
      \label{fig:fifthorder}
    \end{subfigure}
    \caption{Trajectory structure around higher order poles.}
    \label{fig:higherorder}
\end{figure}

Given a quadratic differential $\varphi$ on a Riemann surface $X$, call a \emph{$\varphi$-polygon} any polygon whose sides are trajectories (of possibly different phases $\vartheta$) and whose vertices are points in $X$ (so they may be poles, zeroes, or regular points). Then

\begin{prop}[Teichm\"uller's lemma {\cite[Theorem 14.1]{St84}}]
Let $(X,\varphi)$ be a compact Riemann surface equipped with a meromorphic quadratic differential $\varphi$.
For any $\varphi$-polygon with vertices $\{v_k\}$, denote by $\beta_k$ the following quantity associated to a vertex $v_k$:
\begin{equation}
    \beta_k:= 1 - \frac{\theta_k}{2\pi} (n_k+2)
\end{equation}
where $\theta_k$ is the interior angle of the $\varphi$-polygon at the vertex and $n_k$ is the order of the (critical or regular) point $v_k$. Then we have
\begin{equation}
\sum_{i \in \mathrm{verts}}{\beta_i} = 2 + \sum_{j \in \mathrm{int}}{n_j}
\end{equation}
where the sum on the left goes over all vertices of the $\varphi$-polygon, and the sum on the right over all (critical or regular) points on the interior of the $\varphi$-polygon.
\label{lem:teich}
\end{prop}

We will use the result in determining whether a given set of curves can occur as trajectories or not.

\subsubsection{\bf Definition of spectral networks and its properties}

\begin{dfn}
For any fixed $\vartheta \in {\mathbb R}$, we define
the \emph{(WKB-) spectral network} \emph{$\mathcal{W}_\vartheta(\varphi)$ 
of phase $\vartheta$} as the subset of $X$ which consists of all critical trajectories of $\varphi$ of phase $\vartheta$. 
\end{dfn}

The spectral network agrees with the notion of \emph{Stokes graph} in the WKB literature (c.f., \cite{KT98}), and the critical trajectories are called {\em Stokes curves} (these turn out to be the locus where the Borel resummed WKB solutions have a discontinuity). 
Note that we have
\begin{equation} \label{eq:rotation-of-spectral-network}
\mathcal{W}_\vartheta(\varphi) = \mathcal{W}_0(e^{-2 i\vartheta}\varphi), 
\end{equation} 
and hence, varying $\vartheta$ is equivalent to staying at a fixed 
$\vartheta$ and moving in a certain family of quadratic differentials. 
Thus, when we discuss properties of spectral networks, 
we may assume $\vartheta = 0$ without loss of generality.

It is well known that the each connected component of the complement
$X \setminus {\mathcal W}_{0}(\varphi)$ of the spectral network 
are one of the following (c.f., \cite[\S 3.4]{BS13}):
\begin{enumerate}[i)]
\item 
A {\em half plane} is equivalent to a domain
\begin{equation}
\{ w \in {\mathbb C} ~|~ {\rm Im} \, w > c \}
\end{equation}
for some $c \in {\mathbb R}$
equipped with the quadratic differential $dw^2$, 
through the map \eqref{eq:distinguished-coordinate}. 
A half plane only appears around a pole of $\varphi$ of order $\ge 3$; 
there are always $k-3$ half planes around a pole of order $k$.
Its boundary consists of saddle and separating trajectories.

\item
A {\em horizontal strip} is equivalent to a domain
\begin{equation} \label{eq:horizontal-strip}
\{ w \in {\mathbb C} ~|~ c_1 < {\rm Im} \, w < c_2 \}
\end{equation} 
for some $c_1, c_2 \in {\mathbb R}$ ($c_1 < c_2$)
equipped with the quadratic differential $dw^2$, 
through the map \eqref{eq:distinguished-coordinate}. 
Its boundary consists of saddle and separating trajectories.

\item 
A {\em ring domain} is a domain consisting of any point 
$x \in X \setminus {\rm Crit}$ such that the trajectory 
passing through $x$ is a closed trajectory. 
It is equivalent to $\{ z \in {\mathbb C} ~|~ c_1 < |z| < c_2 \}$
for some $c_1, c_2 \in {\mathbb R}$ ($a < b$) 
equipped with the quadratic differential $r dz^2 / z^2$ 
for some $r \in {\mathbb C}^\ast$. 
We call a ring domain {\em degenerate} if $a = 0$, and 
{\em nondegenerate} otherwise.
The boundary of a ring domain consists of unions of saddle trajectories, or saddle trajectories and isolated points when the ring domain is degenerate. 

\item 
A {\em spiral domain} is defined to be the interior of 
the closure of a recurrent trajectory.
\end{enumerate}

{ 
A degenerate ring domain has two boundary components; a (chain of) saddle trajectory and a single point which must be a second order pole $s$ of $\varphi$ due to the Teichm\"uller's lemma (Proposition \ref{lem:teich}).  
It follows from the definition of trajectories that, if a degenerate ring domain appears in the spectral network ${\mathcal W}_\vartheta (\varphi)$ around a second order pole $s$, then 
\begin{equation} \label{eq:loop-condition}
e^{- i \vartheta}\Res_{x=s} \sqrt{Q(x)} \, dx \in i  {\mathbb R}_{\ne 0}
\end{equation}
must be satisfied. Conversely, if $s$ is a double pole of $\varphi$ satisfying the condition \eqref{eq:loop-condition}, then $s$ is one of the boundary component of a degenerate ring domain (c.f., \cite[\S 3.4]{BS13}). 
}

For our examples, we can conclude the following:
\begin{lemm} \label{lemm:no-recurrent-trajectory}
Let $\varphi_\bullet = Q_\bullet(x) \, dx^2$ 
be the one of the quadratic differentials in Table \ref{table:classical}. 
Then, the nondegenerate ring domain and spiral domain never appear 
in the complement of the spectral networks 
${\mathcal W}_\vartheta(\varphi)$ for any $\vartheta \in {\mathbb R}$.
\end{lemm}

\begin{proof}
If there is a non-degenerate ring domain, its complement consists of 
two domains which we denote $D_1$ and $D_2$. 
Since the quadratic differentials in our examples have at most 
two turning points, both of $D_1$ and $D_2$ must contain poles 
so that the total pole orders in each $D_1$ and $D_2$ is three 
due to Teichm\"uller's lemma (Proposition \ref{lem:teich}). 
This cannot happen in our examples in Table \ref{table:classical}. 

On the other hand, Jenkins' three pole theorem guarantees that recurrent trajectories do not appear in our examples, so we can conclude that the spiral domains never appear either.
\end{proof}

\subsection{Central charge}
\label{sec:centralcharge}
In constructing a BPS structure, we must define the lattice $\Gamma$ as well as the central charge $Z$. 
We define them following \cite{GMN09, BS13}. Recall in \S\ref{section:hypergeometric-curvs-structure} we defined a branched double cover $\pi: \Sigma \rightarrow X $, the \emph{spectral cover} $\Sigma$ associated to $\varphi$ as 
\begin{equation}
    \Sigma = \left \{\lambda \in T^*X \; | \; \lambda^2 - \varphi=0 \right\} \subset T^*X
\end{equation}
and we denoted $\widetilde{\Sigma}$ as the spectral cover with simple poles filled in.

First, we define the central charge $Z(\gamma)$ for any $\gamma \in H_1(\widetilde{\Sigma},\mathbb{Z})$. It is given by the period integral of $\sqrt{\varphi}$: 
\begin{equation}
Z(\gamma) := \oint_{\gamma} \sqrt{Q(x)} \, dx.
\end{equation}
Then we take $\Gamma$ as the sublattice 
of $H_1(\widetilde{\Sigma}, {\mathbb Z})$ given by anti-invariant cycles,
\begin{equation}
\Gamma := \{ \gamma \in H_1(\widetilde{\Sigma}, {\mathbb Z}) 
~|~ \sigma_\ast \gamma = - \gamma \}
\end{equation}
equipped with the intersection pairing $\langle \cdot , \cdot \rangle : \Gamma \times \Gamma \to {\mathbb Z}$. 
The lattice $\Gamma$ is called the hat-homology group in \cite{BS13}. 
We then define the central charge as the restriction of $Z$ to $\Gamma$, and continue to denote it by the same letter.


Suppose we have a horizontal strip $D$ in the spectral network ${\mathcal W}_0(\varphi)$ whose boundary consists of only separating trajectories. 
For such $D$, we associate a homology class $\gamma_D \in H_1(\widetilde{\Sigma}, {\mathbb Z})$, which we call the {\em dual cycle}\footnote{This is called a {\em standard saddle class} in \cite{BS13}}, as follows. 
First, in the description \eqref{eq:horizontal-strip} of $D$, we take the straight line $l$ on the $w$-plane connecting the images of turning points on the different sides of the strip.
$l$ is realized as a path on $X$ connecting two (possibly the same) turning points lying on the boundary of $D$, and its pullback by $\pi: \widetilde{\Sigma} \to X$ defines a closed cycle, up to its orientation. 
Then, the dual cycle $\gamma_D$ is defined (up to the sign) to be the homology class represented by the cycle (see Figure \ref{fig:lift} (A)).
The dual cycle $\gamma_D$ is an element in $\Gamma$.

\begin{figure}[h]
    \centering
    \begin{subfigure}[t]{.35\textwidth}
        \centering
        \includegraphics[width=0.85\linewidth]{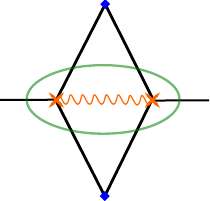}
           \caption{The dual cycle $\gamma_D$ associated with 
           a horizontal strip $D$.}
      \label{fig:gammadiamond}
    \end{subfigure}
    \hspace{2cm}
    \begin{subfigure}[t]{.35\textwidth}
        \centering
        \includegraphics[width=0.85\linewidth,trim={0cm -0.3cm 0cm 0cm}]{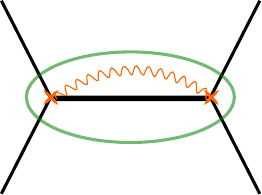}
        \caption{The BPS cycle $\gamma_{\rm BPS}$ associated with 
        a saddle trajectory.}
      \label{fig:sadd}
    \end{subfigure}
    \caption{Elements in $\Gamma$ obtained from a horizontal strip 
    or a saddle trajectory. These homology classes are defined 
    up to sign (orientation of the representatives).}
    \label{fig:lift}
\end{figure}

\begin{lemm}
Suppose ${\mathcal W}_{0}(\varphi)$ is non-degenerate. 
Then, the collection $\{ \gamma_D \}$ of dual cycles for all horizontal strips in ${\mathcal W}_{0}(\varphi)$ forms a basis of $\Gamma$. 
Moreover, for each dual cycle $\gamma_D$, the central charge $Z(\gamma_D)$ has a non-zero imaginary part. 
\label{lem:diamondlemma}
\end{lemm}
\begin{proof}
See for example \cite[Section 3.6]{BS13}.
\end{proof}

\subsection{Saddle trajectories and BPS indices}

For generic $\vartheta$, saddle trajectories are absent and 
the spectral network $\mathcal{W}_\vartheta(\varphi)$ consists only
of separating trajectories of phase $\vartheta$. 
We say the spectral network $\mathcal{W}_\vartheta(\varphi)$  
\emph{degenerate} if it contains a saddle trajectory of phase $\vartheta$. 
These degenerate spectral networks are of paramount importance in this paper and many applications. In the physics of 4d $\mathcal{N}=2$ QFTs, they correspond to BPS states in the spectrum of the theory \cite{GMN08, GMN12}. 
From a mathematical point of view, they correspond to stable objects in a 3-Calabi-Yau category associated with a quiver with potential determined by $\varphi$ \cite{BS13}.

The saddle trajectories can be classified into the following five types: 
\begin{itemize}
\item 
A {\em type I} saddle connects two distinct simple zeros of $\varphi$. 
\item 
A {\em type II} saddle connects a simple zero and a simple pole of $\varphi$.
\item 
A {\em type III} saddle connects two distinct simple poles of $\varphi$.
\item 
A {\em type IV} saddle is a closed curve which forms a boundary of 
a degenerate ring domain.
\item 
A {\em type V} saddle is a closed curve which forms one of 
the boundary components of a non-degenerate ring domain.
\end{itemize}

Since the type IV and type V saddles are closed curves, we simply call them {\em loop-type saddles}. Lemma \ref{lemm:no-recurrent-trajectory} guarantees that, among the loop-type saddles, type V saddles never appear in the spectral network ${\mathcal W}_{\vartheta}(\varphi)$ associated with the quadratic differentials $\varphi = \varphi_\bullet$ in Table \ref{table:classical}. 
Hence the loop-type saddles discussed in this paper are always of type IV, and they appear in the spectral network ${\mathcal W}_{\vartheta}(\varphi)$ if and only if $\varphi$ has a second order pole $s$ satisfying \eqref{eq:loop-condition}.

Given a degenerate spectral network with a saddle trajectory, we can associate a homology class $\gamma_{\rm BPS} \in H_1(\widetilde{\Sigma}, {\mathbb Z})$ represented by the closed cycle (up to orientation) on $\widetilde{\Sigma}$ obtained as the pullback by $\pi : \widetilde{\Sigma} \to X$ of the saddle trajectory (see Figure \ref{fig:lift} (B)). 

In particular, if the degeneration is a type IV saddle around a second order pole $s$ of $\varphi$, then the associated class is $\gamma_{\rm BPS} = \gamma_{s_\pm} - \gamma_{s_\mp}$, where $\gamma_{s_\pm}$ are the residue cycles around the distinct preimages $s_\pm \in \overline{\Sigma}$ of $s$ by the projection map $\pi: \overline{\Sigma} \to X$.
We regard that the cycle $\gamma_{\rm BPS}$ thus obtained is associated with the degenerate ring domain itself (i.e., not the loop-type saddle). 
This viewpoint will be important when we discuss the Legendre example in \S \ref{section:Legendre}.

The cycles $\gamma_{\rm BPS}$ obtained in this manner are also anti-invariant under 
the action induced by the covering involution of $\widetilde{\Sigma}$ (that is, $\gamma_{\rm BPS} \in \Gamma$; see \cite[\S 3.2]{BS13}).

\begin{dfn}
We refer to homology classes $\gamma_{\rm BPS} \in \Gamma$ obtained from saddle trajectories or ring domains as above as \emph{BPS cycles}, and the collection of all BPS cycles as the \emph{BPS spectrum}.
\end{dfn}

Since the real-valued function ${\rm Re} \left( e^{- i \vartheta} \int^{x} \sqrt{Q(x)} \, dx \right)$ is monotone along trajectories of phase $\vartheta$, the central charge $Z(\gamma_{\rm BPS})$ for BPS cycles never vanish. 
By definition, if a BPS cycle $\gamma_{\rm BPS}$ is associated with 
a saddle trajectory in ${\mathcal W}_{\vartheta}(\varphi)$, 
then the phase must satisfy $\vartheta \equiv \arg Z(\gamma_{\rm BPS})$ 
mod $\pi$, and we sometimes say that $\gamma_{\rm BPS}$ 
``appears at the phase $\arg Z(\gamma_{\rm BPS})$".
Note that if $\gamma_{\rm BPS}$ is a BPS cycle appearing at 
$\vartheta \, (\mathrm{mod}\, 2\pi)$, then $-\gamma_{\rm BPS}$ is also 
a BPS cycle and appears at 
$\vartheta = \vartheta+\pi \, (\mathrm{mod}\, 2\pi)$ since 
${\mathcal W}_\vartheta(\varphi) = {\mathcal W}_{\vartheta+\pi}(\varphi)$ 
holds. 
Thus, when determining the BPS spectrum, we may safely restrict 
our attention to those $\gamma_{\rm BPS}$ appearing in 
the range $\vartheta \in [0,\pi)$.

\subsection{Construction of BPS structures} \label{subsection:constructing-BPS-str}
So far, we have summarized a several properties of trajectories of general meromorphic quadratic differentials and spectral networks. We now restrict our attention to the examples in Table \ref{table:classical}, that is, the ``hypergeometric type" quadratic differentials which appear in the WKB analysis of the Gauss hypergeometric differential equation and its confluent degenerations, and explain how BPS structures are constructed from these examples. 
In what follows, the symbol $\bullet$ denotes any of: $\rm{HG}$, $\rm{dHG}$, $\rm{Kum}$, $\rm{Leg}$, $\rm{Bes}$, $\rm{Whi}$, $\rm{Web}$, $\rm{dBes}$, or $\rm{Ai}$. 
Let $\varphi_\bullet({\bm m}) = Q_\bullet(x)dx^2$ denote the corresponding quadratic differential as in Table \ref{table:classical}, for a given value of the parameters ${\bm m} \in M_\bullet$ satisfying Assumption \ref{ass:genericity}. 
We also denote by $\Gamma_\bullet$ and $Z_\bullet$ the lattice and the central charge defined from $\varphi_\bullet$ in \S\ref{sec:centralcharge}.
We sometimes write the central charge as $Z_{\bm m}(\gamma)$ when we emphasize the dependence on the mass parameter ${\bm m}$.

Since $\widetilde{\Sigma}_\bullet$ is a sphere with several punctures, 
the homology group $H_1(\widetilde{\Sigma}_\bullet, {\mathbb Z})$ is generated 
by the {\em residue classes} $\gamma_a$ 
(i.e., the class represented by a positively oriented small circle) around the puncture 
$a \in \overline{\Sigma}_\bullet \setminus \widetilde{\Sigma}_\bullet = D_{\bullet, \infty}$. 
They satisfy
\begin{equation} \label{eq:relation-among-cycles}
\sum_{s \in P_{\bullet, \rm od} \cap D_{\bullet, \infty}} \gamma_s + \sum_{s \in P_{\bullet, \rm ev}} (\gamma_{s_+} + \gamma_{s_-}) = 0.
\end{equation}
See \eqref{eq:sign-convention-preimages} for the sign convention, and recall that we identify a point in $P_{\bullet, \rm od}$ with its unique preimage in $\overline{\Sigma}_\bullet$. 
We note that the intersection pairing $\langle \cdot, \cdot \rangle$ on $\Gamma_\bullet$ is identically $0$ since the residue cycles do not intersect each other.

\label{lemm:miniversal}


In some part of our discussion, we will require that all our parameters ${\bm m}$ in this paper satisfy an additional genericity condition

\begin{dfn} \label{def:generic-locus}
We say the parameter ${\bm m}$ is {\em generic} if it lies on the complement $M_{\bullet}\setminus W_\bullet$ of the set $W_\bullet$ defined by
\begin{equation}
W_\bullet = \{ {\bm m} \in M_\bullet ~|~ 
\text{there exist BPS cycles $\gamma, \gamma'$ with $\gamma \pm \gamma' \ne 0$ and
$Z_{\bm m}(\gamma)/Z_{\bm m}(\gamma') \in {\mathbb R}_{\ne 0}$} \}
\end{equation}
for $\bullet \ne {\rm Leg}$, and $W_{\rm Leg} = \emptyset$ for the Legendre case. 
Otherwise, we say ${\bm m}$ is \emph{non-generic}. We denote by $M_{\bullet}' = M_{\bullet}\setminus W_\bullet$ the locus of generic parameters.
\end{dfn}

This genericity condition implies, in particular, that no two BPS cycles appear at the same phase $\vartheta$ except for the Legendre case\footnote{
The Legendre case is exceptional; that is, the degenerate spectral network in the Legendre case always contains both a type III saddle and a degenerate ring domain simultaneously (see \S \ref{section:Legendre} below).   
}. 
In particular, we do not have any ring domain whose boundary component consists of a chain of type I saddle trajectories (e.g., the ``eyeball" in Figure \ref{fig:loop2}) on the generic locus.

\begin{figure}[h]
    \centering
    \begin{subfigure}[t]{.30\textwidth}
        \centering
        \includegraphics[width=\linewidth,trim={0cm -0.2cm 0cm 0cm}]{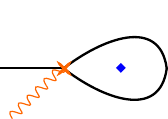}
           \caption{Type IV saddle.}
      \label{fig:loop1}
    \end{subfigure}
    \hspace{1.5cm}             
    \begin{subfigure}[t]{.39\textwidth}
        \centering
        \includegraphics[width=\linewidth,trim={0cm -0.3cm 0cm 0cm}]{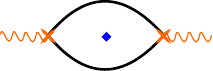}
        \caption{The forbidden ``eyeball" degeneration.}
      \label{fig:loop2}
    \end{subfigure}
    \caption{Degenerations around a second order pole.}
    \label{fig:loops}
\end{figure}

\noindent Thus, together with the results from above, the following holds in all our examples:

\begin{prop}
If ${\bm m} \in M_\bullet'$, then, for any $\vartheta$, a ring domain appearing in the spectral network ${\mathcal W}_{\vartheta}(\varphi_\bullet)$ must be a degenerate ring domain. 
For any second order pole $s$ of $\varphi_\bullet$, a degenerate ring domain appears around $s$ if and only if $\vartheta = \arg{m_s} + \pi/2$ (mod $\pi$) holds, and then, the BPS cycle associated to the degenerate ring domain is given (up to sign) by 
\begin{equation}
\gamma_{\rm BPS} = \gamma_{s_{\pm}} - \gamma_{s_{\mp}},
\end{equation}
where $s_{\pm} \in \overline{\Sigma}_\bullet$ are the preimages of $s$.
\label{prop:loops}
\end{prop}

    

%
%

{ 
Now we introduce the BPS indices which partially generalize the ones considered in \cite{GMN09, GMN12, BS13} by including the contributions from BPS cycles associated with trajectories occurring in the presence of simple poles.  

\begin{dfn}
For each ${\bm m} \in M'_\bullet$, we define the {\em BPS indices} $\{ \Omega(\gamma) \}_{\gamma \in \Gamma}$ as a collection of integers defined as follows\footnote{Altohugh we have Lemma \ref{lemm:no-recurrent-trajectory}, we keep the last line in \eqref{eq:def-of-BPS-indices} with the general case in mind.}: 
\begin{equation} \label{eq:def-of-BPS-indices}
\Omega(\gamma)= \begin{cases} 
      +1 & \quad \text{if $\gamma$ is a BPS cycle associated with a type I saddle,} \\
      +2 & \quad \text{if $\gamma$ is a BPS cycle associated with a type II saddle,} \\
      +4 & \quad \text{if $\gamma$ is a BPS cycle associated with a type III saddle}, \\
      - 1 & \quad \text{if $\gamma$ is a BPS cycle associated with a degenerate ring domain}, \\
      - 2 & \quad \text{if $\gamma$ is a BPS cycle associated with a non-degenerate ring domain},
   \end{cases}
\end{equation}
and we set $\Omega(\gamma) = 0$ for any non-BPS cycles $\gamma \in \Gamma$.
\end{dfn}

\begin{rem}
We note that so far, this definition may appear unmotivated and designed to make our result hold. However, in the sequel to this paper we show that these values of $\Omega$ are natural from the perspective of the jumping behaviour of the Borel resummed Voros symbols of the corresponding quantum curves. It would be interesting to ask if the $\Omega(\gamma)$ in definition \eqref{eq:def-of-BPS-indices} satisfy the Kontsevich-Soibelman wall-crossing formula, or Gaiotto-Moore-Neitzke's 2d/4d wall-crossing formula in \cite{GMN12}) in general (see \cite[\S 6]{FIMS} for a relevant discussion). It would also be interesting to understand the meaning of these BPS cycles from the physical and representation-theoretic perspectives.
\end{rem}

\begin{rem}
There is a discrepancy between our definition \eqref{eq:def-of-BPS-indices} of BPS indices and the one used in \cite[\S 7]{Bri19}; that is, \cite{Bri19} gives $-2$ not only for type V saddles (or non-degenerate ring domains) but also for type IV saddles (or degenerate ring domains). 
This discrepancy does not affect the results of \cite[\S 7]{Bri19} since the associated BPS cycles from type IV saddles are in the kernel of the intersection pairing; thus there is no contribution in the BPS automorphism, defined only in terms of $\Omega(\gamma)\langle \gamma,\cdot \rangle$, from such BPS cycles.
We decided to use the modified definition (including the contributions from simple poles) since it agrees with both the resurgent perspective and the formula of this paper.
\end{rem}
}

\section{Computation of the BPS spectrum}
\label{sec:maincomputation}

In the following sections we will compute the BPS structures associated to the spectral curves of hypergeometric type. Our approach is partly inspired by \cite{MMT15,AMMT14,AT15,AT16} who showed the existence of saddles for some special cases of quadratic differentials.

While all examples are degenerations of the hypergeometric, we will follow a pedagogical order. We will begin with the simple examples of Weber, Whittaker, and Bessel (together with degenerate Bessel, and Airy), which may be viewed as local models for the behaviour of more complicated ones. We then consider the main example of this paper, the hypergeometric spectral curve, before turning to the behaviour of its confluent degenerations --- the Kummer, degenerate hypergeometric, and Legendre curves.

\subsection{Simple examples -- Weber, Whittaker, and Bessel}
Since we already have the lattice $\Gamma_\bullet$ and the central charge $Z_\bullet$, the remaining task is to compute the BPS indices for our examples. 

Here we compute the BPS indices for three of the simplest nontrivial examples, all of which have exactly one BPS cycle (up to sign).
We will use the same notations for homology classes used in \S \ref{subsection:constructing-BPS-str} in what follows.

\subsubsection{\bf BPS structure from the Weber curve}
The simplest example with a nontrivial degeneration is the Weber curve, corresponding to the quadratic differential $\varphi_{\rm Web} = Q_{\rm Web}(x)dx^2$ where
\begin{equation}
Q_{\rm{Web}}(x)= \dfrac{1}{4}x^2 - m_\infty.
\end{equation} 
Under Assumption \ref{ass:genericity} (i.e., $m_{\infty} \in M_{\rm Web} = \mathbb{C}^\ast$), $\varphi_{\rm Web}$ has two simple zeros at $b_1 := 2\sqrt{m_\infty}$ and $b_2=-b_1$, and a unique pole of order $6$ at $\infty$. 
The spectral cover $\Sigma_{\rm{Web}} (= \widetilde{\Sigma}_{\rm Web})$ is of genus 0 with two punctures $\infty_\pm$, and the $H_1(\Sigma_{\rm Web},\mathbb{Z})$ is generated by the corresponding residue cycles $\gamma_{\infty_{\pm}}$ with the relation $\gamma_{\infty_+} + \gamma_{\infty_-} = 0$. Since the covering involution exchanges $\gamma_{\infty_{\pm}} \mapsto \gamma_{\infty_{\mp}}$, we have $\Gamma_{\rm Web} =  {\mathbb Z} \gamma_{\infty_+}$ $(= {\mathbb Z} \gamma_{\infty_-})$ which coincides with the whole $H_1(\Sigma_{\rm Web},\mathbb{Z})$ in this case.

    \begin{figure}[h]
    \centering
    \begin{subfigure}[t]{.28\textwidth}
        \centering
        \includegraphics[width=\linewidth,trim={5cm 5cm 5cm 5cm},clip]{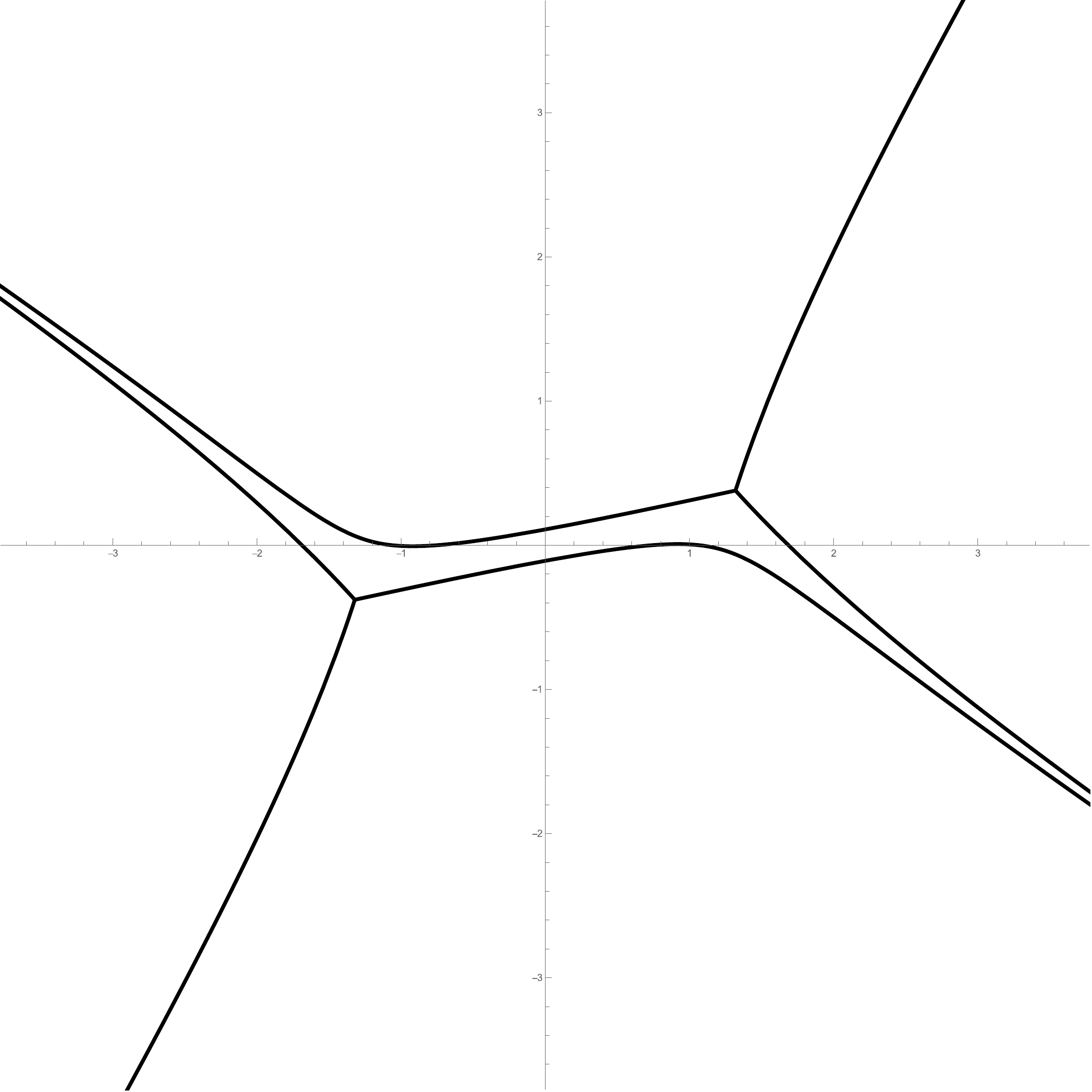}
           \caption{$\vartheta\approx2.03$}
      \label{fig:traj}
    \end{subfigure}
    \hspace{0.75cm}                  
    \begin{subfigure}[t]{.28\textwidth}
        \centering
        \includegraphics[width=\linewidth,trim={5cm 5cm 5cm 5cm},clip]{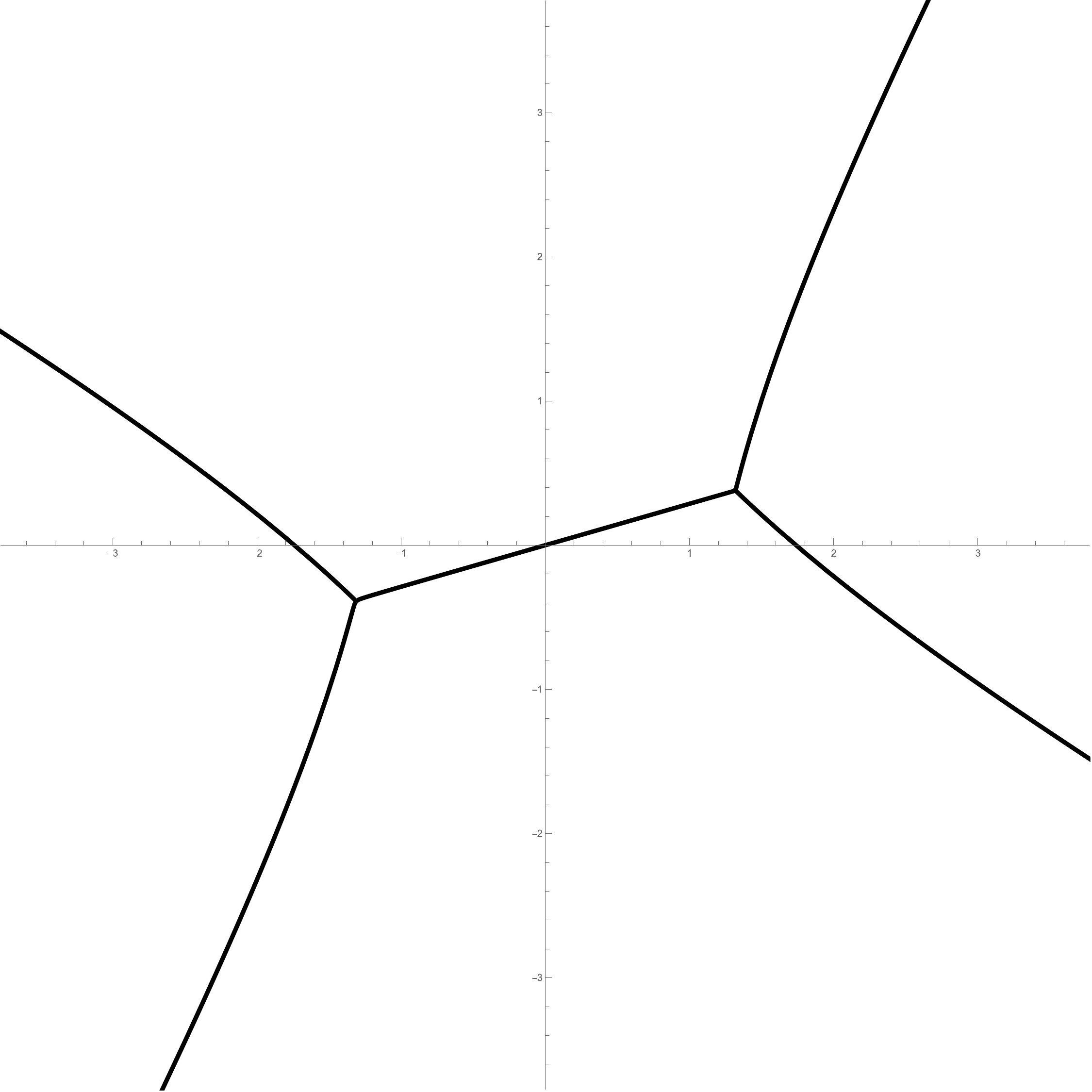}
        \caption{Type I saddle, $\vartheta\approx~2.13$}
      \label{fig:weber2}
    \end{subfigure}
       \hspace{0.75cm}                  
    \begin{subfigure}[t]{.28\textwidth}
        \centering
        \includegraphics[width=\linewidth,trim={5cm 5cm 5cm 5cm},clip]{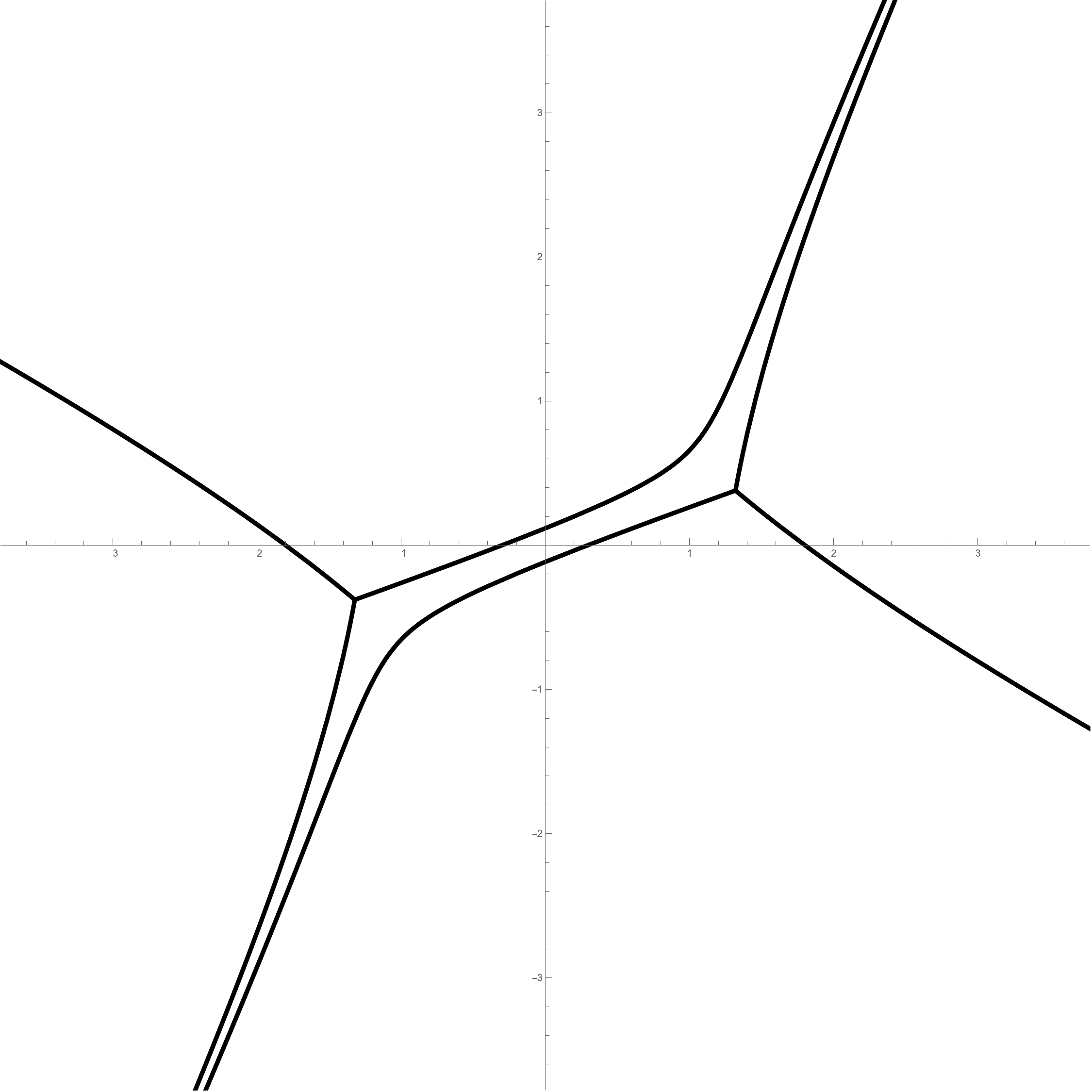}
        \caption{$\vartheta\approx2.23$}
      \label{fig:weber3}
    \end{subfigure}
    \caption{Spectral networks for $\varphi_{\rm Web}$ with $m_\infty \approx 0.4+0.25i$.}
    \label{fig:weber}
\end{figure}

We can draw the spectral network at a generic value $m_\infty \in \mathbb{C}^\ast$, plotted in Figure \ref{fig:weber}, where we can observe a degenerate spectral network, which includes a type I saddle, appearing in Figure \ref{fig:weber2}. 
In this simple setup, we may prove explicitly the following
\begin{prop}  \label{prop:Weber-BPS-structure}
Fix $m_\infty\in M_{\rm Web} = \mathbb{C}^\ast$. 
Then, $\varphi_{\rm Web}$ has exactly one degenerate spectral network in the range $\vartheta \in [0,\pi)$. 
It occurs at $\vartheta = \arg m_\infty+\pi/2$ (mod $\pi$) and contains a type I saddle whose associated BPS cycle is $\gamma_{\infty_\pm}$.
\end{prop}

\begin{proof}
Let $\gamma(t) = (1-t)\,b_1+t\,b_2$ be the straight line connecting $b_1$ and $b_2$. 
It is easy to see that 
\begin{equation}
    \sqrt{Q_{\rm Web}(\gamma(t))} \, d\gamma(t) = \pm 8 i \,m_\infty \sqrt{t(1-t)} \, dt,
\end{equation}
(where $\pm$ depends on the choice of the branch) which has a constant phase $\arg m_\infty + \pi/2$ (mod $\pi$) for any $t \in [0,1]$. 
This means that the straight line $\gamma(t)$ is a geodesic for the metric (defined away from critical points) given by $|{Q_{\rm Web}(x)}|\,|dx|^2$, and hence, it gives a saddle trajectory of phase $\vartheta = \arg m_\infty + \pi/2$. 
Since the saddle trajectory is the unique geodesic among all paths on $X \setminus P_{\rm Web} = {\mathbb C}$ connecting $b_1$ and $b_2$ (c.f., \cite[Theorem 16.2]{St84}), we can conclude that there is no other saddle connection. 
\end{proof}

The central charge is computed by using \eqref{eq:sign-convention-preimages}. 
Thus, we have the BPS structure whose BPS spectrum and BPS indices are summarized in Table \ref{table:bpsangles-Web}. 
Since we have only one BPS cycle (up to sign), the set $W_{\rm Web}$ is empty and we have $M_{\rm Web} = M'_{\rm Web}$. 

\begin{table}[h]
  \begin{tabular}{|c|c|} \hline
    $\vartheta_{\rm BPS}$
    & $\arg{m_\infty} \pm \pi/2$ 
    \\ \hline 
    degeneration
    & type I saddle
     \\ \hline
     $\gamma_{\rm BPS}$ & $\gamma_{\infty_{\pm}}$
     \\ \hline
          $Z(\gamma_{\rm BPS})$ & $\pm 2 \pi i \, m_\infty $
     \\ \hline
    $\Omega(\gamma_{\rm BPS})$ &
    $+1$  \\ \hline
  \end{tabular} 
   \vspace{+1.em}
     \caption{The BPS spectrum of $\varphi_{\rm Web}$. The phase of the BPS ray is denoted by $\vartheta_{\rm BPS}$. }
     \label{table:bpsangles-Web}
  \end{table}

 \subsubsection{\bf BPS structure from the Whittaker curve}
The Whittaker curve $\widetilde{\Sigma}_{\rm Whi}$ is defined by the quadratic differential $\varphi_{\rm Whi} = Q_{\rm Whi}(x) dx^2$ with
\begin{equation}
Q_{{\rm Whi}}(x) = \dfrac{1}{4}-\frac{m_\infty}{x}
\end{equation}
This differential has a simple zero together with a simple pole at 0 and a pole of order 4 at infinity, under the assumption $m_{\infty} \in M_{\rm Whi} = \mathbb{C}^\ast$. 
It is easy to see that $\widetilde{\Sigma}_{\rm Whi}$ is of genus $0$ with two punctures at $\infty_\pm$. 
We can observe that a degenerate spectral network with a type II saddle appears in Figure \ref{fig:whit3}.

   \begin{figure}[h]
    \centering
    \begin{subfigure}[t]{.3\textwidth}
        \centering
        \includegraphics[width=\linewidth]{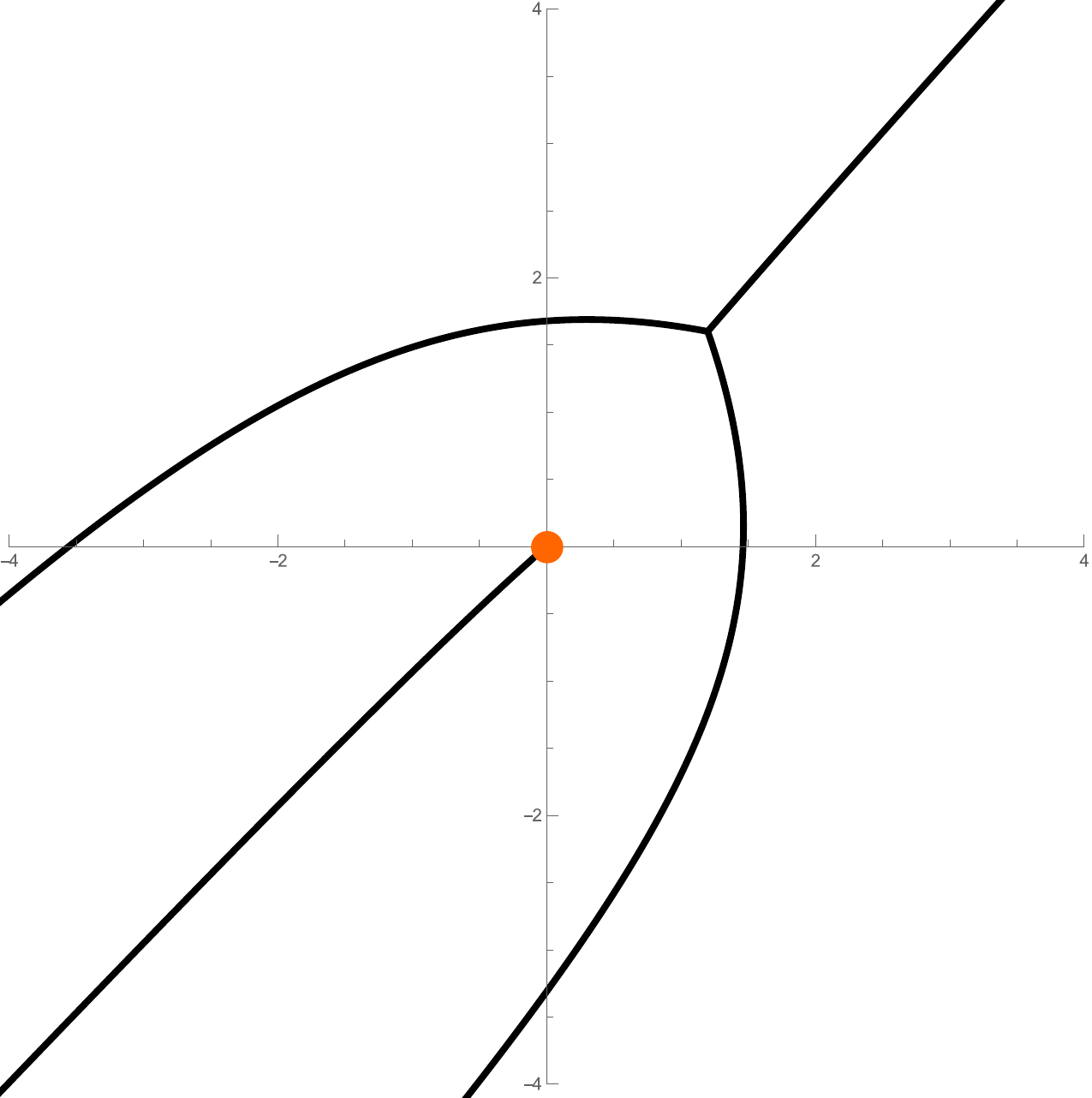}
           \caption{$\vartheta\approx0.82$.}
      \label{fig:whit1}
    \end{subfigure}
    \hspace{0.75cm}              
       \begin{subfigure}[t]{.3\textwidth}
        \centering
        \includegraphics[width=\linewidth]{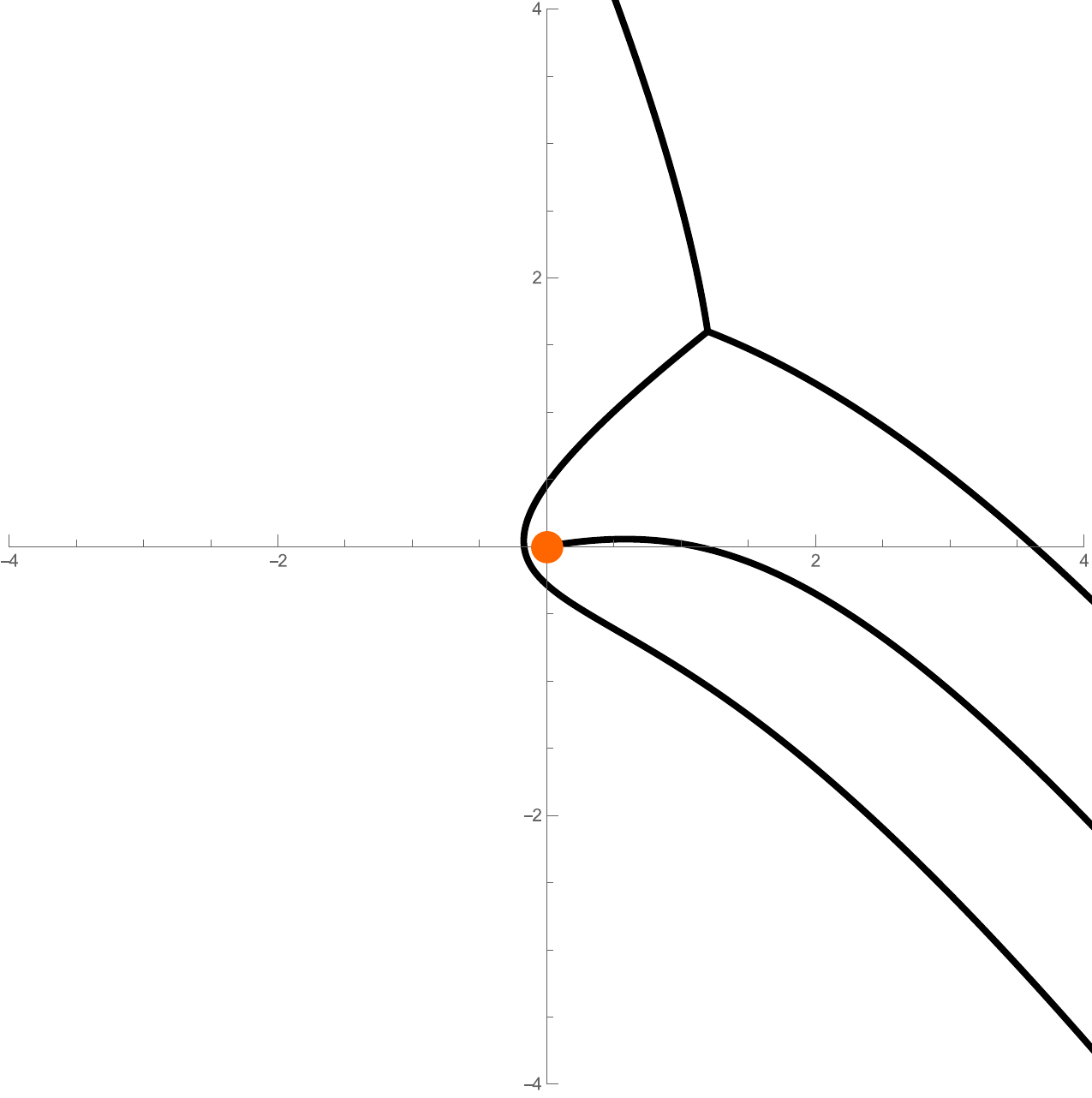}
           \caption{$\vartheta\approx2.12$.}
      \label{fig:whit2}
    \end{subfigure} 
    \\
    \begin{subfigure}[t]{.3\textwidth}
        \centering
        \includegraphics[width=\linewidth]{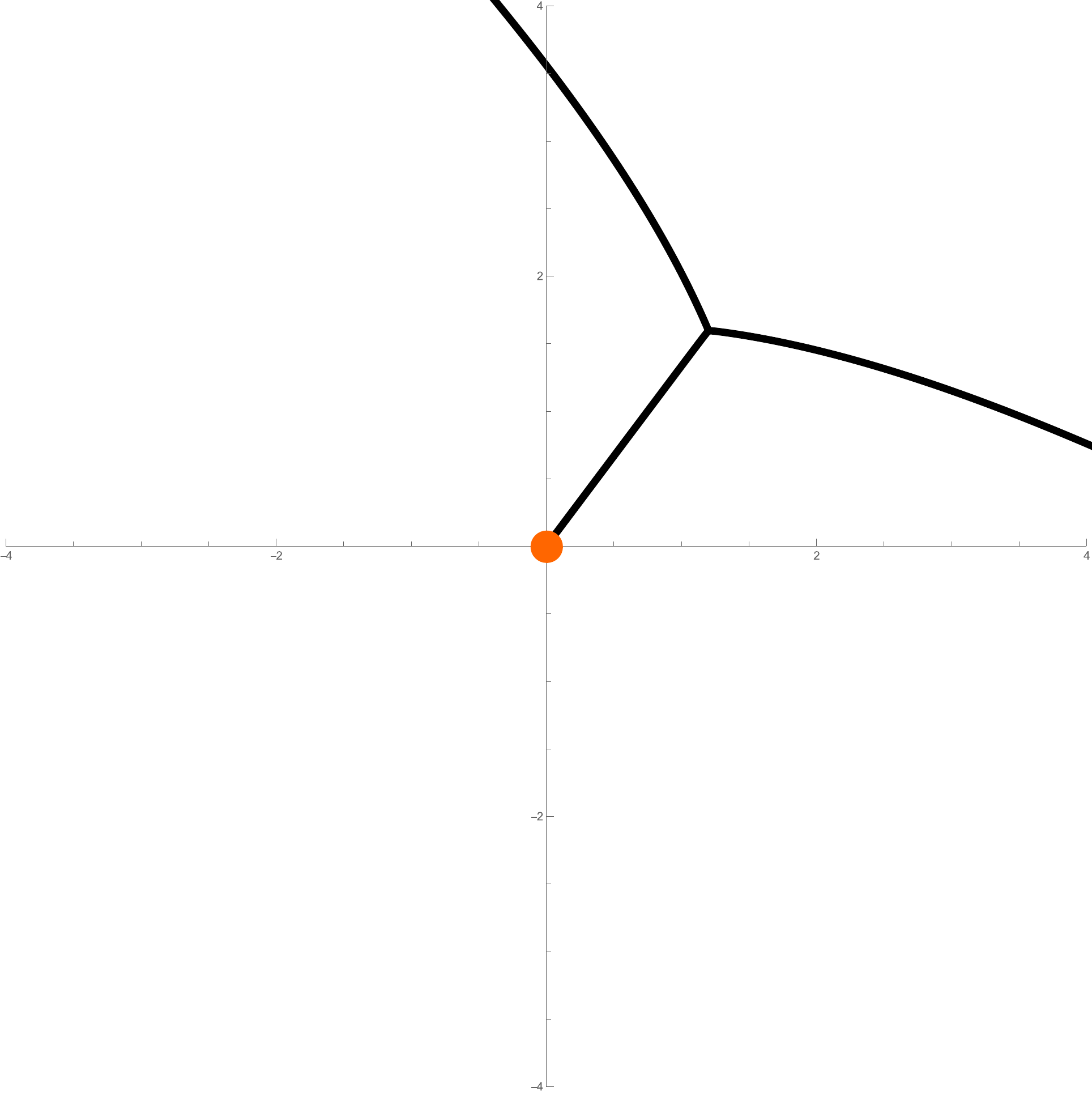}
        \caption{Type II saddle, $\vartheta\approx2.498$.}
      \label{fig:whit3}
    \end{subfigure}
    \hspace{0.75cm}
       \begin{subfigure}[t]{.3\textwidth}
        \centering
        \includegraphics[width=\linewidth]{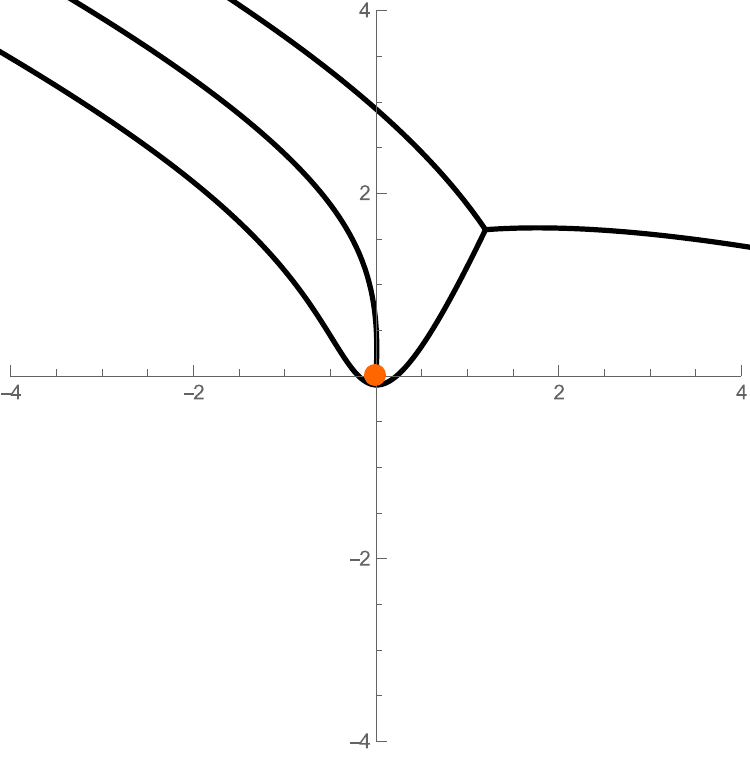}
           \caption{$\vartheta\approx2.79$.}
      \label{fig:whit4}
    \end{subfigure}   
    \caption{Spectral networks for $\varphi_{\rm Whi}$ with $m_0\approx0.3+0.4i$.}
    \label{fig:whit}
\end{figure}

\begin{prop}
Fix $m_\infty \in M_{\rm Whi} = \mathbb{C}^\ast$. Then, $\varphi_{\rm Whi}$ has exactly one degenerate spectral network in the range $\vartheta \in [0,\pi)$. 
It occurs at $\vartheta= \arg m_\infty+\pi/2$ (mod $\pi$) and contains a type II saddle whose associated BPS cycle is $\gamma_{\infty_\pm}$.
\end{prop}
{

\begin{proof}
Suppose that the spectral network ${\mathcal W}_{\vartheta}(\varphi_{\rm Whi})$ with phase $\vartheta$ is nondegenerate; that is, all trajectories are separating. 
Then all four separating trajectories, which we denote by $\gamma_1, \gamma_2, \gamma_3, \gamma_4$,  must terminate at $\infty$. 
Teichmuller's lemma implies that any two separating trajectories emanating from the same simple zero bounding a region without any pole inside must approach $\infty$ with different angles, which must be $\pi$ in this case. 
Therefore, the only possible configuration, topologically, is that of Figure \ref{fig:whitproof} (up to the labeling of the four separating trajectories). 
It contains a unique horizontal strip with dual cycle $\gamma_{\infty_\pm}$, and the associated central charge is $Z(\gamma_{\infty_{\pm}}) = \pm 2 \pi i m_\infty$. 
Therefore, due to Lemma \ref{lem:diamondlemma} (and the equality \eqref{eq:rotation-of-spectral-network}), we must have ${\rm Im} (\pm e^{- i \vartheta} \, 2 \pi i m_\infty) \ne 0$, which is equivalent to $\vartheta \ne \arg m_\infty +\pi/2$ (mod $\pi$), when the spectral network ${\mathcal W}_{\vartheta}(\varphi_{\rm Whi})$ is nondegenerate.

Conversely, if ${\mathcal W}_{\vartheta}(\varphi_{\rm Whi})$ degenerates, then it must contain a saddle trajectory of type II which connects $0$ and $4m_\infty$. 
By a similar argument given in the proof of Proposition \ref{prop:Weber-BPS-structure} for the Weber case, we may prove that the saddle trajectory has phase $\arg m_\infty  +\pi/2$ (mod $\pi$), with the associated BPS cycle $\gamma_{\infty_\pm}$ and central charge $Z(\gamma_{\infty_\pm}) = \pm 2 \pi i m_\infty$. 
\end{proof}

 \begin{figure}[h]
        \centering
        \includegraphics[width=0.24\linewidth,trim={0cm -0.3cm 0cm 0cm}]{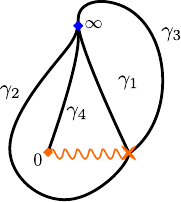}
           \caption{The only possible nondegenerate spectral network for $\varphi_{\rm Whi}$.}
      \label{fig:whitproof}
\end{figure}

}

In summary, we have the BPS structure whose BPS spectrum and BPS indices are listed in Table \ref{table:bpsangles-Whi} ($M_{\rm Whi} = M'_{\rm Whi}$).

\begin{table}[h]
  \begin{tabular}{|c|c|} \hline
    $\vartheta_{\rm BPS}$
    & $\arg m_\infty \pm \pi/2$ 
    \\ \hline 
    degeneration
    & type II saddle
    \\ \hline
    $\gamma_{\rm BPS}$ & $\gamma_{\infty_\pm}$
    \\ \hline
      $Z(\gamma_{\rm BPS})$ &
    $\pm 2 \pi i \, m_\infty$  \\ \hline
    $\Omega(\gamma_{\rm BPS})$ &
    $+2$  \\ \hline
  \end{tabular} 
   \vspace{+1.em}
     \caption{The BPS spectrum of $\varphi_{\rm Whi}$.}
     \label{table:bpsangles-Whi}
\end{table}

 \subsubsection{\bf BPS structure from the Bessel curve}
For $m_0 \in M_{\rm Bes} = {\mathbb C}^\ast$, the quadratic differential $\varphi_{\rm Bes} = Q_{\rm Bes}(x)dx^2$ with
\begin{equation}
Q_{\rm{Bes}}(x)=\dfrac{x+4m_0^2}{4x^2}
\end{equation}
has a single simple zero, a second order pole at the origin, and a pole of order 3 at $\infty$. 
The associated Bessel curve $\Sigma_{\rm{Bes}}$ $(= \widetilde{\Sigma}_{\rm{Bes}})$ is of genus $0$ with three punctures, at $0_{\pm}$ and $\infty$. 

We can draw the spectral network at some chosen value $m_0 \in \mathbb{C}^\ast$. 
The result is a single BPS cycle beginning and terminating at the branch point, projecting to a loop around the origin as in Figure \ref{fig:bessel}. 

    \begin{figure}[h]
    \centering
    \begin{subfigure}[t]{.31\textwidth}
        \centering
        \includegraphics[width=1.05\linewidth,trim={8.75cm 10cm 8.75cm 8.5cm},clip]{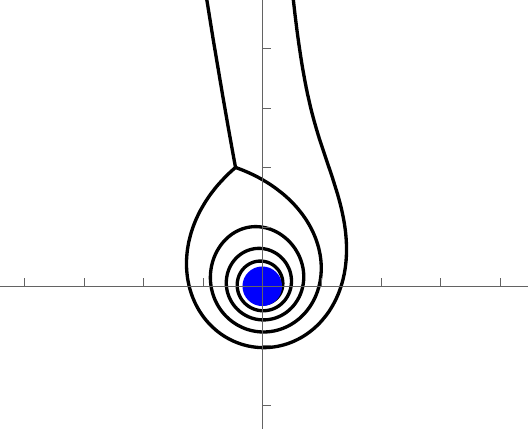}
           \caption{$\vartheta\approx0.84$}
      \label{fig:bessel1}
    \end{subfigure}
    \hspace{0.2cm}                  
    \begin{subfigure}[t]{.31\textwidth}
        \centering
        \includegraphics[width=1.05\linewidth,trim={8.75cm 10cm 8.75cm 8.5cm},clip]{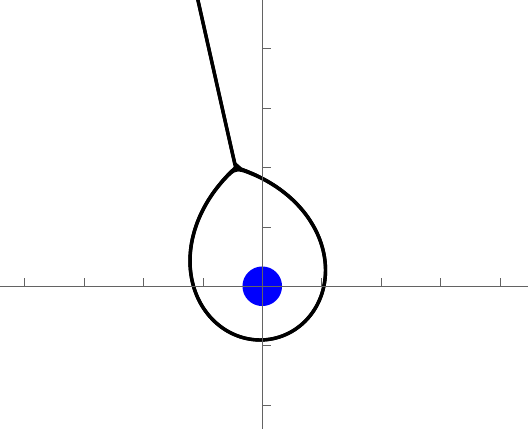}
        \caption{Type IV saddle, $\vartheta\approx0.90$}
      \label{fig:bessel2}
    \end{subfigure}
       \hspace{0.2cm}                  
    \begin{subfigure}[t]{.31\textwidth}
        \centering
        \includegraphics[width=1.05\linewidth,trim={8.75cm 10cm 8.75cm 8.5cm},clip]{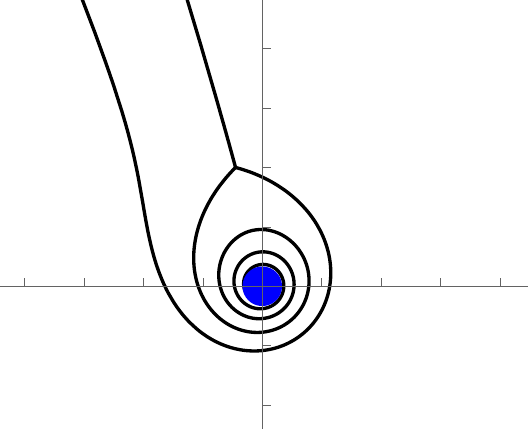}
        \caption{$\vartheta\approx0.96$}
      \label{fig:bessel3}
    \end{subfigure}
    \caption{Spectral networks for $\varphi_{\rm Bes }$ with $m_0 \approx -0.5+0.4i$.}
    \label{fig:bessel}
\end{figure}
 
\noindent We may show
\begin{prop}
Fix $m_0\in M_{\rm Bes}$. Then, $\varphi_{\rm Bes}$ has exactly one degenerate spectral network in the range $\vartheta \in [0,\pi)$, appearing at $\vartheta= \arg m_0 +\pi/2$ (mod $\pi$). At this phase, a degenerate ring domain appears around the origin, and the associated BPS cycle is $\gamma_{0_\pm}-\gamma_{0_\mp}$.
\end{prop}
\begin{proof}
Since there is only one branch point, any saddle trajectory must be of loop type. 
According to Proposition \ref{prop:loops}, a loop can appear around a second order pole if and only if the residue is imaginary, which is equivalent to $\vartheta = \arg m_0 + \pi/2$ (mod $\pi$) in the Bessel case. 
The rest of the claim also follows from Proposition \ref{prop:loops}. 
\end{proof}





In summary, we have the BPS structure whose BPS spectrum and BPS indices are summarized in Table \ref{table:bpsangles-Bes} ($M_{\rm Bes} = M'_{\rm Bes}$). 
We used $Z(\gamma_{0_\pm})= \pm 2\pi i m_0$ to compute the central charge.

\begin{table}[h]
  \begin{tabular}{|c|c|} \hline
    $\vartheta_{\rm BPS}$
    & $\arg m_0 \pm \pi/2$
    \\ \hline 
    degeneration
    & degenerate ring domain
     \\ \hline
     $\gamma_{\rm BPS}$ & $\gamma_{0_\pm} -  \gamma_{0_\mp}$
     \\ \hline
    $Z(\gamma_{\rm BPS})$ &
    $\pm 4 \pi i \, m_0$  \\ \hline
    $\Omega(\gamma_{\rm BPS})$ &
    $-1$  \\ \hline
  \end{tabular} 
   \vspace{+1.em}
    \caption{The BPS spectrum of $\varphi_{\rm Bes}$.}
     \label{table:bpsangles-Bes}
  \end{table}

  \subsubsection{\bf Degenerate Bessel and Airy curves}
  For completeness, we record the trivial cases:
  \begin{equation}
    Q_{\rm Ai}(x) = x , \qquad Q_{\rm dBes}(x) = \frac{1}{x}. 
  \end{equation}
  \begin{prop}
  Let $\bullet = $ $\rm{dBes}$ or ${\rm Ai}$. Then there are no degenerate networks for any $\vartheta \in [0,\pi)$, and the BPS spectrum of $\varphi_\bullet$ is empty.
  \end{prop}
  \begin{proof}
  Both of the spectral covers $\widetilde{\Sigma}_{\rm Ai} (= \Sigma_{\rm Ai})$ and $\widetilde{\Sigma}_{\rm dBes}$ have trivial homology groups. 
  Thus in both cases the lattice and therefore the central charge and $\Omega$, are trivial. 
  \end{proof}

\subsection{Main example -- BPS structure from the Gauss hypergeometric curve}

We turn now to the main example of this paper, the quadratic differential arising from the celebrated Gauss hypergeometric equation. 
It is explicitly given by $\varphi_{\rm HG}({\bm m}) = Q_{{\rm HG }}(x)dx^2$ where
\begin{equation}
Q_{\rm HG}(x) = \frac{m_{\infty}^2 x^2 - (m_\infty^2 - m_1^2 + m_0^2) x + m_0^2}{x^2(x-1)^2}.
\end{equation}
Under the assumption ${\bm m} \in M_{\rm HG}$, it has two simple zeros, and second order poles at $0$, $1$ and $\infty$. 
Thus, the associated Gauss hypergeometric curve $\Sigma_{\rm HG} (= \widetilde{\Sigma}_{\rm HG})$ is of genus 0 with six punctures at $0_\pm$, $1_\pm$ and $\infty_\pm$.

Choosing a generic value for the parameters ${\bm m}$, the degenerations that appear are depicted in Figure \ref{fig:hgsaddles}. 

  \begin{figure}[h!]
    \centering
    \begin{subfigure}[t]{.3\textwidth}
        \centering
        \includegraphics[width=\linewidth,trim={5cm 5cm 5cm 5cm},clip]{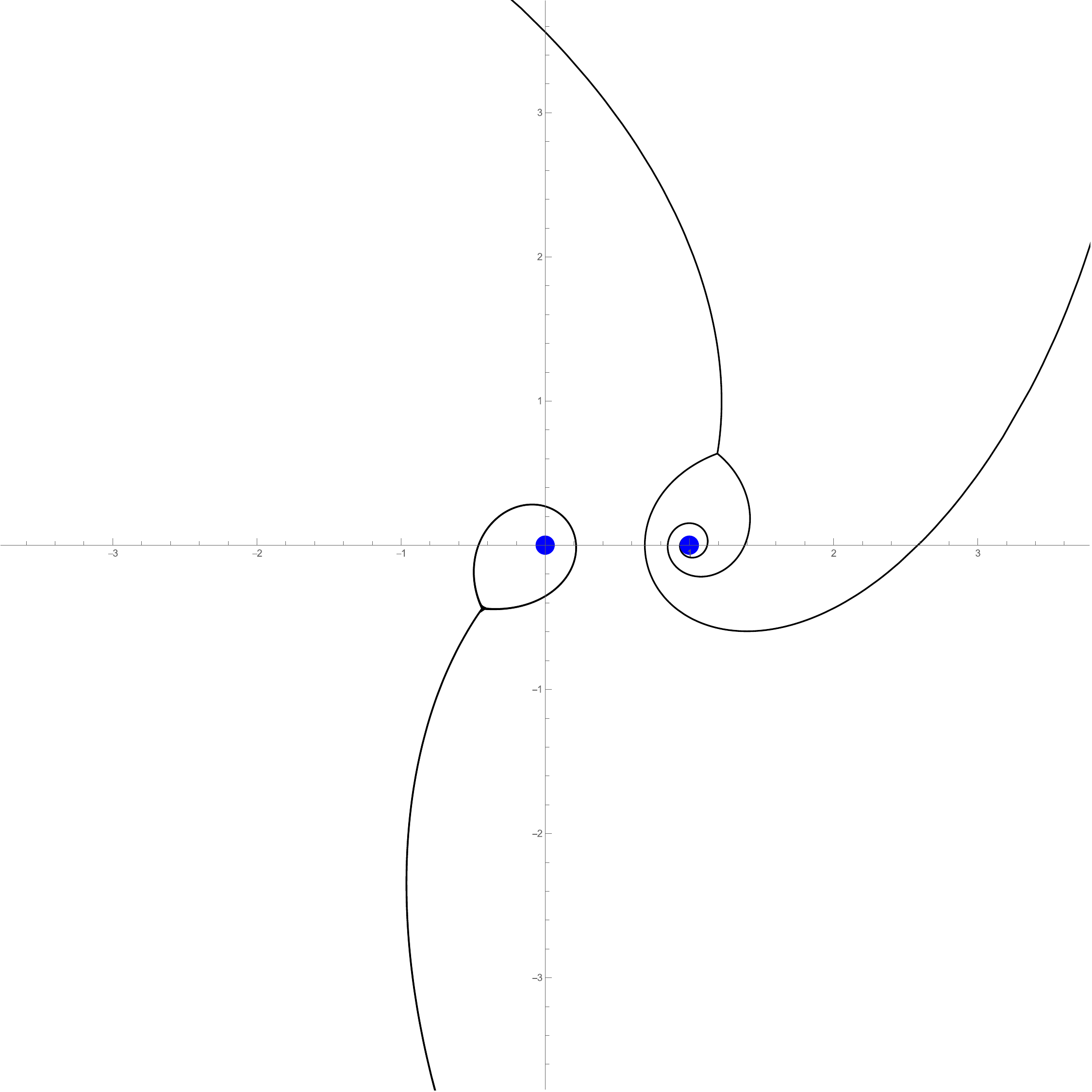}
           \caption{Type IV saddle, $\vartheta\approx0.19$}
      \label{fig:hg1}
    \end{subfigure}       
     \begin{subfigure}[t]{.3\textwidth}
        \centering
        \includegraphics[width=\linewidth,trim={5cm 5cm 5cm 5cm},clip]{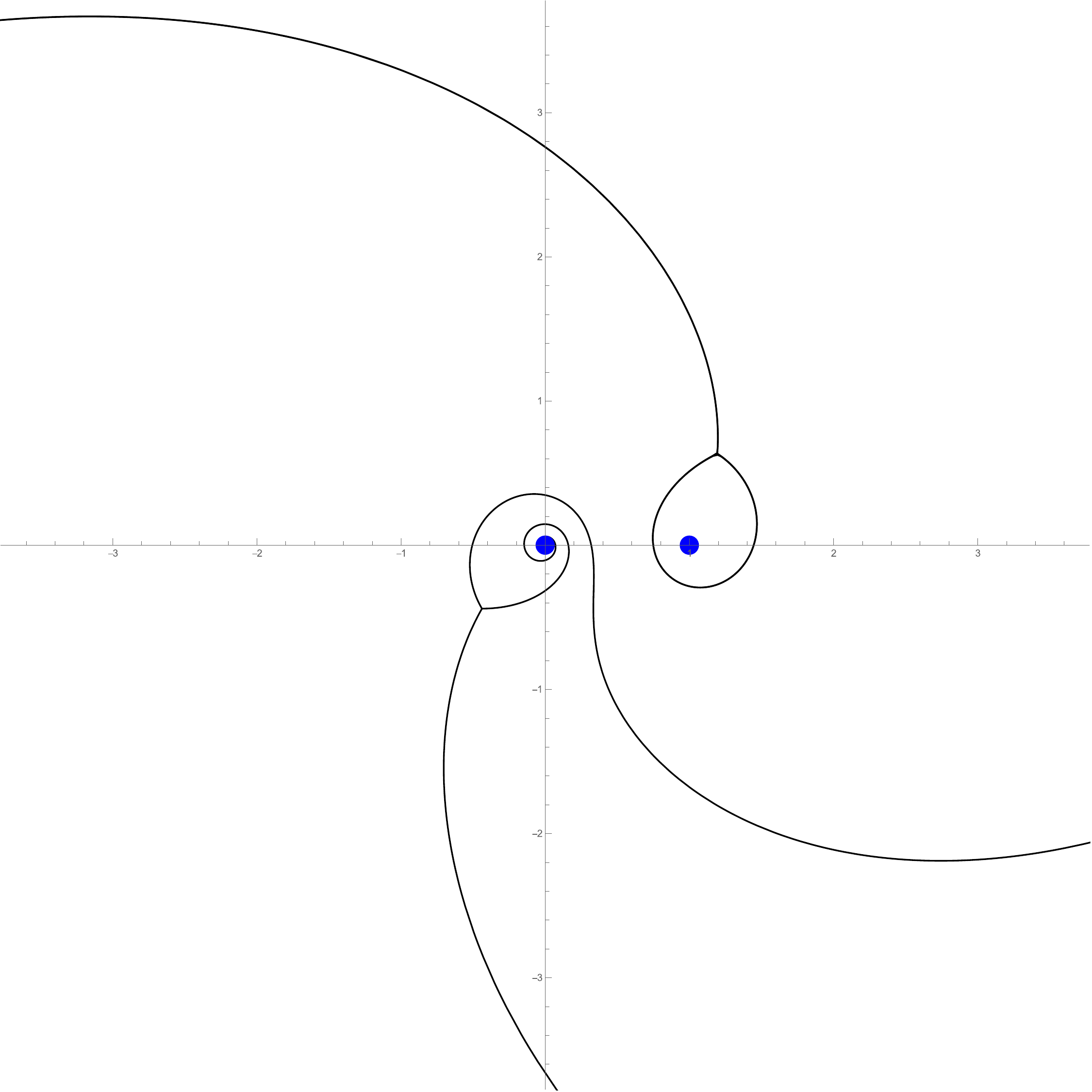}
           \caption{Type IV saddle, $\vartheta\approx0.338$}
      \label{fig:hg3}
    \end{subfigure}
     \begin{subfigure}[t]{.3\textwidth}
        \centering
        \includegraphics[width=\linewidth,trim={5cm 5cm 5cm 5cm},clip]{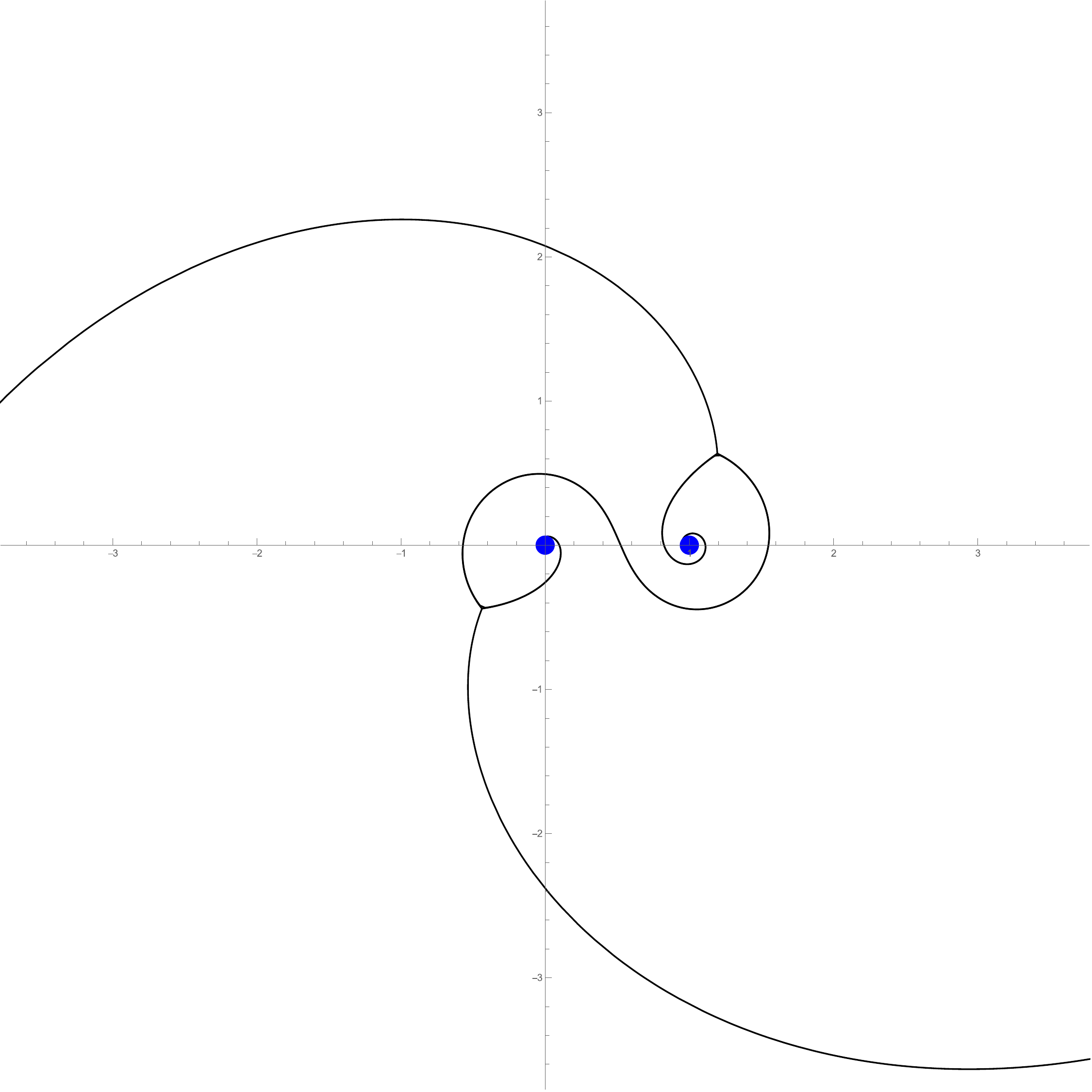}
           \caption{Type I saddle, $\vartheta\approx0.554$}
      \label{fig:hg5}
    \end{subfigure}
     \begin{subfigure}[t]{.3\textwidth}
        \centering
        \includegraphics[width=\linewidth,trim={5cm 5cm 5cm 5cm},clip]{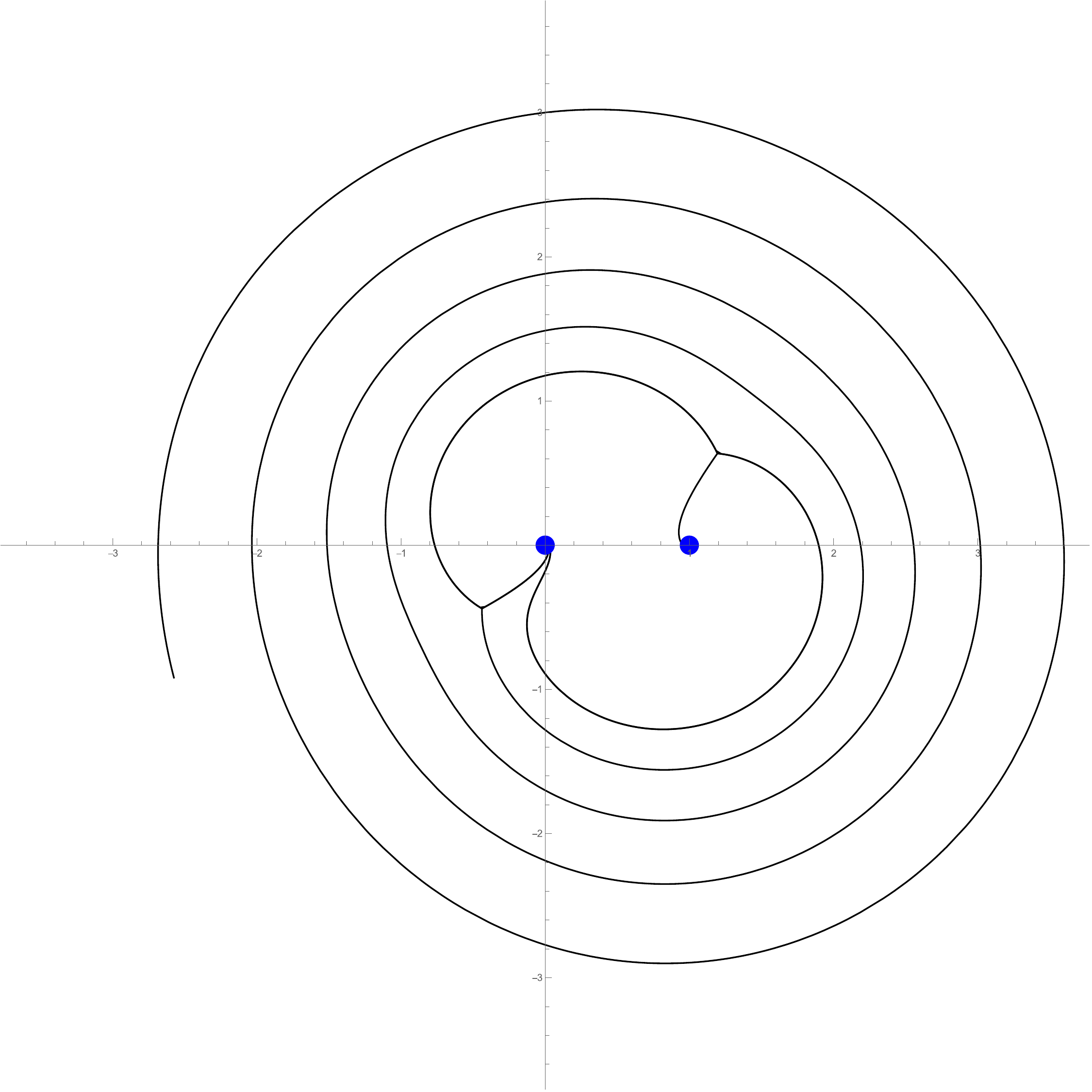}
           \caption{Type I saddle, $\vartheta\approx1.088$}
      \label{fig:hg7}
    \end{subfigure}
        \begin{subfigure}[t]{.3\textwidth}
        \centering
        \includegraphics[width=\linewidth,trim={5cm 5cm 5cm 5cm},clip]{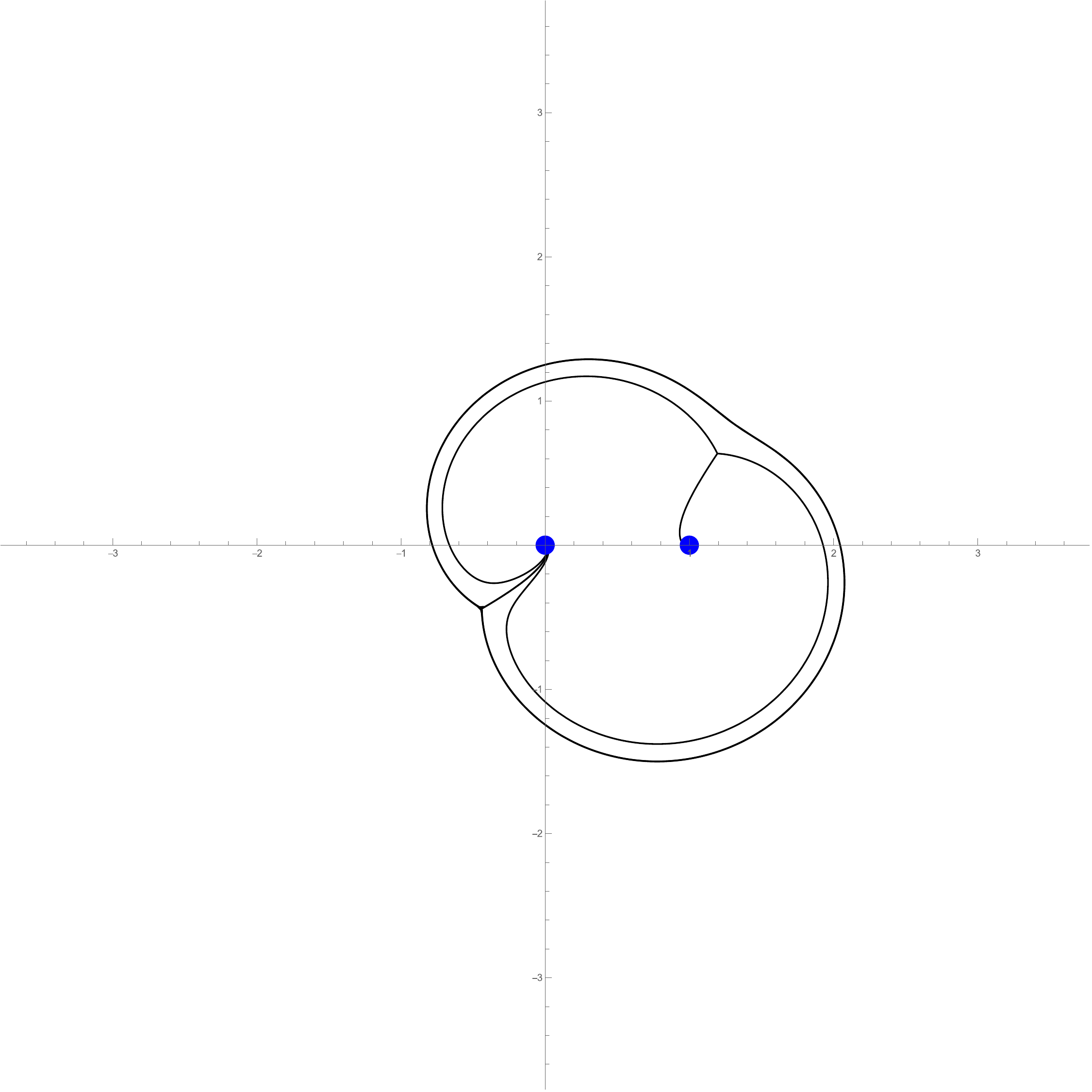}
           \caption{Type IV saddle, $\vartheta\approx1.122$}
      \label{fig:hg9}
    \end{subfigure}
        \begin{subfigure}[t]{.3\textwidth}
        \centering
        \includegraphics[width=\linewidth,trim={5cm 5cm 5cm 5cm},clip]{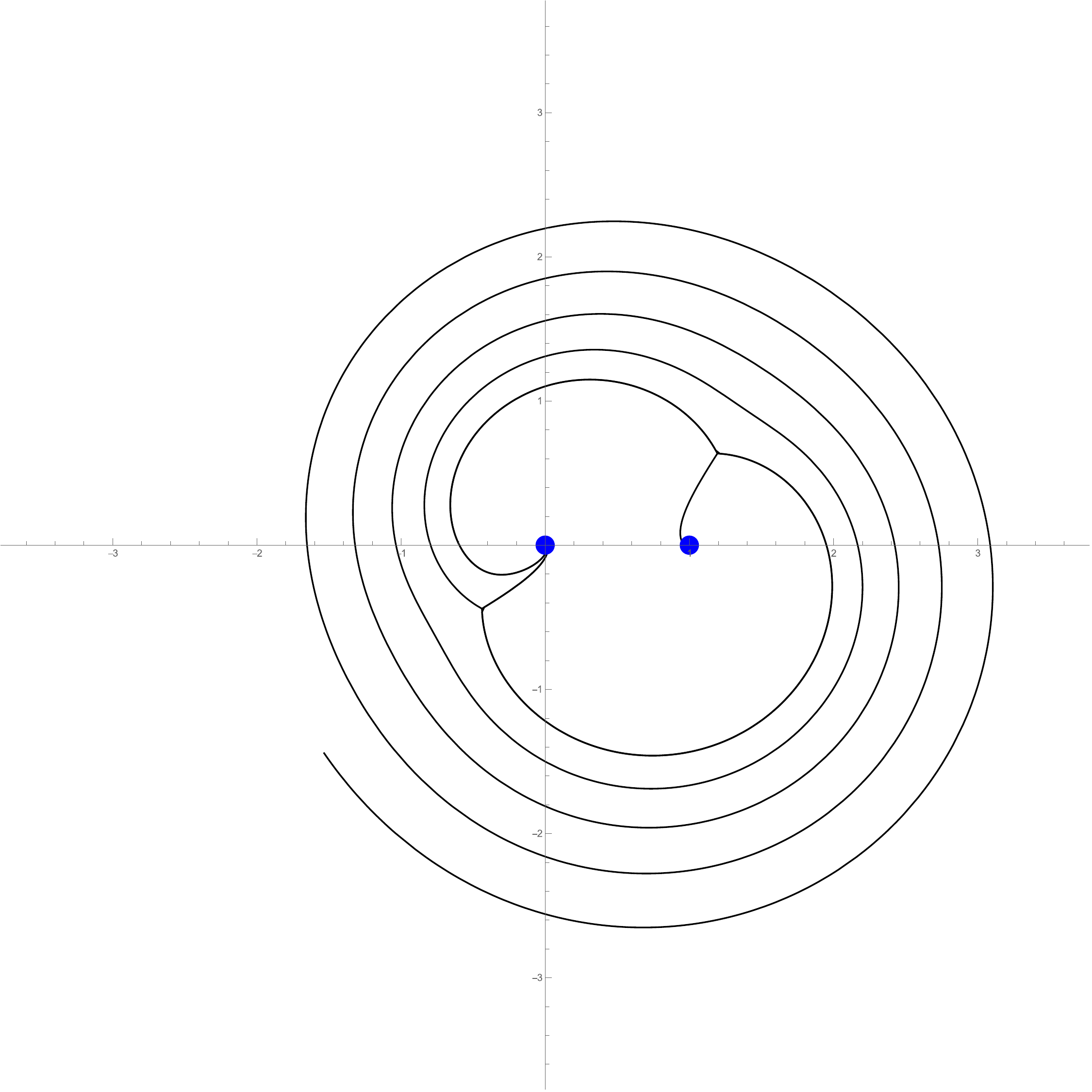}
           \caption{Type I saddle, $\vartheta\approx1.148$}
      \label{fig:hg11}
    \end{subfigure}
        \begin{subfigure}[t]{.3\textwidth}
        \centering
        \includegraphics[width=\linewidth,trim={5cm 5cm 5cm 5cm},clip]{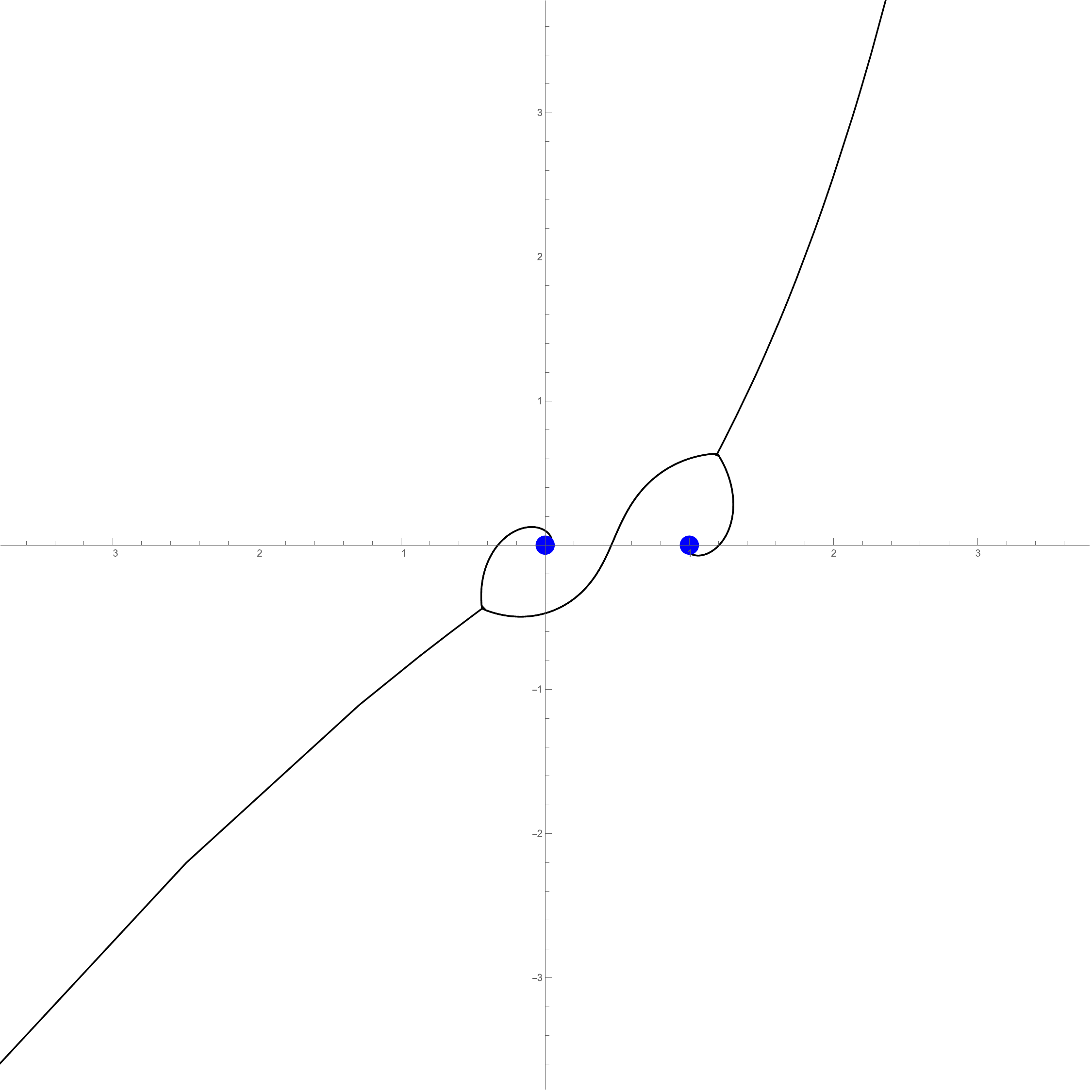}
           \caption{Type I saddle, $\vartheta\approx2.866$}
      \label{fig:hg13}
    \end{subfigure}
    \caption{$\mathcal{W}_{\vartheta_{\rm BPS}}(\varphi_{\rm HG})$ with~$m_0^2\approx~0.5+~0.2i$,~ $m_1^2\approx~0.5+~0.4i$,~$m_\infty^2 \approx~-~0.4+~0.5i$.} 
    \label{fig:hgsaddles}
\end{figure}
  
It is easy to compute the homology classes of each corresponding $\gamma_{\rm BPS}$, so we observe the existence of saddle trajectories with the given classes at seven values of $\vartheta \in [0,\pi)$. 
{ 
Among the seven degenerate spectral networks, four of them (i.e., Figure \ref{fig:hg5}, \ref{fig:hg7}, \ref{fig:hg11}, and \ref{fig:hg13}) contain a type I saddle, while the other three (i.e., Figure \ref{fig:hg1}, \ref{fig:hg3}, and \ref{fig:hg9}) contain a type IV saddle bounding a degenerate ring domain around $0$, $1$ and $\infty$, respectively.
We note that the appearance of four type I saddles was already observed in \cite[Figure 6]{GLPY17}.
}
We collect the classes appearing as $\gamma_{\rm BPS}$ in Tables \ref{table:bpsangles-HG}.


We can again prove that this is in fact exactly the BPS spectrum, for any value of ${\bm m}\in M_{\rm HG}'$.
First, we take care of the ring domains:

{ 
\begin{lemm} \label{lemm:loops-in-HG-cases}
For each $s \in \{0,1,\infty \}$, the spectral network ${\mathcal W}_{\vartheta}(\varphi_{\rm HG})$ contains a degenerate ring domain around $s$ if and only if $\vartheta = \arg m_s + \pi/2$ (mod $\pi$). 
The associated BPS cycle is $\gamma_{s_\pm} - \gamma_{s_\mp}$. 
\end{lemm}
\begin{proof}
This is a consequence of Proposition \ref{prop:loops}. 
\end{proof}

In what follows, we assume that 
$\arg m_0$, $\arg m_1$ and $\arg m_\infty$ are pairwise distinct (mod $\pi$) in order to work on the generic locus. 
Since $\varphi_{\rm HG}$ has no simple poles, the only other possible degenerations arise from type I saddles. 
The following lemma describes the candidate BPS cycles associated with type I saddles (we note a similar result is obtained by \cite[\S 12,4]{BS13} in a representation-theoretic language): 
}

\begin{lemm}
\label{lem:hgbps}
{
The only possible BPS cycles associated with a type I saddle appearing in the spectral network of $\varphi_{\rm HG}$ are in 
\begin{equation} \label{eq:candidate-of-BPS-cycles}
    \{ \gamma_{0_\epsilon} + \gamma_{1_{\epsilon'}} + \gamma_{\infty_{\epsilon''}} 
    ~|~ \epsilon, \epsilon', \epsilon'' \in \{\pm \} \} \subset \Gamma_{\rm HG}.
\end{equation}

}

\end{lemm}

{  
\begin{proof}
It is clear that possible BPS cycles must be represented by pullback of a path on ${\mathbb P}^1 \setminus \{0,1,\infty \}$ connecting two simple zeros of $Q_{\rm HG}(x)$. Since $\Sigma_{\rm HG}$ is of genus $0$, any such cycles are decomposed as a sum of residue cycles at the punctures of the form 
$\gamma_{0_\epsilon} + \gamma_{1_{\epsilon'}} + \gamma_{\infty_{{\epsilon''}}}$ with some signatures $\epsilon, \epsilon', \epsilon'' \in \{\pm \}$, up to sign. 
The signature depends on how the path avoid the poles $0$, $1$, $\infty$, and the orientation of the cycle. 
The relation \eqref{eq:relation-among-cycles} of cycles implies that 
$\gamma_{0_\epsilon} + \gamma_{1_{\epsilon'}} + \gamma_{\infty_{{\epsilon''}}} = - (\gamma_{0_{-\epsilon}} + \gamma_{1_{-\epsilon'}} + \gamma_{\infty_{{-\epsilon''}}})$, and hence, we have the above list of possible classes (this observation also shows that the above chosen cycles belong to the lattice $\Gamma_{\rm HG}$.)
\end{proof}

}

Then we have:
\begin{lemm} \label{lemm:existence-of-saddle-HG}
Let $\gamma$ be any of the eight cycles in \eqref{eq:candidate-of-BPS-cycles}, and set $\vartheta = \arg Z(\gamma)$. 
We assume that $\vartheta \ne \arg m_s + \pi/2$ (mod $\pi$) for all $s \in \{0,1, \infty \}$.
Then, the spectral network ${\mathcal W}_{\vartheta}(\varphi_{\rm HG})$ with phase $\vartheta$ contains a type I saddle with class $\pm\gamma$.  
\end{lemm}

\begin{proof}
Thanks to \eqref{eq:rotation-of-spectral-network}, it suffices to prove the claim when $\vartheta = 0$; that is, when $Z(\gamma) \in {\mathbb R}_{>0}$ and all mass parameters are not pure imaginary. 

Suppose for contradiction that the spectral network ${\mathcal W}_0(\varphi_{\rm HG})$ is nondegenerate. 
Then, as is shown in \cite{Aoki-Tanda}, it is known that possible topological types of ${\mathcal W}_0(\varphi_{\rm HG})$ are one of the graphs of ``type" $(2,2,2)$, $(1,1,4)$, $(1,4,1)$ or $(1,1,4)$. 
Here the tuple $(n_0,n_1,n_\infty)$ denotes the number of critical trajectories which approach to $0$, $1$, and $\infty$, respectively. 
Figure \ref{fig:222cycles} shows a graph of type $(2,2,2)$, while Figure \ref{fig:oneonefour1} depicts a graph of type $(1,1,4)$ for example.
It is easy to see that these nondegenerate spectral networks always contain three horizontal strips, and the associated dual cycles are shown in these figures. 

\begin{figure}[h]
    \centering
    \begin{subfigure}[t]{.38\textwidth}
        \centering
        \includegraphics[width=\linewidth,trim={0cm -0.92cm 0cm 0cm}]{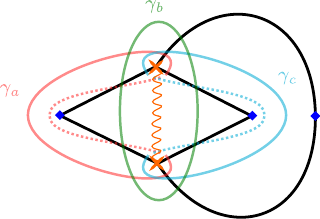}
           \caption{$(2,2,2)$ graph with dual cycles.}
      \label{fig:222cycles}
    \end{subfigure}
    \hspace{2cm}
    \begin{subfigure}[t]{.38
    \textwidth}
        \centering
        \includegraphics[width=\linewidth]{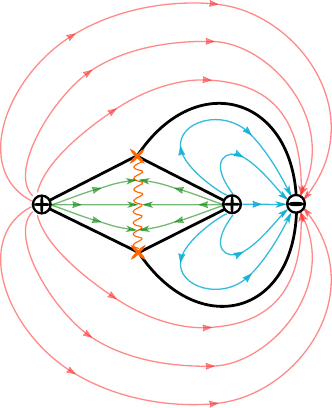}
        \caption{Oriented trajectories .}
      \label{fig:222flow}
    \end{subfigure}
    \caption{Impossibility of the $(2,2,2)$ graph.}
    \label{fig:222fig}
\end{figure}

\begin{figure}[h]
    \centering
    \begin{subfigure}[t]{.40\textwidth}
        \centering
        \includegraphics[width=\linewidth,trim={0cm -0.94cm 0cm 0cm}]{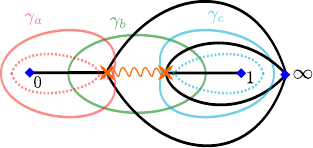}
           \caption{$(1,1,4)$ graph with dual cycles.}
      \label{fig:oneonefour1}
    \end{subfigure}
    \hspace{2cm}              
    \begin{subfigure}[t]{.35\textwidth}
        \centering
        \includegraphics[width=\linewidth]{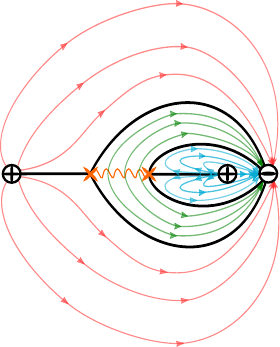}
        \caption{Oriented trajectories.}
      \label{fig:oneonefour2}
    \end{subfigure}
    \caption{Impossibility of the $(1,1,4)$ graph.}
    \label{fig:oneonefour}
\end{figure}

To show that these graphs cannot appear under our assumption, we need a careful treatment of the sign (orientation) of the dual cycles associated with the horizontal strips of these spectral networks. 
For this purpose, it is convenient to use several properties of trajectories enhanced with orientations which are naturally equipped after choosing a branch cut which trivializes the covering $\pi : \overline{\Sigma} \to X (= {\mathbb P}^1)$, and label the sheets as sheets $1$ and $2$. 
That is, we can equip an orientation to each part of trajectories on $X \setminus \{\text{cut}\}$. Namely, identifying $X \setminus \{\text{cut}\}$ with the first sheet of $\Sigma \setminus \pi^{-1}(\text{cut})$, we orient each part of the trajectory so that the real part of the integral appearing in \eqref{eq:trajectories} increases along it. 
Since the integrand in \eqref{eq:trajectories} has different signs on different sheets, the orientation is reversed before and after the branch cuts. 
Figure \ref{fig:222flow} and Figure \ref{fig:oneonefour2} shows an example of the oriented trajectories on the first sheet, with the orange wavy line denoting the branch cut.  
The main property we will use is that any second order pole $s$ must be either a ``source" or a ``sink" for the oriented trajectories, in the obvious sense. 
Here we attach the symbol $\ominus$ (resp. $\oplus$) for the second order poles if it is a source (resp. sink) on the first sheet (see Figure \ref{fig:222flow} and Figure \ref{fig:oneonefour2}). 
In our sign convention \eqref{eq:sign-convention-preimages}, we attach $\oplus$ to $s$ if its preimage on the first sheet is $s_-$, and attach $\ominus$ otherwise. 

Consider now the graph of type $(2,2,2)$, which has three dual cycles $\gamma_a$, $\gamma_b$ and $\gamma_c$ as shown in Figure \ref{fig:222cycles}. Although the dual cycles are defined up to sign, here we take one of them so that its central charge lies on the upper half plane ${\mathbb H}_{+}$ (recall that Lemma \ref{lem:diamondlemma} guarantees that the central charges have a non-zero imaginary part). 
This condition is satisfied if we equip the orientation to each dual cycle $\gamma_D$ so that $(\gamma_D, \beta_D) = - 1$ holds, where $\beta_D \in H_1(\overline{\Sigma}, D_\infty, {\mathbb Z})$ is the relative homology class represented by the preimage of any generic trajectory in $D$ with the orientation given as above, and $(\cdot, \cdot)$ is the intersection pairing normalized as $(\text{$x$-axis}, \text{$y$-axis}) = + 1$. 
With the choice of the first sheet indicated in Figure \ref{fig:222flow}, we have
\begin{equation} \label{eq:dual-cycles-in-2-2-2}
\gamma_a = \gamma_{0_+} + \gamma_{1_-} + \gamma_{\infty_+}, \quad
\gamma_b = \gamma_{0_+} + \gamma_{1_+} + \gamma_{\infty_-}, \quad
\gamma_c = \gamma_{0_-} + \gamma_{1_+} + \gamma_{\infty_+}.
\end{equation}
Now, we set $\gamma_d := \gamma_a + \gamma_b + \gamma_c = \gamma_{0_+} + \gamma_{1_+} + \gamma_{\infty_+}$. 
The central charge $Z(\gamma_d)$ also lies in the upper half plane ${\mathbb H}_{+}$ by construction.  
However, since the set $\{ \pm \gamma_a, \pm \gamma_b, \pm \gamma_c, \pm \gamma_d \}$ coincides with the above set \eqref{eq:candidate-of-BPS-cycles} of all candidate BPS cycles, we have a contradiction to the original assumption (i.e., reality of one of the central charges of the cycles \eqref{eq:candidate-of-BPS-cycles}). 
Thus we have shown that the graph of type $(2,2,2)$ never appears as ${\mathcal W}_{0}(\varphi_{\rm HG})$.

Next, we consider the graph of type $(1,1,4)$ in Figure \ref{fig:oneonefour}. 
We may check that the dual cycles are explicitly given as follows:
\begin{equation} \label{eq:dual-cycles-in-1-1-4}
\gamma_a = \gamma_{0_+} - \gamma_{0_-}, \quad
\gamma_b = - \gamma_{0_+} - \gamma_{1_+} - \gamma_{\infty_-}, \quad
\gamma_c = \gamma_{1_+} - \gamma_{1_-}.
\end{equation}
Here the orientations are determined in the same manner as before, i.e. so that their central charge has positive imaginary part under the choice of the first sheet indicated in Figure \ref{fig:oneonefour2}.
Then we can express all eight cycles in \eqref{eq:candidate-of-BPS-cycles} as a ${\mathbb Z}$-linear combination with all positive or all negative coefficients, as follows:
\begin{equation}
\begin{cases}
\gamma_{0_\pm} + \gamma_{1_\pm} + \gamma_{\infty_\pm} = \pm (\gamma_a + \gamma_b + \gamma_c), & \\
\gamma_{0_\pm} + \gamma_{1_\mp} + \gamma_{\infty_\pm} = \pm (\gamma_a + \gamma_b), & \\
\gamma_{0_\pm} + \gamma_{1_\pm} + \gamma_{\infty_\mp} = \mp \gamma_b, & \\
\gamma_{0_\pm} + \gamma_{1_\mp} + \gamma_{\infty_\mp} = \mp (\gamma_b + \gamma_c).
\end{cases}
\end{equation}
Then we again have a contradiction to the reality of one of the central charges of the cycles \eqref{eq:candidate-of-BPS-cycles}, which eliminates the possibility of the graph of type $(1,1,4)$.
By similar reasoning, we can also check that the other possible cases never appear.

Thus the only possibility is to have a degenerate spectral network at $\vartheta = 0$. 
Since we have also assumed that all mass parameters are not pure imaginary, ${\mathcal W}_{0}(\varphi_{\rm HG})$ does not contain any loop-type saddle. 
Thus, we conclude it contains a saddle trajectory whose type must be I.
\end{proof}

\begin{rem}
The condition that the central charges of the cycles appearing in \eqref{eq:dual-cycles-in-2-2-2} (resp., \eqref{eq:dual-cycles-in-1-1-4}) have a positive imaginary part gives a necessary condition for appearance of a graph of type $(2,2,2)$  (resp., of type $(1,1,4)$) in terms of the mass parameters. 
This agrees with the classification of the Stokes graphs (spectral networks) of the Gauss hypergeometric equation given by Aoki-Tanda \cite{Aoki-Tanda}. 
\end{rem}

It follows from the proof of Lemma \ref{lemm:existence-of-saddle-HG} that, if we further assume that $\arg Z(\gamma') \ne \arg Z(\gamma'')$ holds for any pair $\gamma', \gamma''$ of distinct elements in \eqref{eq:candidate-of-BPS-cycles}, then the BPS cycle associated with the type I saddle thus obtained must be $\pm \gamma$. 
That is, the set \eqref{eq:candidate-of-BPS-cycles} coincides with the set of all BPS cycles associated to a type I saddle, and we have a concrete description of the generic locus $M'_{\rm HG}$ in terms of the mass parameter.
Together with the BPS cycles associated with degenerate ring domains (c.f., Lemma \ref{lemm:loops-in-HG-cases}), we have 

\begin{prop}
Fix ${\bm m}\in M_{\rm HG}'$. Then, there are exactly seven degenerations in the range $\vartheta \in [0,\pi)$: four type I saddles, and three loop-type saddles. 
The BPS spectrum of $\varphi_{\rm HG}$ consists of exactly the fourteen $\gamma_{\rm BPS}$ described in Table \ref{table:bpsangles-HG}.
\end{prop}

  \begin{table}[h]
  \begin{tabular}{|c||c|c|} \hline
    $\vartheta_{\rm BPS}$
    & $\arg (\epsilon m_{0} + \epsilon' m_{1} + \epsilon'' m_{\infty}) + \pi/2$ & 
    $\arg{m_s} \pm \pi/2$  \\ \hline 
    degeneration
    & type I saddle & degenerate ring domain   \\ \hline
    
    $\gamma_{\rm BPS}$
    & $\gamma_{0_\epsilon} + \gamma_{1_{\epsilon'}} + \gamma_{\infty_{\epsilon''}}$ & $\gamma_{s_\pm} - \gamma_{s_\mp}$  \\ \hline 
    $Z(\gamma_{\rm BPS})$
    & $2\pi i (\epsilon m_0 + \epsilon' m_1 + \epsilon'' m_\infty)$ & 
    $\pm 4 \pi i m_s$  \\ \hline

    $\Omega(\gamma_{\rm BPS})$ & 
    $+1$ & $-1$   \\ \hline
  \end{tabular} 
   \vspace{+1.em}
     \caption{
     The BPS spectrum of $\varphi_{\rm HG}$, where $\epsilon, \epsilon', \epsilon'' \in \{ \pm \}$ and $s \in \{0,1,\infty \}$.}
     \label{table:bpsangles-HG}
  \end{table}




\subsection{BPS structure from degenerations of the Gauss hypergeometric curve}
\subsubsection{\bf BPS structure from the Kummer curve}
The Kummer differential is $\varphi_{\rm Kum} = Q_{\rm Kum}(x)dx^2$ where
\begin{equation}
Q_{\rm{Kum}}(x)= \dfrac{x^{2}+4m_{\infty}x + 4 m_{0}^{2}}{4x^{2}} 
\end{equation}
which defines the Kummer curve $\Sigma_{\rm{Kum}} \,(= \widetilde{\Sigma}_{\rm Kum})$. 
It has has two simple zeroes, a second order pole at $0$ and a pole of order $4$ at $\infty$ under the assumption ${\bm m} \in M_{\rm Kum}$. 
It is easy to see that $\Sigma_{\rm Kum}$ is of genus $0$ with four punctures, at $0_\pm$ and $\infty_\pm$.

We can draw the spectral network for a generic value of the parameters ${\bm m}\in M_{\rm Kum}'$. As a result, we observe three BPS cycles, with two type I saddles and a single loop, as depicted in Figure \ref{fig:kummerbps} below.

\begin{figure}[h]
    \centering  
    \begin{subfigure}[t]{.25\textwidth}
        \centering
        \includegraphics[width=1.0\linewidth,trim={7cm 7cm 7cm 7cm},clip]{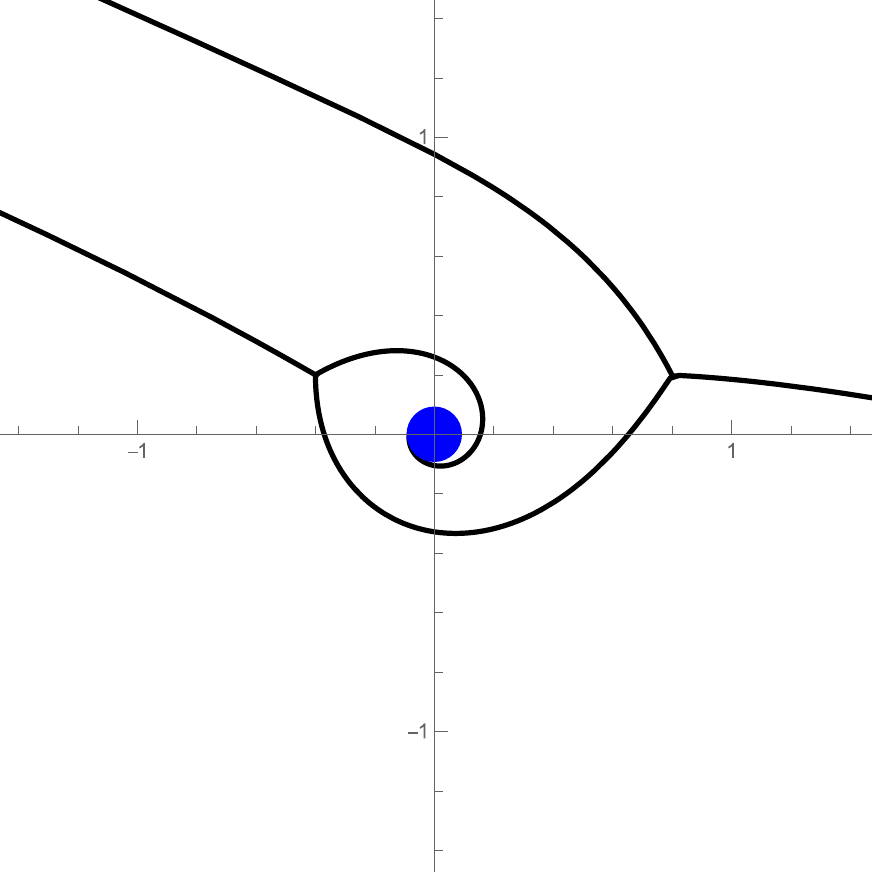}
        \caption{Type I saddle, $\vartheta \approx 1.25$}
      \label{fig:kummer1}
    \end{subfigure}
        \hspace{0.5cm}               
    \begin{subfigure}[t]{.25\textwidth}
        \centering
        \includegraphics[width=1.0\linewidth,trim={7cm 7cm 7cm 7cm},clip]{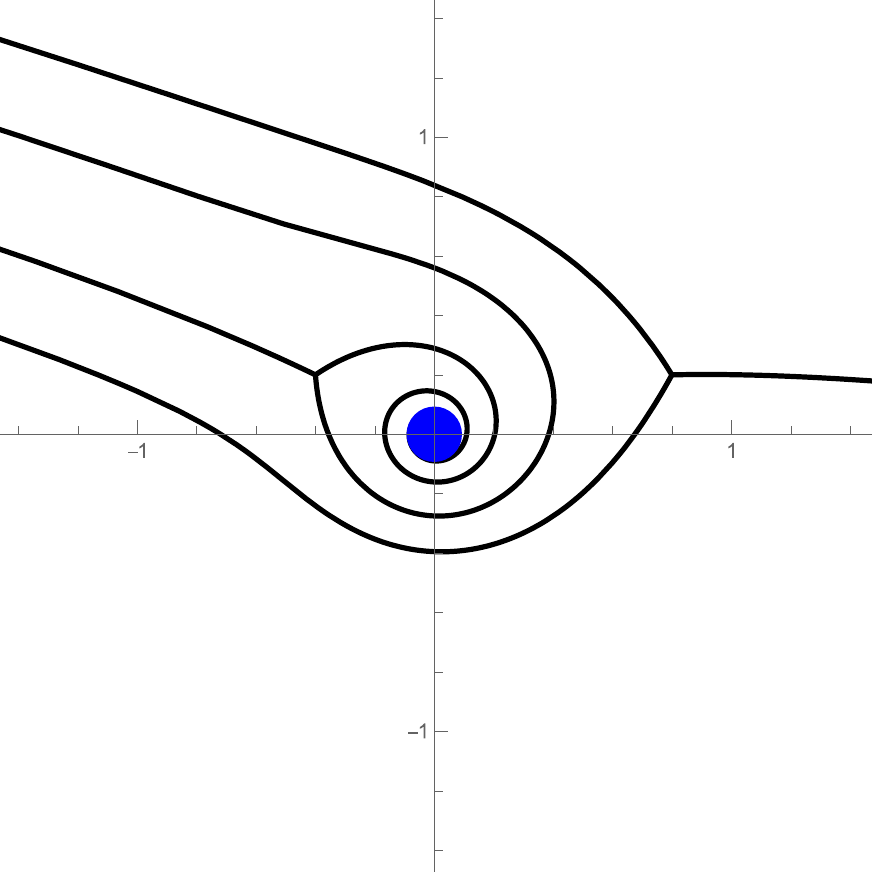}
        \caption{$\vartheta\approx1.36$}
      \label{fig:kummer2}
    \end{subfigure}
        \hspace{0.5cm}               
    \begin{subfigure}[t]{.25\textwidth}
        \centering
        \includegraphics[width=1.0\linewidth,trim={6cm 6cm 6cm 6cm},clip]{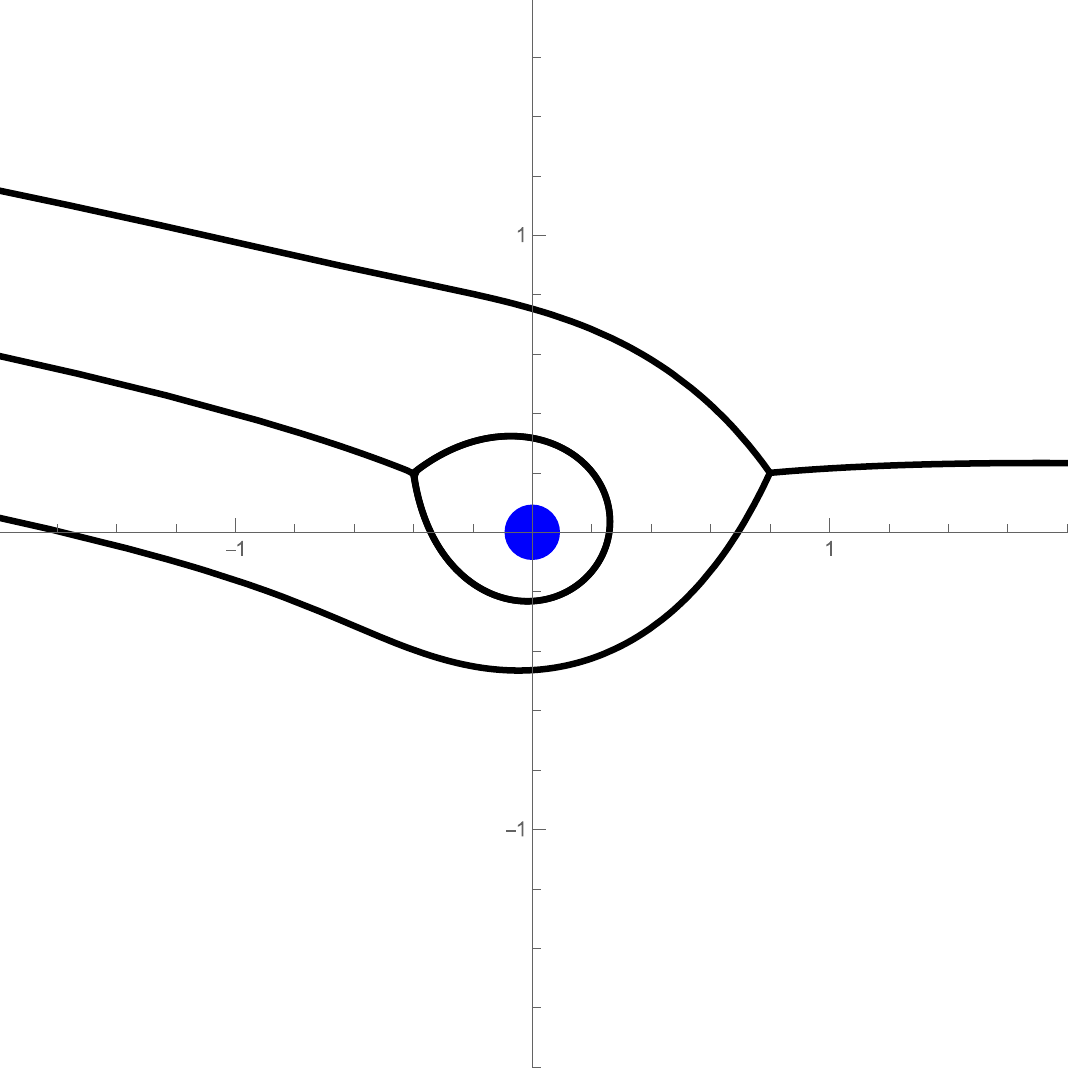}
        \caption{Type IV saddle, $\vartheta\approx1.46$}
      \label{fig:kummer3}
    \end{subfigure}
         \hspace{0.5cm}               
    \begin{subfigure}[t]{.25\textwidth}
        \centering
        \includegraphics[width=1.0\linewidth,trim={6cm 6cm 6cm 6cm},clip]{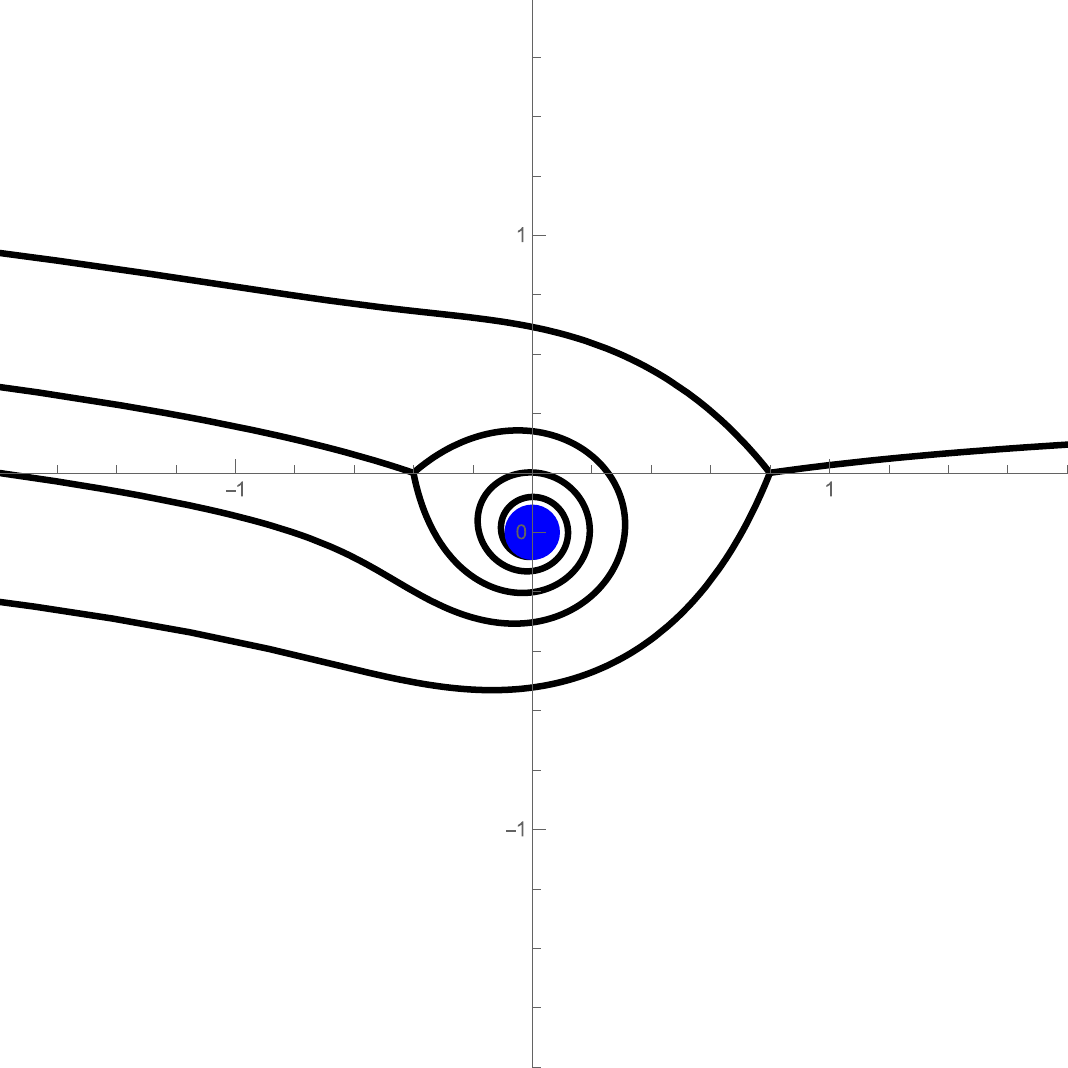}
        \caption{$\vartheta\approx1.54$}
      \label{fig:kummer4}
    \end{subfigure}
         \hspace{0.5cm}               
    \begin{subfigure}[t]{.25\textwidth}
        \centering
        \includegraphics[width=1.0\linewidth,trim={6cm 6cm 6cm 6cm},clip]{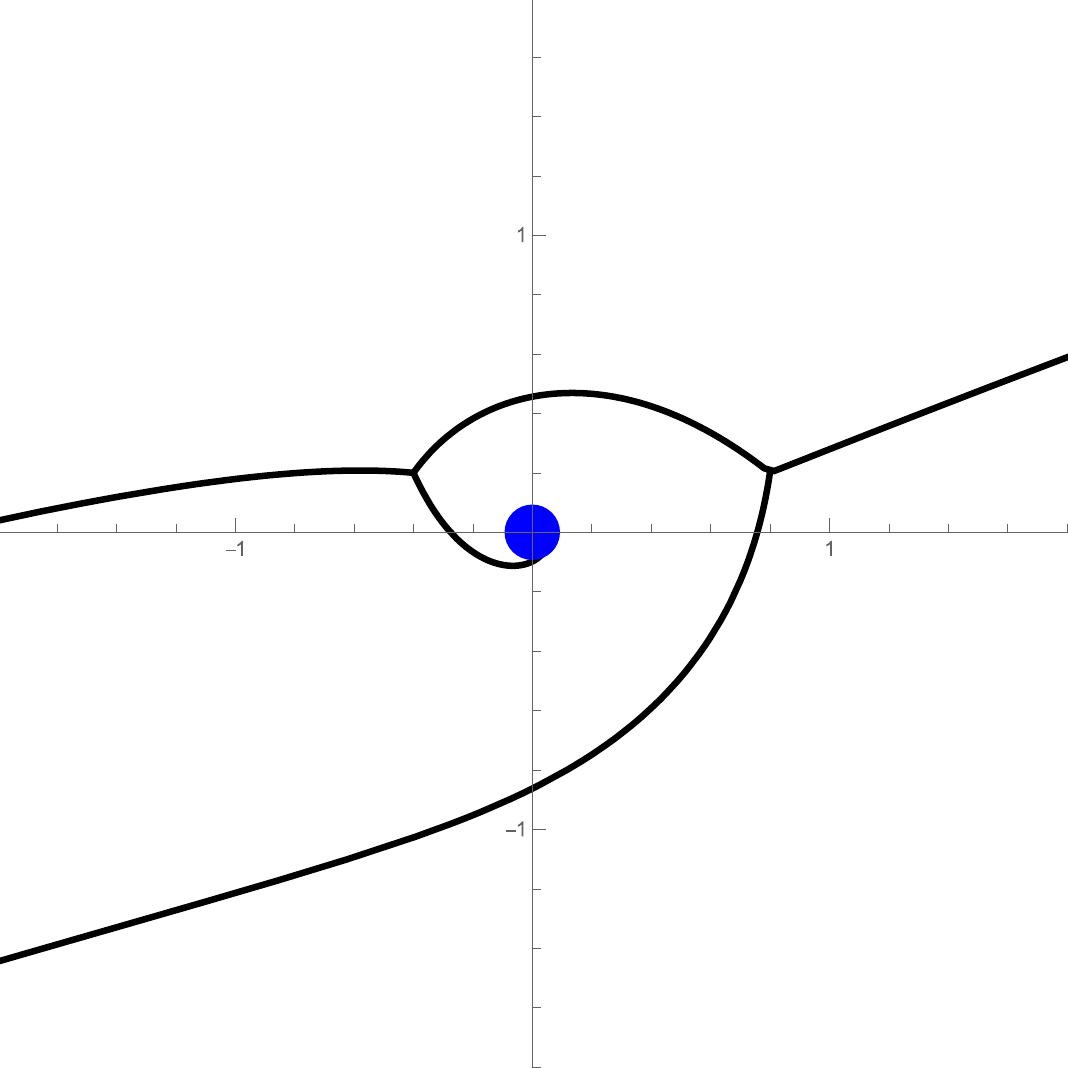}
        \caption{Type I saddle, $\vartheta\approx1.89$}
      \label{fig:kummer5}
    \end{subfigure}
           \hspace{0.5cm}               
    \begin{subfigure}[t]{.25\textwidth}
        \centering
        \includegraphics[width=1.0\linewidth,trim={7.5cm 7.5cm 7.5cm 7.5cm},clip]{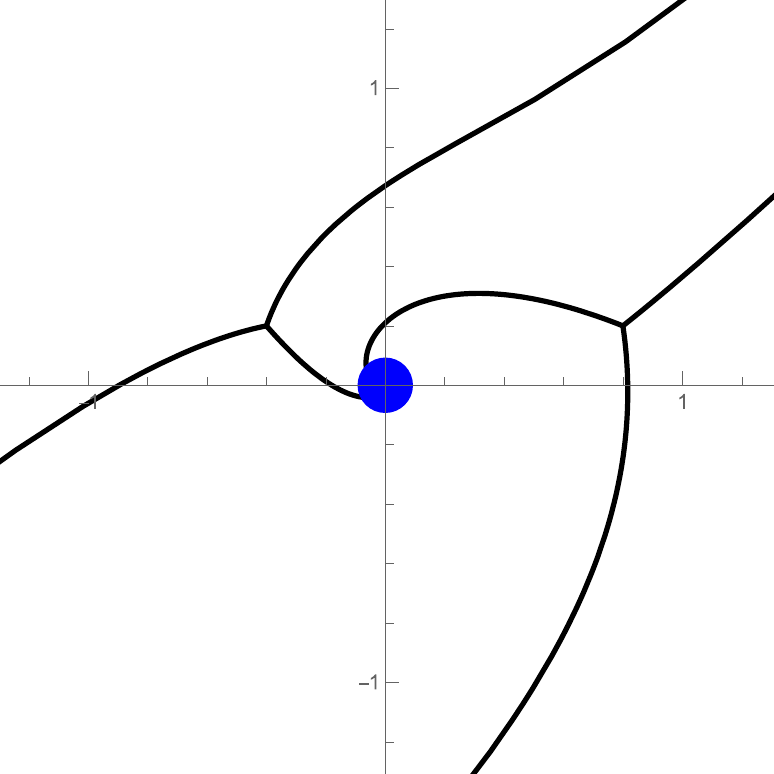}
        \caption{$\vartheta\approx2.33$}
      \label{fig:kummer6}
    \end{subfigure}
    \caption{Spectral networks for $\varphi_\mathrm{Kum}$, $m_0\approx 0.07+0.60, m_\infty \approx -0.4-0.4i$.}
\label{fig:kummerbps}
\end{figure}

Let $Z_a = 2 \pi i (m_0+m_\infty)$, $Z_b = 2 \pi i (m_0-m_\infty)$. We can show
\begin{prop}
{  
Let $\vartheta = \arg Z$ with $Z$ being one of $Z_a$ or $Z_b$, and assume $\vartheta \ne \arg m_0 + \pi/2$ (mod $\pi$).
Then, the spectral network ${\mathcal W}_{\vartheta}(\varphi_{\rm Kum})$ with the phase $\vartheta$ contains a type I saddle.  
If we further assume that $\arg Z_a \ne \arg Z_b$ (mod $\pi$), then the associated BPS cycle is $\gamma_{0_\pm} + \gamma_{\infty_\pm}$ when $Z = Z_a$, or $\gamma_{0_\pm} + \gamma_{\infty_\mp}$ if $Z = Z_b$.
}

Let ${\bm m}\in M_{\rm Kum}'$. 
Then $\varphi_{\mathrm{Kum}}$ has a type I saddle exactly when $Z_a$ or $Z_b$ vanish. The saddles have class $\gamma_a=\pm(\gamma_{0_+}+\gamma_{\infty_+})$ or $\gamma_b=\pm(\gamma_{0_+}+\gamma_{\infty_-})$, respectively.
\end{prop}
\begin{proof}
It is easy to check combinatorially, using Teichm\"uller's lemma and the normal forms, that there are only two possible topological types of spectral networks as in Figure \ref{fig:kummerpic}. 
In both cases, we may proceed by a similar argument to proof of Lemma \ref{lemm:existence-of-saddle-HG} in the hypergeometric case, so we omit the details: 
carefully computing the BPS cycles with orientation determined by the oriented trajectories, we can verify that both of Figures \ref{fig:kummerproof-a} and \ref{fig:kummerproof-b} do not appear under our assumption.  
The rest of the claim can also be proved by a similar argument to the hypergeometric case. 
\end{proof}

\begin{figure}[h]
    \centering
    \begin{subfigure}[t]{.23\textwidth}
        \centering
        \includegraphics[width=0.7\linewidth]{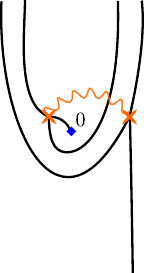}
           \caption{}
      \label{fig:kummerproof-a}
    \end{subfigure}
    \hspace{2cm}               
    \begin{subfigure}[t]{.23\textwidth}
        \centering
        \includegraphics[width=0.7\linewidth]{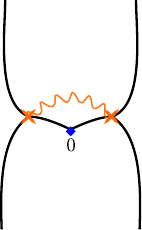}
        \caption{}
      \label{fig:kummerproof-b}
    \end{subfigure}
    \caption{Possible nondegenerate spectral networks for $\varphi_{\rm Kum}$.}
    \label{fig:kummerpic}
\end{figure}



Together with the contribution of the degenerate ring domain around the second order pole $0$, we have the complete list of the BPS spectrum summarized in Table \ref{table:kummerangles}. 

 \begin{table}[h]
  \begin{tabular}{|c||c|c|c|c|c|c|c|c|} \hline
    $\vartheta_{\rm BPS}$
    & $\arg{(m_0+m_\infty)} \pm \pi/2$ 
    & $\arg{(m_0-m_\infty) \pm \pi/2}$ 
    & $\arg{m_0} \pm \pi/2$ 
    \\ \hline 
    degeneration
    & type I saddle
    & type I saddle & degenerate ring domain  \\ \hline
    $\gamma_{\rm BPS}$ &
    $\gamma_{0_\pm}+\gamma_{\infty_\pm}$ & $ \gamma_{0_\pm} + \gamma_{\infty_\mp}$ & $\gamma_{0_\pm}-\gamma_{0_\mp}$ \\ \hline
     $Z(\gamma_{\rm BPS})$ &
    $\pm 2 \pi i (m_0+m_\infty)$ & $\pm 2 \pi i (m_0-m_\infty)$ & $\pm 4 \pi i m_0$ \\ \hline
    $\Omega(\gamma_{\rm BPS})$ &
    $+1$ & $+1$ & $-1$ \\ \hline
  \end{tabular} 
   \vspace{+1.em}
     \caption{The BPS spectrum of $\varphi_{\rm Kum}$.}
     \label{table:kummerangles}
  \end{table}

\subsubsection{\bf BPS structure from the degenerate Gauss curve}
The degenerate Gauss equation is the confluent limit of the hypergeometric equation in the limit when $m_0$ vanishes, leaving a simple pole at the origin. The corresponding quadratic differential is $\varphi_{\rm dHG} = Q_{\rm dHG}(x)dx^2$ where
\begin{equation}
Q_{\rm{dHG}}(x)= \dfrac{m_{\infty}^2 x+m_{1}^2 - m_{\infty}^{2}}{x(x-1)^2}.
\end{equation}
Under the assumption ${\bm m}\in M_{\rm dHG}$, $\varphi_{\rm dHG}$ has two simple zeros, a simple pole at the origin, and a second order pole at $1$ and $\infty$. 
The associated (partially compactified) degenerate Gauss curve $\widetilde{\Sigma}_{\rm{dHG}}$ is of genus $0$ with four punctures, at $1_\pm$ and $\infty_\pm$.

We can compute the BPS spectrum at a generic value of the parameters ${\bm m}=(m_1,m_\infty)$. As the result, we will have four BPS cycles, with two type II saddles and two loops, depicted in Figure \ref{fig:dhgpic1}.

  \begin{figure}[h]
    \centering
    \begin{subfigure}[t]{0.27\textwidth}
        \centering
        \includegraphics[width=\linewidth,trim={5cm 5cm 5cm 5cm},clip]{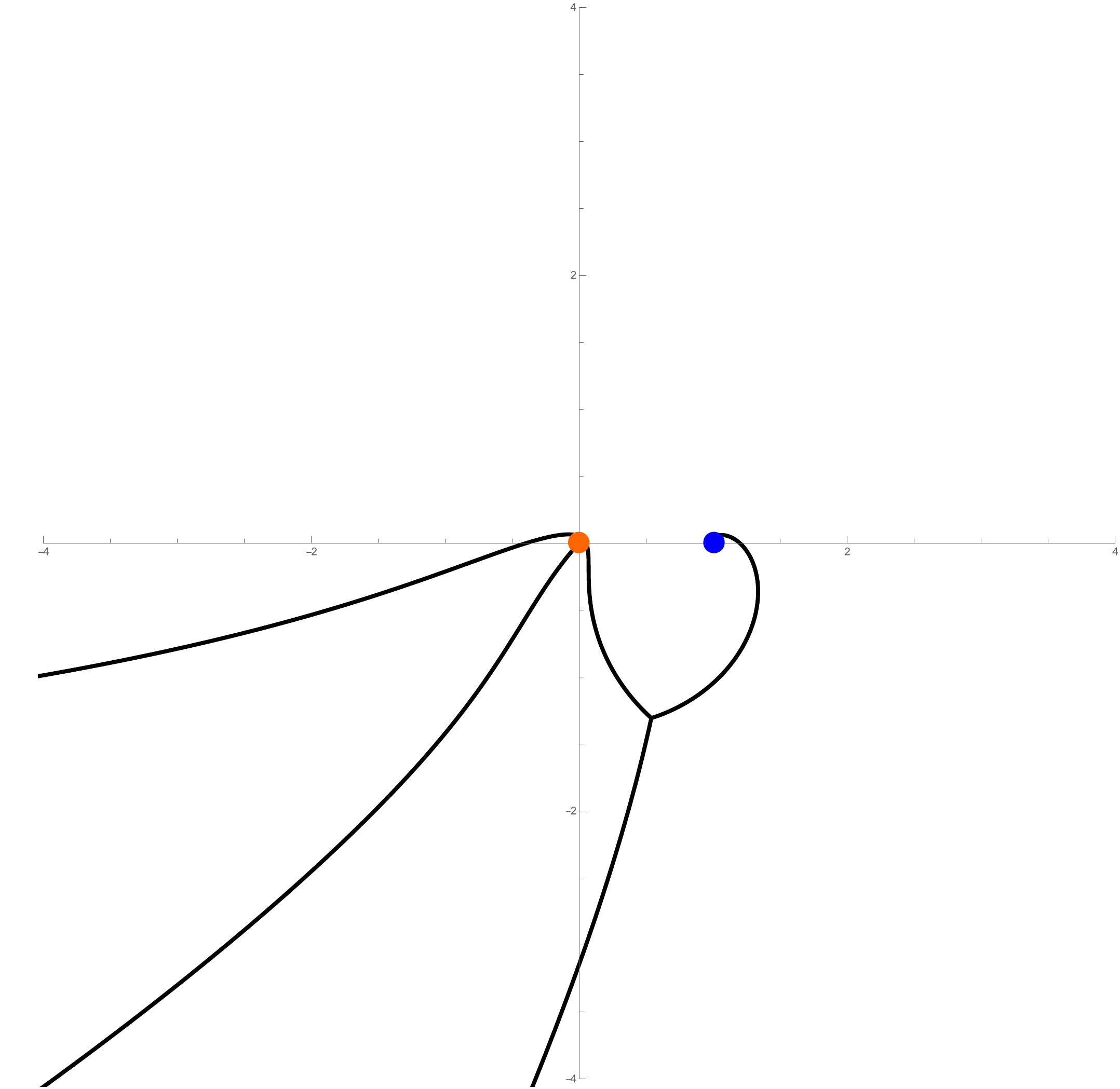}
        \caption{$\vartheta\approx0.31$}
      \label{fig:dhg2}
    \end{subfigure}
        \hspace{0.5cm}               
    \begin{subfigure}[t]{.27\textwidth}
        \centering
        \includegraphics[width=\linewidth,trim={3.75cm 3.75cm 3.75cm 3.75cm},clip]{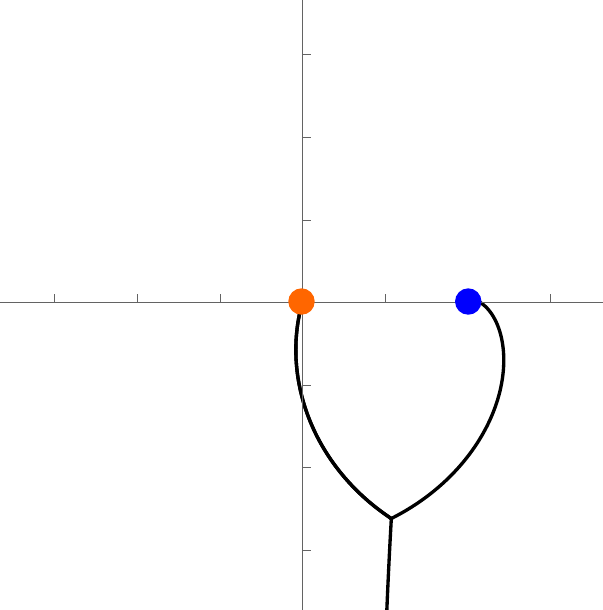}
        \caption{Type II saddle, $\vartheta\approx 0.55$}
      \label{fig:dhg3}
    \end{subfigure}
         \hspace{0.5cm}               
    \begin{subfigure}[t]{.27\textwidth}
        \centering
        \includegraphics[width=\linewidth,trim={4cm 4cm 3cm 3cm},clip]{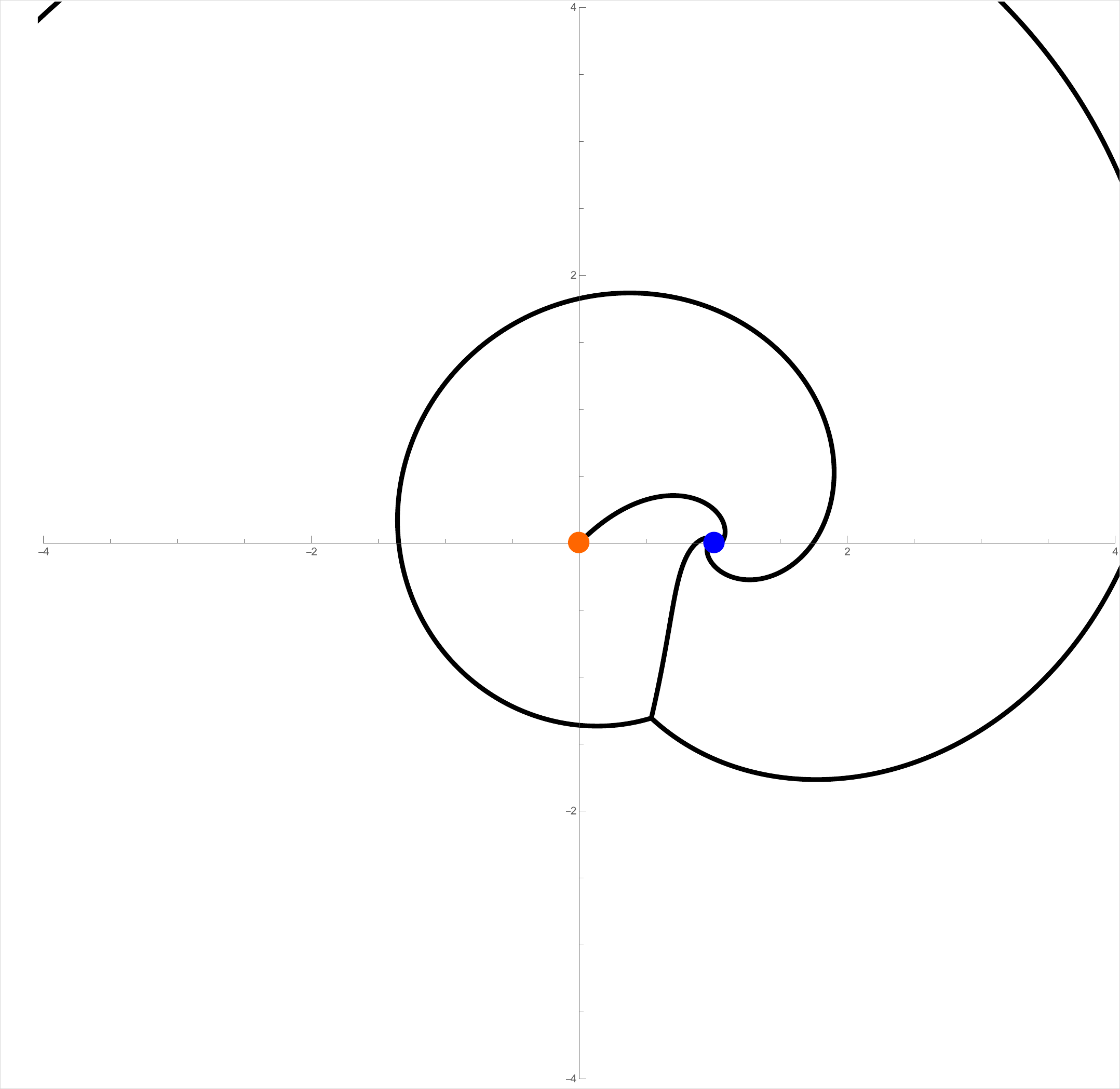}
        \caption{$\vartheta\approx1.87$}
      \label{fig:dhg5}
    \end{subfigure}
           \hspace{0.5cm}               
    \begin{subfigure}[t]{.27\textwidth}
        \centering
        \includegraphics[width=\linewidth,trim={1.8cm 1.8cm 1.8cm 1.8cm},clip]{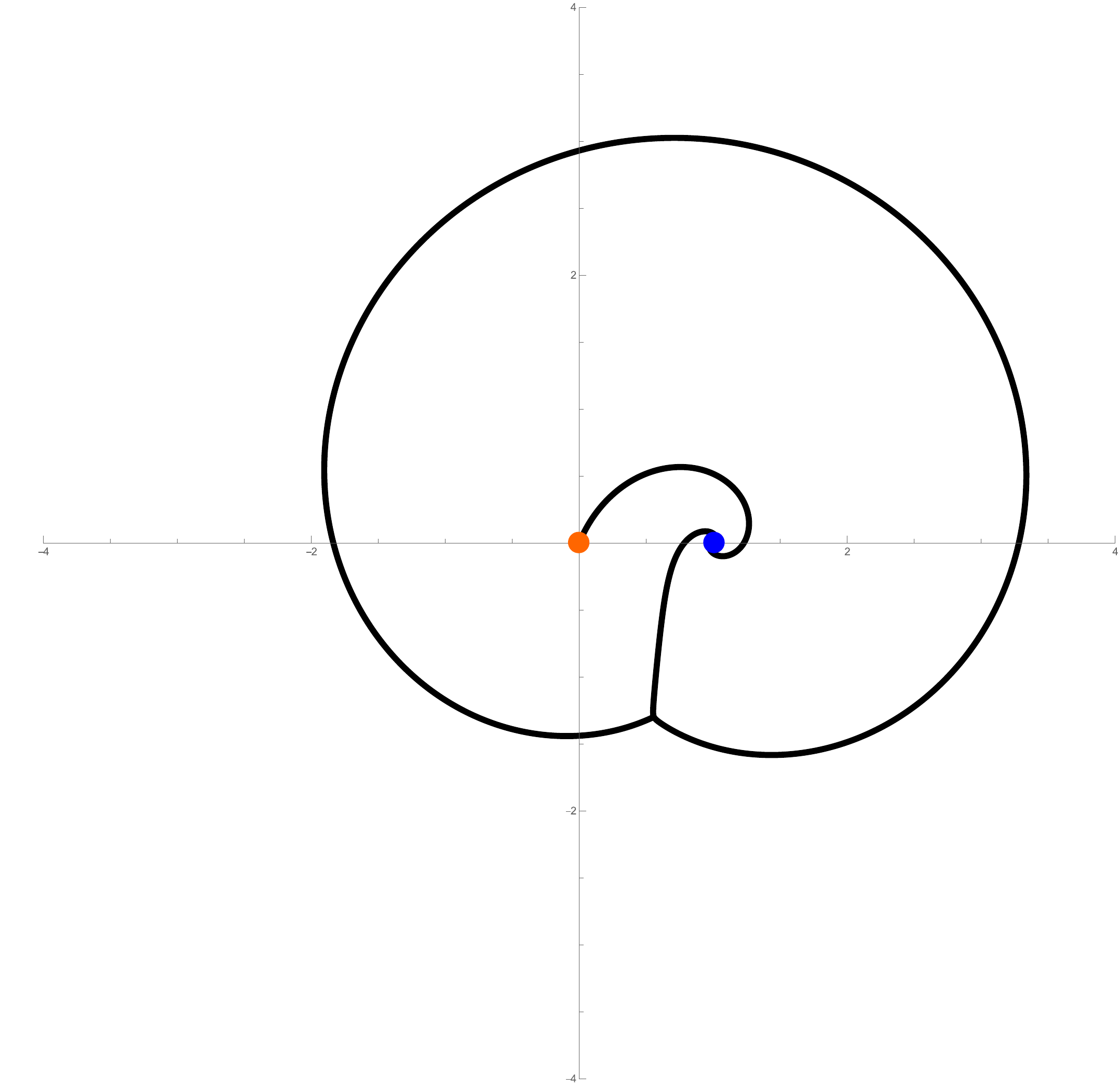}
        \caption{Type IV saddle, $\vartheta\approx2.062$}
      \label{fig:dhg6}
    \end{subfigure}
    \begin{subfigure}[t]{.27\textwidth}
        \centering
        \includegraphics[width=\linewidth,trim={2.5cm 2.5cm 2.5cm 2.5cm},clip]{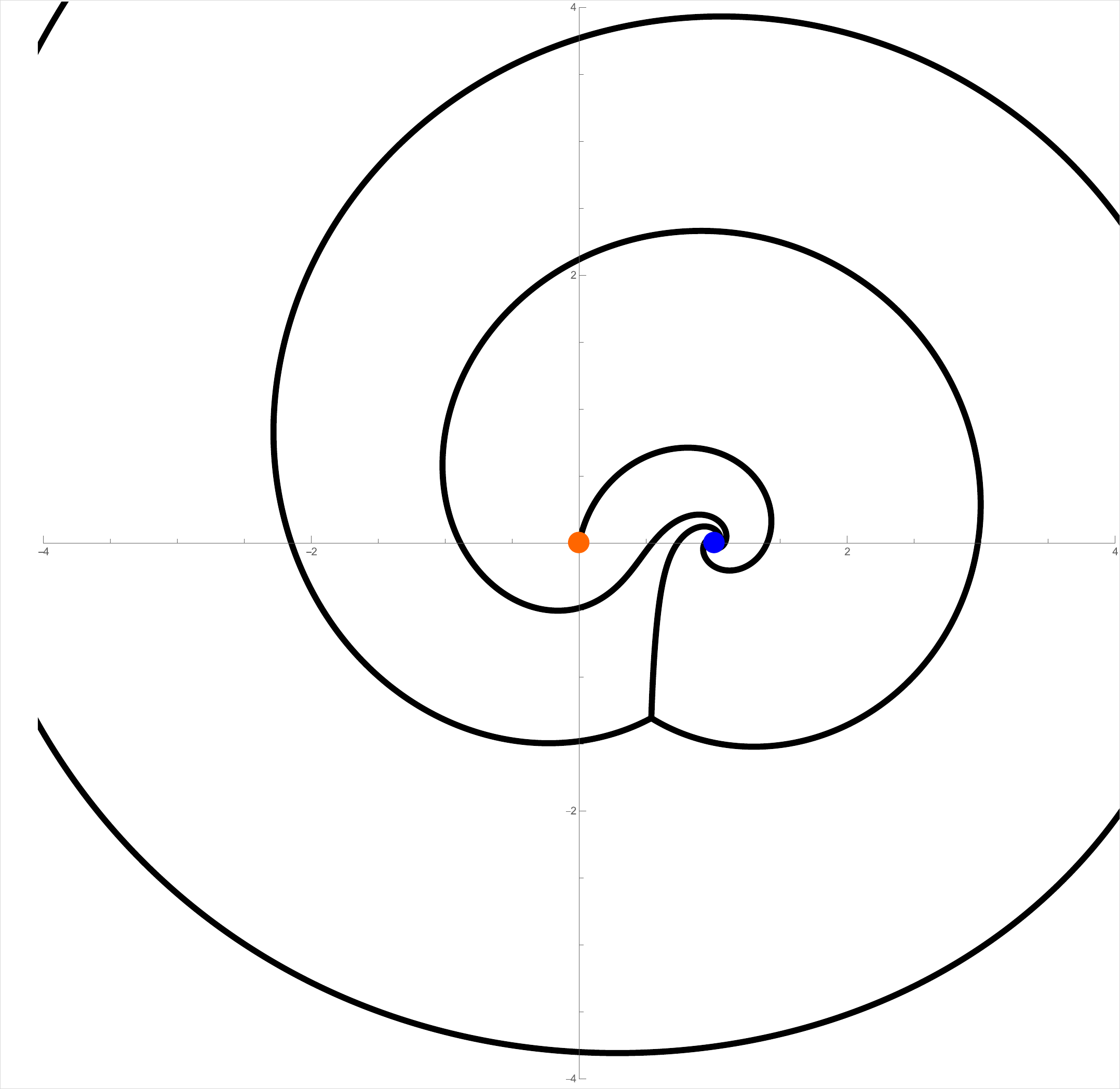}
        \caption{$\vartheta\approx2.16$}
      \label{fig:dhg7}
    \end{subfigure}
    \begin{subfigure}[t]{.27\textwidth}
        \centering
        \includegraphics[width=\linewidth,trim={4cm 4cm 4cm 4cm},clip]{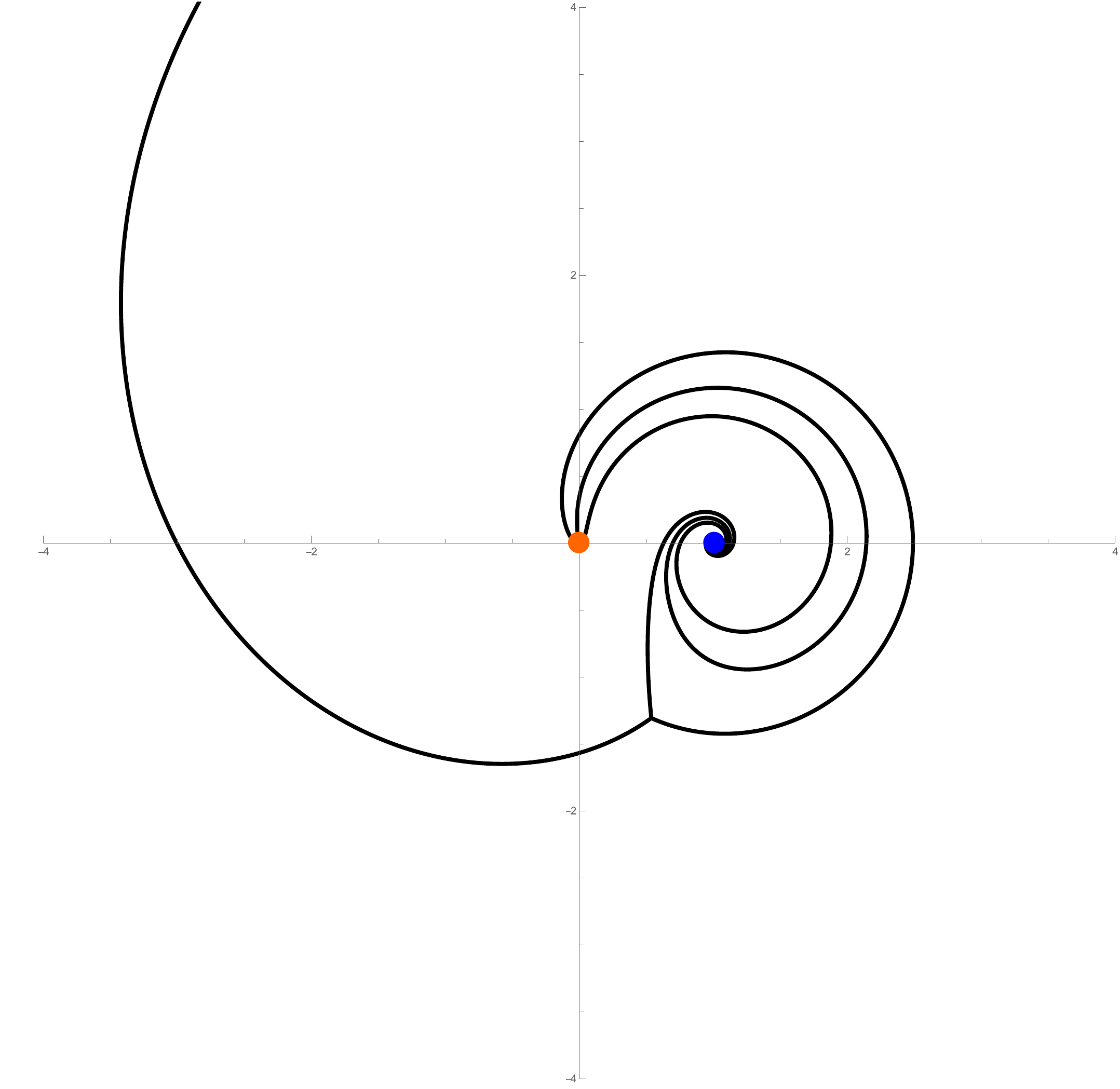}
        \caption{$\vartheta\approx2.356$}
      \label{fig:dhg8}
    \end{subfigure}
        \begin{subfigure}[t]{.27\textwidth}
        \centering
        \includegraphics[width=\linewidth,trim={3cm 3cm 3cm 3cm},clip]{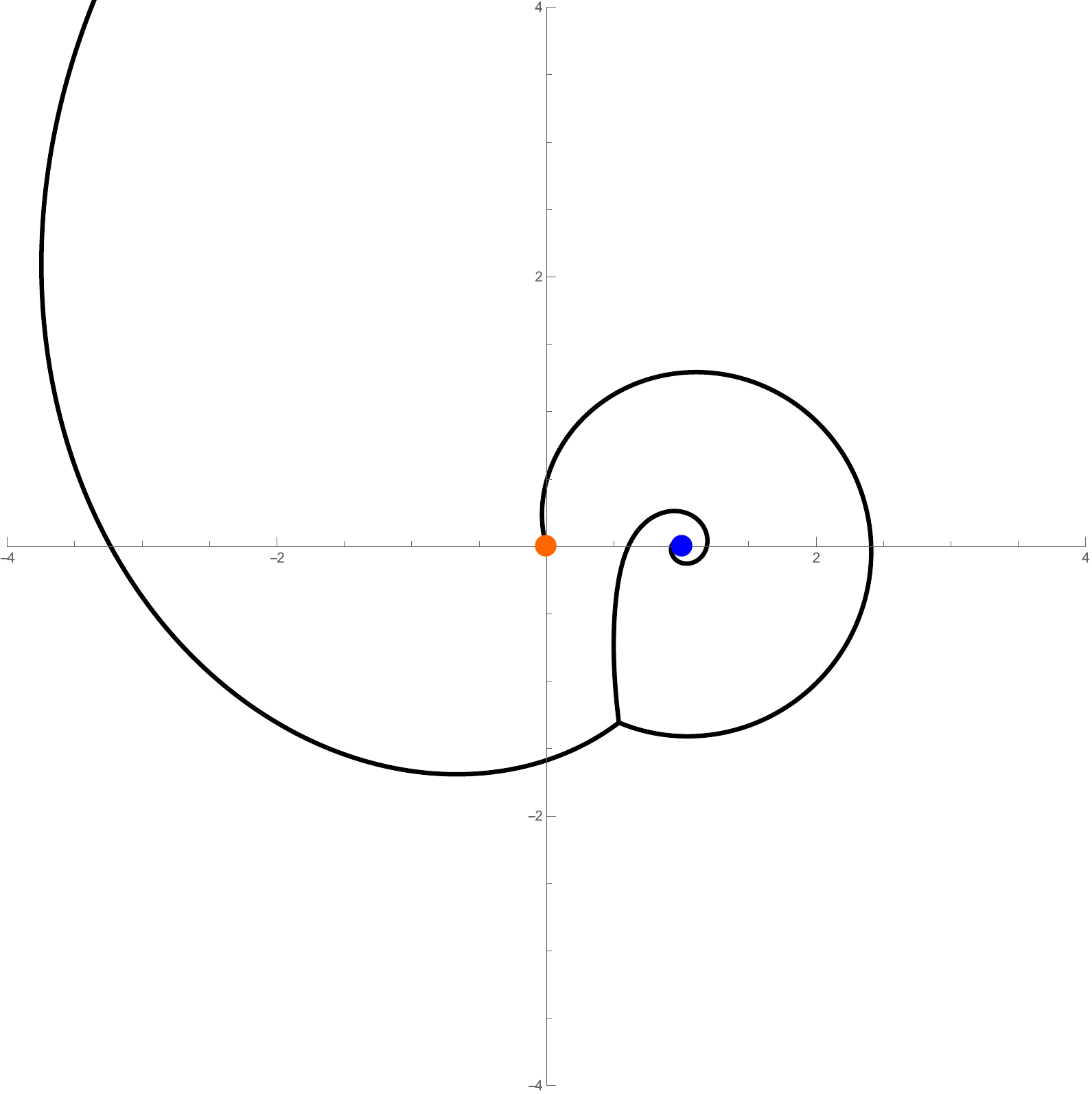}
        \caption{Type II saddle, $\vartheta\approx2.396$}
      \label{fig:dhg9}
    \end{subfigure}
            \begin{subfigure}[t]{.27\textwidth}
        \centering
        \includegraphics[width=\linewidth,trim={0.6cm 0.6cm 0.6cm 0.6cm},clip]{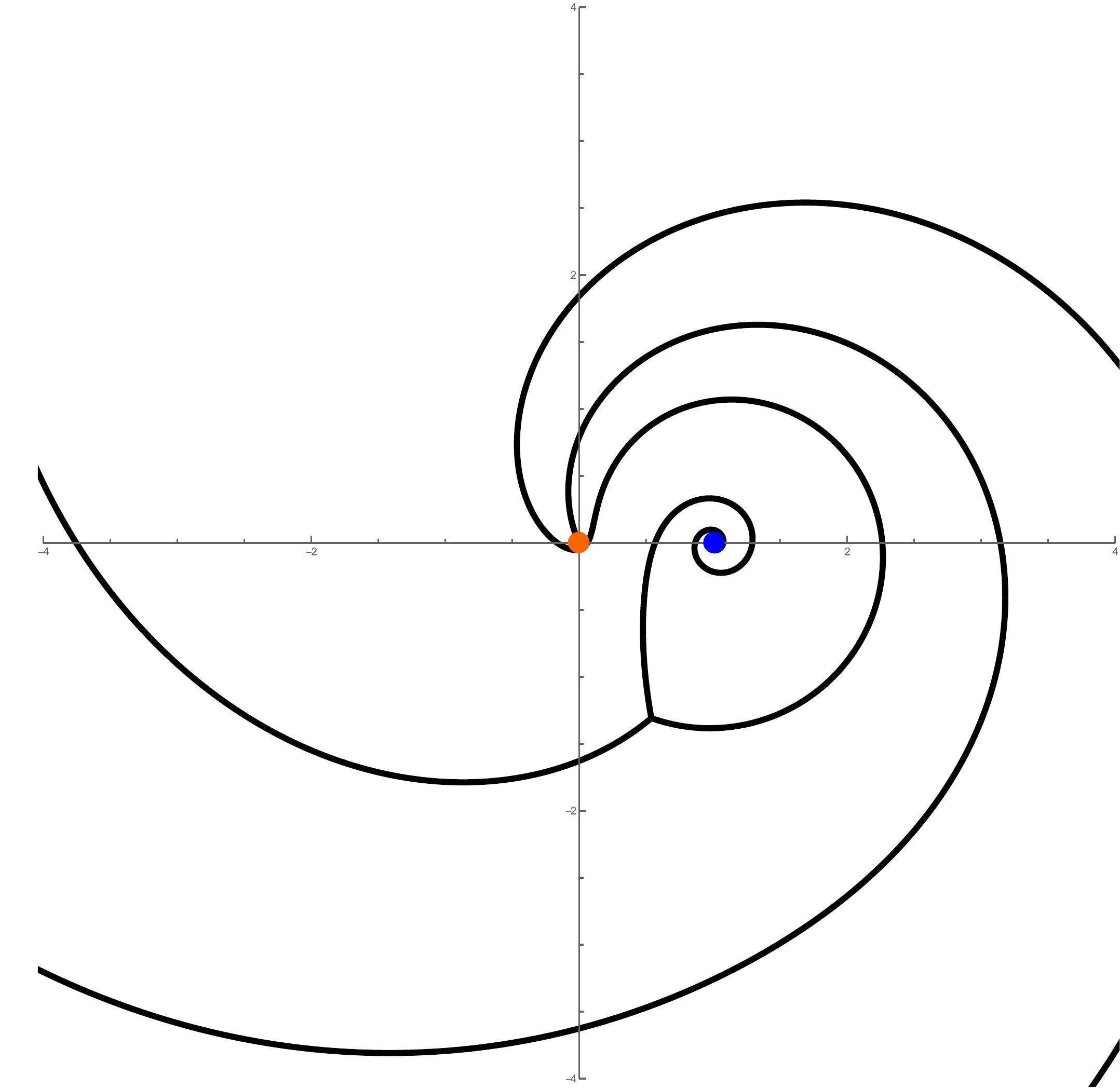}
        \caption{$\vartheta\approx2.474$}
      \label{fig:dhg10}
    \end{subfigure}
            \begin{subfigure}[t]{.27\textwidth}
        \centering
        \includegraphics[width=\linewidth,trim={3cm 3cm 3cm 3cm},clip]{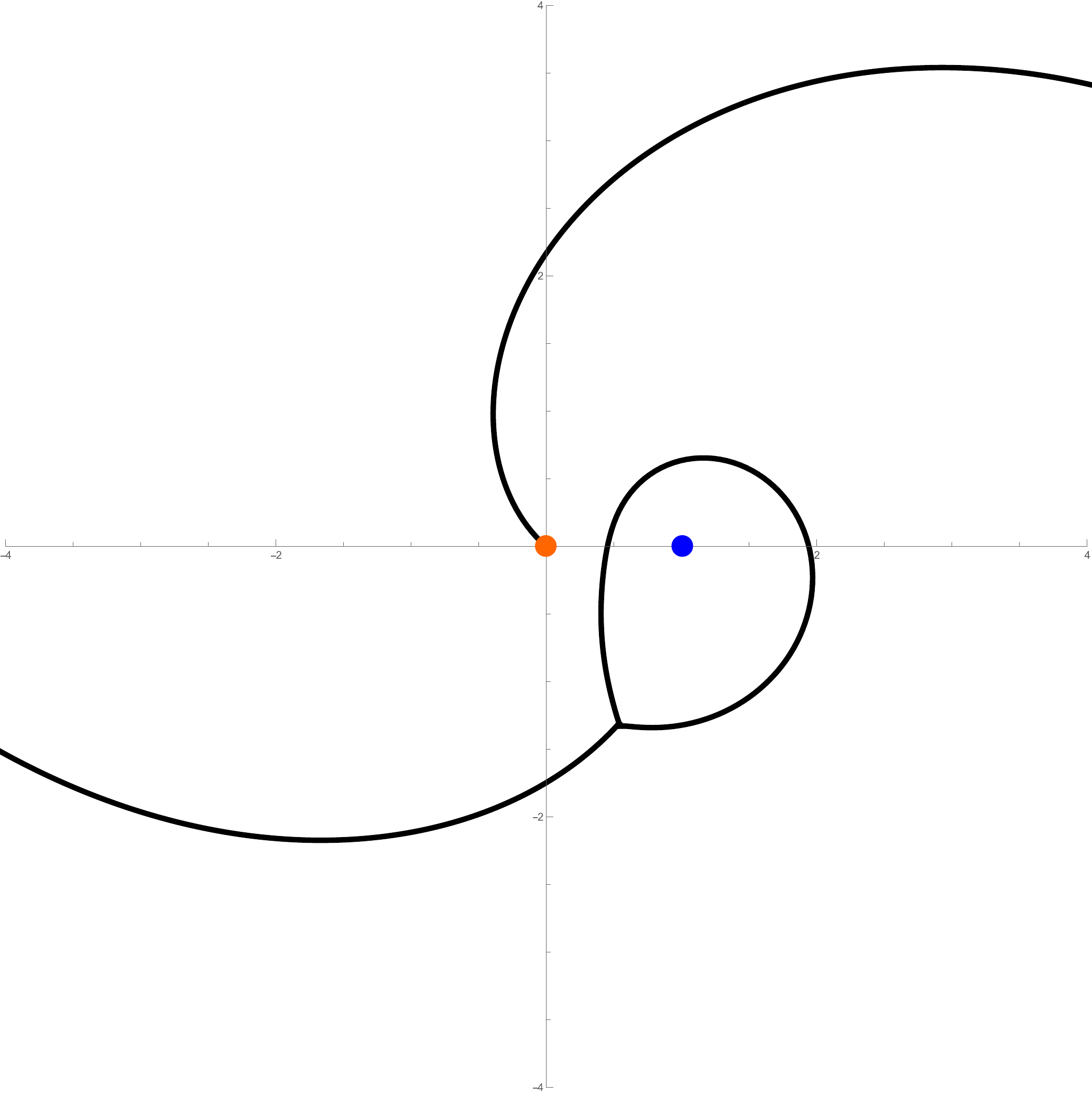}
        \caption{Type IV saddle, $\vartheta\approx2.68$}
      \label{fig:dhg11}
    \end{subfigure}
    
    \caption{Spectral networks for $\varphi_{\mathrm{dHG}}$, $m_1\approx -0.53-0.28i$, $m_\infty\approx-0.32-0.63i$.}
\label{fig:dhgpic1}
\end{figure}

  \begin{prop}
Fix ${\bm m} \in M_{\rm dHG}'$. The BPS spectrum of $\varphi_{\rm dHG}$ consists of exactly the eight BPS cycles given in Table \ref{table:bpsangles-dhg}.
\end{prop}
\begin{proof}
{  
It is easy to check combinatorially that the only possible nondegenerate spectral network, up to permutation of $0$ and $\infty$, is as shown in Figure \ref{fig:dhgproofpic}. 
\begin{figure}[h]
    \centering
        \includegraphics[width=0.2\linewidth,trim={0.0cm -0.0cm 0cm 0cm}]{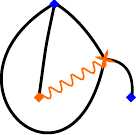}
           \caption{The only possible non-degenerate spectral network for $\varphi_{\rm dHG}$, up to permutation of $1$ and $\infty$}
      \label{fig:dhgproofpic}

\end{figure}

This is because the separating trajectory attached to the simple pole must be a boundary of a horizontal strip; thus the two of separating trajectories attached to the simple zero must have the same endpoint (either $1$ or $\infty$). 
Then the last separating trajectory must end at the other pole, hence we have the topological type shown in Figure \ref{fig:dhgproofpic}.

}

Similar to the previous examples, we may carefully compute the central charges of dual cycles and conclude that if either $2 \pi i (m_1 \pm m_\infty)$ is real, there must be a saddle. The result follows since ${\bm m }$ was generic.
\end{proof}

Together with the contribution from second order poles, we have Table \ref{table:bpsangles-dhg} which summarizes the BPS spectrum of $\varphi_{\rm dHG}$. 

 \begin{table}[h]
  \begin{tabular}{|c||c|c|} \hline
    $\vartheta_{\rm BPS}$
    & $\arg{(\epsilon m_1 + \epsilon' m_\infty)} + \pi/2$ 
    & $\arg m_s \pm \pi/2$   
    \\ \hline 
    degeneration
    & type II saddle
    & degenerate ring domain  \\ \hline
    $\gamma_{\rm BPS}$ & $\gamma_{1_\epsilon} + \gamma_{\infty_{\epsilon'}}$ & 
    $\gamma_{s_\pm} - \gamma_{s_\mp}$ \\ \hline
    $Z(\gamma_{\rm BPS})$ &
    $ 2 \pi i (\epsilon m_1 + \epsilon' m_\infty) $ &  
    $ \pm 4 \pi i m_s $  \\ \hline
    $\Omega(\gamma_{\rm BPS})$ &
    $+2$ &  $-1$ \\ \hline
  \end{tabular} 
   \vspace{+1.em}
     \caption{ 
     The BPS spectrum of $\varphi_{\mathrm{dHG}}$, where $\epsilon, \epsilon' \in \{\pm \}$, and $s \in \{1,\infty\}$.}
     \label{table:bpsangles-dhg}
  \end{table}

\subsubsection{\bf BPS structure from the Legendre curve} 
\label{section:Legendre}
Finally, let us conside the quadratic differential $\varphi_{\rm Leg} = Q_{\rm Leg}(x) dx^2$ obtained from the Legendre equation, where
\begin{equation}
Q_{\rm{Leg}}(x)= \dfrac{m_\infty^2}{x^2 -1}. 
\end{equation}
Under the assumption $m_{\infty} \in M_{\rm Leg} = \mathbb{C}^\ast$, $\varphi_{\rm Leg}$ has two simple poles at $\pm 1$, and a second order pole at infinity. 
The associated (partially compactified) Legendre curve $\widetilde{\Sigma}_{\rm{Leg}}$ is of genus $0$ with two punctures at $\infty_\pm$. 

{ 
This example contains the most unusual behaviour of the degenerations, and the BPS spectrum in this case may be unfamiliar from typical examples in physics. 
In the range $\vartheta \in [0, \pi)$, we may observe that there is a unique phase $\vartheta$ for which the spectral network is degenerate. This degeneration is shown in Figure \ref{fig:leg4}, in which a type III saddle appears between the two simple poles $\pm 1$. However, this is not the end of the story - a family of closed trajectories fill the complement of the type III saddle; that is, a degenerate ring domain around infinity also appears at the same phase (Figure \ref{fig:leg4}). Indeed, our final result cannot hold if we do not regard the ring domain itself as a BPS state. This viewpoint is also justified from the point of view of the Voros symbol jump property in the sequel to this paper.

We understand that the simultaneous degeneration is a special feature of this case among our examples; this is the reason why we treat this example separately in Definition \ref{def:generic-locus}. On the other hand, for more general quadratic differentials we may expect this kind of behaviour to appear whenever two simple poles are present. We arrived at this understanding thanks to the work \cite{KKT14} where the WKB-theoretic transformation to the Legendre equation was discussed. 
}


\begin{figure}[h]
    \centering
    \begin{subfigure}[t]{.27\textwidth}
        \centering
        \includegraphics[width=\linewidth,trim={1.1cm 1.1cm 1.1cm 1.1cm},clip]{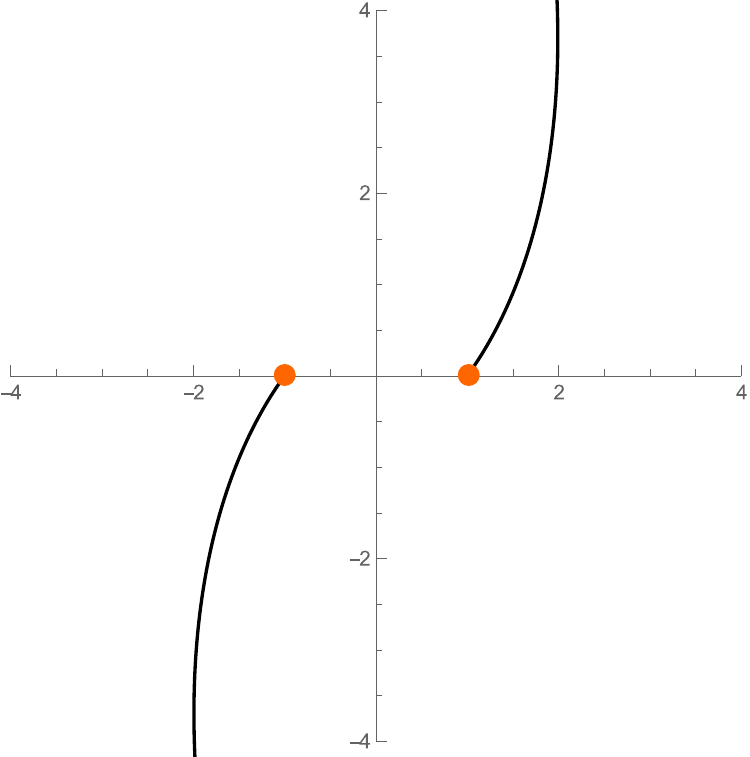}
           \caption{$\vartheta \approx 0.471$}
      \label{fig:leg1}
    \end{subfigure}
    \hspace{1.cm}                  
    \begin{subfigure}[t]{.27\textwidth}
        \centering
        \includegraphics[width=\linewidth,trim={1.1cm 1.1cm 1.1cm 1.1cm},clip]{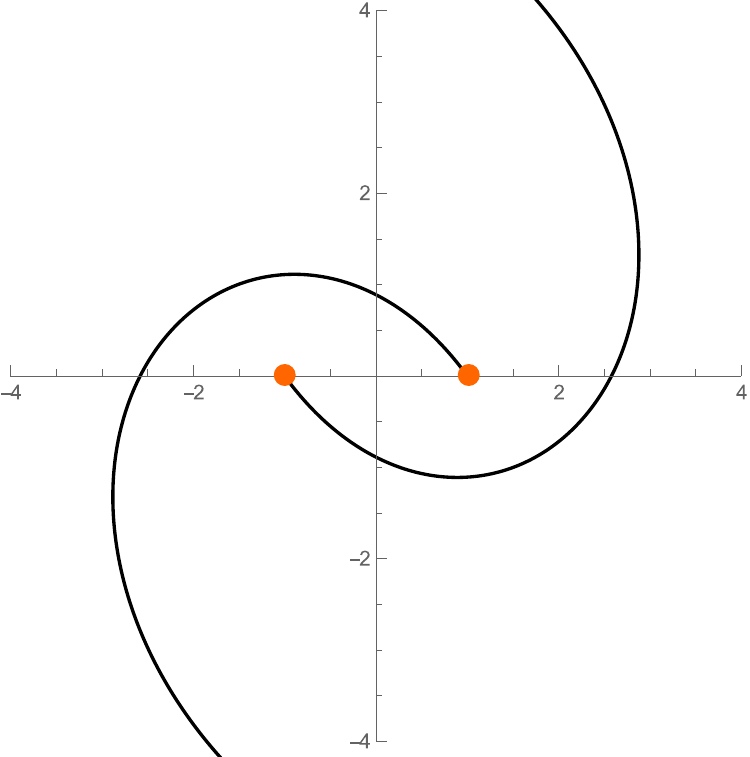}
        \caption{$\vartheta \approx 1.100$}
      \label{fig:leg2}
    \end{subfigure}
    \hspace{1.cm}                  
    \begin{subfigure}[t]{.27\textwidth}
        \centering
        \includegraphics[width=\linewidth,trim={1.1cm 1.1cm 1.1cm 1.1cm},clip]{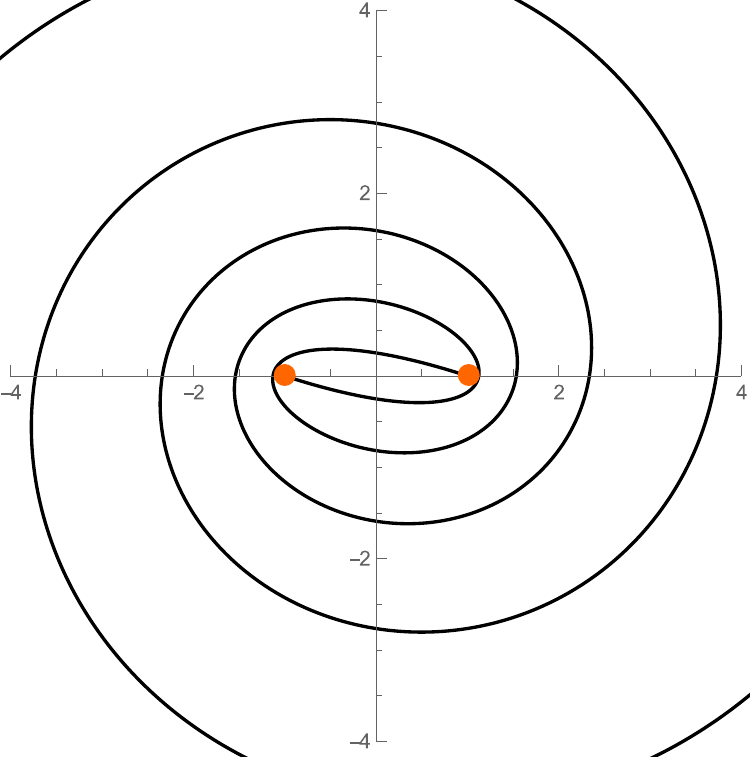}
        \caption{$\vartheta\approx1.414$}
      \label{fig:leg3}
    \end{subfigure}
    \\[+1.em]                 
    \begin{subfigure}[t]{.27\textwidth}
        \centering
        \includegraphics[width=\linewidth,trim={1.1cm 1.1cm 1.1cm 1.1cm},clip]{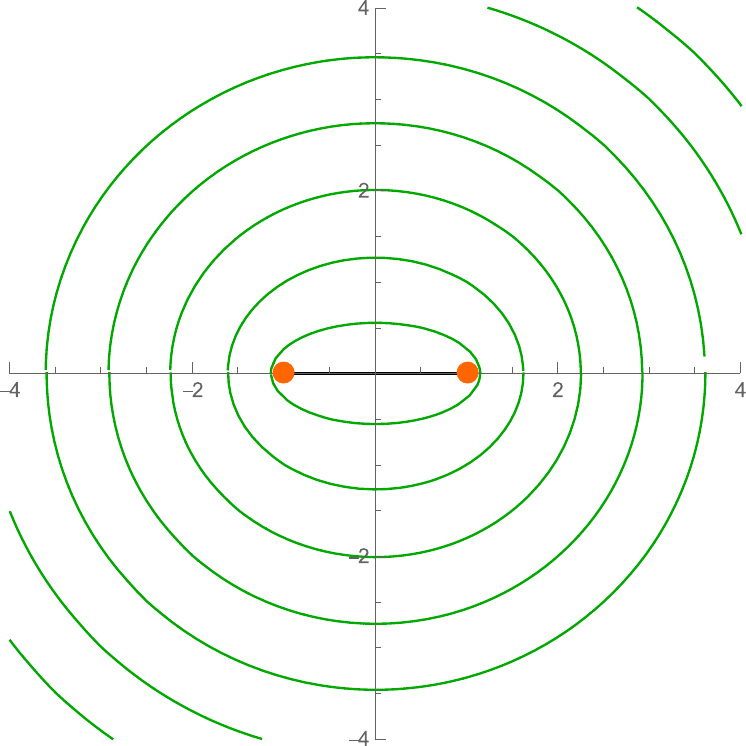}
        \caption{Type III saddle, $\vartheta\approx 1.571$. Closed trajectories forming a degenerate ring domain are shown in green.}
      \label{fig:leg4}
    \end{subfigure}
    \hspace{1.cm}                  
    \begin{subfigure}[t]{.27\textwidth}
        \centering
        \includegraphics[width=\linewidth,trim={1.1cm 1.1cm 1.1cm 1.1cm},clip]{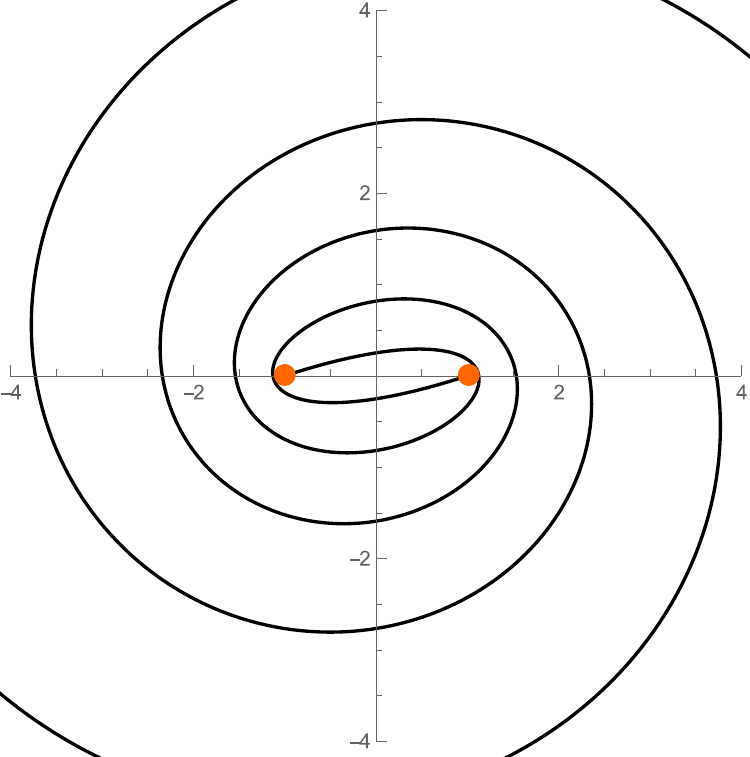}
        \caption{$\vartheta \approx 1.728$}
      \label{fig:leg5}
    \end{subfigure}
    \hspace{1.cm}                  
    \begin{subfigure}[t]{.27\textwidth}
        \centering
        \includegraphics[width=\linewidth,trim={1.1cm 1.1cm 1.1cm 1.1cm},clip]{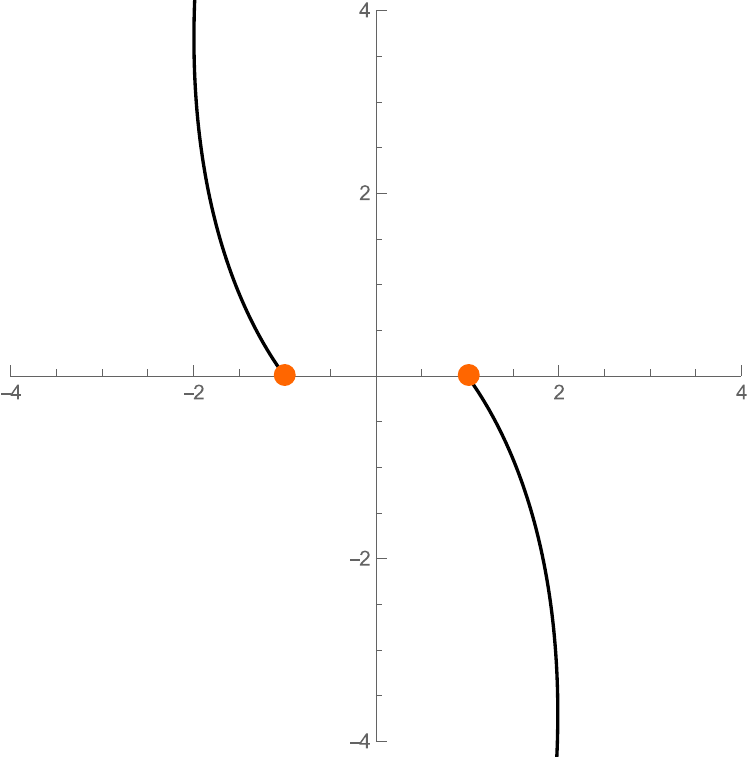}
        \caption{$\vartheta\approx  2.670$}
      \label{fig:leg6}
    \end{subfigure}
    \caption{Spectral networks for $\varphi_{\rm Leg}$ with $m_\infty \approx i$}
    \label{fig:leg}
\end{figure}

 \noindent We have
 
\begin{prop}
Fix $m_\infty \in M_{\rm Leg} = M'_{\rm Leg}$. 
Then, $\varphi_{\rm Leg}$ has exactly one degenerate spectral network in the range $\vartheta \in [0,\pi)$, appearing at $\vartheta = \arg m_\infty + \pi/2$ (mod $\pi$).
It contains a type III saddle with the associated BPS cycle $\gamma_{\infty_\pm}$, and a degenerate ring domain around $\infty$ with the associated BPS cycle $\gamma_{\infty_\pm} - \gamma_{\infty_\mp}$, simultaneously.
\end{prop}
\begin{proof}
A similar technique used in the Weber case allows us to show that the straight line between two simple poles is the unique saddle trajectory of the phase $\arg m_\infty + \pi/2$ (mod $\pi$). 
The phase coincides with the argument of $2 \pi i \Res_{x=\infty} \sqrt{Q_{\rm Leg}(x)} \, dx = \pm 2 \pi i m_\infty$ (mod $\pi$), and hence, a degenerate ring domain must appear around the infinity (c.f., Proposition \ref{prop:loops}). 
This proves that degeneration of the spectral network only occurs at the single phase which is given above. 
\end{proof}

In summary, we have Table \ref{table:bpsangles-Leg}. 

\begin{table}[h]
  \begin{tabular}{|c||c|c|c|c|c|c|c|c|} \hline
    $\vartheta_{\rm BPS}$ 
    & $\arg{m_\infty} \pm \pi/2$ 
    & $\arg{m_\infty} \pm \pi/2$ 
 
    \\ \hline 
    degeneration
    & type III saddle
    & degenerate ring domain  \\ \hline
     $\gamma_{\rm BPS}$
    & $\gamma_{\infty_\pm}$
    & $\gamma_{\infty_\pm}-\gamma_{\infty_\mp}$   \\ \hline
        $Z(\gamma_{\rm BPS})$ &
    $\pm 2 \pi  i m_{\infty} $ & $\pm 4 \pi i m_\infty $ \\ \hline
    $\Omega(\gamma_{\rm BPS})$ &
    $+4$ & $-1$ \\ \hline
  \end{tabular} 
   \vspace{+1.em}
     \caption{The BPS spectrum of $\varphi_{\mathrm{Leg}}$.}
     \label{table:bpsangles-Leg}
  \end{table}

\section{Free energy and BPS spectrum}
\label{sec:final}

We may now present our main result on the relationship between the BPS indices $\Omega(\gamma)$ and the TR free energies for spectral curves 
of hypergeometric type. 

\subsection{Formula for the free energy}

In the previous section, we gave a description of the BPS structure arising from the spectral curves of hypergeometric type.

It remains only to compare these values with Table \ref{table:free-energy-0} of expressions of genus $g$ free energies, which yields 

\begin{prop}
For any $\bullet \in \{{\rm HG}, {\rm dHG}, {\rm Kum}, {\rm Leg}, 
{\rm Bes}, {\rm Whi}, {\rm Web}\}$, fix ${\bm m}\in M'_\bullet$, and any half plane $\mathbb{H}$ whose boundary rays are not BPS. Then we have the following:
\begin{itemize}
\item[(1)] 
The genus $0$ free energy $F_0 = F_0^{\bullet}$ computed from TR is expressed as
\begin{equation} \label{eq:F0-from-BPS-indices}
F_0({\bm m}) \equiv \sum_{\substack{ \gamma \in \Gamma \\ Z_{\bm m}(\gamma)\in  \mathbb{H}}} 
\Omega(\gamma)\cdot \dfrac{1}{2} 
\left( \frac{Z_{\bm m}(\gamma)}{2 \pi i } \right)^2
\log \left( \frac{Z_{\bm m}(\gamma)}{2 \pi i} \right) 
\end{equation}
modulo degree two polynomials of the mass parameter ${\bm m}$.

\item[(2)]
The genus $1$ free energy $F_1 = F_1^{\bullet}$ computed from TR is expressed as
\begin{equation}
F_1({\bm m}) \equiv - \frac{1}{12} 
\log \left(
\prod_{\substack{ \gamma \in \Gamma \\ Z_{\bm m}(\gamma)\in \mathbb{H}}} 
 \left(\frac{Z_{\bm m}(\gamma)}{2 \pi i} \right)^{\Omega(\gamma)} 
 \right)
\end{equation}
modulo additive constants.

\end{itemize}
\end{prop}

\begin{rem}
The expression \eqref{eq:F0-from-BPS-indices} of $F_0$ essentially agree with the {\em prepotential} defined from the lenear data of the {\em Joyce structure} associated with a finite uncoupled BPS structures (see \cite[\S 8.7]{Bri19-2}).  
The expression is also strikingly similar to the solution of the WDVV equation arising from Veselov's $\vee$-systems \cite{Ves99, Ves00} (when $\Omega(\gamma_{\rm BPS}) \in \{0, 1 \}$). 
\end{rem}

The rest of the free energies have a uniform definition and expression.  
Again, we simply inspect the Table \ref{table:free-energy-0} and observe each piece of the sum in the expression for the free energies $F_{\rm{TR}}$ is just $(2\pi i /Z(\gamma_{\rm BPS}))^{2g-2}$ for a BPS cycle $\gamma_{\rm{BPS}}$, weighted by $\Omega(\gamma)$. Then in all our examples, we have:

\begin{thm}
\label{thm:maintheorem-2} For any ${\bm m}\in M'_\bullet$ and their corresponding BPS structures $(\Gamma,Z,\Omega)$,  the equality
 \begin{equation} \label{eq:finalformula}
 F_{g}({\bm m})=\dfrac{B_{2g}}{2g(2g-2)}
\sum_{\substack{ \gamma \in \Gamma \\ Z_{\bm m}(\gamma)\in \mathbb{H}}}\Omega(\gamma) \left( 
{\dfrac{2 \pi i}{Z_{\bm m}(\gamma)}} \right)^{2g-2}
\end{equation}
holds for all $g \ge 2$ and any half plane $\mathbb{H}$ whose boundary rays are not BPS.
\end{thm}
\begin{proof}
The claim follows immediately from Table \ref{table:free-energy-0} of the free energies and Tables \ref{table:bpsangles-Web}--\ref{table:bpsangles-Leg} of the BPS spectrum.
\end{proof}

\subsection{Conjectures}
\label{sec:conjectures}
So far, we only discussed BPS structures arising from a very specific collection of quadratic differentials. We may extend our considerations in two directions.

One the one hand, we may consider more general quadratic differentials. In this case, it is well-known that the \emph{wall-crossing} phenomenon may occur. In such a case, we likely cannot expect our result continues to hold, since the free energies themselves are continuous in the parameters ${\bm m}$, whereas $Z$ and $\Omega$ jump. On the other hand, we may expect that for quadratic differentials whose BPS structure is \emph{uncoupled}, the formula (\ref{eq:finalformula}) continues to hold. This expectation follows from the superposition structure of formula \eqref{eq:finalformula}, which applies on the BPS side to the general solution of the Riemann-Hilbert problem in the uncoupled case (cf. \cite{Bri19}).

The other direction is to consider the generalization to higher degree spectral curves. According to the work \cite{GMN12} by Gaiotto-Moore-Neitzke, it is natural to expect that a tuple of higher differentials (or a higher degree spectral curve) also defines a BPS structure. 
Although many examples of BPS invariants appearing for higher degree spectral curves are studied in \cite{GMN12},  
we still do not have mathematically rigorous description of the BPS structure in these cases because a higher analogue of Bridgeland-Smith's theory \cite{BS13} is missing. Nonetheless, studies of higher rank spectral networks and their degenerations have been made which offer starting points for further investigation \cite{BNR, AKT94, AKT98, GMN12-2, KNPS, MPY, LP, GLPY17, Sasaki17, HK}. On the other hand, TR is applicable to higher degree spectral curves \cite{BHLMR-12, BE-12}, and its relationship to WKB analysis is also discussed in \cite{BE-16} (under a certain admissibility assumption on spectral curves). 

Now, let us make a conjectural statement which generalizes our main results to this class of spectral curves. 
Suppose we are given a tuple 
\begin{equation}
(\varphi_1, \dots, \varphi_N) = (Q_1(x) dx, \dots, Q_N(x) dx^{N})
\end{equation} 
of meromorphic differentials on ${\mathbb P}^1$.
Here $\varphi_r = Q_r(x) dx^r$ is an $r$-differential for each $r \in \{1,\dotsm N \}$. 
Then, we have an associated spectral curve defined by 
\begin{equation} \label{eq:spectral-cover-higher}
\Sigma := \{ \lambda  ~|~ \lambda^{N} + \sum_{r=1}^{N} \varphi_r \lambda^{N-r} = 0  \}  \subset T^\ast {\mathbb P}^1.
\end{equation}
We also define the central charge by the same formula $Z(\gamma) = \oint_{\gamma} \lambda$ as before, and the BPS invariants $\Omega : \Gamma \to {\mathbb Q}$ (or ${\mathbb Z}$) are obtained by the weighted counting of degeneration\footnote{
Since the topological types of degenerate and nondegenerate spectral networks are not fully classified if the degree is greater than 2, the notion of ``degeneration" in the spectral networks is not mathematically rigorous at this moment. It will take further effort to formulate our conjecture more precisely.  
} 
in the associated spectral network as discussed in \cite{GMN12}. Then, we expect 

\begin{conj} \label{conj:1}
Suppose the BPS structure obtained from \eqref{eq:spectral-cover-higher} is uncoupled, and $\mathbb{H}$ is any half plane whose boundary rays are not BPS.
Then, we have the following expression of the $g$-th TR free energy of the spectral curve $\Sigma$: 
\begin{align} 
F_0 & \equiv \sum_{\substack{\gamma \in \Gamma \\ Z(\gamma) \in {\mathbb H}}} 
\frac{\Omega(\gamma)}{2} 
\left( \frac{Z(\gamma)}{2 \pi i } \right)^2
\log \left( \frac{Z(\gamma)}{2 \pi i} \right), \\ \\
F_1 & \equiv - \frac{1}{12} 
\log \left(
\prod_{\substack{\gamma \in \Gamma \\ Z(\gamma) \in {\mathbb H}}} 
 \left(\frac{Z(\gamma)}{2 \pi i} \right)^{\Omega(\gamma)} 
 \right), \\ \\
F_{g} & = \dfrac{B_{2g}}{2g(2g-2)} \sum_{\substack{ \gamma \in \Gamma \\ Z(\gamma)\in  \mathbb{H}}} 
\Omega(\gamma) \left( 
{\dfrac{2 \pi i}{Z(\gamma)}} \right)^{2g-2}
\qquad (g \ge 2) 
\end{align}
modulo the ambiguities in $g=0$ and $g=1$.
\end{conj}

In particular, at least for uncoupled BPS structures, this approach suggests an alternative route to the computation of BPS invariants, without using a spectral network, thus bypassing the need to solve any differential equation. While we do not know how much more complicated the story may be for general (coupled) BPS structures may be, we hope that this approach may offer a starting point for understanding the relationship to the topological recursion side.

\begin{rem}
We formulated our conjecture with the hypothesis that the BPS structure be uncoupled. On the other hand, we are unsure what class of spectral curves should be considered, or what condition on a spectral curve ensures it corresponds to an uncoupled BPS structure. It seems such a property should be closely related to the genus of the spectral curve, but further investigation is needed to understand any details.
\end{rem}

\subsection{Degree 3 examples}
In this last section, we will briefly look at two examples of degree 3 spectral curves (i.e., $N = 3$ in \eqref{eq:spectral-cover-higher}) for which the conjectures are numerically testable and appear to hold. 
The reader may require some terminology about the higher degree situation that can be found in e.g. \cite{AKT94, AKT98, GMN12}.

Recently, Y.\,M.\,Takei computed the TR free energies $F_g$ for some examples of 
degree $3$ spectral curves (\cite{YM-Takei20}). 
Here, by observing several figures of spectral networks, we give numerical evidence that our conjectures continue to hold in these cases.

\subsubsection{\bf The $(1,4)$ curve}

Let us consider the spectral curve arising as the classical limit of the 3rd order hypergeometric differential equation of type $(1,4)$, studied in \cite{Okamoto-Kimura, Hirose, Sasaki17}, which is explicitly given as follows:
\begin{equation}
    \Sigma_{(1,4)} ~:~ 3y^3 + 2 t y^2 + x y - m_\infty = 0.
    \label{eq:14curve}
\end{equation}
Here $m_\infty$ is a parameter assumed to be non-zero, and we regard $t$ as a constant\footnote{The parameter $t$ plays the other independent variable when we view the hypergeometric equation as a system of PDEs; see \cite{Okamoto-Kimura}}. 
The curve $\Sigma_{(1,4)}$ is a genus $0$ curve with two punctures at $\infty_\pm$, and we can verify that $\Res_{x=\infty_\pm} y dx = \pm m_\infty$. 

According to \cite[Theorem 4.6]{YM-Takei20} the $g$-th free energy is given explicitly by
\begin{align}
F_0(m_\infty,t) & = 
\frac{m_\infty^2}{4} \log(- 3 m_\infty^2) - \frac{3m_0^2}{4} + \frac{2t^3 m_0}{27} - \frac{t^2}{972} \equiv \frac{m_\infty^2}{2} \log m_\infty, \\
F_1(m_\infty,t) & = - \frac{1}{12} \log m_\infty, \\
F_g(m_\infty,t) & = \dfrac{B_{2g}}{2g(2g-2)}\dfrac{1}{m_\infty^{2g-2}} \quad (g\geq 2).
\end{align}

We performed a crude numerical experiment choosing values $t\approx-1.4-1.4i$, $m_\infty\approx -1.4-1.4i$, and obtained the networks in Figure \ref{fig:14curvefirsthalf} below.
In this figure, we see there is a degenerate spectral network in Figure \ref{fig:14curve-14}, which is the so-called ``three-string web" observed in \cite{GMN12}. 
Indeed, the spectral networks in this example have already been studied in \cite[\S 7.3]{GMN12} (see also \cite{Hirose, Sasaki17}).  
According to their result, the associated BPS cycle is nothing but the residue class $\gamma_{\infty_\pm}$, and the BPS index assigned so that  $\Omega(\gamma_{\infty_\pm}) = 1$.   
These observations agree with the statements of Conjecture \ref{conj:1} at least numerically.

\subsubsection{\bf The $(2,3)$ curve}

Our last example is the spectral curve arising as the classical limit of the 3rd order hypergeometric differential equation of type $(2,3)$, which was also studied in \cite{Okamoto-Kimura}. The curve is explicitly given as follows:
\begin{equation}
    \Sigma_{(2,3)} ~:~ 4y^3 -2 x y^2 + 2 m_\infty y -t = 0.
    \label{eq:23curve}
\end{equation}

Again, it is a genus $0$ curve with two punctures at $\infty_\pm$, and $\Res_{x=\infty_{\pm}} y \, dx = \pm m_\infty$ is also satisfied.
According to \cite[Theorem 4.6]{YM-Takei20} the free energy for $g \geq 2$ is given by

\begin{align}
F_0(m_\infty,t) & = 
\frac{m_\infty^2}{4} \log(- 2t) \equiv 0, \\
F_1(m_\infty,t) & = - \frac{1}{8} \log t \equiv 0, \\
F_g(m_\infty,t) & = 0 \quad (g\geq 2).
\end{align}

Thus, if we believe Conjecture \ref{conj:1}, we may expect to see no degenerations at all phases in the spectral network.

Again, we can choose values for the parameters and check the numerical evidence. Using the values $t\approx-1.4-1.1i$, $m_\infty\approx -1.8-1.6i$, we obtained the networks in Figure \ref{fig:23curvefirsthalf}.
We show values of $\vartheta$ surrounding what appear to be degenerations, but by the general rules for spectral networks and BPS states predicted from physics (e.g. \cite{GMN12}) and studied by others, these do not count as degenerate due to the ``type" of the trajectories colliding. 
Thus, if we agree with the physics definition of $\Omega$ for this particular family of networks, our expectation is true, and we have verified numerically Conjecture \ref{conj:1} in this case as well.  

\newpage 

\begin{figure}[h]
    \centering
    \begin{subfigure}[t]{.2\textwidth}
        \centering
        \includegraphics[width=\linewidth]{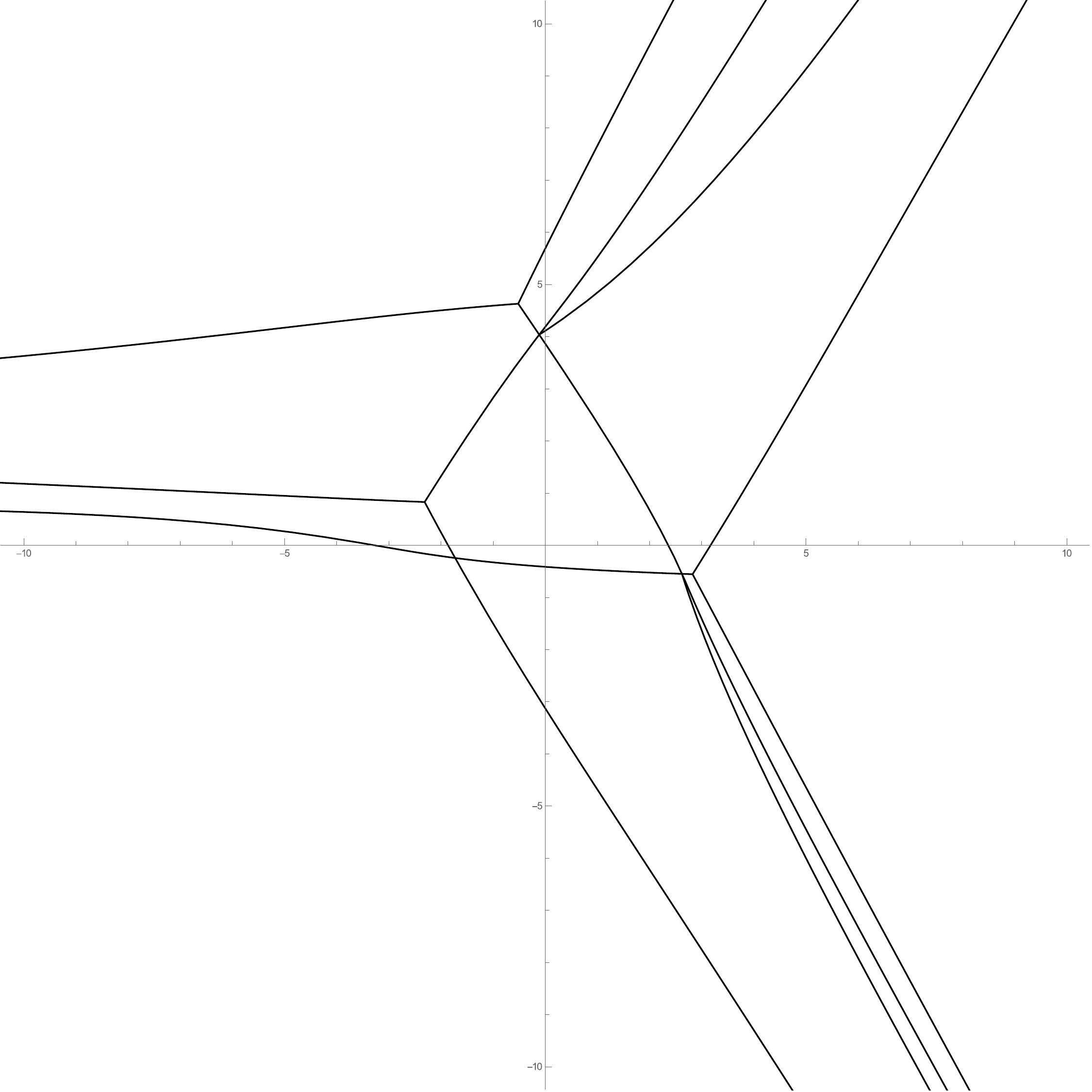}      
        \caption{$\vartheta \approx 0$}
      \label{fig:14curve-1}
    \end{subfigure}
    \hspace{0.5cm}               
    \begin{subfigure}[t]{.2\textwidth}
        \centering
        \includegraphics[width=\linewidth]{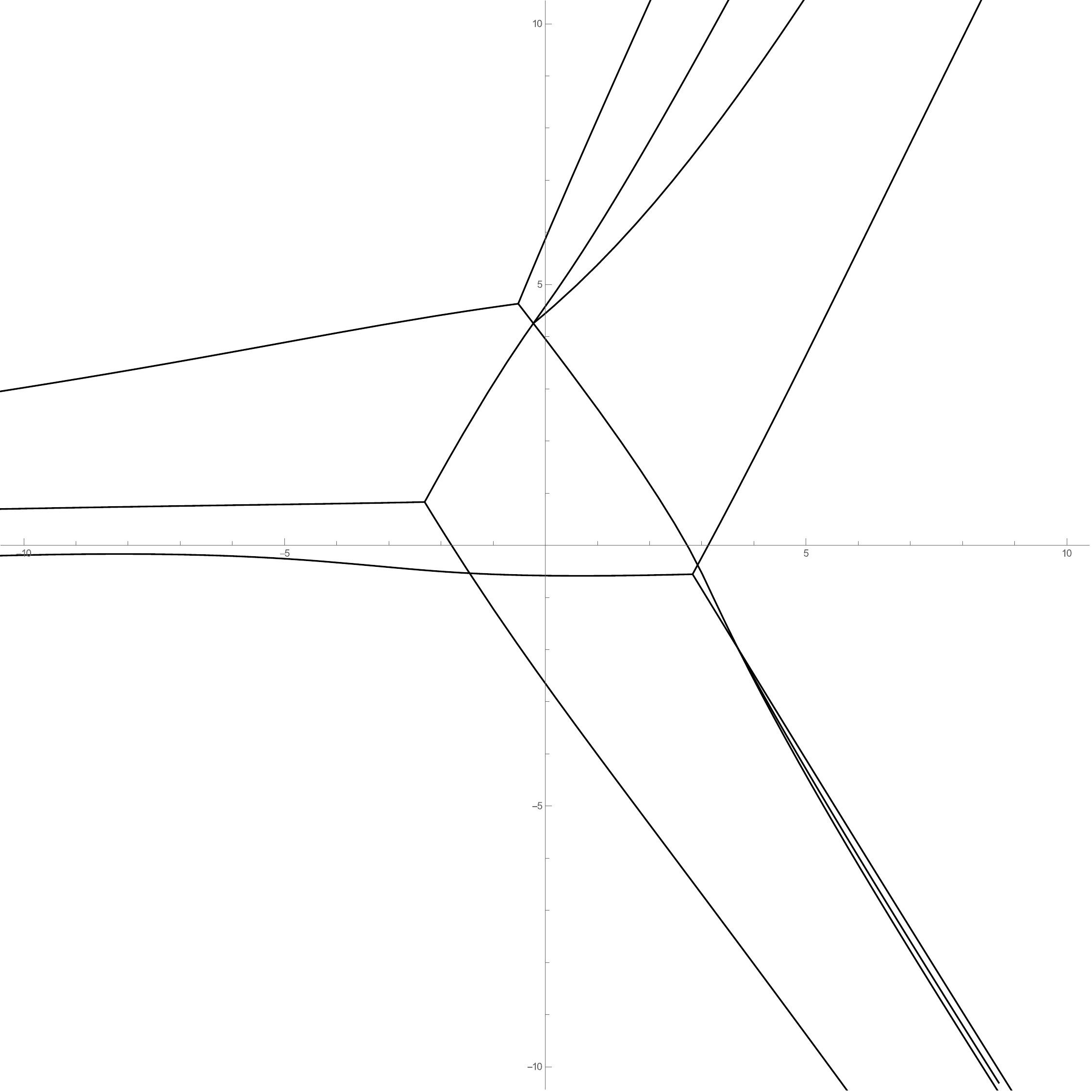}
        \caption{$\vartheta\approx0.09$}
      \label{fig:14curve-25}
    \end{subfigure}
    \hspace{0.5cm}             
    \begin{subfigure}[t]{.2\textwidth}
        \centering
        \includegraphics[width=\linewidth]{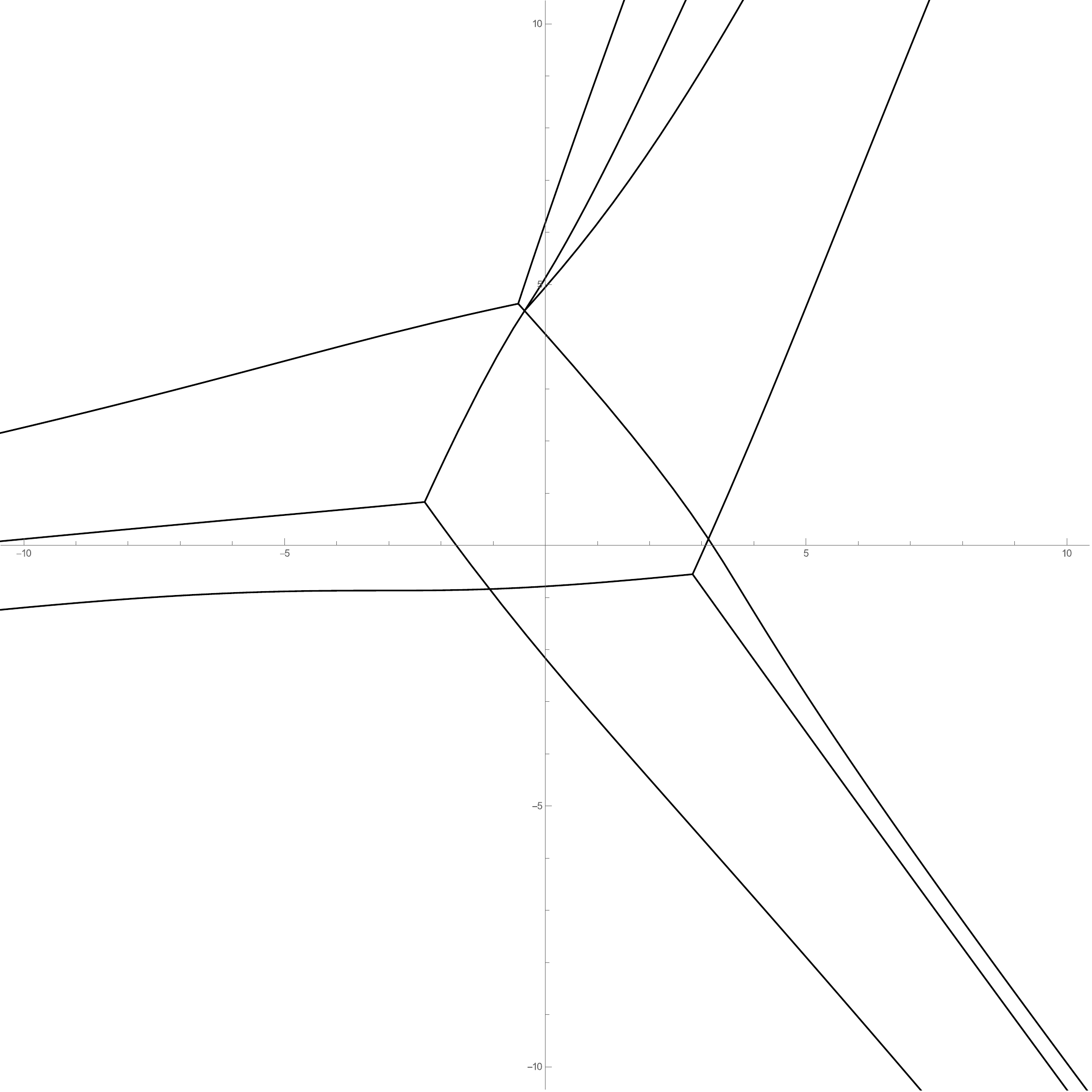}
        \caption{$\vartheta\approx0.2$}
      \label{fig:14curve-3}
    \end{subfigure}
     \hspace{0.5cm}              
    \begin{subfigure}[t]{.2\textwidth}
        \centering
        \includegraphics[width=\linewidth]{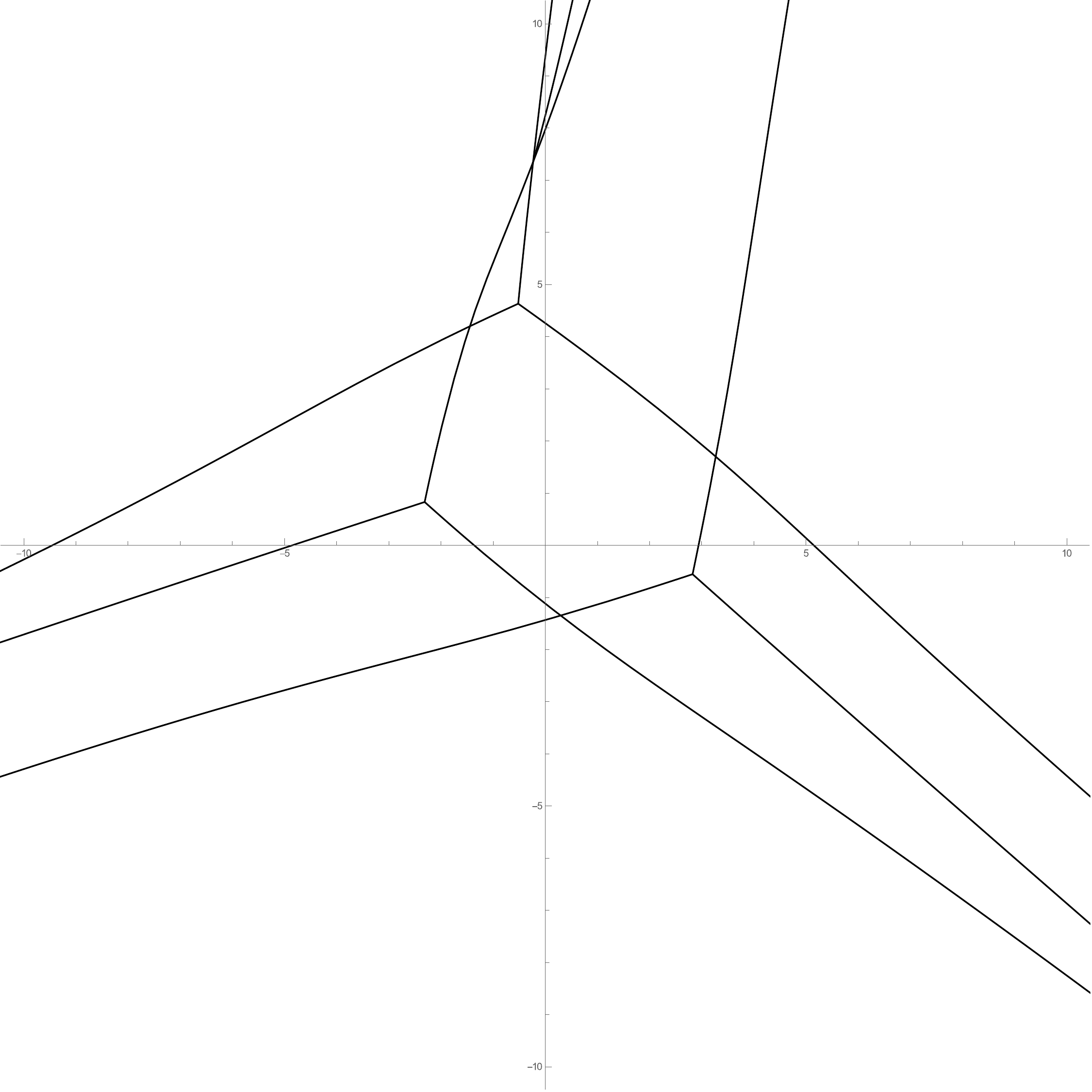}
        \caption{$\vartheta\approx0.53$}
      \label{fig:14curve-4}
    \end{subfigure}
\\[+1.em]
    \begin{subfigure}[t]{.2\textwidth}
        \centering
        \includegraphics[width=\linewidth]{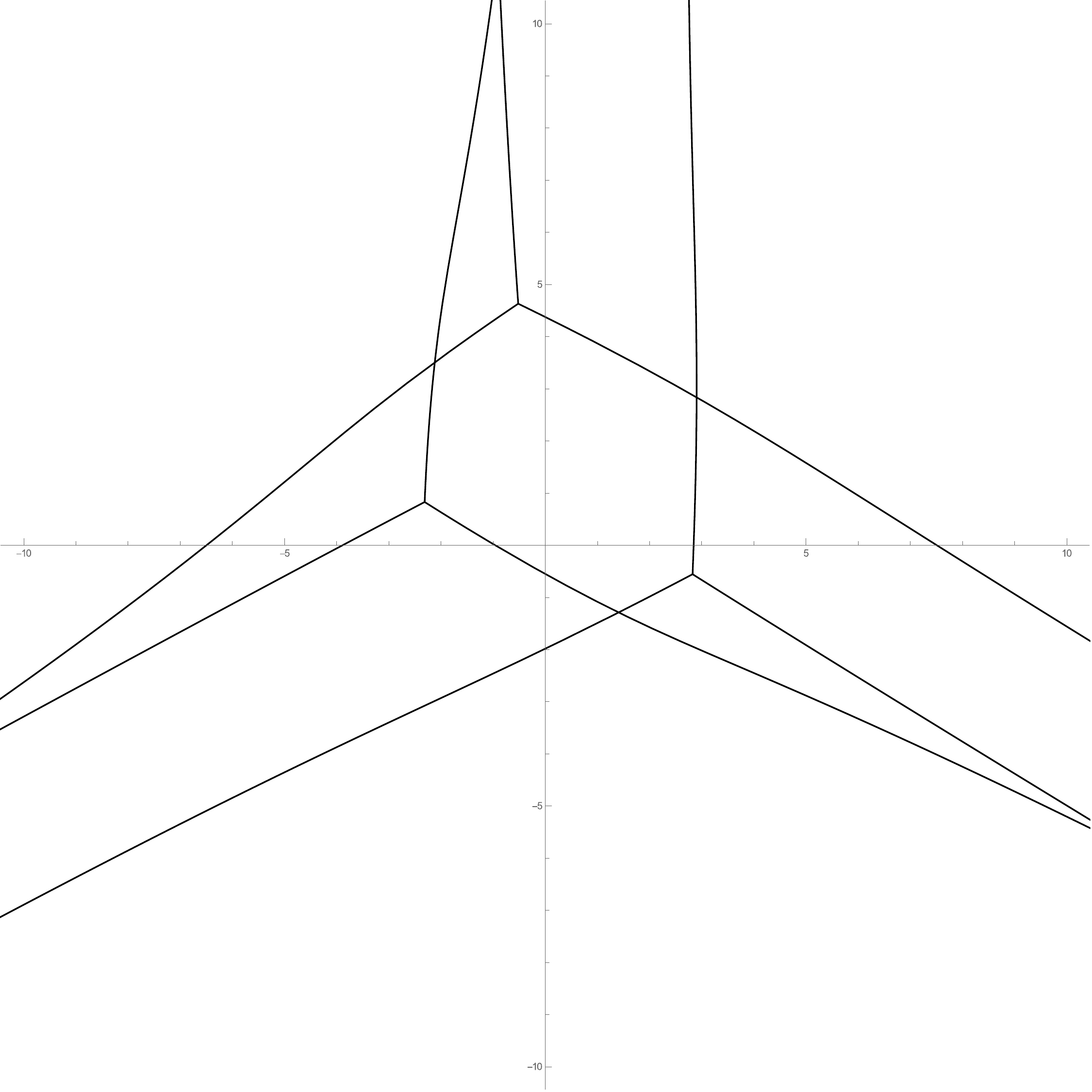}
        \caption{$\vartheta\approx 0.78$}
      \label{fig:14curve-5}
    \end{subfigure}
    \hspace{0.5cm}             
    \begin{subfigure}[t]{.2\textwidth}
        \centering
        \includegraphics[width=\linewidth]{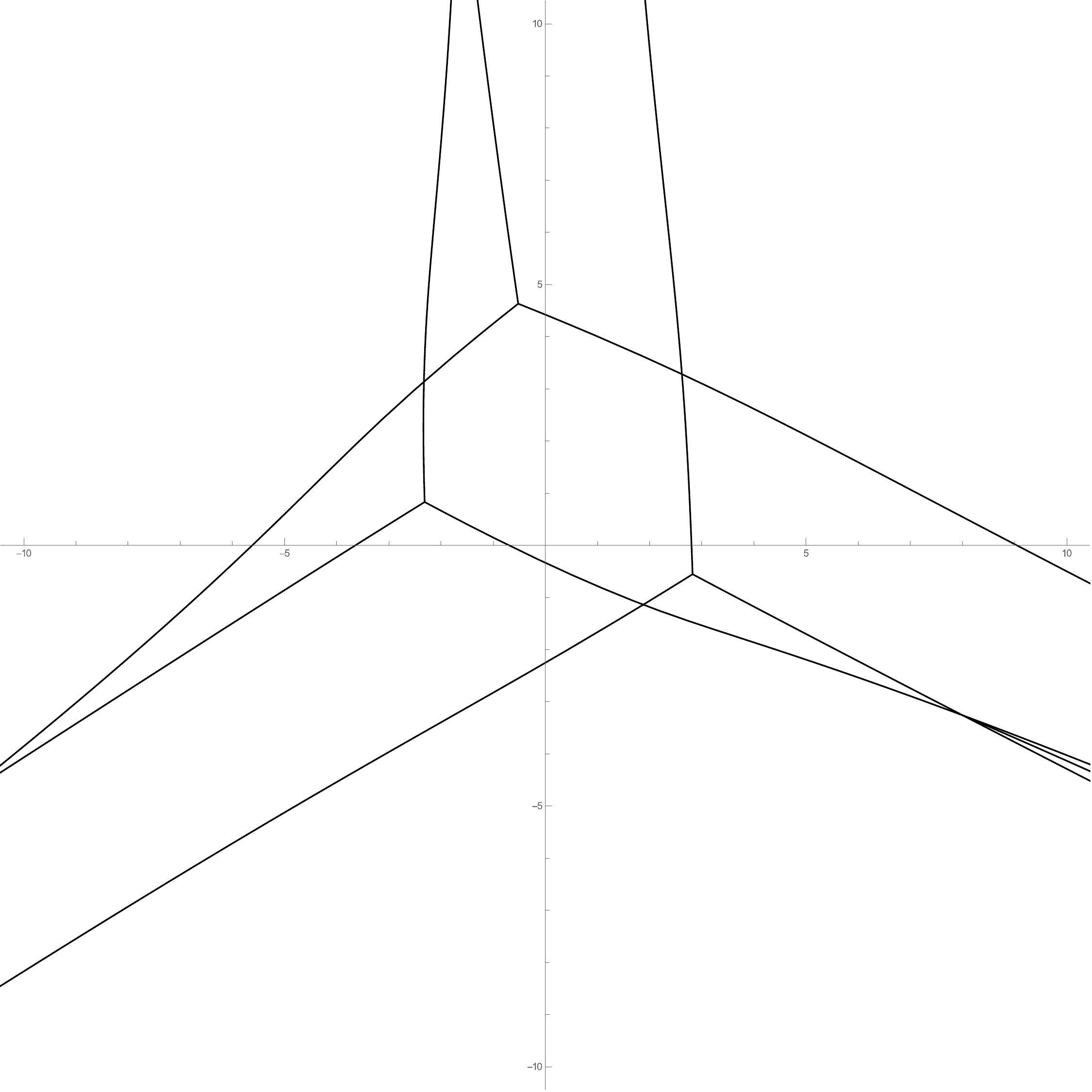}
        \caption{$\vartheta\approx 0.89$}
      \label{fig:14curve-6}
    \end{subfigure}
    \hspace{0.5em}
        \begin{subfigure}[t]{.2\textwidth}
        \centering
        \includegraphics[width=\linewidth]{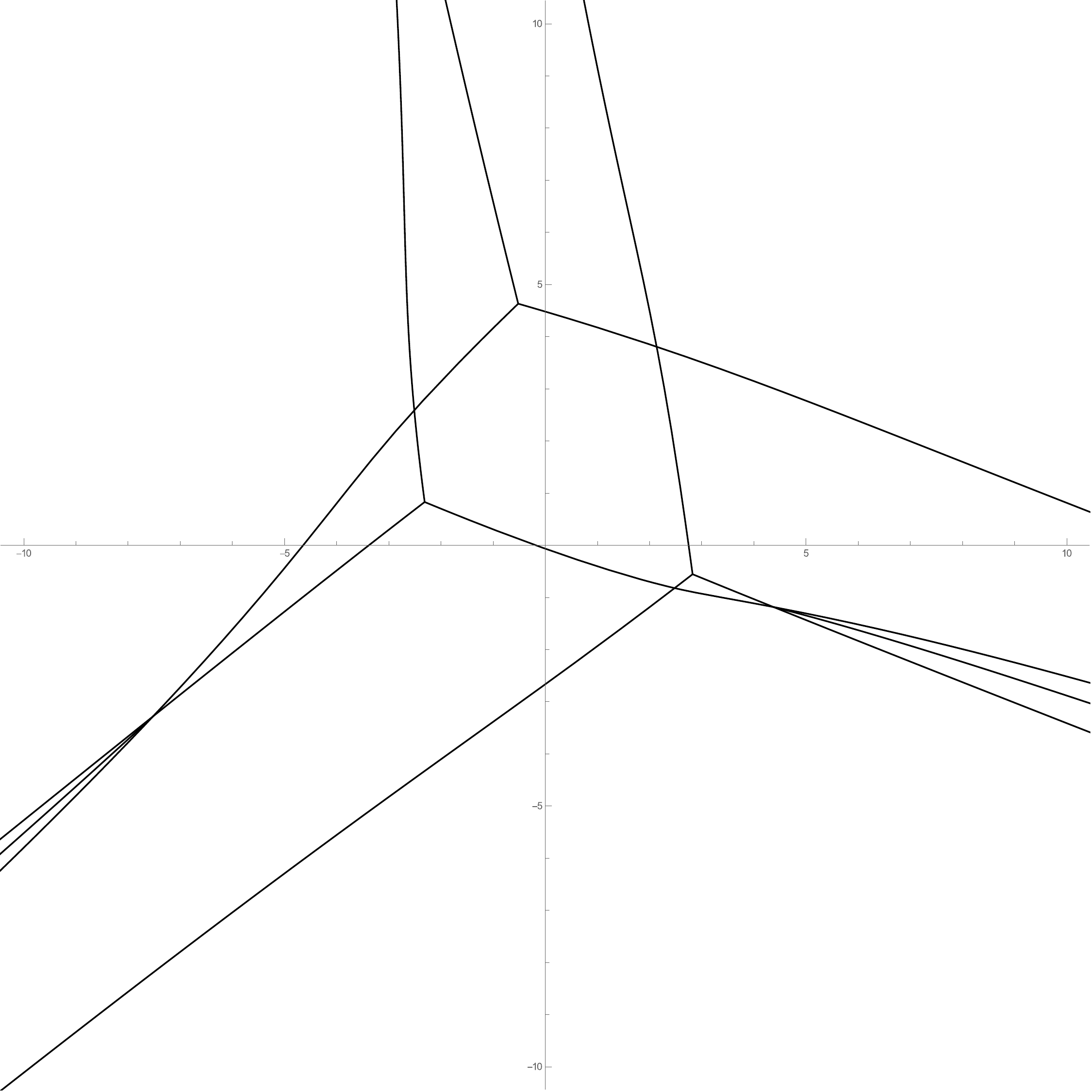}
        \caption{$\vartheta\approx 1.04$}
      \label{fig:14curve-7}
    \end{subfigure}
    \hspace{0.5cm}             
    \begin{subfigure}[t]{.2\textwidth}
        \centering
        \includegraphics[width=\linewidth]{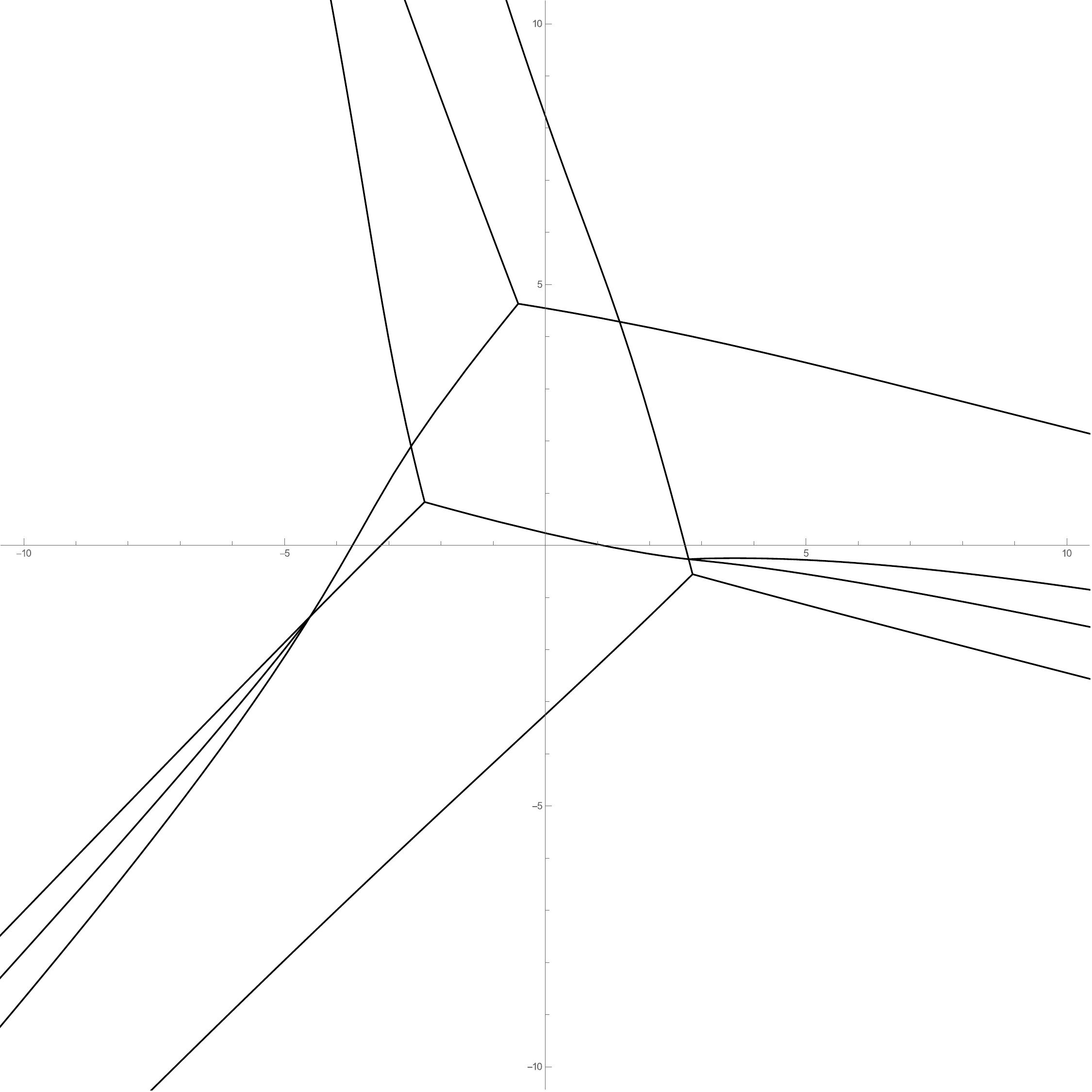}
        \caption{$\vartheta \approx 1.22$}
      \label{fig:14curve-8}
    \end{subfigure}
\\[+1.em]
        \begin{subfigure}[t]{.2\textwidth}
        \centering
        \includegraphics[width=\linewidth]{bps-rhp-14-param-retrelimtiml-1pt4,_-1pt4,_-1pt4,_-1pt4,_1pt04.pdf}
        \caption{$\vartheta\approx1.04$}
      \label{fig:14curve-9}
    \end{subfigure}
    \hspace{0.5cm}             
    \begin{subfigure}[t]{.2\textwidth}
        \centering
        \includegraphics[width=\linewidth]{bps-rhp-14-param-retrelimtiml-1pt4,_-1pt4,_-1pt4,_-1pt4,_1pt22.pdf}
        \caption{$\vartheta\approx 1.22$}
      \label{fig:14curve-10}
    \end{subfigure}
   \hspace{0.5cm}
    \begin{subfigure}[t]{.2\textwidth}
        \centering
        \includegraphics[width=\linewidth]{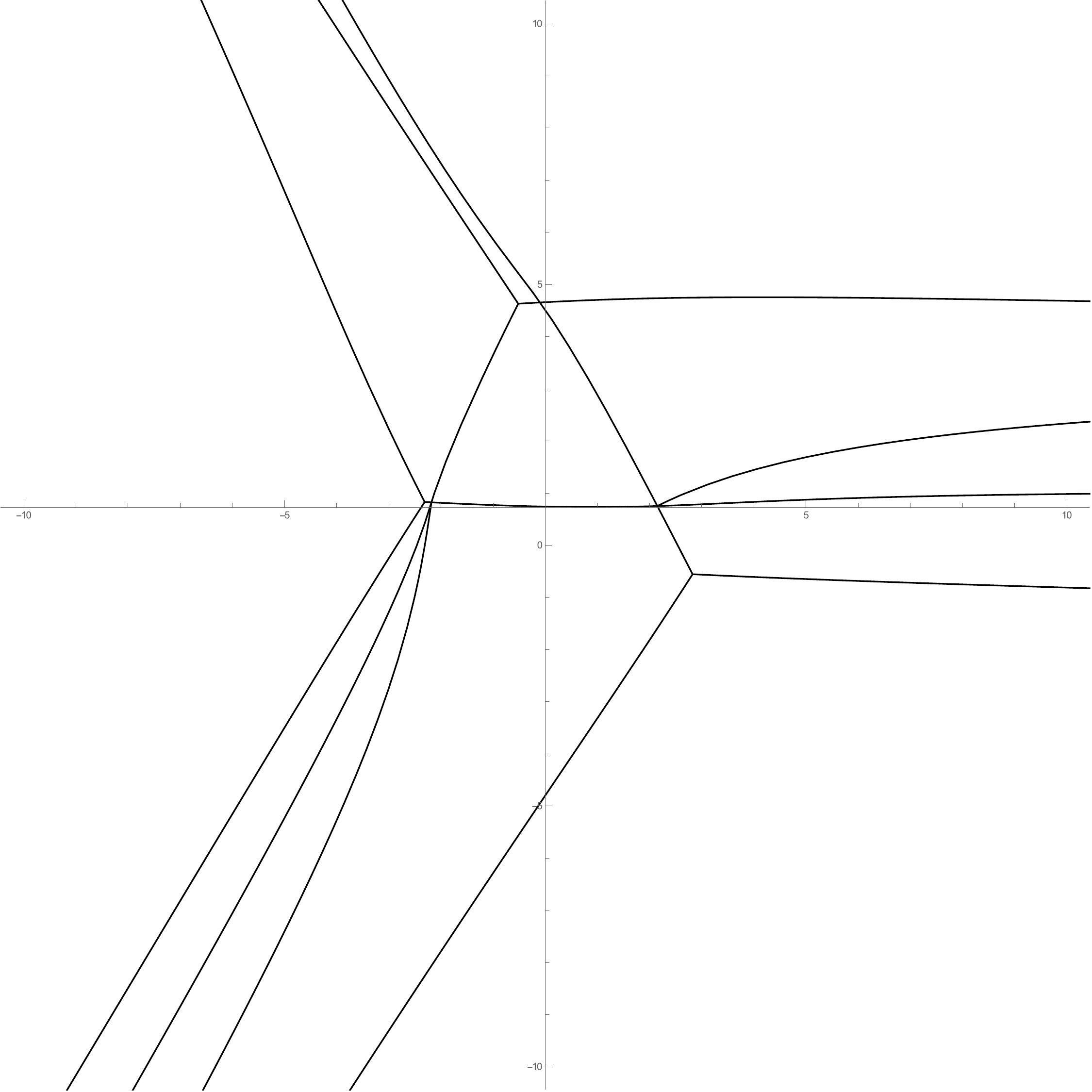}
           \caption{$\vartheta\approx1.55$}
      \label{fig:14curve-11}
    \end{subfigure}
    \hspace{0.5em}
    \begin{subfigure}[t]{.2\textwidth}
        \centering
        \includegraphics[width=\linewidth]{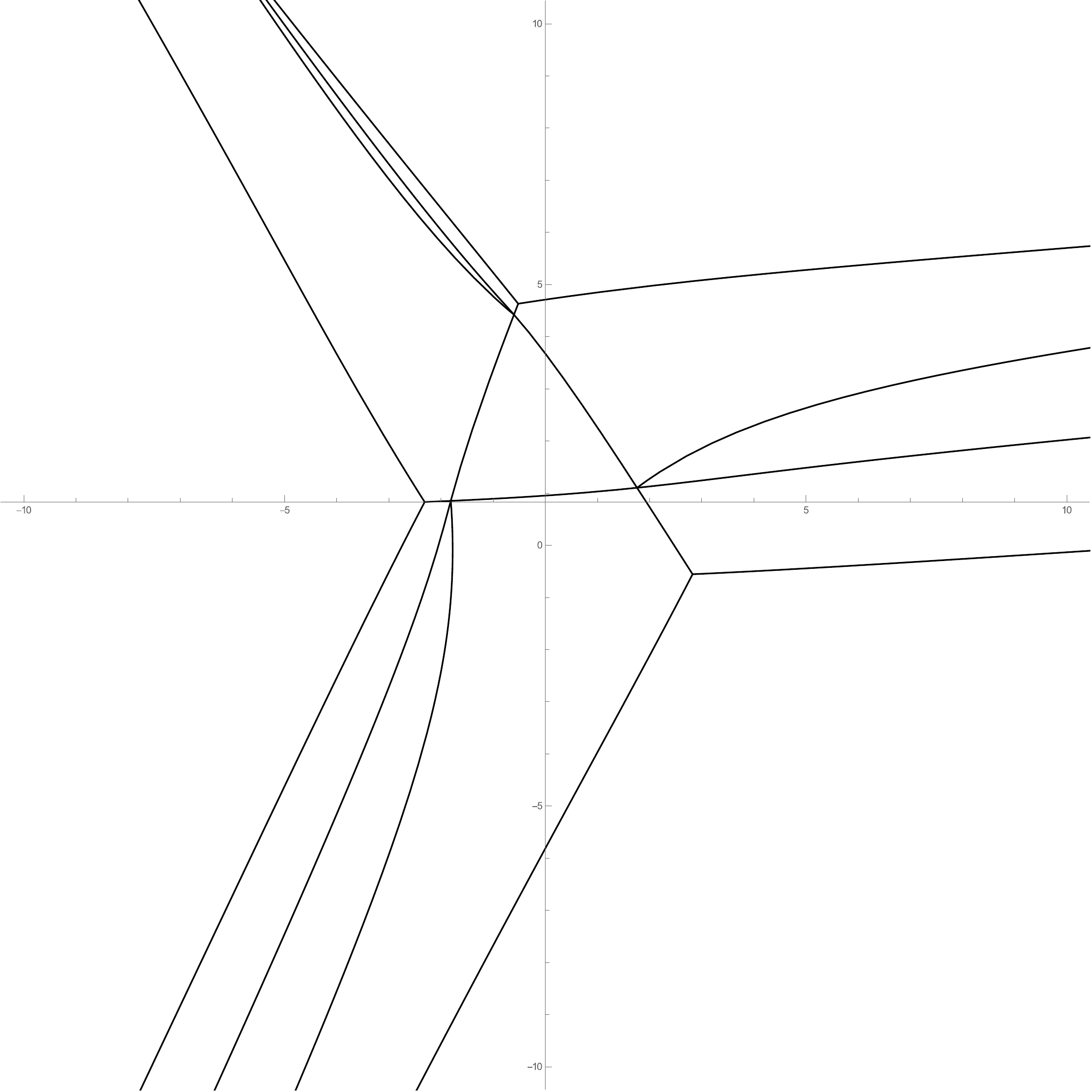}
        \caption{$\vartheta\approx 1.69$}
      \label{fig:14curve-12}
    \end{subfigure}
\\[+1.em]
    \begin{subfigure}[t]{.2\textwidth}
        \centering
        \includegraphics[width=\linewidth]{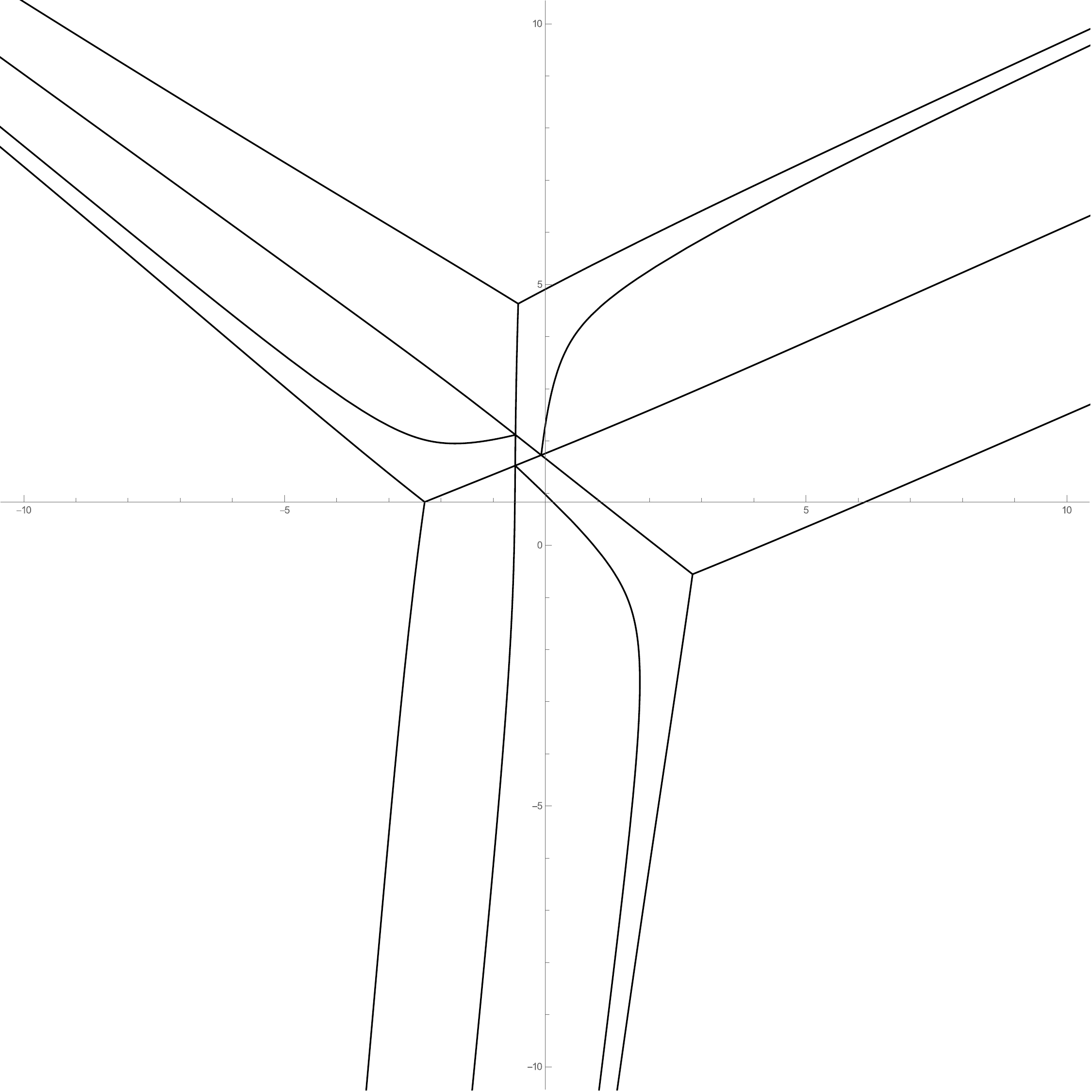}
        \caption{$\vartheta \approx 2.2$}
      \label{fig:14curve-13}
    \end{subfigure}
     \hspace{0.5cm}                
    \begin{subfigure}[t]{.2\textwidth}
        \centering
        \includegraphics[width=\linewidth]{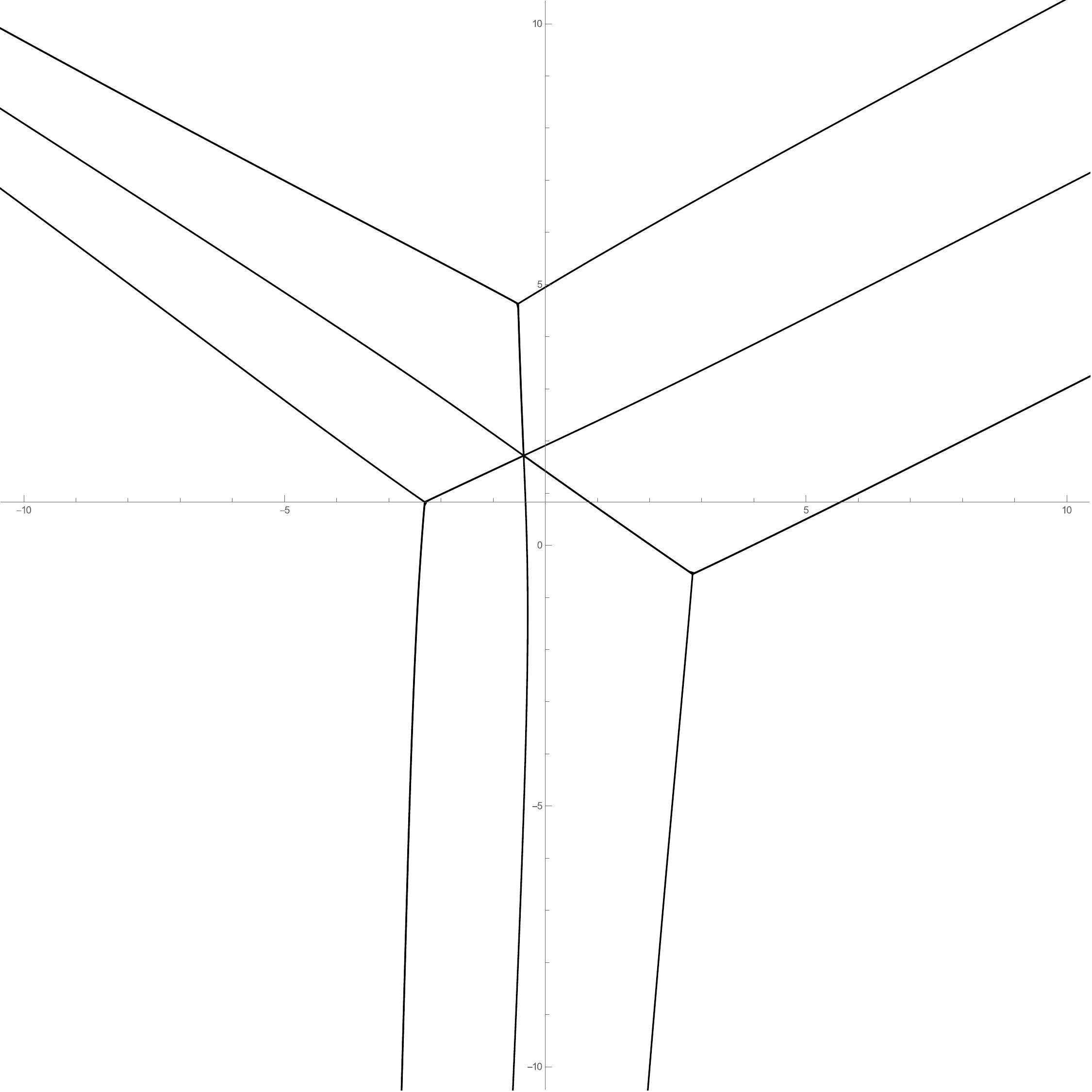}
        \caption{ ``Three-string web", $\vartheta \approx 2.287$}
      \label{fig:14curve-14}
    \end{subfigure}
    \hspace{0.5cm}             
    \begin{subfigure}[t]{.2\textwidth}
        \centering
        \includegraphics[width=\linewidth]{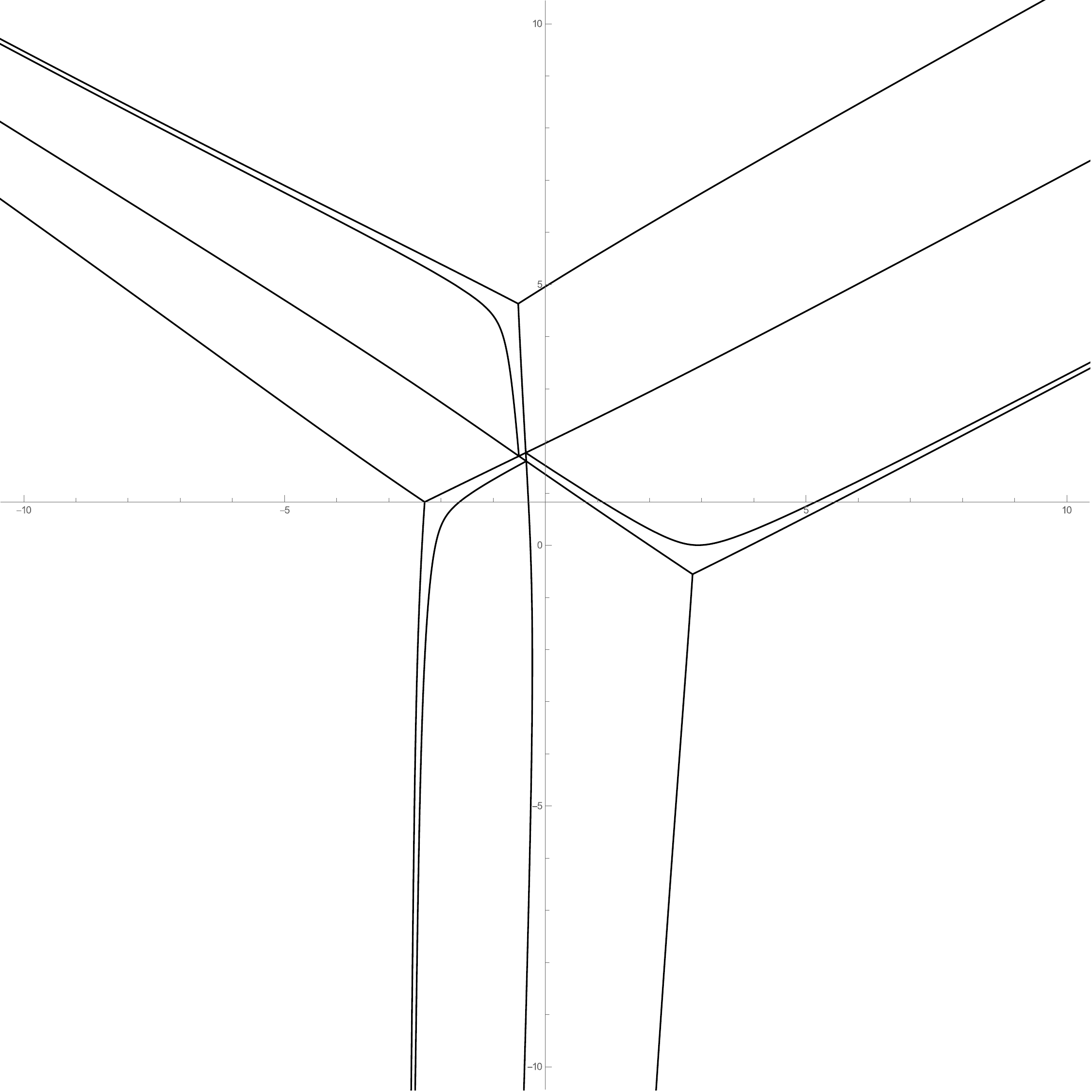}
        \caption{$\vartheta \approx 2.31$}
      \label{fig:14curve-15}
    \end{subfigure}
    \hspace{0.5cm}             
    \begin{subfigure}[t]{.2\textwidth}
        \centering
        \includegraphics[width=\linewidth]{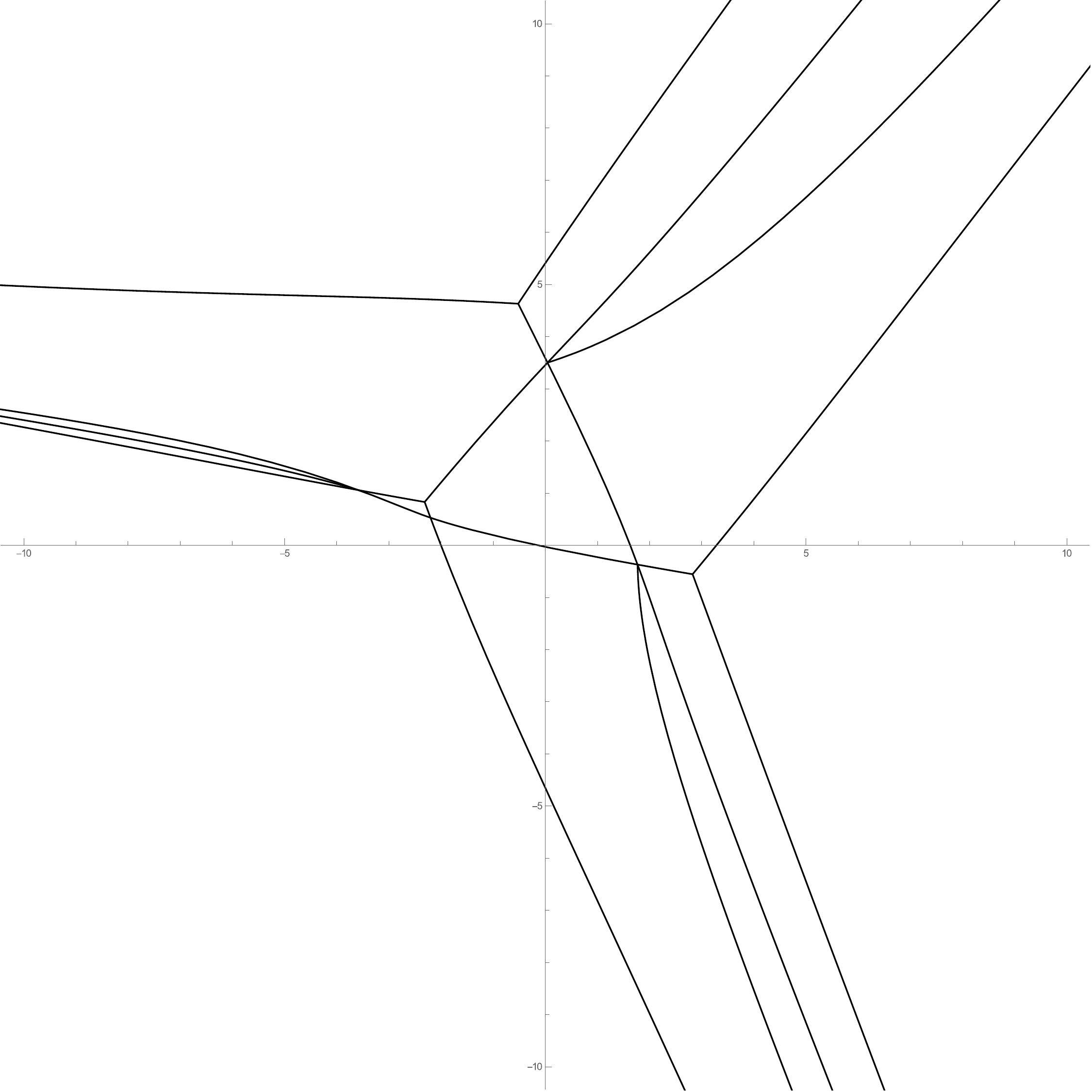}
        \caption{$\vartheta \approx 2.94$}
      \label{fig:14curve-17}
    \end{subfigure}
 \caption{Spectral networks at various $\vartheta$ for the $(1,4)$ spectral curve with $t \approx -1.4-1.4i$ and $m_\infty \approx -1.4-1.4i$. Exactly one degeneration occurs.}
  \label{fig:14curvefirsthalf}
\end{figure}

\vspace{3cm}

\newpage 
\vspace{2cm}
\begin{figure}[h]
    \centering
    \begin{subfigure}[t]{.20\textwidth}
        \centering
        \includegraphics[width=\linewidth]{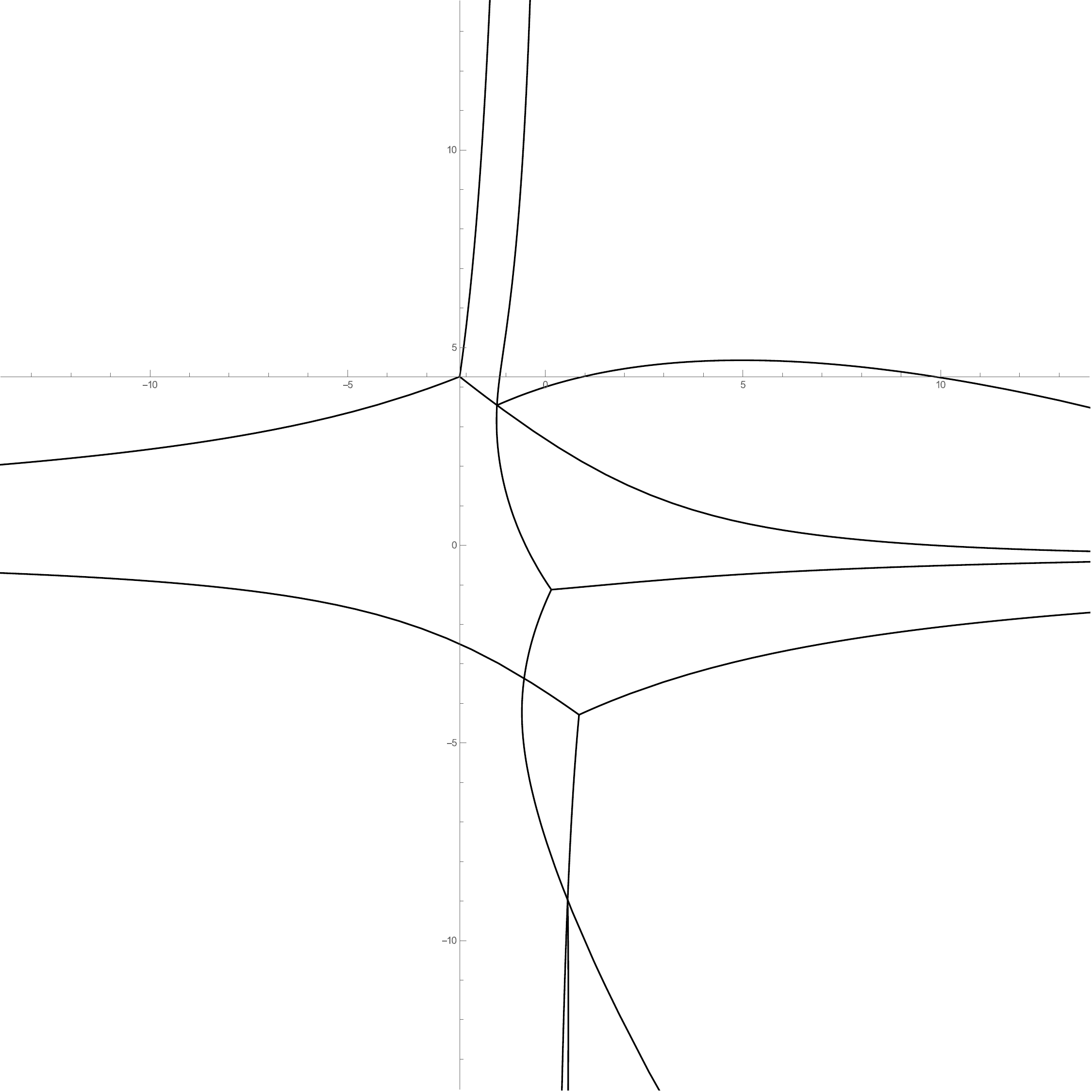}
           \caption{$\vartheta\approx0$}
      \label{fig:23curve-1}
    \end{subfigure}
    \hspace{0.5cm}              
    \begin{subfigure}[t]{.20\textwidth}
        \centering
        \includegraphics[width=\linewidth]{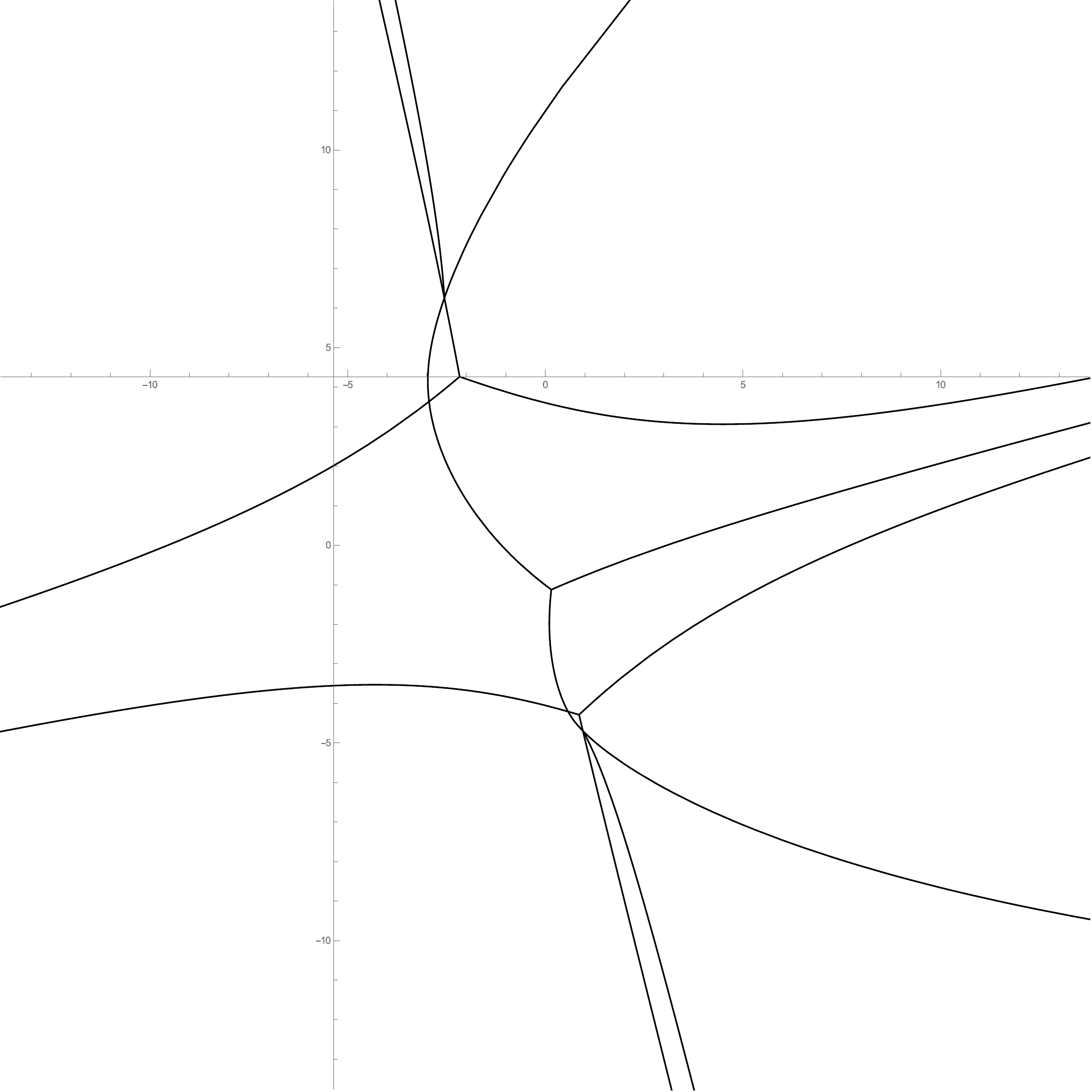}
        \caption{$\vartheta\approx0.49$}
      \label{fig:23curve-2}
    \end{subfigure}
    \hspace{0.5cm}              
    \begin{subfigure}[t]{.20\textwidth}
        \centering
        \includegraphics[width=\linewidth]{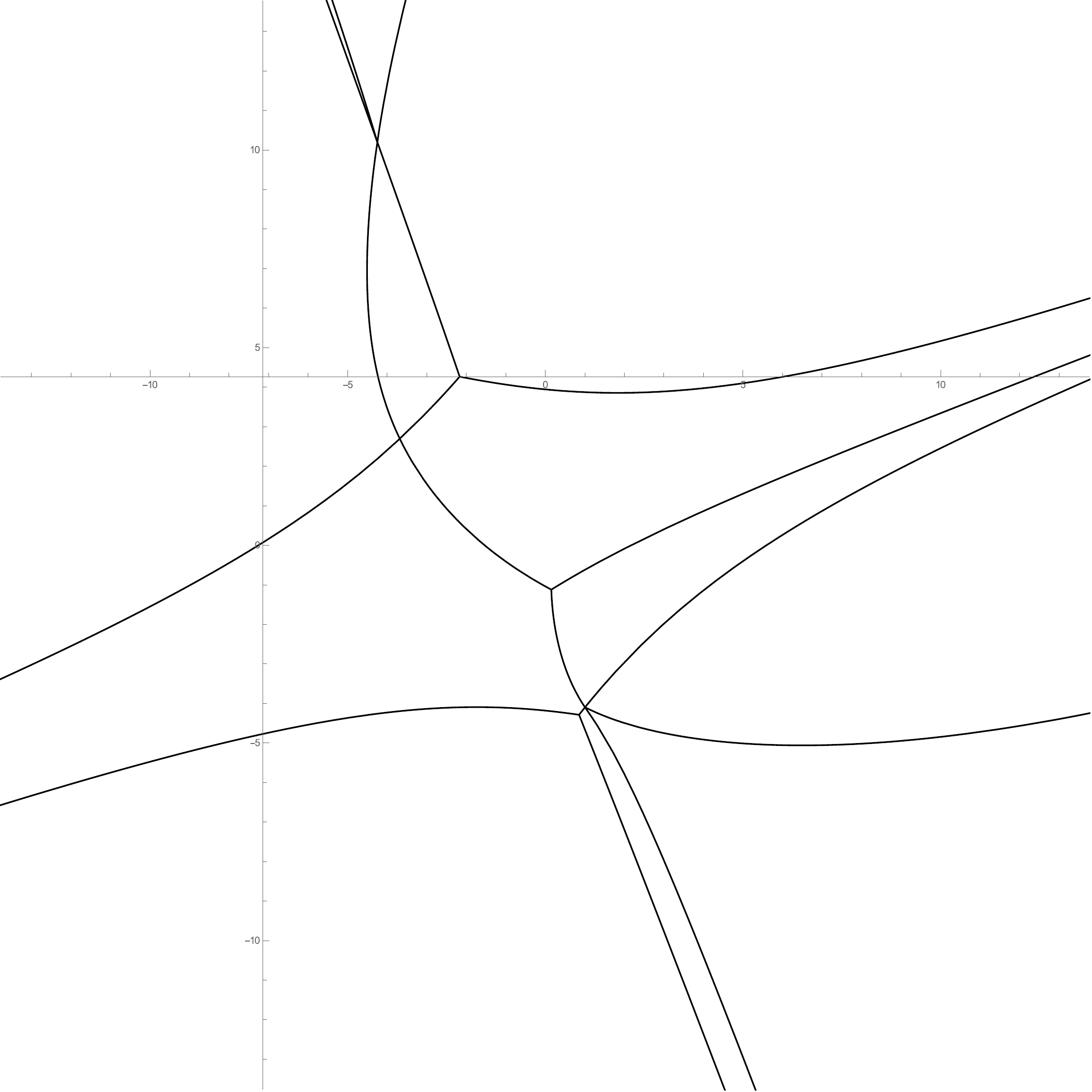}
        \caption{$\vartheta\approx0.71$}
      \label{fig:23curve-3}
    \end{subfigure}
    \\[+1.em]           
    \begin{subfigure}[t]{.20\textwidth}
        \centering
        \includegraphics[width=\linewidth]{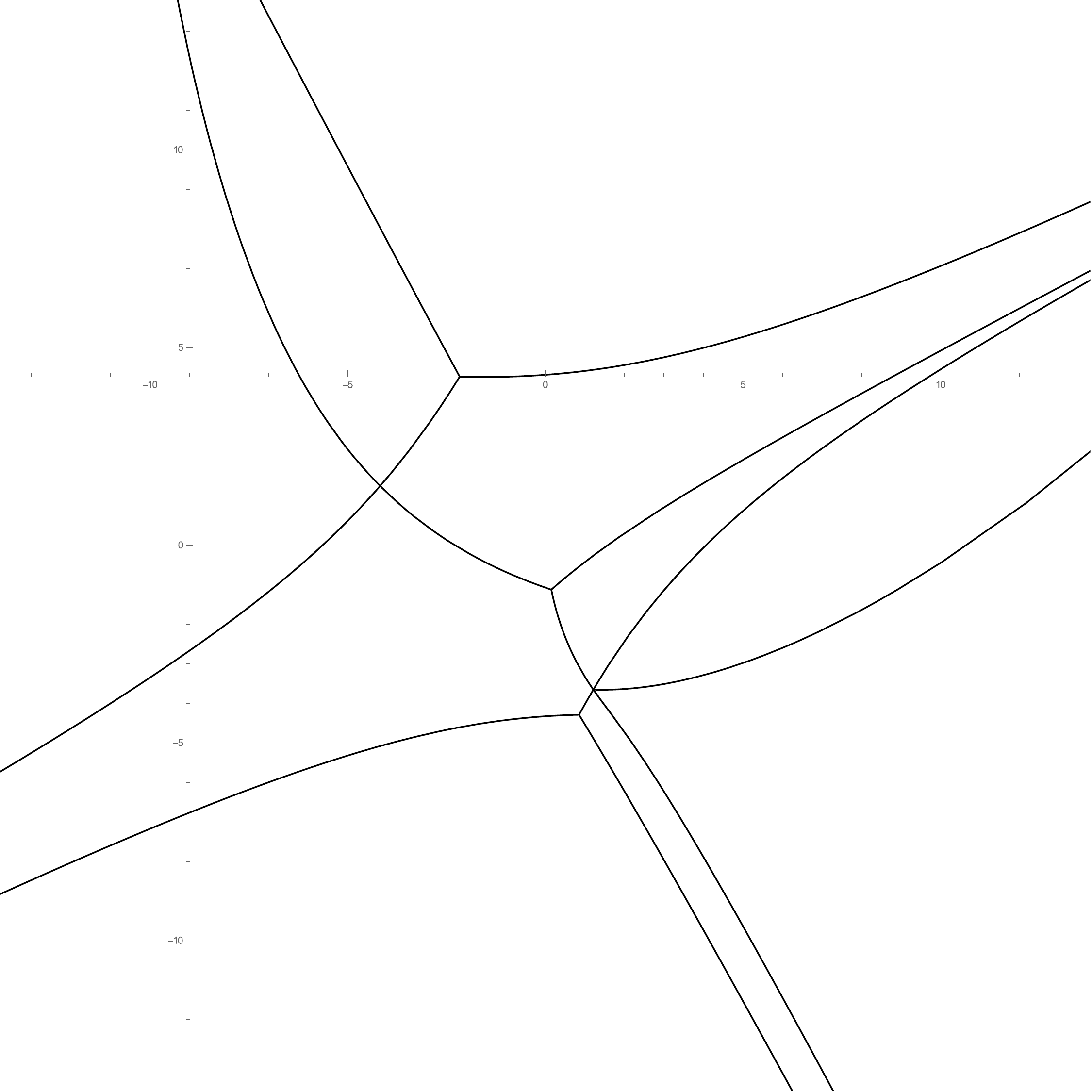}
        \caption{$\vartheta\approx0.96$}
      \label{fig:23curve-4}
    \end{subfigure}
\hspace{0.5cm}
    \begin{subfigure}[t]{.20\textwidth}
        \centering
        \includegraphics[width=\linewidth]{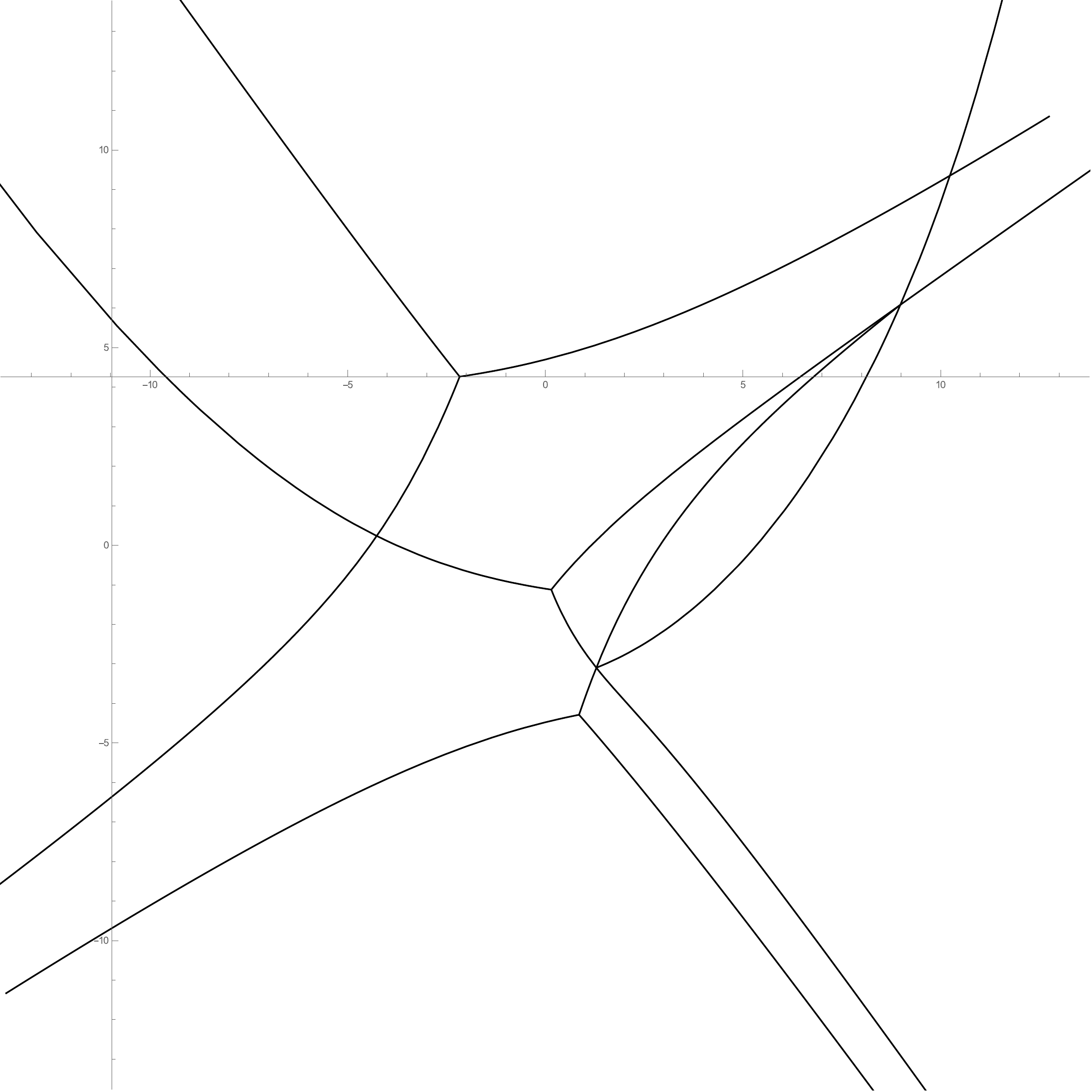}
        \caption{$\vartheta\approx 1.22$}
      \label{fig:23curve-5}
    \end{subfigure}
    \hspace{0.5cm}              
    \begin{subfigure}[t]{.20\textwidth}
        \centering
        \includegraphics[width=\linewidth]{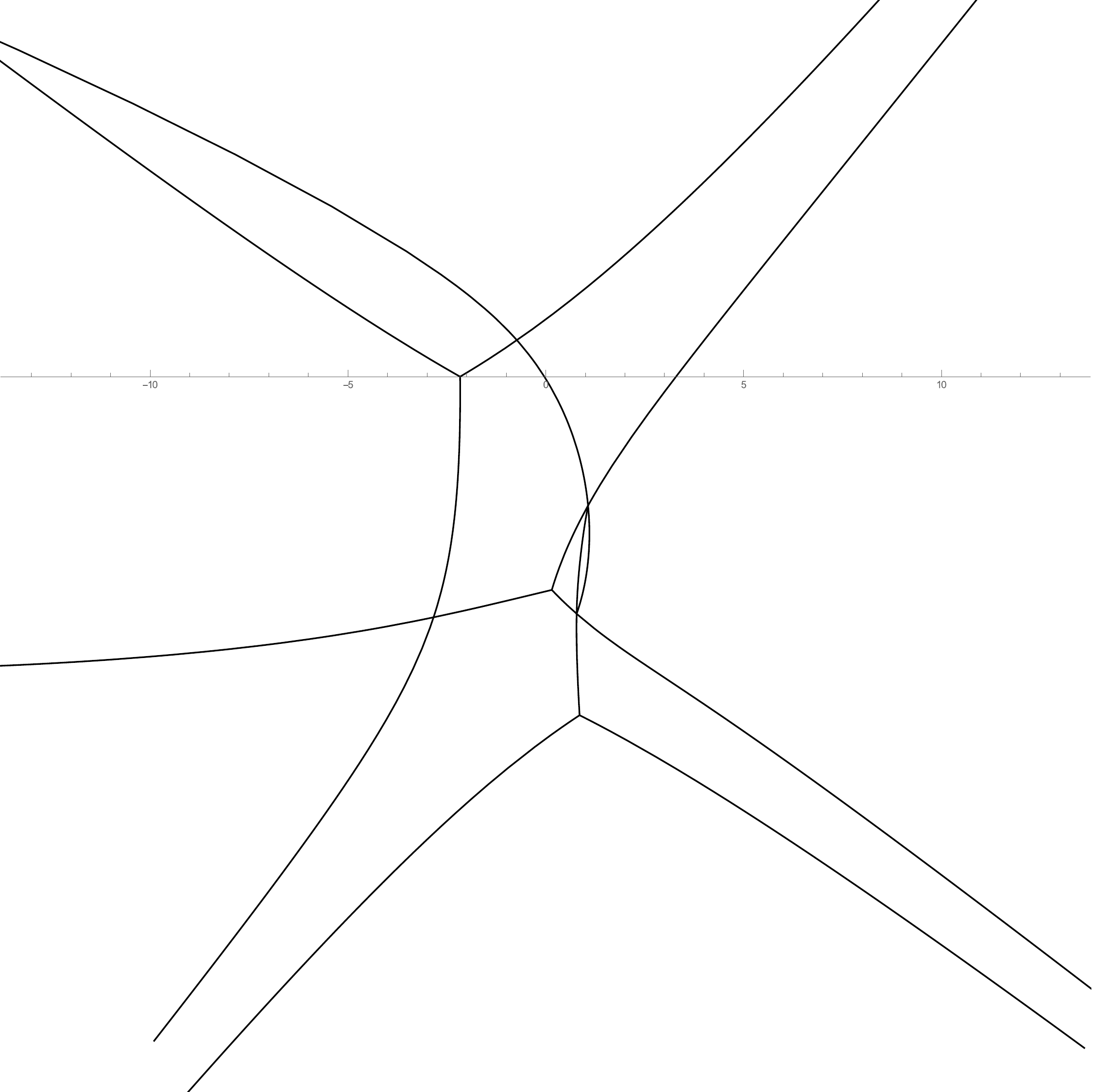}
        \caption{$\vartheta \approx 1.8$}
      \label{fig:23curve-6}
    \end{subfigure}
    \\[+1.em]  
    
  \begin{subfigure}[t]{.20\textwidth}
        \centering
        \includegraphics[width=\linewidth]{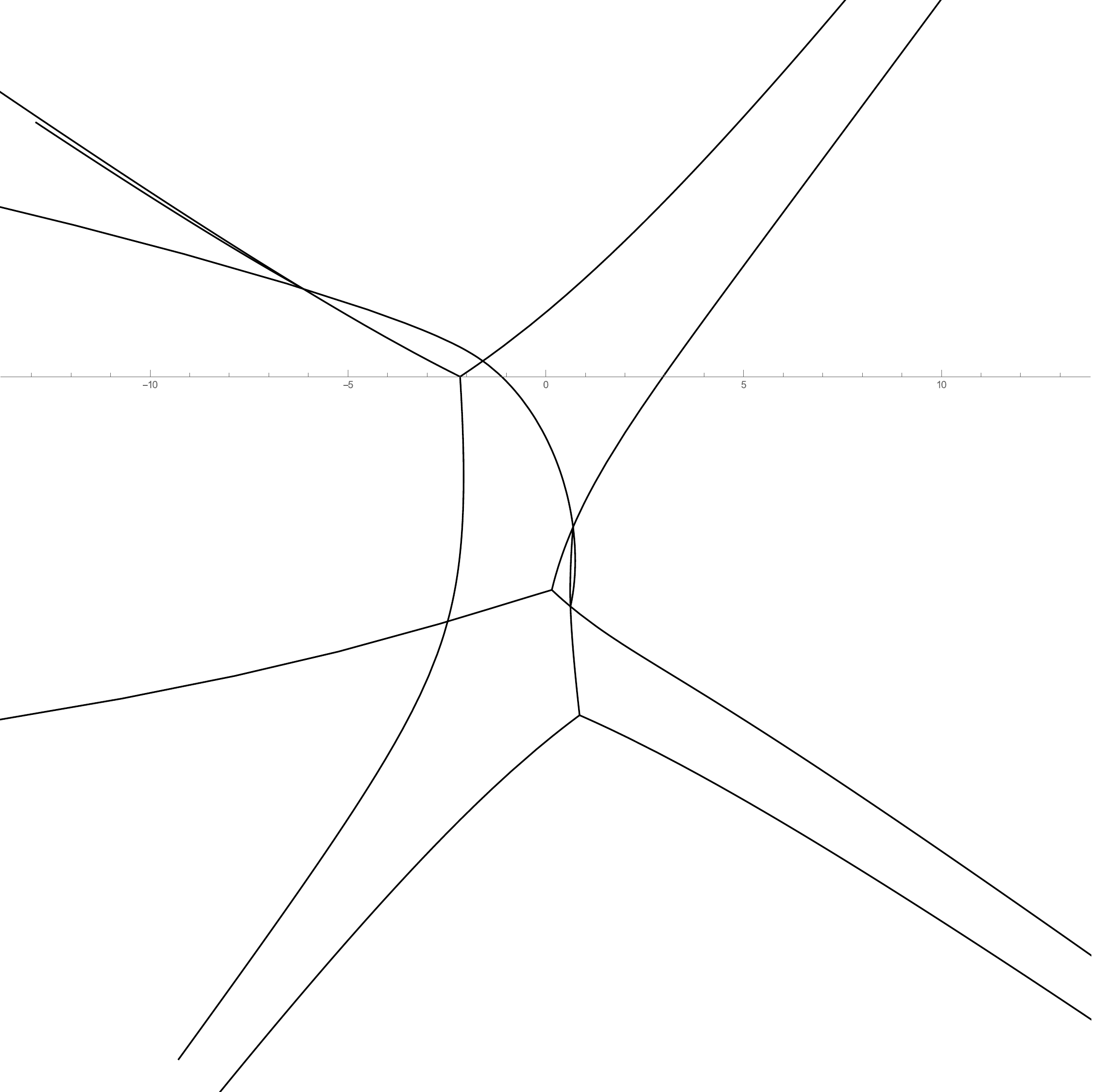}
           \caption{$\vartheta\approx 1.88$}
      \label{fig:23curve-7}
    \end{subfigure}
    \hspace{0.5cm}
    \begin{subfigure}[t]{.20\textwidth}
        \centering
        \includegraphics[width=\linewidth]{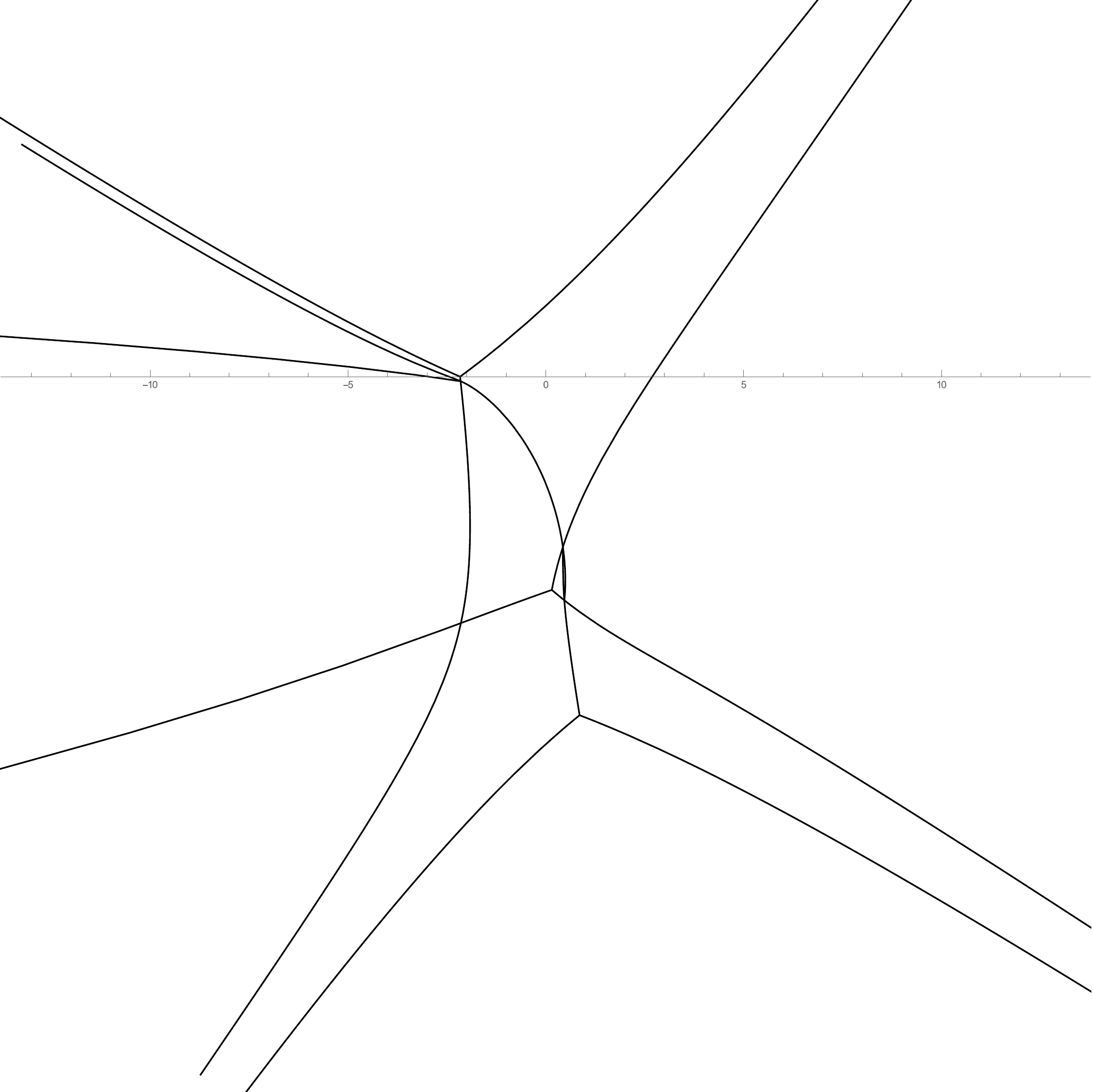}
        \caption{$\vartheta\approx1.95$}
      \label{fig:23curve-8}
    \end{subfigure}
\hspace{0.5cm}
        \begin{subfigure}[t]{.20\textwidth}
        \centering
        \includegraphics[width=\linewidth]{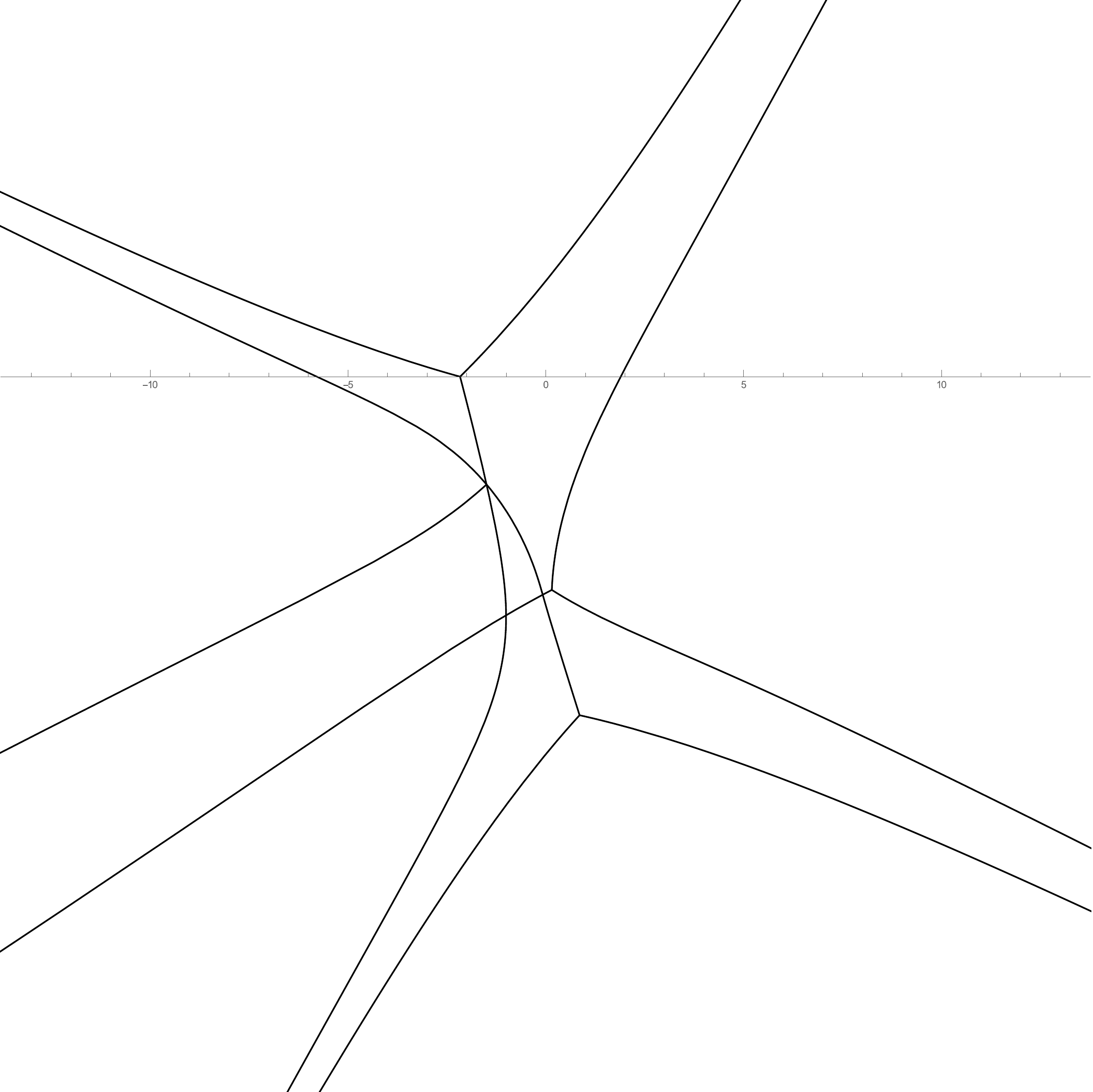}
        \caption{$\vartheta\approx2.17$}
      \label{fig:23curve-9}
    \end{subfigure}
\\[+1.em]           
    \begin{subfigure}[t]{.20\textwidth}
        \centering
        \includegraphics[width=\linewidth]{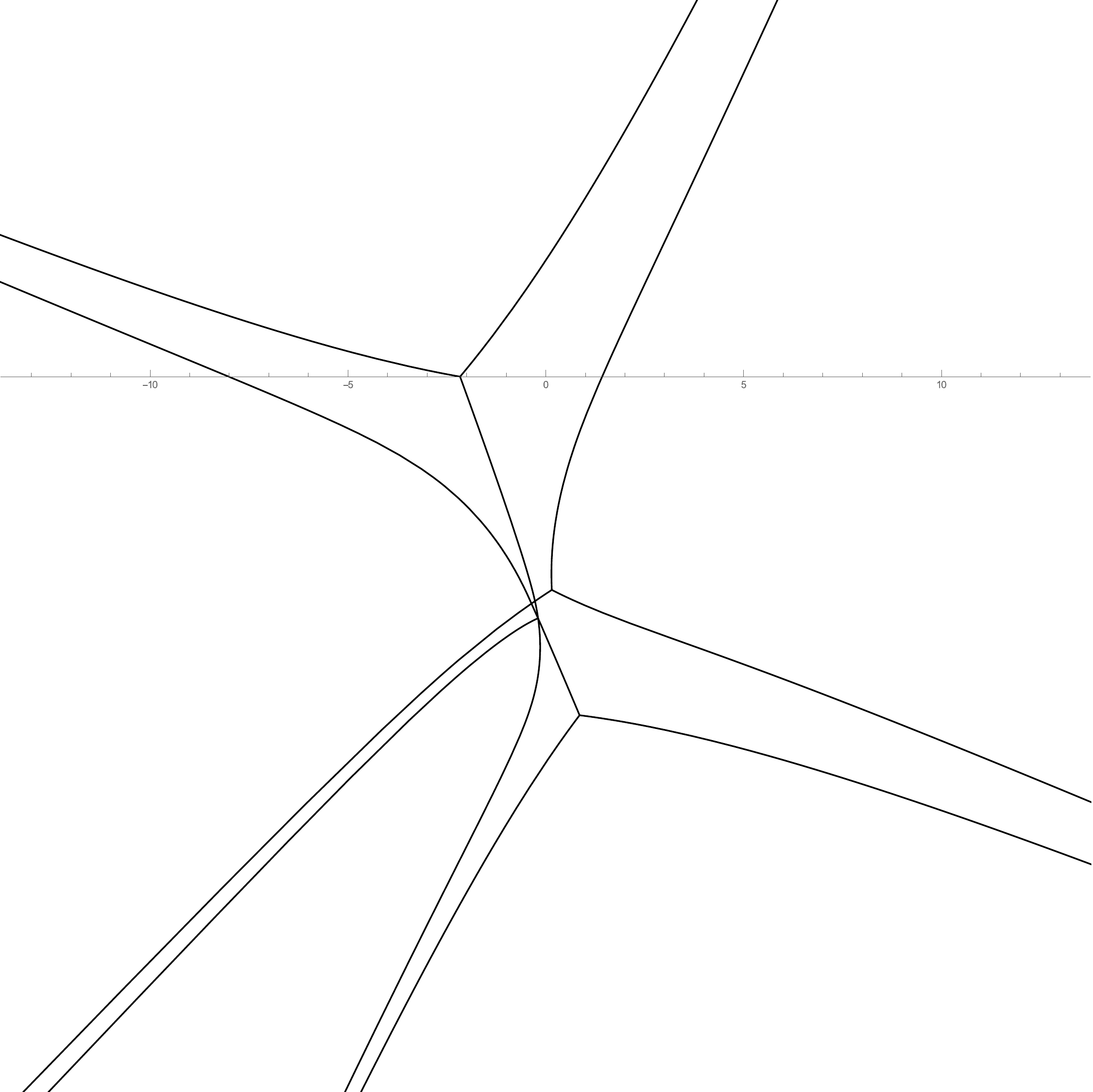}
        \caption{$\vartheta\approx2.31$}
      \label{fig:23curve-10}
    \end{subfigure}
    \hspace{0.5cm}              
    \begin{subfigure}[t]{.20\textwidth}
        \centering
        \includegraphics[width=\linewidth]{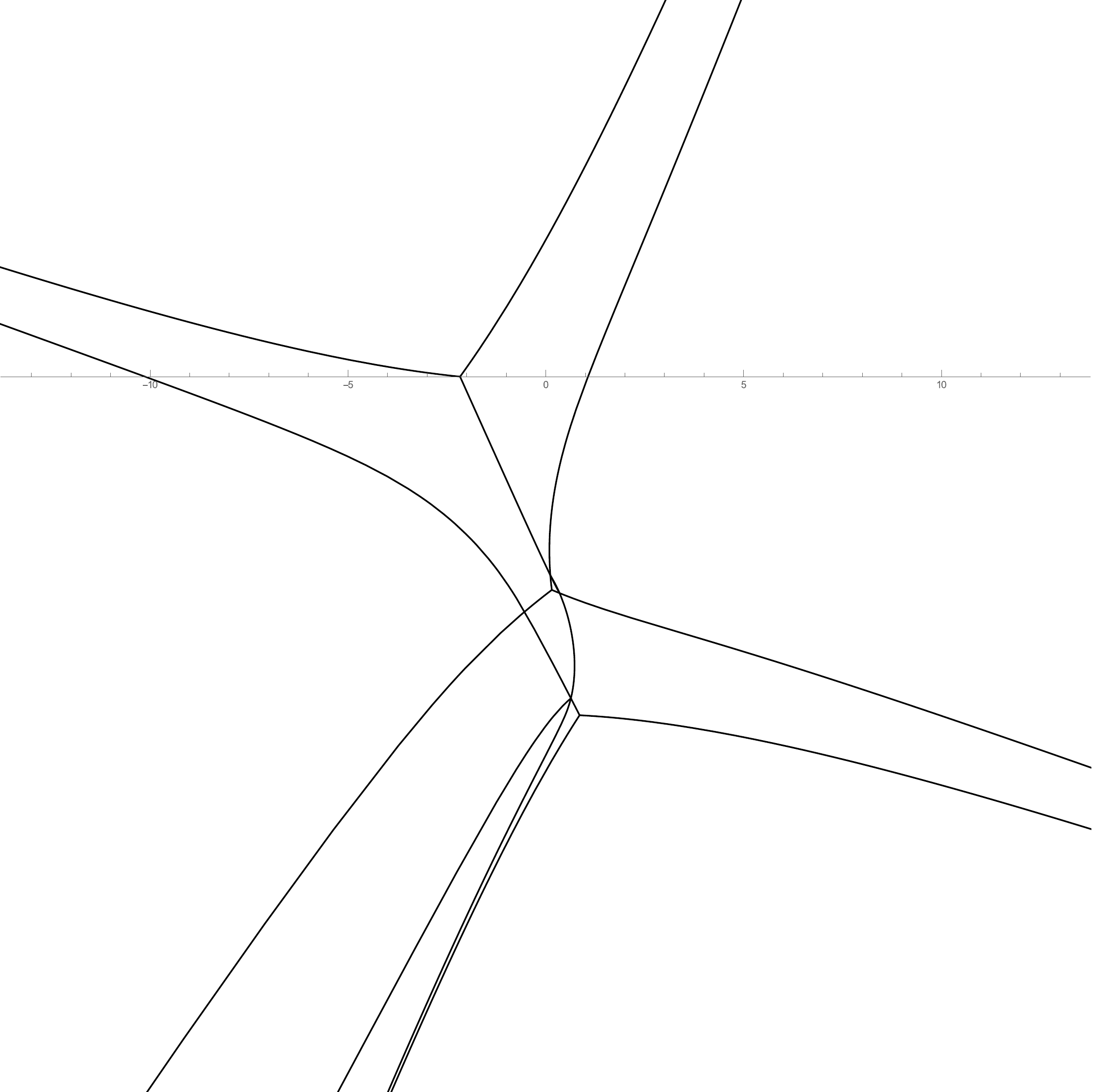}
        \caption{$\vartheta \approx 2.42$}
      \label{fig:23curve-11}
    \end{subfigure}
    \hspace{0.5cm}              
    \begin{subfigure}[t]{.20\textwidth}
        \centering
        \includegraphics[width=\linewidth]{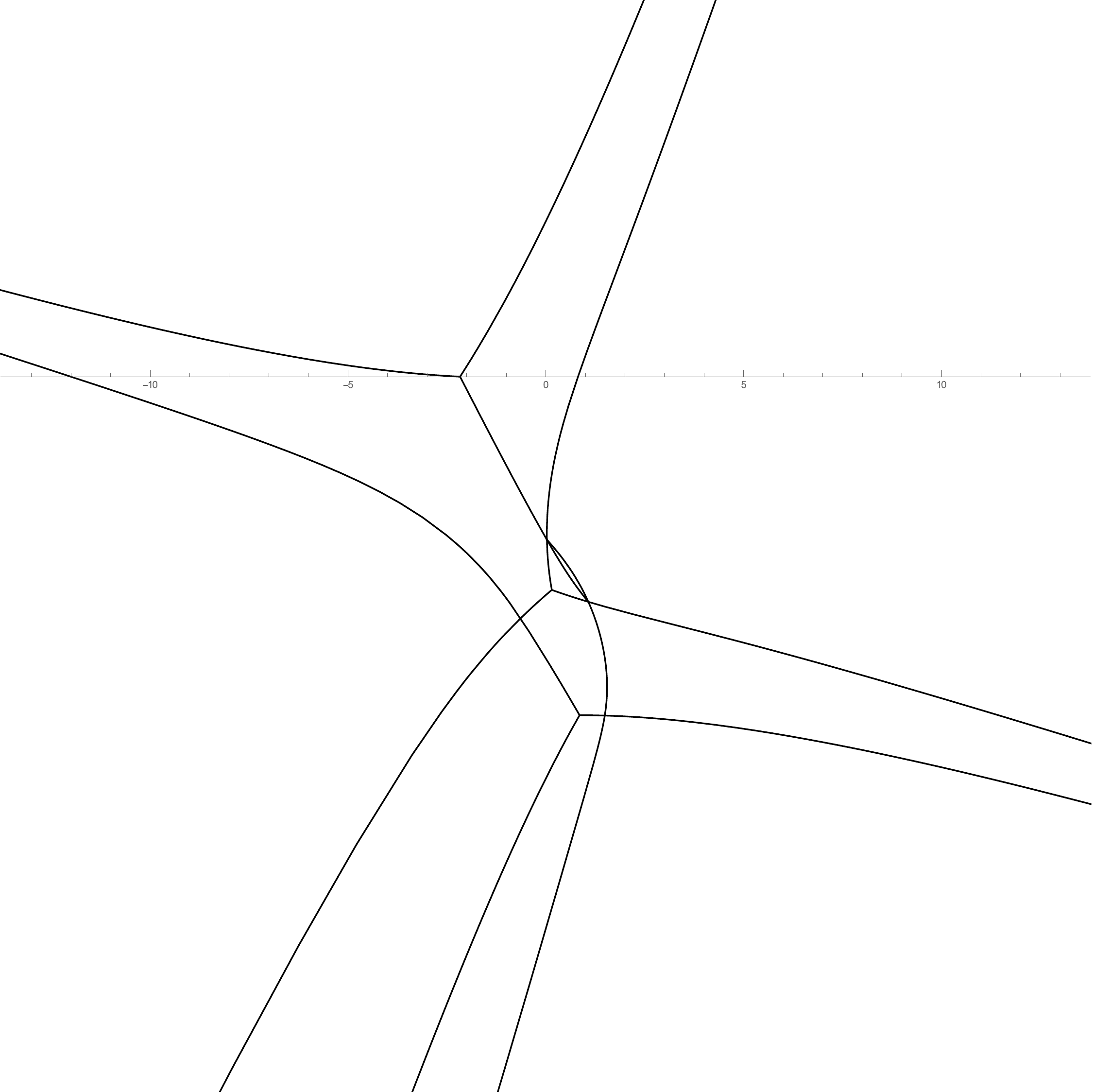}
        \caption{$\vartheta \approx 2.5$}
      \label{fig:23curve-12}
    \end{subfigure}
\\[+1.em]
    \begin{subfigure}[t]{.20\textwidth}
        \centering
        \includegraphics[width=\linewidth]{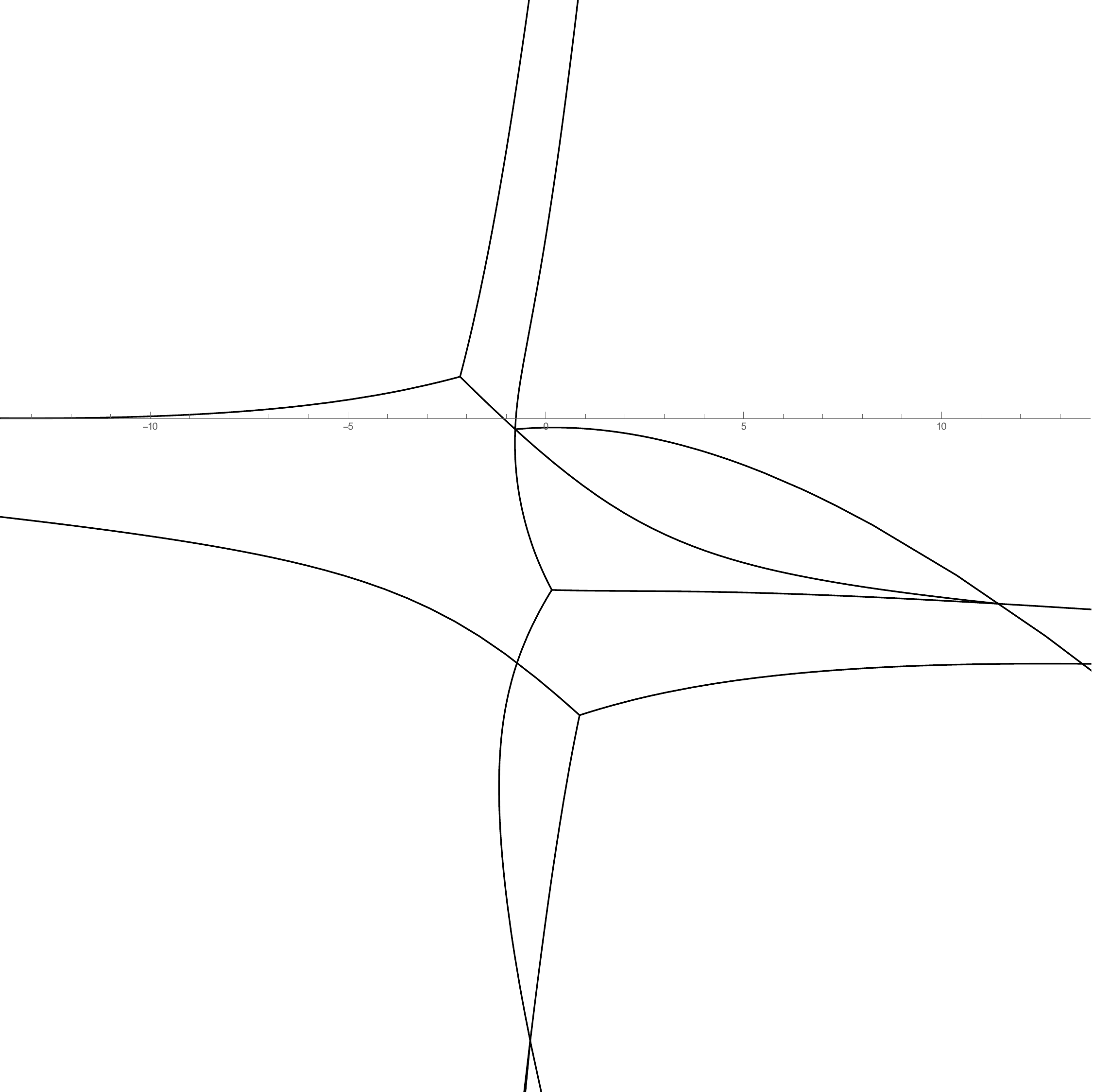}
        \caption{$\vartheta \approx 2.97$}
      \label{fig:23curve-13}
    \end{subfigure}
    \hspace{0.5cm}              
    \begin{subfigure}[t]{.20\textwidth}
        \centering
        \includegraphics[width=\linewidth]{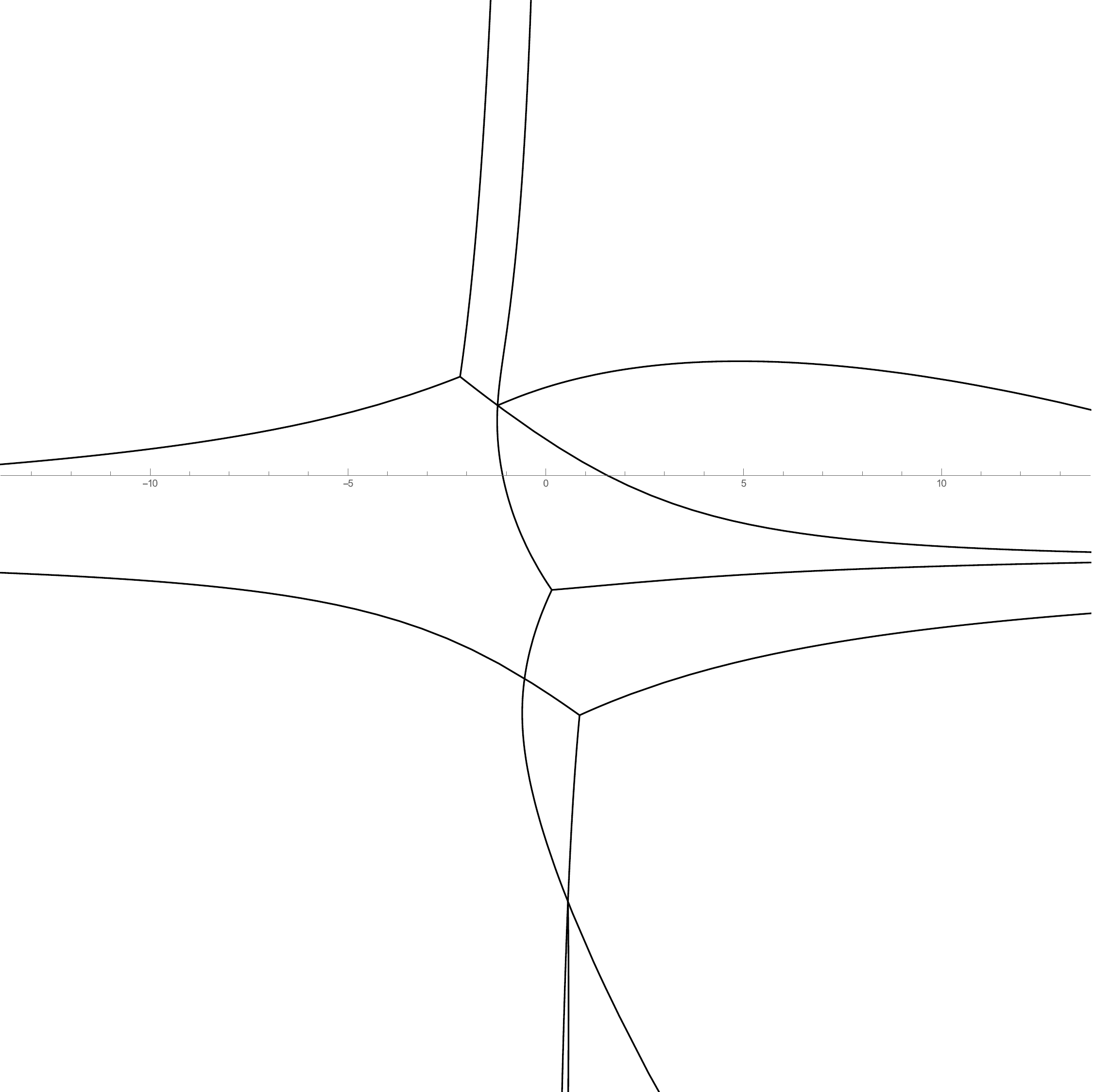}
        \caption{$\vartheta \approx 3.14$}
      \label{fig:23curve-14}
    \end{subfigure}
    \caption{Spectral networks at various $\vartheta$ for the $(2,3)$ spectral curve with $t \approx -1.44-1.1i$ and $m_\infty \approx -1.8-1.6i$.}
    \label{fig:23curvefirsthalf}
\end{figure}

\end{document}